%% file: arxiv.tex
\documentclass[11pt]{article}
\pdfoutput=1
\usepackage{yrz}

\usepackage[margin=1in]{geometry}

\usepackage{listings}

\usepackage{subfigure}

\usepackage{tikz}
\usetikzlibrary{math}
\usepackage{circuitikz}
\usetikzlibrary{arrows, arrows.meta, shapes, positioning}

\tikzstyle{layer0}=[draw, fill = red!50, circle, minimum size = 10pt,inner sep = 0pt]
\tikzstyle{layer1}=[draw, fill = blue!50, circle, minimum size = 10pt,inner sep = 0pt]
\tikzstyle{layer2}=[draw, fill = green!50, circle, minimum size = 10pt,inner sep = 0pt]
\tikzstyle{arrow}=[arrows={{Latex[scale=0.5]}-}, thick]
\tikzstyle{layer01}=[draw, fill = red!20, circle, minimum size = 10pt,inner sep = 0pt]
\tikzstyle{layer11}=[draw, fill = blue!20, circle, minimum size = 10pt,inner sep = 0pt]
\tikzstyle{layer02}=[draw, fill = red!50, circle, minimum size = 10pt,inner sep = 0pt]
\tikzstyle{layer12}=[draw, fill = blue!50, circle, minimum size = 10pt,inner sep = 0pt]
\tikzstyle{layer03}=[draw, fill = red!80, circle, minimum size = 10pt,inner sep = 0pt]
\tikzstyle{layer13}=[draw, fill = blue!80, circle, minimum size = 10pt,inner sep = 0pt]
\def\tikzdist{5.5}

\usepackage{upgreek}
\usepackage[only,llbracket,rrbracket]{stmaryrd}

\def\logit{\mathrm{logit}}
\def\offdiag{\mathrm{off}\text{-}\mathrm{diag}}
\def\sgn{\mathrm{sgn}}
\def\Dir{\mathrm{Dir}}
\def\PG{\mathrm{PG}}

\def\perm{\mathrm{Perm}}
\def\marg{\mathrm{Marg}}
\def\shuff{\mathrm{Shuff}}
\def\law{\mathrm{Law}}

\usepackage{xurl}

\title{Bayesian Deep Generative Models for Multiplex Networks with Multiscale Overlapping Clusters}
\author{
Yuren Zhou
\thanks{Department of Statistical Science, Duke University, Durham, NC, USA.}
\and
Yuqi Gu
\thanks{Department of Statistics, Columbia University, New York, NY, USA.}
\and
David B. Dunson
\thanks{Department of Statistical Science \& Mathematics, Duke University, Durham, NC, USA.}
}

\begin{document}

\maketitle

\begin{abstract}
\input{abstract}
\end{abstract}


\input{main_paper}

\bibliography{ref}
\bibliographystyle{apalike}

\pagebreak
\appendix

\tableofcontents

\input{supp_paper}

\end{document}

%% file: abstract.tex
Our interest is in multiplex network data with multiple network samples observed across the same set of nodes.
Examples originate from a variety of fields, including brain connectivity, international trade networks, and social networks, among others. Our goal is to infer a hierarchical structure of the nodes at a population level, while performing multi-resolution clustering of the individual replicates. To accomplish this, we propose a Bayesian hierarchical model, provide theoretical support in terms of identifiability and posterior consistency, and design efficient methods for posterior computation. We provide novel technical tools for proving model identifiability, which are of independent interest. Our proposed methodology is demonstrated through numerical simulation and an application to brain connectome data.

%% file: keywords.tex
Bayesian inference; Deep generative model; Identifiability; Community detection; Overlapping clustering.

%% file: main_paper.tex
\section{Introduction}\label{sec:intro}

In recent decades, the study of network data has been a popular topic in the statistics literature. An undirected network with $p$ nodes is commonly represented by a symmetric binary matrix $\Xb$, where $X_{i,j} = 1$ signifies the existence of an edge between the nodes $i$ and $j$, and $X_{i,j} = 0$ indicates the absence of an edge. In this paper, we consider multiplex network data in which each sample collects a different undirected network that shares the same nodes. Motivating applications include connectivity networks between brain regions for different individuals \citep{park2013structural}, trade networks between countries for different agricultural products \citep{macdonald2022latent}, passing networks between soccer players in different matches \citep{han2018multiresolution}, contact networks between students collected at different times \citep{durante2016locally}, etc.

More specifically, our overarching goal is to develop statistical methods that can reliably uncover a hierarchical structure of the nodes at a population level and perform multi-resolution overlapping clustering of individual samples based on the idiosyncrasies of their networks. To enable reliable model interpretation, it is crucial that the latent model structure be identifiable and that our methods provide a realistic characterization of uncertainty. We address the former problem with appropriate identifiability theory, while tackling the latter through a Bayesian inferential framework.

There is a vast statistical literature on modeling a single network. One of the most common goals in single-network analysis is community detection, which groups nodes into communities such that nodes within a community are more densely connected. Traditional methods include spectral clustering \citep{shi2000normalized, meilua2001random, ng2001spectral}, stochastic block models \citep[SBMs,][]{holland1983stochastic, nowicki2001estimation, kemp2006learning, abbe2017community, van2018bayesian}, 
extensions of SBMs to allow degree heterogeneity of nodes \citep{karrer2011stochastic, jin2015fast, peng2016bayesian} and many others.
There are also more flexible extensions of SBMs that allow mixed memberships in multiple communities
\citep{airoldi2008mixed, jin2023mixed}, as well as network models that avoid clustering of nodes entirely including exponential random graphs \citep{frank1986markov}, latent space models \citep{hoff2002latent, handcock2007model}, random dot product graphs \citep{young2007random}, etc.

Multiplex network data arise in diverse applications spanning across neuroscience, ecology, sociology, and sports, and have been referred to by various alternative names in the literature including multiple network data, populations of networks, replicated networks, and ensembled networks.
To model multiplex networks, researchers have proposed various generative models \citep{durante2017nonparametric, nielsen2018multiple, wang2019common}.
These approaches typically define a connection probability matrix for each network sample, take an element-wise logit transform, and then define a hierarchical matrix factorization model.
To find the community structure of the nodes at the population level, community detection methods have been developed for multiplex networks \citep{han2015consistent, mukherjee2017clustering, arroyo2021inference, jing2021community, duan2023bayesian, josephs2023nested, mantziou2024bayesian}.
To our knowledge, we propose the first method for simultaneously inferring the community structure of network nodes, obtaining lower-dimensional network representations, and grouping individual networks into multiscale clusters.
Our model features a multi-latent-layer structure that performs these tasks at multiple resolutions.

\input{tikz/tikz_figure_intro}

We illustrate the population-level model using a graphical model representation in Figure \ref{fig:tikz_A_intro} and give their sample realizations in Figure \ref{fig:tikz_X_intro}.
The bottom layer in Figure \ref{fig:tikz_A_intro} corresponds to shared nodes of the observed adjacency matrices.
Between any two consecutive layers in the population model (e.g., between layers 0 and 1, or between layers 1 and 2), the nodes in the upper/deeper layer are interpreted as the communities for the nodes in the lower/shallower layer.
We allow overlapping community memberships, meaning that for each node in the lower layer, there could be multiple nodes in the upper layer having directed edges pointing to it.
As illustrated in Figure \ref{fig:tikz_X_intro}, each individual network has a distinct latent adjacency matrix at each layer, which can be viewed as a representation of its observed adjacency matrix at a lower resolution.
The entries of these latent adjacency matrices can be used for grouping individuals into clusters that are easily interpretable.
We define a multilayer directed graphical model (i.e., a Bayesian network) with latent adjacency matrices at each resolution.
The distribution of an individual's adjacency matrix at one layer is defined conditional on the adjacency matrix in the above layer.
It is worth emphasizing the distinction between the observed edges in the multiplex networks and the directed edges in our directed graphical model or Bayesian network:
The former determine entries of the observed adjacency matrices, while the latter indicate statistical dependence of variables and are unobserved.

We develop nontrivial identifiability theory for the proposed model, which is key to the interpretability and reproducibility of the results.
Our proof introduces a new technique to show generic identifiability, whose definition will be given later in Section \ref{sec:theo}.
Most existing proofs \citep{allman2009identifiability, gu2020partial, gu2021bayesian} for generic identifiability in related models use properties of algebraic varieties in algebraic geometry, and we have substantially generalized the approach by incorporating the concept of holomorphic functions from complex analysis.
This new technique may be of independent theoretical interest for studying the identifiability of other models.
Based on the novel identifiability result, we further use Doob's and Schwartz's theories \citep{doob1949application, schwartz1965bayes} to prove posterior consistency of the model parameters in Bayesian inference.

The rest of our paper is organized as follows. We describe our model in Section \ref{sec:model} and provide a rich set of theoretical results on its identifiability and posterior consistency in Section \ref{sec:theo}.
Our new technique for proving generic identifiability is also briefly introduced in Section \ref{sec:theo}.
In Section \ref{sec:post}, we specify the prior distribution and discuss its theoretical properties. Details of the Gibbs sampler for posterior computation and scalable versions are deferred to the supplementary materials. Section \ref{sec:appl} applies the model to brain connectome data. We discuss the implications of our model and future directions in Section \ref{sec:dis}.
Additional theoretical developments, comprehensive simulation studies, and extended data analyses are also provided in the supplementary materials.

\paragraph{Notations.}
We use bold capital letters (e.g. $\Xb$) to denote matrices and tensors, bold letters (e.g. $\xb$) to denote vectors, and nonbold letters (e.g. $X, x$) to denote scalars. For a square matrix $\Mb$, we let $\diag(\Mb)$ denote the diagonal matrix obtained by setting all off-diagonal entries of $\Mb$ to zero, and let $\offdiag(\Mb) := \Mb - \diag(\Mb)$ denote the matrix obtained by setting all diagonal entries of $\Mb$ to zero. For integer valued variables and integers $k_1, k_2 \in \ZZ$, we let $[k_1, k_2]$ denote the collection of integers $\{k_1, k_1 + 1, \ldots, k_2\}$.
For $k \in \ZZ$, we abbreviate $[k] := [1, k]$ and let $\langle k \rangle := \{(i, j):~ 1 \leq i < j \leq k\}$ denote the collection of all $\frac{k (k - 1)}{2}$ distinct pairs of integers in $[k]$.

\section{Deep Generative Model for Multiplex Networks}\label{sec:model}

\subsection{Two-layer Model}\label{ssec:model_2}

In this subsection, we first present the intuition behind the hierarchical structure of our model by introducing its two-layer version. Its general form with $K + 1$ layers will be formally introduced in Section \ref{ssec:model_k}.

We start with the community detection problem of a single network with $p$ nodes, where an adjacency matrix $\Xb \in \RR^{p \times p}$ is observed.
Assuming that the $p$ nodes each belong to one of the $\tilde{p}$ communities, we let the binary matrix $\Ab \in \RR^{p \times \tilde{p}}$ denote the community membership of the $p$ nodes, so that node $i$ belongs to the community $s$ if and only if $A_{i, s} = 1$.
Following traditional community detection methods for a single network, we can model the probability of nodes $i, j$ being adjacent as
\begin{equation}\label{eq:sbm}
\PP(X_{i, j} = 1 | \Ab, C, \gamma)
=
\frac{\exp\left(
C + \gamma \sum_{s = 1}^{\tilde{p}} A_{i, s} A_{j, s}
\right)}{1 + \exp\left(
C + \gamma \sum_{s = 1}^{\tilde{p}} A_{i, s} A_{j, s}
\right)}
,
\end{equation}
where $C, \gamma \in \RR$ are parameters.
The above formulation is a reparameterization of the basic stochastic block model \citep{holland1983stochastic}, where the probability that two nodes of the same community are adjacent is $\frac{\exp(C + \gamma)}{1 + \exp(C + \gamma)}$, and the probability that two nodes of different communities are adjacent is $\frac{\exp(C)}{1 + \exp(C)}$.
A common assumption in SBMs is that two nodes in the same community are more likely to be adjacent than nodes from different communities, that is, $\gamma > 0$.

For multiplex networks with $N$ observed adjacency matrices $\Xb^{(1:N)} := \{\Xb^{(n)}\}_{n = 1}^N$, the above model can be easily extended if we assume the $N$ networks share the common community memberships $\Ab$ and that the networks are conditionally independent given the community memberships:
\begin{equation}\label{eq-naive2layer}
\PP(\Xb^{(1:N)} | \Ab, C, \gamma)
=
\prod_{n = 1}^N \PP(\Xb^{(n)} | \Ab, C, \gamma)
=
\prod_{n = 1}^N \prod_{1 \le i < j \le p}
\frac{\exp\left( X_{i, j}^{(n)} \left(
C + \gamma \sum_{s = 1}^{\tilde{p}} A_{i, s} A_{j, s}
\right) \right)}{1 + \exp\left(
C + \gamma \sum_{s = 1}^{\tilde{p}} A_{i, s} A_{j, s}
\right)}
.
\end{equation}
However, this model does not account for heterogeneity among the individual networks.

We leverage information from multiplex networks to improve flexibility over SBMs while characterizing differences across the networks. With this in mind, we introduce \emph{latent adjacency matrices} for each network, which characterize connections between communities: $\tilde\Xb^{(1:N)} := \{\tilde\Xb^{(n)}\}_{n = 1}^N$.
Each $\tilde\Xb^{(n)} = (\tilde{X}^{(n)}_{s,t})$ is a binary $\tilde p \times \tilde p$ matrix, where $\tilde{X}^{(n)}_{s,t} = 1$ or $0$ indicates whether the latent community $s$ and the latent community $t$ are adjacent.
In \eqref{eq:sbm}, the basic SBM assumes that the nodes are more likely to be adjacent if they belong to the same community. We additionally assume that nodes in the $n$th individual network are more likely to be adjacent if they belong to adjacent communities in $\tilde\Xb^{(n)}$, controlled by the parameter $\delta > 0$. For individual $n$ with adjacency matrix $\tilde\Xb^{(n)}$, we let the probability of two nodes $i, j$ being adjacent be
\begin{equation}\label{eq-2layer2}
\PP(X_{i, j}^{(n)} = 1 | \Ab, C, \gamma, \delta, \tilde\Xb^{(n)})
=
\frac{\exp\left(
C + \gamma \sum_{s = 1}^{\tilde{p}} A_{i, s} A_{j, s} + \delta \sum_{1\leq s < t \leq \tilde{p}}  A_{i, s} A_{j, t} \tilde{X}_{s, t}^{(n)}
\right)}{1 + \exp\left(
C + \gamma \sum_{s = 1}^{\tilde{p}} A_{i, s} A_{j, s} + \delta \sum_{1\leq s < t \leq \tilde{p}} A_{i, s} A_{j, t} \tilde{X}_{s, t}^{(n)}
\right)}
.
\end{equation}
When two nodes are from different but adjacent communities, we have $\sum_{s = 1}^{\tilde{p}} A_{i, s} A_{j, s} = 0$ but $ \sum_{1\leq s < t \leq \tilde{p}} A_{i, s} A_{j, t} \tilde{X}_{s, t}^{(n)} = 1$, so the adjacency probability of these two nodes is $\frac{\exp(C + \delta)}{1 + \exp(C + \delta)}$. Compared to \eqref{eq-naive2layer}, the model \eqref{eq-2layer2} can accommodate heterogeneity between multiple networks through individual-specific latent adjacency matrices $\tilde\Xb^{(n)}$ for $n=1,\ldots,N$.
See Figure \ref{fig:tikz_layer2} for an illustration with $N = 3$ individual networks.

\input{tikz/tikz_figure_layer2}

Flexible SBMs typically allow the adjacency probabilities to vary according to the community memberships.
Similarly, we can allow $\gamma$ in \eqref{eq-2layer2} to be different for each community $s$, and $\delta$ to be different for each pair of communities $s, t$.
To this end, we define
$$
\bGamma
:=
\left(\begin{matrix}
\gamma_{1, 1} & \delta_{1, 2} & \ldots & \delta_{1, \tilde{p}} \\
\delta_{2, 1} & \gamma_{2, 2} & \ldots & \delta_{2, \tilde{p}} \\
\vdots & \vdots & \ddots & \vdots \\
\delta_{\tilde{p}, 1} & \delta_{\tilde{p}, 2} & \ldots & \gamma_{\tilde{p}, \tilde{p}}
\end{matrix}\right)
$$
as a symmetric $\tilde{p} \times \tilde{p}$ matrix with positive entries.
Rewriting \eqref{eq-2layer2} in matrix form, the probability of two nodes $i, j$ in individual network $n$ being adjacent is:
$$
\PP(X_{i, j}^{(n)} = 1 | \Ab, C, \bGamma)
=
\frac{\exp\left(
C + \ab_i^\top (\tilde\Xb^{(n)} * \bGamma) \ab_j
\right)}{1 + \exp\left(
C + \ab_i^\top (\tilde\Xb^{(n)} * \bGamma) \ab_j
\right)}
,
$$
where $\ab_i^\top \in \RR^{1 \times \tilde{p}}$ denotes the $i$th row of $\Ab$ and $*$ denotes the Hadamard product.
We avoid modeling self-adjacency of nodes and fix the diagonal entries of the adjacency matrix to all ones for notation simplicity. Letting $\bTheta := (C, \bGamma)$ denote the model parameters, the model can be written as
\begin{equation}\label{eq:half_model_2layer}
\PP(\Xb^{(1:N)} | \Ab, \bTheta, \tilde\Xb^{(1:N)})
=
\prod_{n = 1}^N \prod_{1 \le i < j \le p} \frac{\exp\left( X_{i, j}^{(n)} \left(
C + \ab_i^\top (\tilde\Xb^{(n)} * \bGamma) \ab_j
\right) \right)}{1 + \exp\left(
C + \ab_i^\top (\tilde\Xb^{(n)} * \bGamma) \ab_j
\right)}
.
\end{equation}

Restricting each node to have a single membership is often insufficiently flexible. For this reason, we consider an \emph{overlapping membership} model.
Overlapping membership models \citep{yang2013overlapping}
characterize nodes as belonging to a community or not instead of assigning weights characterizing proportional memberships as in mixed membership models \citep{airoldi2008mixed}. In our model, each row of $\Ab$ is a binary vector in $\{0,1\}^{\tilde p}$ with no restriction, such as all entries summing up to one. 
For two nodes $i, j$ of an individual network $n$, if \emph{either} the number of their shared communities is larger ($\sum_{s = 1}^{\tilde{p}} A_{i, s} A_{j, s}$ larger), \emph{or} the number of adjacent pairs of communities to which they separately belong is larger ($\sum_{1\leq s < t \leq \tilde{p}} A_{i, s} A_{j, t} \tilde{X}_{s, t}^{(n)}$ larger), then the probability that nodes $i, j$ are adjacent should also be higher. We make these effects additive, so that for each community $s$ to which the nodes $i, j$ both belong, we increase $\logit\{\PP(X_{i, j} = 1)\}$ by $\Gamma_{s, s}$, and for every adjacent pair of communities $s, t$ to which the nodes $i, j$ separately belong, we increase $\logit\{\PP(X_{i, j} = 1)\}$ by $\Gamma_{s, t}$. Our model with overlapping communities can still be written as \eqref{eq:half_model_2layer}.
Although we focus on $\Gamma_{s, t} > 0$ in this paper, our model can accommodate disassortative mixing, that is, nodes from different communities are more likely to be connected, allowing $\Gamma_{s, t} < 0$.

We model the distribution of $\tilde\Xb^{(1:N)}$ without structural constraints.
There are $2^{\frac{\tilde{p} (\tilde{p} - 1)}{2}}$ distinct $\tilde{p} \times \tilde{p}$ adjacency matrices, so the distribution of $\tilde\Xb^{(1:N)}$ can be characterized via a categorical distribution with parameter $\bnu$ belonging to the probability simplex $\cS^{2^{\frac{\tilde{p} (\tilde{p} - 1)}{2}} - 1} \subset \RR^{2^{\frac{\tilde{p} (\tilde{p} - 1)}{2}}}$. Our model is most useful practically when the dimension $\tilde{p}$ is small so that the adjacency matrix is at most $4 \times 4$ or so. 

We have now arrived at the 2-layer version of our model with the following joint distribution of observed $\Xb^{(1:N)}$ and latent $\tilde\Xb^{(1:N)}$:
$$
\PP(\Xb^{(1:N)}, \tilde\Xb^{(1:N)} | \Ab, \bTheta)
=
\prod_{n = 1}^N \left(
\PP(\tilde\Xb^{(n)} | \bnu) \prod_{1 \le i < j \le p} \frac{\exp\left( X_{i, j}^{(n)} \left(
C + \ab_i^\top (\tilde\Xb^{(n)} * \bGamma) \ab_j
\right) \right)}{1 + \exp\left(
C + \ab_i^\top (\tilde\Xb^{(n)} * \bGamma) \ab_j
\right)}
\right)
.
$$
As the latent adjacency matrices  $\tilde\Xb^{(1:N)}$ form a sample of multiplex networks, the same approach of this subsection can be applied to model the distribution of the latent $\tilde\Xb^{(1:N)}$ and our model can be formulated hierarchically across multiple layers, forming a deep generative model for multiplex networks.
This model is described in the following subsection.

\subsection{General \texorpdfstring{$(K + 1)$}{}-layer Model}\label{ssec:model_k}

With the intuition for the two-layer model established in the previous section, we now describe the general version of our model with $K + 1$ layers.
Let $\cN$ be a Bayesian network of $(K + 1)$ layers with the 0th layer $\cN_0$ at the top and the $K$th layer $\cN_K$ at the bottom. Each layer $k$ contains $p_k$ nodes with $p_k \ge p_{k - 1}$ for all $k \in [K]$, such that the network $\cN$ is pyramid-shaped. The 0th layer typically contains very few nodes (e.g. $p_0 = 2, 3, 4$). Only directed edges between nodes in two consecutive layers are allowed, which always point from the deeper layer to the shallower layer.
Between two consecutive layers $\cN_{k - 1}$ and $\cN_k$, we use a binary connection matrix {\em} $\Ab_k \in \{0, 1\}^{p_k \times p_{k - 1}}$ to denote the bipartite connections between the $p_k$ nodes of $\cN_k$ and the $p_{k - 1}$ nodes of $\cN_{k - 1}$. Using the intuition from the previous subsection, we can view nodes in the deeper layer $\cN_{k - 1}$ as defining communities for nodes in the shallower layer $\cN_k$, such that the connection matrix $\Ab_k$ encodes the overlapping community membership.
See Figure \ref{fig:tikz_A} for an illustration with three layers.

\input{tikz/tikz_figure_layer3}

For a sample of $N$ individual networks, we view the observed adjacency matrix of each individual as a realization from the population distribution of adjacency matrices modeled by our Bayesian network as follows. For individual $n$, the adjacency between the $p_k$ nodes of layer $\cN_k$ is indicated by the symmetric binary adjacency matrix $\Xb_k^{(n)} \in \{0, 1\}^{p_k \times p_k}$. We view each node as being adjacent to itself, such that all diagonal entries of each adjacency matrix are ones. Only the bottom layer adjacency matrix $\Xb_K^{(n)}$ is observed, while the upper layer adjacency matrices $\Xb_k^{(n)}$'s ($k \in [0, K - 1]$) are viewed as latent variables and the connection matrices $\Ab_k$'s ($k \in [K]$) are viewed as unknown model parameters.
See Figure \ref{fig:tikz_X} for an illustration.
Let $\Xb^{(n)} := \{\Xb_k^{(n)}\}_{k \in [0, K]}$ denote the collection of adjacency matrices and $\Ab := \{\Ab_k\}_{k \in [K]}$ denote the collection of connection matrices.
The joint distribution of $\Xb^{(n)}$ given $\Ab$ is
$$
\PP(\Xb^{(n)} | \Ab, \bTheta)
:=
\PP(\Xb_0^{(n)} | \bnu) \prod_{k = 1}^K \PP(\Xb_k^{(n)} | \Xb_{k - 1}^{(n)}, \Ab_k, \bTheta_k)
,
$$
where $\bTheta := \{\bnu, \bTheta_k: k \in [K]\}$ are model parameters, and $\PP(\Xb_0^{(n)} | \bnu)$ is a categorical distribution over all $p_0 \times p_0$ adjacency matrices parameterized by $\bnu$ in the probability simplex $\cS^{2^{\frac{p_0(p_0 - 1)}{2}} - 1}$.
For $k \in [K]$, each $\Xb_k^{(n)}$ is symmetric, has a diagonal of all ones, and we assume the lower-triangular entries are conditionally independent given $\Xb_{k - 1}^{(n)}, \Ab_k, \bTheta_k$ following the Bernoulli distribution
\begin{equation}\label{eq:model_entry}
\PP(X_{k, i, j}^{(n)} | \Xb_{k - 1}^{(n)}, \Ab_k, \bTheta_k)
=
\frac{\exp\left(X_{k, i, j}^{(n)} \left(
C_k + \ab_{k, i}^\top (\bGamma_k * \Xb_{k - 1}^{(n)}) \ab_{k, j}
\right) \right)}{1 + \exp\left(
C_k + \ab_{k, i}^\top (\bGamma_k * \Xb_{k - 1}^{(n)}) \ab_{k, j}
\right)}
.
\end{equation}
Here $X_{k, i, j}^{(n)}$ denotes the $(i, j)$th entry of matrix $\Xb_k^{(n)}$, $\ab_{k, i}^\top \in \RR^{1 \times p_{k - 1}}$ denotes the $i$th row of matrix $\Ab_k$, and $\{C_k, \bGamma_k\}$ together form the model parameter $\bTheta_k$ of the $k$th layer.
Parameter $C_k \in \RR$ is a scalar, and $\bGamma_k \in \RR_+^{p_{k - 1} \times p_{k - 1}}$ is a symmetric matrix with positive entries. 
Between every two consecutive layers $\cN_k$ and $\cN_{k - 1}$, the parameters $\Ab_k, C_k, \bGamma_k$ are as we introduced in the previous section for the two-layer version of our model.

The joint distribution of observed and latent adjacency matrices for all $N$ individuals at all $(K + 1)$ layers under our model is
\begin{equation}\label{eq:model}
\PP(\Xb^{(1:N)} | \Ab, \bTheta)
=
\prod_{n = 1}^N \left(
\PP(\Xb_0^{(n)} | \bnu) \prod_{k = 1}^K \prod_{1 \le i < j \le p_k} \frac{\exp\left( X_{k, i, j}^{(n)} \left(
C_k + \ab_{k, i}^\top (\bGamma_k * \Xb_{k - 1}^{(n)}) \ab_{k, j}
\right) \right)}{1 + \exp\left(
C_k + \ab_{k, i}^\top (\bGamma_k * \Xb_{k - 1}^{(n)}) \ab_{k, j}
\right)}
\right)
.
\end{equation}
Marginalizing out all latent adjacency matrices $\Xb_0^{(1:N)},\ldots,\Xb_{K-1}^{(1:N)}$ gives the marginal distribution of the observed adjacency matrices $\Xb_K^{(1:N)}$ for the multiplex networks.
Let $\cA_0$ denote the parameter space containing all $\Ab$'s with each $\Ab_k$ being a $p_k \times p_{k - 1}$ binary matrix:
\begin{equation}\label{eq:A0_def}
\cA_0
:=
\Big\{
\Ab:~
\forall k \in [K]
,~
\Ab_k \in \{0, 1\}^{p_k \times p_{k - 1}}
\Big\}
,
\end{equation}
and let $\cT_0$ denote the continuous parameter space containing all $\bTheta$'s:
\begin{equation}\label{eq:T0_def}
\cT_0
:=
\Big\{
\bTheta:~
\bnu \in \cS_+^{2^{\frac{p_0 (p_0 - 1)}{2}} - 1}
,~
\forall k \in [K]
,~
i, j \in [p_k]
,~
C_k \in \RR,~ \Gamma_{k, i, j} = \Gamma_{k, j, i} > 0
\Big\}
,
\end{equation}
where for any $m \in \NN$, $\cS_+^m := \{\bnu \in \RR^{m + 1}: ~ \nu_i > 0~ \forall i \in [m + 1],~ \one^\top \bnu = 1\}$ denotes the interior of the probability simplex $\cS^m \subset \RR^{m + 1}$.
Let $\cX_k$ denote the collection of all $p_k \times p_k$ adjacency matrices for each $k \in [0, K]$, defined as
\begin{equation}\label{eq:Xk_def}
\cX_k
:=
\big\{
\Xb_k \in \{0, 1\}^{p_k \times p_k}:~
\forall i, j \in [p_k],~
X_{k, i, i} = 1,~ X_{k, i, j} = X_{k, j, i}
\big\}
.
\end{equation}
The parameter spaces of our model considered in this work will be subsets of $\cA_0 \times \cT_0$, and each $\cX_k$ is the space of adjacency matrices $\Xb_k$.

\section{Theoretical Properties}\label{sec:theo}

In this section, we establish novel theoretical results on identifiability and posterior consistency of suitable Bayesian estimators.
At a high level, identifiability of our model means all unknown parameters (including both continuous $\bTheta$ and discrete connection matrices $\Ab$) can be uniquely identified and recovered from the marginal distribution of the observed adjacency matrix.
There exists the usual trivial type of non-identifiability issue of label switching \citep{gyllenberg1994non}.
In particular, latent adjacency matrices  $\Xb_k^{(n)}$ for $k \in [0, K - 1]$ are discrete latent class variables and the posterior is invariant to permutations of their rows and columns.
We follow standard practice in dealing with the label switching issue using post-processing methods in \citet{stephens2000dealing}, and focus on the nontrivial task of proving identifiability of our model up to a permutation of the nodes in each layer.
Proofs of all theorems in this section are deferred to the supplementary material Sections \ref{sec:iden}, \ref{sec:class_M}, \ref{sec:post_theo}.

Let $\cR \subset \cA_0 \times \cT_0$ denote a parameter space.
For a parameter $(\Ab, \bTheta) \in \cR$, we denote the set of all its permutations in $\cR$ by
\begin{align}
\perm_\cR(\Ab, \bTheta)
:=
\Big\{
(\Ab', \bTheta') \in \cR:~
&
\forall k \in [K]
,~
\Ab_k'
=
\Pb_k \Ab_k \Pb_{k - 1}^\top
,~
C_k'
=
C_k
,~
\bGamma_k'
=
\Pb_{k - 1} \bGamma_k \Pb_{k - 1}^\top
,\notag\\&
\forall \Xb_0 \in \cX_0
,~
\PP(\Pb_0 \Xb_0 \Pb_0^\top | \bnu') = \PP(\Xb_0 | \bnu)
\Big\}
,
\label{eq:sigma_def}
\end{align}
where $\Pb_k$ are arbitrary $p_k \times p_k$ permutation matrices for $k \in [0, K - 1]$ and $\Pb_K := \Ib_{p_K}$.
Any combination of permutation matrices $\Pb_0, \ldots, \Pb_{K - 1}$ corresponds to a way of permuting the nodes in each layer of our Bayesian network, so that the set $\perm_\cR(\Ab, \bTheta)$ contains all model parameters that are equivalent to $(\Ab, \bTheta)$ up to label switching.
We also denote the set of all parameters that generate the same marginal distribution of $\Xb_K$ by
\begin{equation}\label{eq:pi_def}
\marg_\cR(\Ab, \bTheta)
:=
\Big\{
(\Ab', \bTheta') \in \cR:~
\forall \Xb_K \in \cX_K,~
\PP(\Xb_K | \Ab, \bTheta) = \PP(\Xb_K | \Ab', \bTheta')
\Big\}
.
\end{equation}
Since the marginal distribution $\PP(\Xb_K | \Ab, \bTheta)$ is invariant under the permutations of $(\Ab, \bTheta)$, we always have $\perm_\cR(\Ab, \bTheta) \subseteq \marg_\cR(\Ab, \bTheta)$.
We say a parameter $(\Ab, \bTheta)$ is identifiable in parameter space $\cR$ if $\perm_\cR(\Ab, \bTheta) = \marg_\cR(\Ab, \bTheta)$, and non-identifiable in $\cR$ if $\perm_\cR(\Ab, \bTheta) \subsetneq \marg_\cR(\Ab, \bTheta)$.
Following \cite{allman2009identifiability}, we next define strict and generic notions of model identifiability.

\begin{definition}[Strict Identifiability]\label{defi:strict}
The model with parameter space $\cR$ is \textit{strictly identifiable} if $\perm_\cR(\Ab, \bTheta) = \marg_\cR(\Ab, \bTheta)$ for all $(\Ab, \bTheta) \in \cR$.
\end{definition}

Strict identifiability means that each parameter in the parameter space is identifiable.
While it is a necessary condition for posterior consistency under all possible true parameters, it is often hard to establish without constraining the parameter space \citep{gu2021bayesian}.
A slightly weaker notion, called generic identifiability, means that the set of non-identifiable parameters have Lebesgue measure zero.
Generic identifiability often suffices for data analysis purposes \citep{allman2009identifiability} and can be used to show posterior consistency for almost all true parameters in the prior support.
While originally defined for a continuous parameter space, the definition of generic identifiability can be easily extended to our model, where there are both discrete ($\Ab$) and continuous ($\bTheta$) parameters.
Viewing $\cA_0$ as a discrete space and $\cT_0$ as a Riemannian manifold of dimension $\big(2^{\frac{p_0 (p_0 - 1)}{2}} - 1\big) + \sum_{k = 1}^K \big(1 + \frac{p_k (p_k + 1)}{2}\big)$,
we let $\upmu$ denote the counting measure over $\cA_0$, let $\uplambda$ denote the Lebesgue measure over $\cT_0$, and take their product $\upmu \times \uplambda$ to be the $\sigma$-finite base measure on $\cA_0 \times \cT_0$.

\begin{definition}[Generic Identifiability]\label{defi:generic}
The model with parameter space $\cR$ is \textit{generically identifiable} if $(\upmu \times \uplambda) \big\{ (\Ab, \bTheta) \in \cR:~ \perm_\cR(\Ab, \bTheta) \subsetneq \marg_\cR(\Ab, \bTheta) \big\} = 0$.
\end{definition}

To guarantee the asymptotic correctness of Bayesian inference, one often tries to establish the theoretical property known as \emph{posterior consistency}, as defined in the following.
Equiping the space $\cA_0 \times \cT_0$ with the $L_1$ metric, the induced topology on $\cA_0 \times \cT_0$ is the product of the discrete topology on $\cA_0$ and the Euclidean topology on $\cT_0$.
For a parameter space $\cR \subset \cA_0 \times \cT_0$ with the same metric and the subspace topology,
we denote the $\epsilon$-neighborhood of $(\Ab, \bTheta)$ as $\cO_\epsilon(\Ab, \bTheta)$ and the union of $\epsilon$-neighborhoods surrounding permutations of $(\Ab, \bTheta)$ as
\begin{equation}\label{eq:eps_nbhd}
\tilde\cO_\epsilon(\Ab, \bTheta)
:=
\bigcup_{(\Ab', \bTheta') \in \perm_\cR(\Ab, \bTheta)} \cO_\epsilon(\Ab', \bTheta')
.
\end{equation}

\begin{definition}[Posterior Consistency]\label{defi:post}
For the model with parameter space $\cR$ and prior distribution $\PP(\Ab, \bTheta)$, the posterior distribution $\PP\big(\Ab, \bTheta | \Xb_K^{(1:n)}\big)$ is \emph{(strongly) consistent} at $(\Ab^*, \bTheta^*) \in \cR$ if for any $\epsilon > 0$,
$$
\lim_{n \to \infty} \PP\left( \left.
(\Ab, \bTheta) \in \tilde\cO_\epsilon(\Ab^*, \bTheta^*)^c
\right| \Xb_K^{(1:n)}
\right)
=
0
$$
$\PP\big(\Xb_K^{(1:\infty)} | \Ab^*, \bTheta^*\big)$-almost surely.
\end{definition}

We note the slight difference between Definition \ref{defi:post} and the standard definition of posterior consistency \citep{ghosal2017fundamentals}.
The standard definition requires the posterior distribution to asymptotically concentrate in an arbitrarily small $\epsilon$-neighborhood $\cO_\epsilon(\Ab^*, \bTheta^*)$.
In our model, to deal with the trivial type of non-identifiability issue due to permutations, we loosen the $\epsilon$-neighborhood $\cO_\epsilon(\Ab^*, \bTheta^*)$ to $\tilde\cO_\epsilon(\Ab^*, \bTheta^*)$, which is the union of $\epsilon$-neighborhoods surrounding permutations of $(\Ab^*, \bTheta^*)$.
With $\cR$ being a separable metric space, Definition \ref{defi:post} is equivalent to saying that the posterior distribution is asymptotically supported on $\perm_\cR(\Ab^*, \bTheta^*)$.

In the following, we first present a theorem that establishes the strict identifiability of our model within a subset of $\cA_0 \times \cT_0$.
Let $\cA_1 \subset \cA_0$ be the set of all $\Ab$'s satisfying that each $\Ab_k$ contains two distinct identity submatrices $\Ib_{p_{k - 1}}$.
Formally, $\cA_1$ is defined as
\begin{equation}\label{eq:A1_def}
\cA_1
:=
\Big\{
\Ab \in \cA_0:~
\forall k \in [K],~
\exists~ \Pb_k, \Bb_k \text{ such that }
\Pb_k \Ab_k
=
\big(
\Ib_{p_{k - 1}} ~ \Ib_{p_{k - 1}} ~ \Bb_k
\big)^\top
\Big\}
,
\end{equation}
where each $\Pb_k \in \RR^{p_k \times p_k}$ is an arbitrary permutation matrix and each $\Bb_k \in \RR^{p_{k - 1} \times (p_k - 2p_{k - 1})}$ is an arbitrary matrix.
For each $\Ab \in \cA_1$, it is implicitly required that $\Ab$ characterizes a Bayesian network structure in the form of a steep pyramid with $p_k \ge 2p_{k - 1}$ for all $k \in [K]$.
Additionally, each community (node at upper layer) must have two members (nodes at lower layer) that only belong to this community.
Such nodes with a unique community are sometimes referred to as pure nodes, whose existence is a common condition to guarantee model identifiability of network models \citep[e.g.,][]{jin2023mixed}.

\begin{theorem}[Strict Identifiability]\label{theo:strict}
The model with parameter space $\cA_1 \times \cT_0$ is strictly identifiable.
\end{theorem}

To provide further intuition into the condition in \eqref{eq:A1_def}, we briefly discuss the proof techniques used for Theorem \ref{theo:strict} with $K = 1$.
The marginal distribution $\PP(\Xb_1 | \Ab, \bTheta)$ can be transformed into a three-way tensor and the conditional distribution $\PP(\Xb_1 | \Xb_0, \Ab_1, \bTheta_1)$ can be viewed as its CP decomposition \citep{kolda2009tensor}.
The strict identifiability of our model requires the uniqueness of this CP decomposition, for which the Kruskal's condition \citep{kruskal1977three} is known to be sufficient.
The existence of two identity blocks in $\Ab_1$ ensures the Kruskal's condition to hold.
Our model can also be viewed as a constrained latent class model, so a comparable result of strict identifiability can be found in \citet{gu2021bayesian}.

The identifiability of model parameters is a prerequisite for posterior consistency.
With the strict identifiability established for our model with parameter space $\cA_1 \times \cT_0$, we now present results for the posterior consistency of all parameters in the prior support.

\begin{theorem}[Schwartz's Posterior Consistency]\label{theo:schwartz}
For the model with parameter space $\cR \subset \cA_0 \times \cT_0$ and true parameter $(\Ab^*, \bTheta^*)$ in the support of the prior distribution $\PP(\Ab, \bTheta)$, the posterior distribution $\PP\big(\Ab, \bTheta | \Xb_K^{(1:n)}\big)$ is consistent at $(\Ab^*, \bTheta^*)$ if one of the following two conditions holds:
\begin{enumerate}[(i)]
\item
$\cR$ is compact and $\perm_\cR(\Ab^*, \bTheta^*) = \marg_\cR(\Ab^*, \bTheta^*)$;
\item
$\cR \subset \cA_1 \times \cT_0$.
\end{enumerate}
\end{theorem}

The proof of Theorem \ref{theo:schwartz} is based on Schwartz's Theorem \citep{schwartz1965bayes}, which establishes the posterior consistency of models in the Wasserstein space of probability measures.
We note that consistent probability measures are not the same as consistent model parameters.
To fill in this gap, we need the identifiability of the true parameter and some topological property of the parameter space $\cR$.
In condition (i), compactness of $\cR$ together with the assumed identifiability at $(\Ab^*, \bTheta^*)$ suffice.
In condition (ii), the parameter space $\cA_1 \times \cT_0$ ensures the identifiability at all parameters by Theorem \ref{theo:strict} and satisfies the desired topological property, so any subspace $\cR \subset \cA_1 \times \cT_0$ automatically inherits both properties.

Despite being the strongest possible notion of identifiability, strict identifiability requires the connection matrices to be in $\cA_1$, i.e., containing at least two identity submatrices, which could be a stringent constraint.
When running MCMC algorithms (introduced in Section \ref{sec:post}) to obtain posterior samples of $(\Ab, \bTheta)$, the samples of $\Ab$ do not always fall into $\cA_1$.
While we could certainly enforce constraints to guarantee $\Ab \in \cA_1$, this can lead to slow MCMC mixing. 
Hence, we hope to establish identifiability in a parameter space larger than $\cA_1$, so that constraining  $\Ab$ to $\cA_1$ will not be necessary.
This is where the notion of generic identifiability comes into play.

For a square matrix $\Mb$, we let $\mb_i$ denote its $i$th column and let $\offdiag(\Mb) := \Mb - \diag(\Mb)$ be the matrix obtained by setting the diagonal of $\Mb$ to zero.
Let $\langle d \rangle$ denote the set $\{(i, j):~ 1 \leq i < j \leq d\}$, which is the collection of all $\frac{d (d - 1)}{2}$ distinct pairs of integers in $[d]$.
Let $2^{\langle d \rangle}$ denote the power set of $\langle d \rangle$, containing all $2^{\frac{d (d - 1)}{2}}$ subsets of $\langle d \rangle$.
We relax the submatrices requirement in \eqref{eq:A1_def} from identity matrices to the following enlarged class of binary square matrices.

\begin{definition}[Class $\cM_d$]\label{defi:M}
We define $\cM_d$ to be the class of binary square matrices $\Mb \in \{0, 1\}^{d \times d}$ satisfying the following two conditions:
\begin{enumerate}[(i)]
\item
$\diag(\Mb) = \Ib_d$;
\item
for any two distinct subsets $\cS_1, \cS_2 \in 2^{\langle d \rangle}$,
$$
\offdiag\Big(
\sum_{(i, j) \in \cS_1} (\mb_i \mb_j^\top + \mb_j \mb_i^\top)
\Big)
\ne
\offdiag\Big(
\sum_{(i, j) \in \cS_2} (\mb_i \mb_j^\top + \mb_j \mb_i^\top)
\Big)
.
$$
\end{enumerate}
\end{definition}

The class $\cM_d$ is a large class of binary square matrices that includes the identity matrix $\Ib_d$ as an instance.
We defer other examples of class $\cM_d$ to the supplementary material Section \ref{sec:class_M}.
Therein we also discuss the connection of class $\cM_d$ to common linear algebra concepts such as distinct rows and columns, symmetry, invertibility, and rank.
We define an enlargement of the space $\cA_1$ as
\begin{equation}\label{eq:A21_def}
\cA_{2, 1}
:=
\Big\{
\Ab \in \cA_0:~
\forall k \in [K],~
\exists~ \Pb_k, \Bb_k \text{ such that }
\Pb_k \Ab_k
=
\big(
\Mb_k ~ \Mb_k' ~ \Bb_k
\big)^\top
\Big\}
,
\end{equation}
where each $\Pb_k \in \RR^{p_k \times p_k}$ is an arbitrary permutation matrix, each $\Bb_k \in \RR^{p_{k - 1} \times (p_k - 2p_{k - 1})}$ is an arbitrary matrix, and $\Mb_k, \Mb_k'$ are binary square matrices in $\cM_{p_{k - 1}}$.
The all ones diagonals of $\Mb_k, \Mb_k'$ ensure that each community contains at least two nodes.
We consider $\Pb_k = \Ib_{p_k}$ for intuition.
The condition (ii) in Definition \ref{defi:M} satisfied by $\Mb_k$ is equivalent to $\prod_{1 \le i < j \le p_{k - 1}} \exp\big( C_k + \ab_{k, i}^\top (\bGamma_k * \Xb_{k - 1}) \ab_{k, j} \big)$ being distinct for each $\Xb_{k - 1} \in \cX_{k - 1}$, which allows Kruskal's condition to hold almost everywhere.

Two additional conditions are needed. 
Note that each $\Ab_k$ needs to have distinct columns, otherwise entries in the identical columns of $\bGamma_k$ can be swapped without changing the conditional distribution $\PP(\Xb_k | \Xb_{k - 1}, \Ab_k, \bTheta_k)$.
We incorporate this condition to the space
\begin{equation}\label{eq:A22_def}
\cA_{2, 2}
:=
\Big\{
\Ab \in \cA_0:~
\forall k \in [K],~
\forall \Pb_{k - 1},~
\Ab_k \Pb_{k - 1}
\ne
\Ab_k
\Big\}
,
\end{equation}
where each $\Pb_{k - 1} \in \RR^{p_{k - 1} \times p_{k - 1}}$ is an arbitrary permutation matrix.
We also note from \eqref{eq:model_entry} and \eqref{eq:model} that the conditional distribution $\PP(\Xb_k | \Xb_{k - 1}, \Ab_k, \bTheta_k)$ depends on $C_k$ and the diagonal vector $\bgamma_k$ of $\bGamma_k$ only through linear combinations $C_k + (\ab_{k, i}^\top * \ab_{k, j}^\top) \bgamma_k$ for $(i, j) \in \langle p_k \rangle$.
Let matrix $\Db_k \in \RR^{\frac{p_k (p_k - 1)}{2} \times (1 + p_{k - 1})}$ denote the stack of row vectors $\big( 1 ~ \ab_{k, i}^\top * \ab_{k, j}^\top \big)$, then we need $\Db_k$ to have full column rank for the parameters $C_k, \bgamma_k$ to be identifiable.
Accordingly, define
\begin{equation}\label{eq:A23_def}
\cA_{2, 3}
:=
\Big\{
\Ab \in \cA_0:~
\forall k \in [K],~
\rank(\Db_k) = p_{k - 1} + 1
\Big\}
.
\end{equation}

By taking intersection over the spaces $\cA_{2, 1}, \cA_{2, 2}, \cA_{2, 3}$, we define the space $\cA_2 \supset \cA_1$ as
\begin{equation}\label{eq:A2_def}
\cA_2
:=
\cA_{2, 1} \cap \cA_{2, 2} \cap \cA_{2, 3}
.
\end{equation}
Despite the complicated definition of $\cA_2$, in practice we only need to verify the condition in $\cA_{2, 1}$, as shown in Proposition \ref{prop:simp_A2} in the following section.

\begin{theorem}[Generic Identifiability]\label{theo:generic}
The model with parameter space $\cA_2 \times \cT_0$ is generically identifiable.
\end{theorem}

In the literature, all previous techniques for showing generic identifiability rely on algebraic varieties as in  \citet{allman2009identifiability}.
In the proof of Theorem \ref{theo:generic}, we instead rely on zero sets of holomorphic functions, instead of just polynomials.
This suggests potential of this approach to be applied to general model classes built on other nonlinear functions.
Details are included in the supplementary material Section \ref{sec:iden}.

Based on the generic identifiability conclusion, the following Theorem \ref{theo:doob} establishes the posterior consistency for our model with parameter space in $\cA_2 \times \cT_0$.
We show that the null set of posterior inconsistency is negligibly small under the prior distribution, so the posterior is consistent at almost every true parameter.

\begin{theorem}[Doob's Posterior Consistency]\label{theo:doob}
Let parameter space $\cR$ be a Borel subset of $\cA_2 \times \cT_0$.
For a prior distribution $\PP(\Ab, \bTheta)$ absolutely continuous w.r.t. $\upmu \times \uplambda$, the posterior distribution $\PP\big(\Ab, \bTheta | \Xb_K^{(1:n)}\big)$ is consistent at $(\Ab^*, \bTheta^*)$, for $\PP(\Ab, \bTheta)$-almost every $(\Ab^*, \bTheta^*) \in \cR$
\end{theorem}

In its proof, we define a quotient parameter space and its quotient $\sigma$-algebra, in which each $\perm_\cR(\Ab, \bTheta)$ is viewed as a fiber (equivalence class) of the quotient map and $\tilde\cO_\epsilon(\Ab, \bTheta)$ is its $\epsilon$-neighborhood.
This provides further insight into our definition of posterior consistency.
Applying Doob's theorem of posterior consistency \citep{doob1949application} to the quotient parameter space along with measure theoretic arguments yields Theorem \ref{theo:doob}.

\section{Posterior Computation}\label{sec:post}

\subsection{Prior Specification}\label{ssec:prior}

With the likelihood function of our model given in \eqref{eq:model}, to conduct Bayesian inference we need to specify prior distributions on the model parameters $(\Ab, \bTheta)$.
For simplicity, we impose independent priors on each row $\ab_{k, i}^\top$ of $\Ab_k$, each $C_k$, each entry $\Gamma_{k, i, j}$ of $\bGamma_k$, and $\bnu$, such that the prior distribution $\PP(\Ab, \bTheta)$ takes the factorized form
\begin{equation}\label{eq:prior}
\PP(\Ab, \bTheta)
:=
\PP(\bnu)
\prod_{k = 1}^K \left(
\PP(C_k)
~
\prod_{i = 1}^{p_k} \PP(\ab_{k, i})
~
\prod_{1 \le i \le j \le p_{k - 1}} \PP(\Gamma_{k, i, j})
\right)
.
\end{equation}

We let the prior on each $C_k$ be the normal distribution $N(\mu_C, \sigma_C^2)$, the prior on each diagonal entry of $\bGamma_k$ be the positively truncated normal distribution $N_+(\mu_\gamma, \sigma_\gamma^2)$, and the prior on each off-diagonal entry of $\bGamma_k$ be the positively truncated normal distribution $N_+(\mu_\delta, \sigma_\delta^2)$.
The difference between diagonal and off-diagonal priors gives flexibility in modeling the adjacency probability of two nodes from the same community to be larger than that from two adjacent communities.
For parameter $\bnu$ of the categorical distribution on $\cX_0$, we assign the Dirichlet prior $\Dir(\alpha \one)$ over the probability simplex $\cS_+^{2^{\frac{p_0(p_0 - 1)}{2}} - 1}$.
Here $\mu_C, \sigma_C^2, \mu_\gamma, \sigma_\gamma^2, \mu_\delta, \sigma_\delta^2, \alpha$ are all hyperparameters.

To encourage statistical parsimony and prevent each node from belonging to too many communities, we enforce sparsity in the connection matrices.
For each $\ab_{k, i}$, we let its prior be the uniform distribution $\mathrm{Unif}(\cV_k)$, where $\cV_k := \big\{ \vb \in \{0, 1\}^{p_k}: 1 \le \one^\top \vb \le S \big\}$ is the set of all binary vectors under a sparsity constraint.
The hyperparameter $S$ is the maximum number of communities allowed for each node, which we recommend choosing as 1 or 2 for interpretability and computational efficiency. 

Let $\cA_3$ denote the space of all $\Ab$ satisfying the sparsity constraint.
The prior distribution $\PP(\Ab, \bTheta)$ in \eqref{eq:prior} is supported on $\cA_3 \times \cT_0$ and absolutely continuous with respect to the base measure $\upmu \times \uplambda$, so Theorems \ref{theo:strict} and \ref{theo:generic} respectively hold on parameter spaces $(\cA_1 \cap \cA_3) \times \cT_0$ and $(\cA_2 \cap \cA_3) \times \cT_0$.
We also note that the conditions in defining $\cA_2$ can be simplified when sparsity hyperparameter $S \le 2$, as shown in the following proposition.

\begin{proposition}\label{prop:simp_A2}
Let $S \in \{1, 2\}$ and $p_0 \ge 3$.
If $\Ab \in \cA_3$ does not have an all-ones column, then $\Ab \in \cA_{2, 1}$ is equivalent to $\Ab \in \cA_2$.
\end{proposition}

In the following Section \ref{ssec:post_comp}, we will introduce a MCMC sampler with data augmentation over space $\cA_3 \times \cT_0$ for posterior computation.
To accommodate conditions for generic identifiability, we could afterwards remove the small proportion of samples outside $(\cA_2 \cap \cA_3) \times \cT_0$.
Proposition \ref{prop:simp_A2} essentially suggests that we only need to verify samples with $\Ab$ lying in $\cA_{2, 1}$ instead of $\cA_2$, largely simplifying this post-processing step.
Its proof is provided in the supplementary material Section \ref{sec:class_M}.

\subsection{MCMC Sampling and Spectral Initialization}\label{ssec:post_comp}

To conduct posterior inference, we design a Gibbs sampler using Polya-Gamma data augmentation \citep{polson2013bayesian}.
We associate each entry $X_{k, i, j}^{(n)}$ with an augmented Polya-Gamma variable $\omega_{k, i, j}^{(n)}$, so that the full conditional distributions of each parameter and each latent variable are all tractable through semi-conjugacy.
Conditioned on all other parameters and latent variables, each row $\ab_{k, i}^\top$ in $\Ab_k$ follows a categorical distribution over $\cV_k$, each $C_k$ follows a normal distribution, each entry $\Gamma_{k, i, j}$ in $\bGamma_k$ follows a truncated normal distribution, $\bnu$ follows a Dirichlet distribution, each entry $X_{k, i, j}^{(n)}$ in $\Xb_k^{(n)}$ follows a Bernoulli distribution, and each augmented variable $\omega_{k, i, j}^{(n)}$ follows a Polya-Gamma distribution.
Details of this Gibbs sampler are provided in the supplementary material Section \ref{sec:post_comp_supp}.

As is typical with complex and hierarchical Bayesian models, concerns can arise regarding the efficiency of the Gibbs sampler when applied to large network datasets.
To address this concern, we further design an efficient algorithm leveraging techniques including subsampling MCMC \citep{quiroz2018speeding, quiroz2018subsampling, maire2019informed} and spectral initialization using the Mixed-SCORE algorithm \citep{jin2023mixed}.
Subsampling MCMC significantly reduces the computational complexity of obtaining one posterior sample while only slightly compromising the accuracy of stationary distribution.
Spectral initialization, through theoretical justification, allows an initialization close to the posterior mode, which further reduces the burn-in period in MCMC sampling.
Due to space limitations, we defer the computational details to the supplementary material Section \ref{sec:post_comp_supp}.
We also empirically demonstrate the consistency of posterior inference and model selection for our model through abundant simulation studies in the supplementary material Section \ref{sec:sim}.

\section{Application to Brain Connectome Data}\label{sec:appl}

As one of the motivating examples of multiplex networks, we apply our model to brain structural connectome data from the Human Connectome Project (HCP) \citep{van2012human, elam2021human}.
Following the Desikan atlas \citep{desikan2006automated}, the brain is divided into 68 anatomical regions of interest (ROI), and the observed data are the $68 \times 68$ adjacency matrices for each of the 1065 individuals.
In addition to the network data, a variety of cognitive traits of each individual are also available, including test scores for skills such as linguistics and memory, use of tobacco and marijuana, etc.

As discussed in Section \ref{sec:post}, the sparsity hyperparameter $S$ controls the number of deeper-layer communities that each node can belong to, and we recommend $S \in \{1,2\}$ for interpretability and computational/statistical efficiency.
For this application, we focus on three-layer Bayesian network structures ($K = 2$) and carry out our analyses for both $S = 1$ and $S = 2$.
We conduct model selection based on the Watanabe-Akaike information criterion (WAIC) in \citet{watanabe2010asymptotic} and make posterior inference through a combination of spectral initialization and MCMC sampling, with details included in the supplementary material Section \ref{ssec:model_sel}.
Under sparsity parameter $S = 1$, the optimal model has structure $p_0 = 4, p_1 = 21$, and under sparsity parameter $S = 2$, the optimal model has structure $p_0 = 4, p_1 = 18$.

In Section \ref{ssec:appl_part_clus}, we demonstrate that our model can successfully identify and cluster together the ROIs that are spatially close on the cerebral cortex, even though no such location information is used in our analysis.
In Section \ref{ssec:appl_cluster}, we show that our model can group the $N = 1065$ individuals into clusters interpretable by significant differences in cognitive traits, even though they are formed purely based on the observed adjacency matrices without supervision.
In the supplementary material Section \ref{ssec:appl_relate}, we show that the latent features inferred by our model exhibit substantially stronger associations with cognitive traits than those obtained from alternative methods.
In the supplementary material Section \ref{ssec:mask}, we further show that our model provides substantially more accurate prediction of missing edges in multiplex networks than alternative methods, highlighting the benefits of full Bayesian inference.

\subsection{Multi-resolution Partitioning and Clustering of Brain Regions}\label{ssec:appl_part_clus}

Our first goal is to group the 68 brain ROIs into meaningful hierarchical partitions and clusters.
As explained in Section \ref{sec:model}, the $p_1$ nodes in the middle layer of our model can be viewed as communities for the $p_2 = 68$ ROIs in the observed adjacency matrices, and the $p_0$ nodes in the top (deepest) layer can be viewed as communities for the $p_1$ nodes in the middle layer.
The connection matrices $\Ab =  \{\Ab_1, \Ab_2\}$ indicate the (overlapping) memberships of these community assignments, which naturally represent a hierarchical partitioning (when $S = 1$) or overlapping clustering (when $S = 2$) of the 68 ROIs.

\begin{figure}[ht]
\centering
\subfigure[
Visualization of $p_0 = 4$ higher-level clusters of ROIs on cerebral cortex.
]{
\includegraphics[width = 0.4\textwidth]{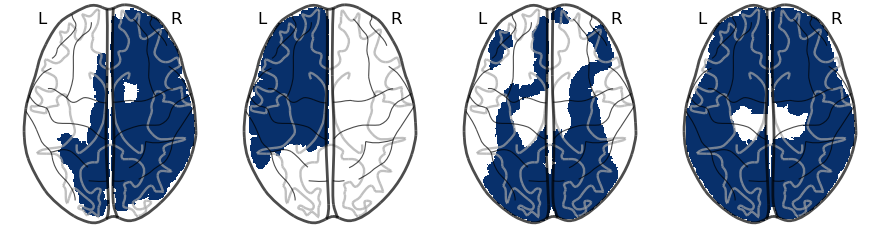}
\label{fig:app_brain_S2_A1}
}
\hfill
\subfigure[
Visualization of $p_1 = 18$ lower-level clusters of ROIs on cerebral cortex.
]{
\includegraphics[width = 0.8\textwidth]{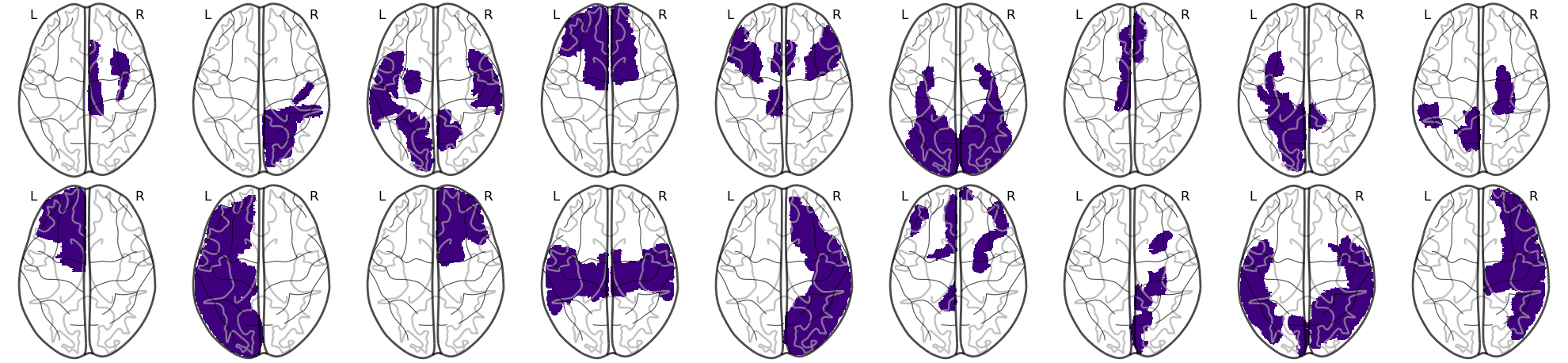}
\label{fig:app_brain_S2_A2}
}
\caption{
Posterior distributions of the higher-level and lower-level clusters given by connection matrices $\Ab$, under sparsity hyperparameter $S = 2$ and number of nodes in each Bayesian network layer $p_0 = 4, p_1 = 18, p_2 = 68$.
}
\label{fig:app_brain_S2}
\end{figure}

For each posterior sample $\Ab^{(t)}$, we obtain $p_1$ clusters through $\Ab_2^{(t)}$ and $p_0$ higher-level clusters of these $p_1$ clusters through $\Ab_1^{(t)}$.
Under $S = 1$, we define an ROI as belonging to a higher-level (deeper-layer) cluster if there exists a path from this ROI to that higher-level cluster in the DAG structure visualized in Figure \ref{fig:tikz_A}.
Under $S = 2$, we further require the number of paths from an ROI to its higher-level cluster to be at least two, in order to adjust for the increased number of edges in the DAG structure.
Therefore, the $j$th (lower-level) cluster is formed by the ROIs in $\{i \in [p_2]:~ A_{2, i, j}^{(t)} = 1\}$ and the $\ell$th higher-level cluster is formed by the ROIs in $\{i \in [p_2]:~ (\Ab_2^{(t)} \Ab_1^{(t)})_{i, \ell} \ge S\}$.
The posterior samples are post-processed to handle label switching \citep{stephens2000dealing} before being used to calculate posterior summaries.

We present the posterior distributions of the lower-level and higher-level clusters for $S = 2$ and defer the results for $S = 1$ to the supplementary material Section \ref{sec:more_app}.
For $S = 2$, the optimal number of nodes in each layer of our model is $p_0 = 4, p_1 = 18, p_2 = 68$.
Under this optimal choice, the marginal posterior distributions of $p_0$ higher-level clusters and $p_1$ lower-level clusters are visualized in Figure \ref{fig:app_brain_S2} on the cerebral cortex in a top-down view.
An interesting observation is that, almost all lower-level and higher-level clusters are formed by brain regions that are geographically close together on the cerebral cortex, even though we do not include spatial locations of nodes in our analysis.
This suggests our model's excellent ability to extract important neurobiological information from the brain connectivity data alone.

The learned lower-level and higher-level clusters are highly interpretable.
The lower-level clusters in Figure \ref{fig:app_brain_S2_A2} form 18 overlapping regions on the cerebral cortex roughly corresponding to the frontal, parietal, temporal, and occipital lobes of the brain.
Some of clusters are restricted to one hemisphere, while others include regions that are symmetric in the left and right hemispheres.
In Figure \ref{fig:app_brain_S2_A1}, the four higher-level clusters exhibit clear anatomical structure: two roughly correspond to the left and right hemispheres, one aligns with the occipital lobe, and the remaining cluster captures the combination of frontal, temporal, and occipital lobes.
Overall, these clusters closely reflect known anatomical organization of the human brain.

\subsection{Clustering Individuals}\label{ssec:appl_cluster}

We are also interested in grouping the 1065 individuals into interpretable clusters based on their brain connectivity.
In the following, we demonstrate that the entries in $\{\Xb_1^{(n)}: n = 1, \ldots, 1065\}$ can be naturally used for hierarchical clustering.
For an individual $n$ and $L$ arbitrary entries of $\Xb_1^{(n)}$, we let $\bar{x}_1^{(n)}, \ldots, \bar{x}_L^{(n)}$ denote the posterior means of these entries estimated from the MCMC samples.
By picking $L$ thresholds $m_1, \ldots, m_L \in (0, 1)$, we can group the individuals into a total of $2^L$ clusters consisting of
$$
\left\{
n \in [N]:~
\forall \ell \in [L],~
1_{\bar{x}_\ell^{(n)} > m_\ell} = z_\ell
\right\}
$$
for each binary vector $\zb \in \{0, 1\}^L$ indicating the cluster.
By choosing an order of the $L$ entries, these $2^L$ clusters can be viewed as the leaves of a complete binary tree, where each parent cluster has two children clusters obtained from further partitioning by an entry of $\Xb_1$.
We recommend choosing $m_\ell$ as the median of $\bar{x}_\ell^{(1)}, \ldots, \bar{x}_\ell^{(N)}$.
It is also possible to form more than two clusters using one entry of $\Xb_1$, in which case the multiple thresholds could be chosen as quantiles of $\bar{x}_\ell^{(1)}, \ldots, \bar{x}_\ell^{(N)}$.

We assess how different clusters of individuals vary in their cognitive traits.
A significant difference would help further interpret the clusters.
For instance, one cluster could correspond to people with stronger addiction to marijuana and the other cluster could correspond to people with weaker addiction.
For each active entry of $\Xb_1$, we divide individuals into two clusters based on whether their posterior means $\bar{x}^{(n)}$ are above or below the median $m$.
We test for the difference between these two groups over the 175-dimensional vectors of cognitive traits using Hotelling's $T^2$ test \citep{hotelling1931generalization}, conducted for each active entry and both of our $S = 1$ and $S = 2$ models.
For either choice of $S$, we find that around half of the active entries in $\Xb_1$ have significant $p$-values ($< 0.1$), indicating that our unsupervised model is inferring latent adjacency edges that are related to many cognitive traits.
Histograms of these $p$-values are provided in the supplementary material Section \ref{sec:more_app}.

\begin{figure}[ht]
\centering
\subfigure[
Model: $S = 1, p_0 = 4, p_1 = 21, p_2 = 68$.
Entry: $X_{1, 18, 7}$.
Cognitive trait: age of first cigarette.
]{
\includegraphics[width=0.45\textwidth]{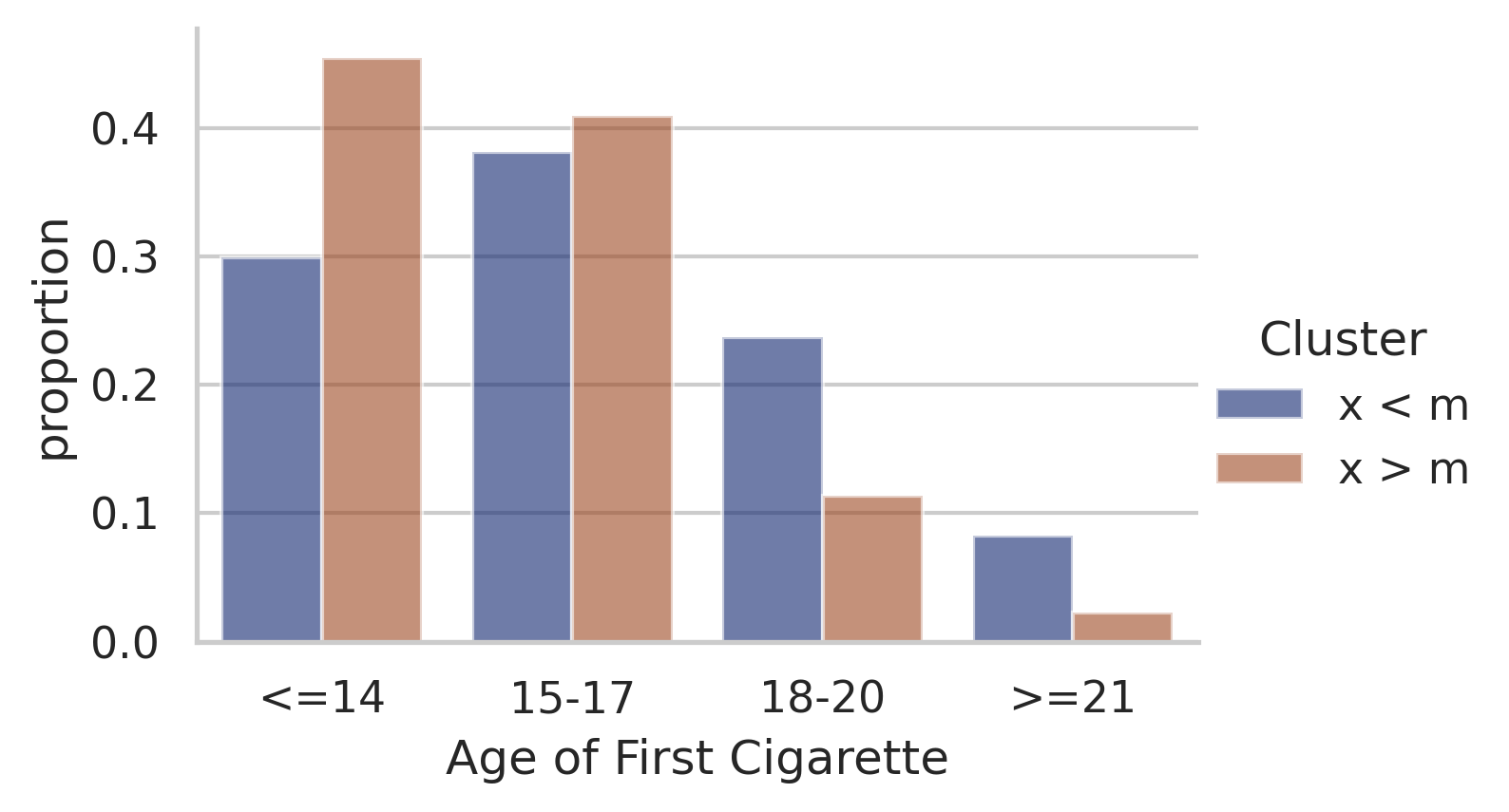}
\label{fig:cluster_S1}
}
\hfill
\subfigure[
Model: $S = 2, p_0 = 4, p_1 = 18, p_2 = 68$.
Entry: $X_{1, 14, 10}$.
Cognitive trait: (age-adjusted) grip strength.
]{
\includegraphics[width=0.45\textwidth]{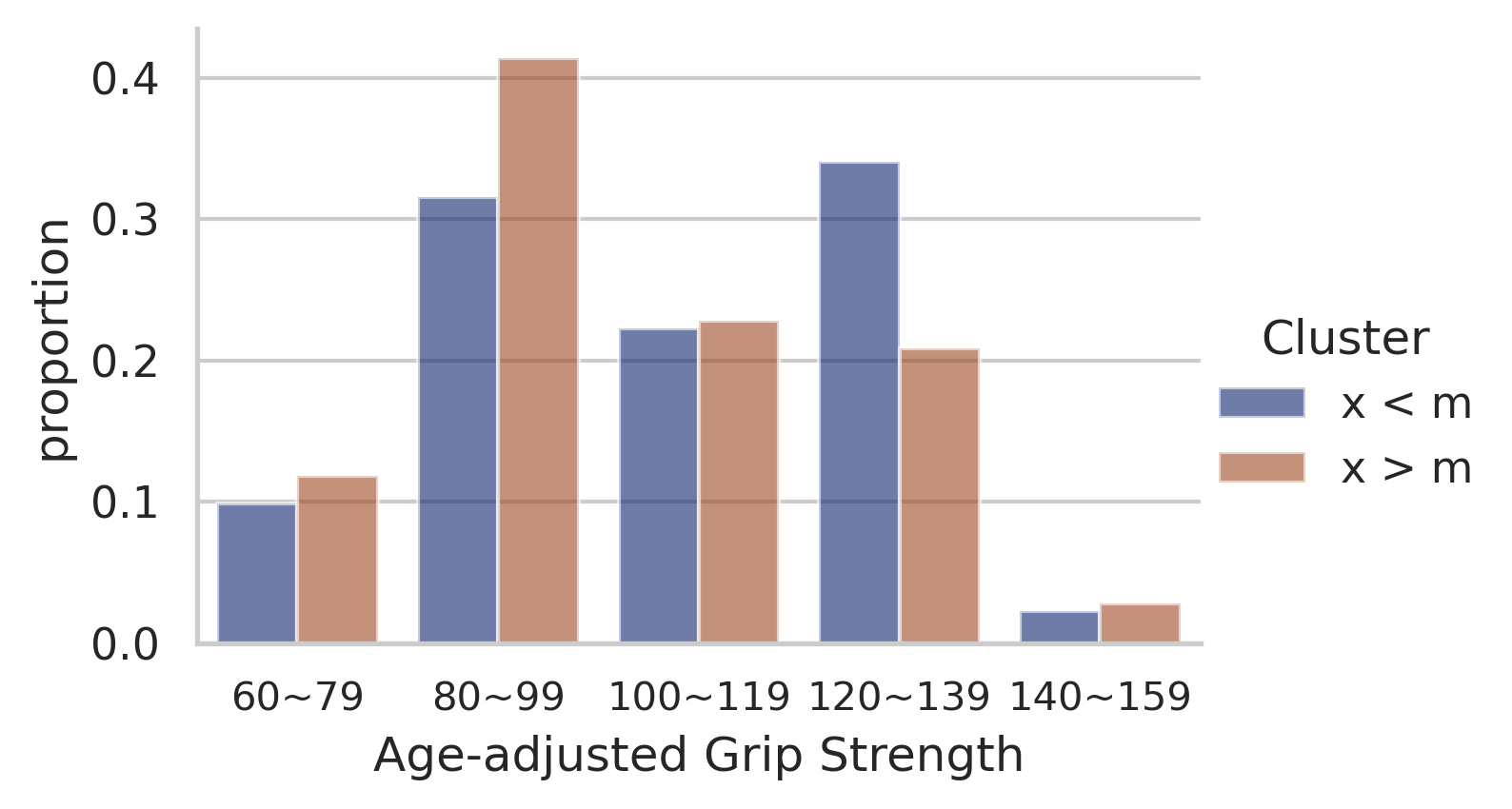}
\label{fig:cluster_S2}
}
\caption{
The distributions of some cognitive traits in the two clusters of individuals with entrywise posterior mean $\bar{x}^{(n)} < m$ and $\bar{x}^{(n)} > m$, with $m$ the median of $\bar{x}^{(1:N)}$.
}
\label{fig:cluster}
\end{figure}

We provide two examples of active entries in $\Xb_1$ for which the formed clusters have significantly varying cognitive traits.
For the model with $S = 1, p_0 = 4, p_1 = 21, p_2 = 68$, the cluster of individuals with larger posterior mean of $X_{1, 18, 7}$ tends to have an earlier age of smoking compared to the cluster of individuals with smaller posterior mean of $X_{1, 18, 7}$.
The 18th community in this model roughly corresponds to the occipital lobe of the right hemisphere, and the 7th community roughly corresponds to the parietal lobe of the left hemisphere.
This suggests that clusters of individuals formed by different strengths of brain connectivity between these two brain regions differ in their age of first cigarette use.
We restrict the population to the smoking individuals in this example.
The distributions of age of first cigarette use for both clusters are plotted in Figure \ref{fig:cluster_S1}.

For the model with $S = 2, p_0 = 4, p_1 = 18, p_2 = 68$, we find that the cluster of individuals with smaller posterior mean of $X_{1, 14, 10}$ tends to have a stronger (age-adjusted) grip strength compared to the cluster of individuals with larger posterior mean of $X_{1, 14, 10}$.
The 14th community in this model roughly corresponds to the right hemisphere, and the 10th community roughly corresponds to the frontal lobe of the left hemisphere.
This suggests that clusters of individuals formed by different strength of brain connectivity between these two brain regions differ in their grip strength.
The distributions of age-adjusted grip strength for both clusters are plotted in Figure \ref{fig:cluster_S2}.
While our conclusions are based on an unsupervised exploratory data analysis, closely related results can be found in the neurobiology literature \citep{cheng2019decreased, shaughnessy2020narrative, jiang2022associations}.

\section{Discussion}\label{sec:dis}

There are multiple directions for future work.
First, we are currently using a simple unstructured prior distribution for the model parameters.
Using structured priors such as low-rank representations for the top layer adjacency matrix could further reduce the dimension of model parameters.
Second, the sparsity of our model is enforced as a hard constraint.
Ideas of constraint relaxation could potentially give more flexible control of model sparsity as well as further improvement to the efficiency of MCMC.
Third, our identifiability theory tools can be potentially applied to a much broader range of models, encompassing more flexible data types and alternative latent variable structures.
Fourth, our current MCMC sampler involves a Gibbs sampler using all the data following a subsampling Gibbs sampler.
An interesting question is how to perform model selection by WAIC or WBIC using these approximate posterior samples.
Fifth, while our model is motivated from the application to brain connectome data, it is of immediate interest to consider applications to multiplex network data in ecological and social networks among other areas; in some of these cases the replicates have a temporal ordering, suggesting time series extensions of our model.
Finally, our model can be viewed as a deep generative model for multiplex networks, and there is essentially no conceptual or theoretical difficulty in using our model with a much larger $K$ (i.e., much deeper structure) than considered in this paper.
To this end, it remains to explore how to further scale up computation to adapt to truly deep model architectures.

\section*{Acknowledgements}
This work was partially supported by the National Institutes of Health (grant ID R01ES035625), by the European Research Council under the European Union’s Horizon 2020 research and innovation programme (grant agreement No 856506), by NSF grant DMS-2210796, by Office of Naval Research (ONR) projects N00014-21-1-2510 and N00014-24-1-2626.

\section*{Supplementary Material}

The Supplementary Material contains additional discussions and technical details supporting the proposed model and methodology. This includes further literature on multiplex networks, connections to related models, additional model illustrations, and extensions of prior specifications. We provide complete proofs for all theorems and propositions, along with the development of novel technical tools for establishing model identifiability.

Detailed descriptions of posterior computation are given, including the data-augmented Gibbs sampler, its subsampling variant, spectral initialization, and practical guidance on model selection. Additional theoretical results on posterior consistency are also included.

We further present comprehensive simulation studies, including sensitivity analysis, large-$p_k$ small-$N$ regimes, and comparisons with hierarchical community detection methods. Additional data analyses are provided, covering model selection, multi-resolution clustering of brain regions, relationships between inferred latent features and cognitive traits, and prediction of missing edges in multiplex networks. Implementation details and computational scalability are also discussed.

%% file: tikz/tikz_figure_intro.tex
\begin{figure}[ht]
\centering
\subfigure[
Population graphical model structure.
]{
\centering
\begin{tikzpicture}[scale = 0.9]
\node (n21)[layer2] at (0, 0.2) {};
\node (n22)[layer2] at (1, 0.2) {};
\node (n23)[layer2] at (2, 0.2) {};
\node (n24)[layer2] at (3, 0.2) {};
\node (n25)[layer2] at (4, 0.2) {};
\node (n26)[layer2] at (5, 0.2) {};
\node (n11)[layer1] at (1, 2.2) {};
\node (n12)[layer1] at (2, 2.2) {};
\node (n13)[layer1] at (3, 2.2) {};
\node (n14)[layer1] at (4, 2.2) {};
\node (n01)[layer0] at (2, 4.2) {};
\node (n02)[layer0] at (3, 4.2) {};
\draw (n21) edge[arrow] (n11);
\draw (n22) edge[arrow] (n11);
\draw (n23) edge[arrow] (n11);
\draw (n23) edge[arrow] (n12);
\draw (n24) edge[arrow] (n12);
\draw (n24) edge[arrow] (n13);
\draw (n24) edge[arrow] (n14);
\draw (n25) edge[arrow] (n13);
\draw (n25) edge[arrow] (n14);
\draw (n26) edge[arrow] (n14);
\draw (n11) edge[arrow] (n01);
\draw (n12) edge[arrow] (n01);
\draw (n13) edge[arrow] (n02);
\draw (n14) edge[arrow] (n02);
\draw[style = dashed] (-0.5, -0.6 + 0.2) -- (5.5, -0.6 + 0.2) -- (5.5, 0.6 + 0.2) -- (-0.5, 0.6 + 0.2) -- (-0.5, -0.6 + 0.2);
\draw[style = dashed] (-0.5, 1.4 + 0.2) -- (5.5, 1.4 + 0.2) -- (5.5, 2.6 + 0.2) -- (-0.5, 2.6 + 0.2) -- (-0.5, 1.4 + 0.2);
\draw[style = dashed] (-0.5, 3.4 + 0.2) -- (5.5, 3.4 + 0.2) -- (5.5, 4.6 + 0.2) -- (-0.5, 4.6 + 0.2) -- (-0.5, 3.4 + 0.2);
\node (l2)[anchor = west] at (5.7, 0 + 0.2) {\textcolor{black}{\small{Layer 2}}};
\node (l1)[anchor = west] at (5.7, 2 + 0.2) {\textcolor{black}{\small{Layer 1}}};
\node (l0)[anchor = west] at (5.7, 4 + 0.2) {\textcolor{black}{\small{Layer 0}}};
\node (dm)[anchor = west] at (5.7, -0.4) {\textcolor{white}{Dummy}};
\end{tikzpicture}
\label{fig:tikz_A_intro}
}
\hfill
\subfigure[
Individual adjacency matrices at each layer.
]{
\centering
\begin{tikzpicture}[inner sep = 0pt, scale = 0.9]
\tikzmath{ let \hdist = 2.3; let \vdist = 2;
for \x in {1,...,6}{ for \y in {1,...,6}{ if{
(\x == \y)
|| (\x < 4 && \y < 4)
|| (\x > 3 && \y > 3)
|| (\x == 3 && \y == 4) || (\x == 4 && \y == 3)
|| (\x == 2 && \y == 6) || (\x == 6 && \y == 2)
}then{
let \c = green!50;
}else{
let \c = black!10;
}; {
\fill[\c] (0.2 * \x, \vdist - 0.2 * \y - 0.5) rectangle + (0.19, 0.19);
}; }; };
for \x in {1,...,6}{ for \y in {1,...,6}{ if{
(\x == \y)
|| (\x < 3 && \y < 3)
|| (\x > 3 && \y > 3)
|| (\x == 3 && \y == 4) || (\x == 4 && \y == 3)
|| (\x == 2 && \y == 3) || (\x == 3 && \y == 2)
|| (\x == 2 && \y == 4) || (\x == 4 && \y == 2)
}then{
let \c = green!50;
}else{
let \c = black!10;
}; {
\fill[\c] (\hdist + 0.2 * \x, \vdist - 0.2 * \y - 0.5) rectangle + (0.19, 0.19);
}; }; };
for \x in {1,...,6}{ for \y in {1,...,6}{ if{
(\x == \y)
|| (\x < 4 && \y < 4)
|| (\x > 4 && \y > 4)
|| (\x == 3 && \y == 4) || (\x == 4 && \y == 3)
|| (\x == 4 && \y == 5) || (\x == 5 && \y == 4)
|| (\x == 2 && \y == 5) || (\x == 5 && \y == 2)
}then{
let \c = green!50;
}else{
let \c = black!10;
}; {
\fill[\c] (\hdist * 2 + 0.2 * \x, \vdist - 0.2 * \y - 0.5) rectangle + (0.19, 0.19);
}; }; };
for \x in {1,...,4}{ for \y in {1,...,4}{ if{
(\x == \y)
|| (\x < 3 && \y < 3)
|| (\x > 2 && \y > 2)
||(\x + \y == 5)
}then{
let \c = blue!50;
}else{
let \c = black!10;
}; {
\fill[\c] (0.3 * \x - 0.1, \vdist * 2 - 0.3 * \y - 0.5) rectangle + (0.29, 0.29);
}; }; };
for \x in {1,...,4}{ for \y in {1,...,4}{ if{
(\x == \y)
|| (\x < 3 && \y < 3)
|| (\x > 2 && \y > 2)
|| (\x == 1 && \y == 3) || (\x == 3 && \y == 1)
}then{
let \c = blue!50;
}else{
let \c = black!10;
}; {
\fill[\c] (\hdist + 0.3 * \x - 0.1, \vdist * 2 - 0.3 * \y - 0.5) rectangle + (0.29, 0.29);
}; }; };
for \x in {1,...,4}{ for \y in {1,...,4}{ if{
(\x == \y)
|| (\x < 3 && \y < 3)
|| (\x > 2 && \y > 2)
}then{
let \c = blue!50;
}else{
let \c = black!10;
}; {
\fill[\c] (\hdist * 2 + 0.3 * \x - 0.1, \vdist * 2 - 0.3 * \y - 0.5) rectangle + (0.29, 0.29);
}; }; };
for \x in {1, 2}{ for \y in {1, 2}{ if{
(\x == \y)
|| (\x + \y == 3)
}then{
let \c = red!50;
}else{
let \c = black!10;
}; {
\fill[\c] (0.6 * \x - 0.4, \vdist * 3 - 0.6 * \y - 0.5) rectangle + (0.59, 0.59);
}; }; };
for \x in {1, 2}{ for \y in {1, 2}{ if{
(\x == \y)
}then{
let \c = red!50;
}else{
let \c = black!10;
}; {
\fill[\c] (\hdist + 0.6 * \x - 0.4, \vdist * 3 - 0.6 * \y - 0.5) rectangle + (0.59, 0.59);
}; }; };
for \x in {1, 2}{ for \y in {1, 2}{ if{
(\x == \y)
}then{
let \c = red!50;
}else{
let \c = black!10;
}; {
\fill[\c] (\hdist * 2 + 0.6 * \x - 0.4, \vdist * 3 - 0.6 * \y - 0.5) rectangle + (0.59, 0.59);
}; }; };
}
\draw[style = dashed] (0 - 0.3, 1 - 0.8) -- (\hdist * 3 - 0.3, 1 - 0.8) -- (\hdist * 3 - 0.3, 1 + 0.7) -- (0 - 0.3, 1 + 0.7) -- (0 - 0.3, 1 - 0.8);
\draw[style = dashed] (0 - 0.3, 3 - 0.8) -- (\hdist * 3 - 0.3, 3 - 0.8) -- (\hdist * 3 - 0.3, 3 + 0.7) -- (0 - 0.3, 3 + 0.7) -- (0 - 0.3, 3 - 0.8);
\draw[style = dashed] (0 - 0.3, 5 - 0.8) -- (\hdist * 3 - 0.3, 5 - 0.8) -- (\hdist * 3 - 0.3, 5 + 0.7) -- (0 - 0.3, 5 + 0.7) -- (0 - 0.3, 5 - 0.8);
\node[anchor = west] at (\hdist * 3 + 0.2 - 0.3, 4.9) {\textcolor{black}{\small{Layer 0}}};
\node[anchor = west] at (\hdist * 3 + 0.2 - 0.3, 2.9) {\textcolor{black}{\small{Layer 1}}};
\node[anchor = west] at (\hdist * 3 + 0.2 - 0.3, 0.9) {\textcolor{black}{\small{Layer 2}}};
\node[anchor = west] at (\hdist * 0 - 0.22, 6.2) {\textcolor{black}{\small{Individual 1}}};
\node[anchor = west] at (\hdist * 1 - 0.22, 6.2) {\textcolor{black}{\small{Individual 2}}};
\node[anchor = west] at (\hdist * 2 - 0.22, 6.2) {\textcolor{black}{\small{Individual 3}}};
\draw[color = black!50] (\hdist * 1 - 0.3, 0) -- (\hdist * 1 - 0.3, \vdist * 3 + 0.5);
\draw[color = black!50] (\hdist * 2 - 0.3, 0) -- (\hdist * 2 - 0.3, \vdist * 3 + 0.5);
\end{tikzpicture}
\label{fig:tikz_X_intro}
}
\caption{
A simple illustration of our model.
(a) a graphical model with a three-layer population hierarchical structure, where layer 2 is observed and layers 0 and 1 are latent. 
(b) a sample of $N = 3$ individual networks with their own realizations in the three layers.
}
\label{fig:tikz_intro}
\end{figure}

%% file: tikz/tikz_figure_layer2.tex
\begin{figure}[ht]
\centering
\begin{tikzpicture}[scale = 0.8]
\node[anchor = west] at (-2.5, 0.5) {\small{Nodes}};
\node[anchor = west] at (-3.3, 2.5) {\small{Communities}};
\draw (-1.9, 0.7) edge[arrow] (-1.9, 2.3);
\node (n11)[layer11] at (0, 0.5) {};
\node (n12)[layer11] at (1, 0) {};
\node (n13)[layer11] at (1, 1) {};
\node (n14)[layer12] at (2, 0) {};
\node (n15)[layer12] at (2, 1) {};
\node (n16)[layer13] at (3, 0) {};
\node (n17)[layer13] at (3, 1) {};
\node (n18)[layer13] at (4, 0.5) {};
\draw (n11) -- (n12);
\draw (n11) -- (n13);
\draw (n12) -- (n13);
\draw (n14) -- (n15);
\draw (n16) -- (n17);
\draw (n17) -- (n18);
\draw[style = dashed] (-0.3, -0.3) -- (1.3, -0.3) -- (1.3, 1.3) -- (-0.3, 1.3) -- (-0.3, -0.3);
\draw[style = dashed] (1.7, -0.3) -- (2.3, -0.3) -- (2.3, 1.3) -- (1.7, 1.3) -- (1.7, -0.3);
\draw[style = dashed] (2.7, -0.3) -- (4.3, -0.3) -- (4.3, 1.3) -- (2.7, 1.3) -- (2.7, -0.3);
\node (c11)[layer01] at (0.5, 2.5) {};
\node (c12)[layer02] at (2, 2.5) {};
\node (c13)[layer03] at (3.5, 2.5) {};
\draw (0.5, 1.3) edge[arrow] (c11);
\draw (2, 1.3) edge[arrow] (c12);
\draw (3.5, 1.3) edge[arrow] (c13);
\node (n21)[layer11] at (0 + \tikzdist, 0.5) {};
\node (n22)[layer11] at (1 + \tikzdist, 0) {};
\node (n23)[layer11] at (1 + \tikzdist, 1) {};
\node (n24)[layer12] at (2 + \tikzdist, 0) {};
\node (n25)[layer12] at (2 + \tikzdist, 1) {};
\node (n26)[layer13] at (3 + \tikzdist, 0) {};
\node (n27)[layer13] at (3 + \tikzdist, 1) {};
\node (n28)[layer13] at (4 + \tikzdist, 0.5) {};
\draw (n21) -- (n22);
\draw (n22) -- (n23);
\draw (n24) -- (n25);
\draw (n26) -- (n27);
\draw (n26) -- (n28);
\draw (n27) -- (n28);
\draw[style = dashed] (-0.3 + \tikzdist, -0.3) -- (1.3 + \tikzdist, -0.3) -- (1.3 + \tikzdist, 1.3) -- (-0.3 + \tikzdist, 1.3) -- (-0.3 + \tikzdist, -0.3);
\draw[style = dashed] (1.7 + \tikzdist, -0.3) -- (2.3 + \tikzdist, -0.3) -- (2.3 + \tikzdist, 1.3) -- (1.7 + \tikzdist, 1.3) -- (1.7 + \tikzdist, -0.3);
\draw[style = dashed] (2.7 + \tikzdist, -0.3) -- (4.3 + \tikzdist, -0.3) -- (4.3 + \tikzdist, 1.3) -- (2.7 + \tikzdist, 1.3) -- (2.7 + \tikzdist, -0.3);
\node (c21)[layer01] at (0.5 + \tikzdist, 2.5) {};
\node (c22)[layer02] at (2 + \tikzdist, 2.5) {};
\node (c23)[layer03] at (3.5 + \tikzdist, 2.5) {};
\draw (0.5 + \tikzdist, 1.3) edge[arrow] (c21);
\draw (2 + \tikzdist, 1.3) edge[arrow] (c22);
\draw (3.5 + \tikzdist, 1.3) edge[arrow] (c23);
\node (n31)[layer11] at (0 + \tikzdist * 2, 0.5) {};
\node (n32)[layer11] at (1 + \tikzdist * 2, 0) {};
\node (n33)[layer11] at (1 + \tikzdist * 2, 1) {};
\node (n34)[layer12] at (2 + \tikzdist * 2, 0) {};
\node (n35)[layer12] at (2 + \tikzdist * 2, 1) {};
\node (n36)[layer13] at (3 + \tikzdist * 2, 0) {};
\node (n37)[layer13] at (3 + \tikzdist * 2, 1) {};
\node (n38)[layer13] at (4 + \tikzdist * 2, 0.5) {};
\draw (n31) -- (n32);
\draw (n31) -- (n33);
\draw (n32) -- (n33);
\draw (n34) -- (n35);
\draw (n36) -- (n37);
\draw (n36) -- (n38);
\draw (n37) -- (n38);
\draw[style = dashed] (-0.3 + \tikzdist * 2, -0.3) -- (1.3 + \tikzdist * 2, -0.3) -- (1.3 + \tikzdist * 2, 1.3) -- (-0.3 + \tikzdist * 2, 1.3) -- (-0.3 + \tikzdist * 2, -0.3);
\draw[style = dashed] (1.7 + \tikzdist * 2, -0.3) -- (2.3 + \tikzdist * 2, -0.3) -- (2.3 + \tikzdist * 2, 1.3) -- (1.7 + \tikzdist * 2, 1.3) -- (1.7 + \tikzdist * 2, -0.3);
\draw[style = dashed] (2.7 + \tikzdist * 2, -0.3) -- (4.3 + \tikzdist * 2, -0.3) -- (4.3 + \tikzdist * 2, 1.3) -- (2.7 + \tikzdist * 2, 1.3) -- (2.7 + \tikzdist * 2, -0.3);
\node (c31)[layer01] at (0.5 + \tikzdist * 2, 2.5) {};
\node (c32)[layer02] at (2 + \tikzdist * 2, 2.5) {};
\node (c33)[layer03] at (3.5 + \tikzdist * 2, 2.5) {};
\draw (0.5 + \tikzdist * 2, 1.3) edge[arrow] (c31);
\draw (2 + \tikzdist * 2, 1.3) edge[arrow] (c32);
\draw (3.5 + \tikzdist * 2, 1.3) edge[arrow] (c33);
\node[anchor = west] at (0.7, 3.5) {\small{Individual 1}};
\node[anchor = west] at (0.7 + \tikzdist, 3.5) {\small{Individual 2}};
\node[anchor = west] at (0.7 + \tikzdist * 2, 3.5) {\small{Individual 3}};
\draw[thick, style = dotted] (-0.4, -0.4) -- (4.4, -0.4) -- (4.4, 4) -- (-0.4, 4) -- (-0.4, -0.4);
\draw[thick, style = dotted] (-0.4 + \tikzdist, -0.4) -- (4.4 + \tikzdist, -0.4) -- (4.4 + \tikzdist, 4) -- (-0.4 + \tikzdist, 4) -- (-0.4 + \tikzdist, -0.4);
\draw[thick, style = dotted] (-0.4 + \tikzdist * 2, -0.4) -- (4.4 + \tikzdist * 2, -0.4) -- (4.4 + \tikzdist * 2, 4) -- (-0.4 + \tikzdist * 2, 4) -- (-0.4 + \tikzdist * 2, -0.4);
\draw (n12) -- (n14);
\draw (n12) -- (n15);
\draw (n13) -- (n15);
\draw (n15) -- (n17);
\draw (c11) -- (c12);
\draw (n23) -- (n25);
\draw (n24) -- (n26);
\draw (n25) -- (n26);
\draw (n25) -- (n27);
\draw (c22) -- (c23);
\draw (n32) -- (n34);
\draw (n32) -- (n35);
\draw (n33) -- (n35);
\draw (n34) -- (n36);
\draw (n35) -- (n36);
\draw (c31) -- (c32);
\draw (c32) -- (c33);
\end{tikzpicture}
\caption{
An illustration of node adjacency and community adjacency in three individuals, representing the two-layer version of our model with single community membership.
}
\label{fig:tikz_layer2}
\end{figure}

%% file: tikz/tikz_figure_layer3.tex
\begin{figure}[ht]
\centering
\subfigure[
Population Bayesian network $\cN$ and connection matrices $\Ab_k$ between layers $(k - 1), k$.
]{
\centering
\begin{tikzpicture}[scale = 0.9]
\node (n21)[layer2] at (0, 0) {\small{$\cN_{2, 1}$}};
\node (n22)[layer2] at (1, 0) {\small{$\cN_{2, 2}$}};
\node (n23)[layer2] at (2, 0) {\small{$\cN_{2, 3}$}};
\node (n24)[layer2] at (3, 0) {\small{$\cN_{2, 4}$}};
\node (n25)[layer2] at (4, 0) {\small{$\cN_{2, 5}$}};
\node (n26)[layer2] at (5, 0) {\small{$\cN_{2, 6}$}};
\node (n11)[layer1] at (1, 2) {\small{$\cN_{1, 1}$}};
\node (n12)[layer1] at (2, 2) {\small{$\cN_{1, 2}$}};
\node (n13)[layer1] at (3, 2) {\small{$\cN_{1, 3}$}};
\node (n14)[layer1] at (4, 2) {\small{$\cN_{1, 4}$}};
\node (n01)[layer0] at (2, 4) {\small{$\cN_{0, 1}$}};
\node (n02)[layer0] at (3, 4) {\small{$\cN_{0, 2}$}};
\draw (n21) edge[arrow] (n11);
\draw (n22) edge[arrow] (n11);
\draw (n23) edge[arrow] (n11);
\draw (n23) edge[arrow] (n12);
\draw (n24) edge[arrow] (n12);
\draw (n24) edge[arrow] (n13);
\draw (n24) edge[arrow] (n14);
\draw (n25) edge[arrow] (n13);
\draw (n25) edge[arrow] (n14);
\draw (n26) edge[arrow] (n14);
\draw (n11) edge[arrow] (n01);
\draw (n12) edge[arrow] (n01);
\draw (n13) edge[arrow] (n02);
\draw (n14) edge[arrow] (n02);
\draw[style = dashed] (-0.5, -0.6) -- (5.5, -0.6) -- (5.5, 0.6) -- (-0.5, 0.6) -- (-0.5, -0.6);
\draw[style = dashed] (-0.5, 1.4) -- (5.5, 1.4) -- (5.5, 2.6) -- (-0.5, 2.6) -- (-0.5, 1.4);
\draw[style = dashed] (-0.5, 3.4) -- (5.5, 3.4) -- (5.5, 4.6) -- (-0.5, 4.6) -- (-0.5, 3.4);
\node (l2)[anchor = west] at (5.7, 0) {\textcolor{black}{\small{Layer 2}}};
\node (l1)[anchor = west] at (5.7, 2) {\textcolor{black}{\small{Layer 1}}};
\node (l0)[anchor = west] at (5.7, 4) {\textcolor{black}{\small{Layer 0}}};
\draw (l2) edge[arrow] (l1);
\draw (l1) edge[arrow] (l0);
\node[anchor = west] at (6.5, 1) {\textcolor{black}{\small{$\Ab_2$}}};
\node[anchor = west] at (6.5, 3) {\textcolor{black}{\small{$\Ab_1$}}};
\node[anchor = west] at (0.5, 5) {\textcolor{black}{\small{Bayesian Network $\cN$}}};
\node (dm)[anchor = west] at (0.5, -0.55) {\textcolor{white}{Dummy}};
\end{tikzpicture}
\label{fig:tikz_A}
}
\hfill
\subfigure[
Individual adjacency matrices $\Xb_k^{(1:N)}$ at each layer $k$.
]{
\centering
\begin{tikzpicture}[inner sep = 0pt, scale = 0.9]
\tikzmath{ let \hdist = 2.3; let \vdist = 2;
for \x in {1,...,6}{ for \y in {1,...,6}{ if{
(\x == \y)
|| (\x < 4 && \y < 4)
|| (\x > 3 && \y > 3)
|| (\x == 3 && \y == 4) || (\x == 4 && \y == 3)
|| (\x == 2 && \y == 6) || (\x == 6 && \y == 2)
}then{
let \c = green!50;
}else{
let \c = black!10;
}; {
\fill[\c] (0.2 * \x, \vdist - 0.2 * \y - 0.5) rectangle + (0.19, 0.19);
}; }; };
for \x in {1,...,6}{ for \y in {1,...,6}{ if{
(\x == \y)
|| (\x < 3 && \y < 3)
|| (\x > 3 && \y > 3)
|| (\x == 3 && \y == 4) || (\x == 4 && \y == 3)
|| (\x == 2 && \y == 3) || (\x == 3 && \y == 2)
|| (\x == 2 && \y == 4) || (\x == 4 && \y == 2)
}then{
let \c = green!50;
}else{
let \c = black!10;
}; {
\fill[\c] (\hdist + 0.2 * \x, \vdist - 0.2 * \y - 0.5) rectangle + (0.19, 0.19);
}; }; };
for \x in {1,...,6}{ for \y in {1,...,6}{ if{
(\x == \y)
|| (\x < 4 && \y < 4)
|| (\x > 4 && \y > 4)
|| (\x == 3 && \y == 4) || (\x == 4 && \y == 3)
|| (\x == 4 && \y == 5) || (\x == 5 && \y == 4)
|| (\x == 2 && \y == 5) || (\x == 5 && \y == 2)
}then{
let \c = green!50;
}else{
let \c = black!10;
}; {
\fill[\c] (\hdist * 2 + 0.2 * \x, \vdist - 0.2 * \y - 0.5) rectangle + (0.19, 0.19);
}; }; };
for \x in {1,...,4}{ for \y in {1,...,4}{ if{
(\x == \y)
|| (\x < 3 && \y < 3)
|| (\x > 2 && \y > 2)
||(\x + \y == 5)
}then{
let \c = blue!50;
}else{
let \c = black!10;
}; {
\fill[\c] (0.3 * \x - 0.1, \vdist * 2 - 0.3 * \y - 0.5) rectangle + (0.29, 0.29);
}; }; };
for \x in {1,...,4}{ for \y in {1,...,4}{ if{
(\x == \y)
|| (\x < 3 && \y < 3)
|| (\x > 2 && \y > 2)
|| (\x == 1 && \y == 3) || (\x == 3 && \y == 1)
}then{
let \c = blue!50;
}else{
let \c = black!10;
}; {
\fill[\c] (\hdist + 0.3 * \x - 0.1, \vdist * 2 - 0.3 * \y - 0.5) rectangle + (0.29, 0.29);
}; }; };
for \x in {1,...,4}{ for \y in {1,...,4}{ if{
(\x == \y)
|| (\x < 3 && \y < 3)
|| (\x > 2 && \y > 2)
}then{
let \c = blue!50;
}else{
let \c = black!10;
}; {
\fill[\c] (\hdist * 2 + 0.3 * \x - 0.1, \vdist * 2 - 0.3 * \y - 0.5) rectangle + (0.29, 0.29);
}; }; };
for \x in {1, 2}{ for \y in {1, 2}{ if{
(\x == \y)
|| (\x + \y == 3)
}then{
let \c = red!50;
}else{
let \c = black!10;
}; {
\fill[\c] (0.6 * \x - 0.4, \vdist * 3 - 0.6 * \y - 0.5) rectangle + (0.59, 0.59);
}; }; };
for \x in {1, 2}{ for \y in {1, 2}{ if{
(\x == \y)
}then{
let \c = red!50;
}else{
let \c = black!10;
}; {
\fill[\c] (\hdist + 0.6 * \x - 0.4, \vdist * 3 - 0.6 * \y - 0.5) rectangle + (0.59, 0.59);
}; }; };
for \x in {1, 2}{ for \y in {1, 2}{ if{
(\x == \y)
}then{
let \c = red!50;
}else{
let \c = black!10;
}; {
\fill[\c] (\hdist * 2 + 0.6 * \x - 0.4, \vdist * 3 - 0.6 * \y - 0.5) rectangle + (0.59, 0.59);
}; }; };
}
\node[anchor = west] at (\hdist * 0 + 1.4, \vdist * 1 - 1.1) {\textcolor{black}{\small{$\Xb_2^{(1)}$}}};
\node[anchor = west] at (\hdist * 0 + 1.4, \vdist * 2 - 1.1) {\textcolor{black}{\small{$\Xb_1^{(1)}$}}};
\node[anchor = west] at (\hdist * 0 + 1.4, \vdist * 3 - 1.1) {\textcolor{black}{\small{$\Xb_0^{(1)}$}}};
\node[anchor = west] at (\hdist * 1 + 1.4, \vdist * 1 - 1.1) {\textcolor{black}{\small{$\Xb_2^{(2)}$}}};
\node[anchor = west] at (\hdist * 1 + 1.4, \vdist * 2 - 1.1) {\textcolor{black}{\small{$\Xb_1^{(2)}$}}};
\node[anchor = west] at (\hdist * 1 + 1.4, \vdist * 3 - 1.1) {\textcolor{black}{\small{$\Xb_0^{(2)}$}}};
\node[anchor = west] at (\hdist * 2 + 1.4, \vdist * 1 - 1.1) {\textcolor{black}{\small{$\Xb_2^{(3)}$}}};
\node[anchor = west] at (\hdist * 2 + 1.4, \vdist * 2 - 1.1) {\textcolor{black}{\small{$\Xb_1^{(3)}$}}};
\node[anchor = west] at (\hdist * 2 + 1.4, \vdist * 3 - 1.1) {\textcolor{black}{\small{$\Xb_0^{(3)}$}}};
\draw[style = dashed] (0, 1 - 0.8) -- (\hdist * 3, 1 - 0.8) -- (\hdist * 3, 1 + 0.7) -- (0, 1 + 0.7) -- (0, 1 - 0.8);
\draw[style = dashed] (0, 3 - 0.8) -- (\hdist * 3, 3 - 0.8) -- (\hdist * 3, 3 + 0.7) -- (0, 3 + 0.7) -- (0, 3 - 0.8);
\draw[style = dashed] (0, 5 - 0.8) -- (\hdist * 3, 5 - 0.8) -- (\hdist * 3, 5 + 0.7) -- (0, 5 + 0.7) -- (0, 5 - 0.8);
\node[anchor = west] at (\hdist * 3 + 0.2, 4.9) {\textcolor{black}{\small{Layer 0}}};
\node[anchor = west] at (\hdist * 3 + 0.2, 2.9) {\textcolor{black}{\small{Layer 1}}};
\node[anchor = west] at (\hdist * 3 + 0.2, 0.9) {\textcolor{black}{\small{Layer 2}}};
\node[anchor = west] at (\hdist * 0, 6.2) {\textcolor{black}{\small{Individual 1}}};
\node[anchor = west] at (\hdist * 1, 6.2) {\textcolor{black}{\small{Individual 2}}};
\node[anchor = west] at (\hdist * 2, 6.2) {\textcolor{black}{\small{Individual 3}}};
\draw[color = black!50] (\hdist * 1 - 0.06, 0) -- (\hdist * 1 - 0.08, \vdist * 3 + 0.5);
\draw[color = black!50] (\hdist * 2 - 0.06, 0) -- (\hdist * 2 - 0.08, \vdist * 3 + 0.5);
\end{tikzpicture}
\label{fig:tikz_X}
}
\caption{
An illustration of our model, where the Bayesian network $\cN$ in (a) has $(K + 1) = 3$ layers with $p_k = 2(k + 1)$ nodes in each layer $k$, and the sample in (b) has $N = 3$ individuals.
Connection matrices: $\mathbf{A}_1 = (1,0;~ 1,0;~ 0,1;~ 0,1)$, and $\mathbf{A}_2 = (1,0,0,0;~ 1,0,0,0;~ 1,1,0,0;~ 0,1,1,1;~ 0,0,1,1;~ 0,0,0,1)$.
}
\label{fig:tikz_layer3}
\end{figure}

%% file: supp_paper.tex
\paragraph{Notations.}
We use bold capital letters (e.g. $\Xb$) to denote matrices and tensors, bold letters (e.g. $\xb$) to denote vectors, and non-bold letters (e.g. $X, x$) to denote scalars.
Following \cite{kolda2009tensor}, we use $\Xb * \Yb$, $\Xb \otimes \Yb$, and $\Xb \odot \Yb$ to denote the Hadamard product, Kronecker product, and Khatri-Rao product, of two matrices $\Xb, \Yb$, and use $\vec(\cdot)$ to denote the vectorization of tensors.
For matrices $\Xb, \Yb$ of the same shape, we define inner product $\langle \Xb, \Yb \rangle := \tr(\Xb \Yb^\top)$.
For a vector $\vb$, we let $\diag(\vb)$ denote the diagonal matrix with diagonal entries being entries of $\vb$.
For a square matrix $\Mb$, we let $\diag(\Mb)$ denote the diagonal matrix obtained by setting all off-diagonal entries of $\Mb$ to zero, and let $\offdiag(\Mb) := \Mb - \diag(\Mb)$ denote the matrix obtained by setting all diagonal entries of $\Mb$ to zero.
We let $\eb_i$ denote the indicator vector with the $i$th entry equal to one and all other entries equal to zero.
The logit function is defined as $\logit(p) := \log\frac{p}{1 - p}$ for $p \in (0, 1)$.
We denote the cardinality of a set $\cB$ by $|\cB|$.
For a sequence of random variables $\{Z^{(n)}\}_{n = 1}^\infty$, we let $Z^{(N_1:N_2)}$ denote the collection $\{Z^{(n)}\}_{n = N_1}^{N_2}$ and let $Z^{(1:\infty)}$ denote the whole sequence.
For integer valued variables and integers $k_1, k_2 \in \ZZ$, we let $[k_1, k_2]$ denote the collection of integers $\{k_1, k_1 + 1, \ldots, k_2\}$.
For $k \in \ZZ$, we abbreviate $[k] := [1, k]$ and let $\langle k \rangle := \{(i, j):~ 1 \leq i < j \leq k\}$ denote the collection of all $\frac{k (k - 1)}{2}$ distinct pairs of integers in $[k]$.
We use $\cS^m$ to denote the $m$-dimensional probability simplex embedded in  $\RR^{m + 1}$, defined as $\cS^m := \{\bnu \in \RR^{m + 1}: ~ \nu_i \ge 0~ \forall i \in [m + 1],~ \one^\top \bnu = 1\}$, and let $\cS_+^m$ denote its interior $\{\bnu \in \RR^{m + 1}: ~ \nu_i > 0~ \forall i \in [m + 1],~ \one^\top \bnu = 1\}$.
For two measures $\upmu_1, \upmu_2$, we denote $\upmu_1$ absolutely continuous with respect to $\upmu_2$ by $\upmu_1 \ll \upmu_2$.

\section{Further Discussions}\label{sec:fur_disc}

In this section, we provide further discussions that are not included in the main paper due to space limitations.
We include more related literature in Section \ref{ssec:more_lit}, discuss the connection of our model with deep belief networks in Section \ref{ssec:dbn}, present a graphical illustration of the additive effects in our model formulation in Section \ref{ssec:add_model_ill}, provide more advanced ideas of designing prior distributions in Section \ref{ssec:adv_prior}, and discuss the practical benefit of potentially truncating the continuous parameters in Section \ref{ssec:trunc_param}.

\subsection{More Literature on Multiplex Networks}\label{ssec:more_lit}

In recent years, many works have emerged that study multiplex networks.
We refer readers broadly interested in this field to \citet{stanley2016clustering, reyes2016stochastic, le2018estimating, diquigiovanni2019analysis, lei2020consistent, paul2020random, arroyo2020simultaneous, jing2021community, mantziou2021bayesian, chen2022global, fan2022alma, young2022clustering, josephs2023nested}.

\subsection{Connection with Deep Belief Networks}\label{ssec:dbn}

There is an interesting connection between the proposed model and deep belief networks (DBNs) and their variants \citep{neal1992connectionist, hinton2006fast, gu2021bayesian}. Although DBNs are used to model simpler multivariate data rather than multiplex networks, our model and DBNs share a common feature in adopting multiple binary latent layers to model the hierarchical data generative process. 
Intuitively speaking, our model can be viewed as nuanced ``folding'' of the vector in each layer of a DBN to be a network adjacency matrix, so that observed adjacency matrices are generated from multiscale latent adjacency matrices hierarchically. Despite this resemblance, a distinguishing property of our model is \emph{sparsity} of the between layer directed edges.
In the context of multiplex networks, since the deeper latent nodes can be viewed as communities for shallower nodes in the graphical model, the sparsity between layers can be interpreted as each node belonging to a few, instead of all, deeper communities.
The between-layer sparsity not only enhances model interpretability and parsimony, but also is key to showing identifiability of our model, whereas DBNs typically assume fully connected edges between layers.

\subsection{Additional Model Illustrations}\label{ssec:add_model_ill}

We now provide some further intuition of our model through graphical illustrations.

Recall that our model has a Bayesian network structure given in Figure \ref{fig:tikz_A}, with the $k$th layer of the Bayesian network denoted by $\cN_k$ and the nodes at the $k$th layer denoted by $\cN_{k, i}$'s for $i \in [p_k]$.
As discussed in Section \ref{sec:model}, the formulation of our model \eqref{eq:model_entry}, \eqref{eq:model} is designed such that two nodes $\cN_{k, i}$ and $\cN_{k, j}$ in layer $\cN_k$ are more likely to be adjacent, i.e. $\PP\big( X_{k, i, j}^{(n)} = 1 | \Xb_{k - 1}^{(n)}, \Ab_k, \bTheta_k \big)$ is relatively large, if they are connected to the same node $\cN_{k - 1, s}$ in layer $\cN_{k - 1}$, i.e. $A_{k, i, s} = A_{k, j, s} = 1$, or if they are separately connected to adjacent nodes $\cN_{k - 1, s}$ and $\cN_{k - 1, t}$ in layer $\cN_{k - 1}$, i.e. $A_{k, i, s} = A_{k, j, t} = 1$ and $X_{k - 1, s, t}^{(n)} = 1$.
For simplicity, we refer to the latter event of $A_{k, i, s} = A_{k, j, t} = 1$ and $X_{k - 1, s, t}^{(n)} = 1$ as nodes $\cN_{k, i}$ and $\cN_{k, j}$ being \textit{upper-connected}.
The logit of the conditional probability of $X_{k, i, j}^{(n)} = 1$ in \eqref{eq:model_entry} can be rewritten as
$$
\logit~\PP\left(
X_{k, i, j}^{(n)} = 1
~\Big|~
\Xb_{k - 1}^{(n)}, \Ab_k, \bTheta_k
\right)
=
C_k + \sum_{s = 1}^{p_{k - 1}} \sum_{t = 1}^{p_{k - 1}} \Gamma_{k, s, t} A_{k, i, s} A_{k, j, t} X_{k - 1, s, t}^{(n)}
.
$$

When two nodes $\cN_{k, i}$ and $\cN_{k, j}$ are not upper-connected, the logit equals the baseline $C_k$, suggesting a small probability $\frac{\exp(C_k)}{1 + \exp(C_k)}$ of these two nodes being adjacent.
For each shared node $\cN_{k - 1, s}$ that both $\cN_{k, i}$ and $\cN_{k, j}$ are connected to, the logit is increased by $\Gamma_{k, s, s}$.
For each pair of adjacent nodes $\cN_{k - 1, s}$ and $\cN_{k - 1, t}$ through which nodes $\cN_{k, i}$ and $\cN_{k, j}$ are upper-connected, the logit is further increased by $\Gamma_{k, s, t}$.
A graphical illustration of this additive effect in the logit conditional probability is provided in Figure \ref{fig:additive_effects}.

\input{tikz/tikz_figure_additive}

\subsection{Advanced Ideas of Prior Distributions}\label{ssec:adv_prior}

While we have restricted our attention to the simple independent prior distributions in Section \ref{sec:post}, there are multiple more advanced ideas for designing prior distributions.
These ideas are briefly discussed in the following and left for future work.

Within our current model framework, the number of nodes $p_k$ at each layer of our Bayesian network is treated as a hyperparameter and fixed during model inference.
Each choice of $p_k (k \in [0, K])$ suggests a different model and we pick the best model through model selection approaches.
As discussed in Section \ref{sec:sim}, we recommend using WAIC \citep{watanabe2010asymptotic} or WBIC \citep{watanabe2013widely} as the information criteria for model selection.
An alternative to our model selection approach is to directly incorporate the $p_k$'s into our model as parameters.
We could potentially design our model in an over-complete formulation by choosing each $p_k$ large enough and adopt cumulative shrinkage priors to control the effective number of nodes at each layer $k$.
Such cumulative shrinkage priors are initially proposed for factor models \citep{bhattacharya2011sparse, legramanti2020bayesian} and have been extended to various models such as Bayesian networks \citep{gu2021bayesian}.

As discussed in Section \ref{sec:post}, we enforce sparsity in the connection matrices $\Ab$ through the hyperparameter $S$, which controls the maximum number of communities allowed for each node, or equivalently the maximum number of ones within each row of each $\Ab_k$.
Our Gibbs sampler treats each row of $\Ab_k$ as a block having support $\cV_k = \big\{ \vb \in \{0, 1\}^{p_{k - 1}}:~ 1 \le \one^\top \vb \le S \big\}$.
The set $\cV_k$ contains $\sum_{s = 1}^S \binom{p_{k - 1}}{s}$ binary vectors, each requiring an evaluation of the likelihood function during Gibbs sampling, which incurs a high computational cost when $S > 1$.
An alternative to this hard constrained sparsity is constraint relaxation \citep{duan2020bayesian}.
A constraint-relaxed posterior distribution is more smoothly varying, leading to potentially better mixing of Gibbs sampler and allowing us to sample each connection matrix $\Ab_k$ entrywise instead of rowwise, which could further reduce the computational complexity of the Gibbs sampler.

Our model formulation as given in Section \ref{sec:model} includes the parameter $\bnu \in \cS^{2^{\frac{p_0 (p_0 - 1)}{2}} - 1}$ that parameterizes the categorical distribution of $\Xb_0$ over $\cX_0$.
So far we have restricted attention to when the dimension $p_0$ is small ($p_0 \le 4$), partly for the sake of model interpretability and partly due to the quick growth in the dimension of $\bnu$ as $p_0$ increases.
When a larger $p_0$ is preferred, we could replace the categorical distribution of $\PP(\Xb_0 | \bnu)$ with more structured distributions.
For instance, we could let the upper-triangular, off-diagonal entries be conditionally independent given a joint one-dimensional parameter $v$, with the conditional distributions given by the Bernoulli distribution $\mathrm{Ber}(v)$.
A Beta distribution prior on $v$ can then be specified to preserve conditional conjugacy.

\subsection{Truncation of Continuous Parameters}\label{ssec:trunc_param}

In practice, we find that it is often helpful to truncate the unbounded distributions of the continuous parameters $C_k, \Gamma_{k, i, j}$'s to bounded intervals.
For instance, we could restrict each $C_k$ to the region $[\tilde{C}_-, \tilde{C}_+]$, each diagonal entry $\Gamma_{k, i, i}$ to the region $[\tilde\gamma_-, \tilde\gamma_+]$, and each off-diagonal entry $\Gamma_{k, i, j} (i \ne j)$ to the region $[\tilde\delta_-, \tilde\delta_+]$.

The lower and upper bounds $\tilde{C}_-$ and $\tilde{C}_+$ specify the range $\left[ \frac{\exp(\tilde{C}_-)}{1 + \exp(\tilde{C}_-)}, \frac{\exp(\tilde{C}_+)}{1 + \exp(\tilde{C}_+)} \right]$ allowed for the default probability of connection between two nodes that are not upper-connected, which could be chosen using prior knowledge of the problem.
For instance, if we want the default probability of connection to be at most 0.001, then we could choose $\tilde{C}_+ \approx -7$.

The lower bounds $\tilde\gamma_-$ and $\tilde\delta_-$ are small positive constants that could help improve the mixing of the Gibbs sampler.
Empirically, we have observed many local posterior modes around each $\Gamma_{k, i, j} = 0$.
Enforcing the positive lower bounds rule out these uninteresting local modes.

The upper bounds $\tilde\gamma_+$ and $\tilde\delta_+$ have some theoretical benefits.
For instance, the compactness of the parameter space allows application of Theorem \ref{theo:schwartz} (i) to all identifiable true parameters $(\Ab^*, \bTheta^*)$.
In practice, we find that when $\tilde\gamma_+, \tilde\delta_+$ are chosen to be large enough (e.g. $\tilde\gamma_+ = 20$, $\tilde\delta_+ = 10$), the difference between the truncated and non-truncated posterior distributions of the parameters $C_k, \Gamma_{k, i, j}$ are negligible.

\subsection{Extensions to Node Covariates and Weighted Edges}

The proposed model can be extended in several natural ways to incorporate additional information, such as node covariates or weighted edges.
While these extensions are conceptually straightforward at the modeling level, their identifiability properties differ from the binary-edge setting considered in this work and require additional assumptions to ensure theoretical guarantees.

\paragraph{Node covariates}
Node-level attributes can be incorporated by augmenting the generative mechanism for the latent adjacency matrices.
For example, one may allow the log-odds of an edge to depend on both the latent community interactions and observed node covariates through an additive or multiplicative structure.
This type of extension has been studied in related Bayesian network models, where covariates help explain heterogeneity beyond latent communities and often improve empirical clustering performance.
From an identifiability perspective, the presence of covariates modifies the parameterization of the link probabilities and can weaken or strengthen identifiability depending on how they interact with the latent structure.
In particular, identifiability typically requires that covariate effects and latent community effects do not compensate for each other in the likelihood; constraints such as centered covariate effects or identifiability restrictions on regression coefficients are commonly imposed to separate the latent and observed components.
A related set of identifiability arguments and sufficient conditions have been developed in recent work on covariate-dependent latent class models (for example, refer to  \cite{zhou2025bayesian}), and similar techniques would extend to the multiplex network setting considered here.

\paragraph{Weighted edges}
The model can also be generalized to settings where edges take values in $\RR$ or $\RR_+$, rather than being binary.
A simple extension replaces Bernoulli likelihoods with conditionally independent distributions appropriate for the edge weights, such as Gaussian, Poisson, or exponential families.
The latent adjacency matrices would then encode higher-level structural patterns that modulate the mean or dispersion of the weighted edges. 
Identifiability in this setting depends on the choice of edge-weight distribution.
When the conditional distribution belongs to a one-parameter exponential family with known variance structure, identifiability can often be established by adapting arguments for generalized linear models. 
However, weighted networks introduce a richer parameter space, and additional conditions, such as monotonicity or separation between community-specific parameters, are typically required to ensure that distinct latent structures induce distinguishable distributions on the network.
These ideas parallel existing identifiability results developed for weighted stochastic block models in recent literature.

Overall, both extensions are compatible with the framework of multilayer Bayesian networks and would broaden the applicability of the model.
A full identifiability theory for covariate-augmented or weighted-edge versions of the model is feasible but requires technical developments beyond the scope of the present work.
We view these directions as promising avenues for future research.

\subsection{Sparsity Level of Adjacency Matrices}\label{ssec:sparse_adjmat}

The proposed model, in its basic form, may generate relatively dense networks if the intercept parameter $C_k$ is kept fixed as the number of nodes $p_k$ grows.
In many real-world networks, however, the average degree remains bounded as $p_k \to \infty$, motivating the need to understand how sparsity can be incorporated into the model.

In our formulation, the edge probabilities at layer $k$ take the form
$$
\PP(X_{k, i, j} = 1 | \Xb_{k - 1}, \Ab_k, C_k, \bGamma_k)
=
\frac{1}{
1 + \exp\left(
-C_k - \ab_{k, i}^\top (\bGamma_k * \Xb_{k - 1}) \ab_{k, j}
\right)
}
.
$$
For node $i$ at layer $k$, we let $deg_k(i)$ denote its node degree in the adjacency matrix $\Xb_k$.
We show that the sparsity level of $\Xb_k$ can be be controlled through the scaling of $C_k$ as $p_k \to \infty$.

\begin{lemma}\label{lemm:node_deg}
As $p_k \to \infty$, suppose $C_k \asymp -\log p_k$ while $\bGamma_k$ remains fixed.
Then $\EE[deg_k(i)] = O(1)$.
\end{lemma}

\begin{proof}[Proof of Lemma \ref{lemm:node_deg}]
We have
\begin{align*}
\EE[deg_k(i)]
&=
\sum_{j = 1}^{p_k} \PP(X_{k, i, j} = 1 | \Xb_{k - 1}, \Ab_k, C_k, \bGamma_k)
\\&=
1 + \sum_{j \in [p_k], j \ne i} \frac{1}{
1 + \exp\left(
-C_k - \ab_{k, i}^\top (\bGamma_k * \Xb_{k - 1}) \ab_{k, j}
\right)
}
.
\end{align*}
Using the bound $\ab_{k, i}^\top (\bGamma_k * \Xb_{k - 1}) \ab_{k, j} \le \one^\top \bGamma_k \one$ with $\bGamma_k$ fixed, we obtain
$$
\EE[deg_k(i)]
\le
1 + \sum_{j \in [p_k], j \ne i} \frac{1}{
1 + \exp\left(
-C_k - \one^\top \bGamma_k \one
\right)
}
\lesssim
p_k \exp(C_k)
.
$$
Setting $C_k \asymp -\log p_k$ yields $\EE[\deg_k(i)] = O(1)$.
\end{proof}

This result shows that sparsity of adjacency matrices can be achieved in our framework without introducing an explicit multiplicative factor as in common network models.
Instead, the intercept $C_k$ plays an analogous role and can be scaled with $p_k$ to control the expected degree.
In particular, choosing $C_k \asymp -\log p_k$ ensures that the network operates in a sparse regime with bounded expected degree.

We note that this scaling also highlights a trade-off between sparsity and signal strength: as $C_k$ decreases, edges become rarer, which may require a larger number of network samples $N$ to maintain sufficient signal-to-noise ratio for reliable inference.
A more detailed investigation of this interplay is an interesting direction for future work.

\section{Identifiability Theory}\label{sec:iden}

We here present the proofs of Theorem \ref{theo:strict} on strict identifiability and Theorem \ref{theo:generic} on generic identifiability.
Our proof technique from tensor decomposition is overviewed in Section \ref{ssec:cp_kruskal}.
We provide the proof of Theorem \ref{theo:strict} in Section \ref{ssec:iden_strict}, using auxiliary lemmas proved in Section \ref{ssec:iden_strict_aux}, 
Additional proof techniques from complex analysis and geometric measure theory are introduced in Sections \ref{ssec:holo} and \ref{ssec:geom}.
We provide the proof of Theorem \ref{theo:generic} in Section \ref{ssec:iden_generic}, using auxiliary lemmas proved in Section \ref{ssec:iden_generic_aux}.
The necessity of our conditions for strict and generic identifiability is established in Section \ref{ssec:iden_nece}.

\subsection{CP Decomposition and Kruskal's Theorem}\label{ssec:cp_kruskal}

We start by introducing some definitions.
Let $\Yb \in \RR^{R_1 \times R_2 \times \cdots \times R_n}$ be a \emph{$n$-way tensor}, denoted entrywise by $Y_{i_1, i_2, \ldots, i_n}$ for each index $i_d \in [R_d]$.
For each $d \in [n]$, by fixing all indices other than $i_d$, the $R_d$-dimensional vector
$$
\big(
Y_{i_1, \ldots, i_{d - 1}, 1, i_{d + 1}, \ldots, i_n}
~
Y_{i_1, \ldots, i_{d - 1}, 2, i_{d + 1}, \ldots, i_n}
~
\cdots
~
Y_{i_1, \ldots, i_{d - 1}, R_d, i_{d + 1, \ldots, i_n}} \big)
\in
\RR^{R_d}
$$
is called a \emph{mode-$d$ fiber} of $\Yb$.
By varying each index $i_k (k \ne d)$ in $[R_k]$, $\Yb$ has a total of $\prod_{k \in [n], k \ne d} R_k$ mode-$d$ fibers.
For a matrix $\Ub \in \RR^{p \times R_d}$, we let $\Yb \times_d \Ub$ denote the $n$-way tensor of dimensions $R_1 \times \cdots R_{d - 1} \times p \times R_{d + 1} \times \cdots R_n$ obtained from linearly transforming each mode-$d$ fiber of $\Yb$ through the left multiplication of $\Ub$.
Entrywise we have
$$
(\Yb \times_d \Ub)_{i_1, \ldots, i_n}
:=
\sum_{j = 1}^{R_d} U_{i_d, j} Y_{i_1, \ldots, i_{d - 1}, j, i_{d + 1}, \ldots, i_n}
,\quad
\forall i_k \in [R_k]~ (k \ne d)
,~
i_d \in [p]
.
$$

We now restrict attention to the three-way tensors.
Let $\tilde\Ib \in \RR^{R \times R \times R}$ denote the three-way tensor with ones along its super-diagonal and zeros elsewhere such that its entry $\tilde{I}_{i_1, i_2, i_3} = 1_{i_1 = i_2 = i_3}$ for all $i_1, i_2, i_3 \in [R]$.
For matrices $\Ub_1 \in \RR^{p_1 \times R}, \Ub_2 \in \RR^{p_2 \times R}, \Ub_3 \in \RR^{p_3 \times R}$, we define
\begin{equation}\label{eq:cp_decom}
\llbracket \Ub_1, \Ub_2, \Ub_3 \rrbracket
:=
\tilde\Ib \times_1 \Ub_1 \times_2 \Ub_2 \times_3 \Ub_3
,
\end{equation}
which is a three-way tensor of dimensions $p_1 \times p_2 \times p_3$.
Given $\Yb \in \RR^{p_1 \times p_2 \times p_3}$ and $R$, finding these matrices $\Ub_1, \Ub_2, \Ub_3$ such that $\Yb = \llbracket \Ub_1, \Ub_2, \Ub_3 \rrbracket$ is known as the \emph{CP decomposition of rank $R$} \citep{kolda2009tensor}.

For a matrix $\Ab$, we let its \emph{Kruskal rank}, denoted by $\rank_K(\Ab)$, be the maximal number $k$ such that any $k$ columns of $\Ab$ are linearly independent.
In the following Theorem \ref{theo:kruskal}, we state the seminal result of \citet{kruskal1977three} that provides a condition under which the uniqueness of CP decomposition of rank $R$ is guaranteed up to column rescaling and permutation.

\begin{theorem}[Kruskal's Theorem]\label{theo:kruskal}
Let $\Ub_1 \in \RR^{p_1 \times R}, \Ub_2 \in \RR^{p_2 \times R}, \Ub_3 \in \RR^{p_3 \times R}$ be matrices satisfying the condition
$$
\rank_K(\Ub_1) + \rank_K(\Ub_2) + \rank_K(\Ub_3)
\ge
2R + 2
.
$$
Then for any matrices $\tilde\Ub_1 \in \RR^{p_1 \times R}, \tilde\Ub_2 \in \RR^{p_2 \times R}, \tilde\Ub_3 \in \RR^{p_3 \times R}$ such that
$$
\llbracket \Ub_1, \Ub_2, \Ub_3 \rrbracket
=
\llbracket \tilde\Ub_1, \tilde\Ub_2, \tilde\Ub_3 \rrbracket
,
$$
there exist diagonal matrices $\Qb_1, \Qb_2, \Qb_3 \in \RR^{R \times R}$ satisfying $\Qb_1 \Qb_2 \Qb_3 = \Ib$ and a permutation matrix $\Pb \in \RR^{R \times R}$ such that
$$
\tilde\Ub_1 = \Ub_1 \Qb_1 \Pb
,\quad
\tilde\Ub_2 = \Ub_2 \Qb_2 \Pb
,\quad
\tilde\Ub_3 = \Ub_3 \Qb_3 \Pb
.
$$
\end{theorem}

In simple terms, if a three-way tensor admits a CP decomposition $\llbracket \Ub_1, \Ub_2, \Ub_3 \rrbracket$ of rank $R$, with the Kruskal ranks of $\Ub_1, \Ub_2, \Ub_3$ summing to be $\ge 2R + 2$, then such a CP decomposition of rank $R$ is unique up to column rescaling and permutation of $\Ub_i$'s.
Proofs of Theorem \ref{theo:kruskal} can be found in \citet{kruskal1977three} and \citet{rhodes2010concise}.

\subsection{Proof of Theorem \ref{theo:strict}}\label{ssec:iden_strict}

Recall that $\cX_k$ denotes the collection of all $p_k \times p_k$ adjacency matrices.
The cardinality of $\cX_k$ is $|\cX_k| = 2^{\frac{p_k (p_k - 1)}{2}}$ due to the symmetry and the all-ones diagonal.
We distinguish the notations $\law$ and $\PP$.
For instance, $\law(\Xb_k | \Ab, \bTheta)$ denotes the categorical distribution of $\Xb_k | \Ab, \bTheta$ over $\cX_k$, whereas $\PP(\Xb_k | \Ab, \bTheta)$ denotes the probability for a specific $\Xb_k \in \cX_k$ under this distribution.

We let $\sX_k$ be an arbitrary bijective map from $\cX_k$ to the set of integers $[|\cX_k|] = \{1, 2, \ldots, |\cX_k|\}$.
We characterize the marginal distribution $\law(\Xb_k | \Ab, \bTheta)$ through a $|\cX_k|$-dimensional probability vector $\vb_k \in \cS^{|\cX_k| - 1}$, such that
\begin{equation}\label{eq:def_vk}
\PP(\Xb_k | \Ab, \bTheta)
=
v_{k, \sX_k(\Xb_k)}
\quad
\forall \Xb_k \in \cX_k
,
\end{equation}
for $v_{k, i}$ denoting the $i$th entry of $\vb_k$.
Similarly, for $X_{k, i, j}$ denoting the $(i, j)$th entry in $\Xb_k$, we characterize its conditional distribution $\law(X_{k, i, j} | \Xb_{k - 1}, \Ab_k, \bTheta_k)$ through a $(2 \times |\cX_{k - 1}|)$-dimensional matrix $\bLambda_{k, (i, j)}$, such that
\begin{equation}\label{eq:def_lambdakij}
\PP(X_{k, i, j} = 1_{\ell = 1} | \Xb_{k - 1}, \Ab, \bTheta)
=
\Lambda_{k, (i, j), \ell, \sX_{k - 1}(\Xb_{k - 1})}
\quad
\forall \Xb_{k - 1} \in \cX_{k - 1}
,
\ell \in \{1, 2\}
,
\end{equation}
for $\Lambda_{k, (i, j), \ell, r}$ denoting the $(\ell, r)$th entry of $\bLambda_{k, (i, j)}$.
That is, the first row of $\bLambda_{k, i, j}$ represents the probabilities of $X_{k, i, j} = 1$, the second row represents the probabilities of $X_{k, i, j} = 0$, and the $r$th column represents conditioning on $\Xb_{k - 1} = \sX_{k - 1}^{-1}(r)$.
Recall that $\langle p \rangle := \{(i, j):~ 1 \le i < j \le p\}$ denotes the collection of all pairs of distinct indices in $[p]$.
By taking Khatri-Rao product \citep{kolda2009tensor} of the matrices $\bLambda_{k, (i, j)}$ over all $(i, j) \in \langle p_k \rangle$, we obtain the $|\cX_k| \times |\cX_{k - 1}|$ matrix
\begin{equation}\label{eq:def_lambdak}
\bLambda_k
:=
\odot_{(i, j) \in \langle p_k \rangle} \bLambda_{k, (i, j)}
,
\end{equation}
which fully characterizes the conditional distribution $\law(\Xb_k | \Xb_{k - 1}, \Ab_k, \bTheta_k)$ due to the conditional independence of entries in $\Xb_k$.
Note that the order of pairs $(i, j)$'s in $\langle p_k \rangle$ when taking the Khatri-Rao product does not matter in our proof.

Noticing that
$$
\PP(\Xb_k | \Ab, \bTheta)
=
\sum_{\Xb_{k - 1} \in \cX_{k - 1}} \PP(\Xb_k | \Xb_{k - 1}, \Ab_k, \bTheta_k) \PP(\Xb_{k - 1} | \Ab, \bTheta)
,
$$
by viewing each $\vb_k$ as a column vector, for all $k \in [K]$ we have
\begin{equation}\label{eq:vk_lambdak_vk1}
\vb_k
=
\bLambda_k \vb_{k - 1}.
\end{equation}
Since $\vb_k$ and each column of $\bLambda_k$ are probability vectors, they satisfy the conditions $\vb_{k - 1}^\top \one = 1$ and $\bLambda_k^\top \one = \one$.
In the following lemma, we show that among the probability vectors and matrices satisfying these conditions, $\bLambda_k$ and $\vb_{k - 1}$ can be uniquely recovered from their product $\bLambda_k \vb_{k - 1}$ up to shuffles of columns in $\bLambda_k$ and entries in $\vb_{k - 1}$.

\begin{lemma}\label{lemm:strict_tensor}
Let parameter $(\Ab, \bTheta) \in \cA_1 \times \cT_0$ and layer $k \in [K]$.
For any matrix $\tilde\bLambda_k \in \RR^{|\cX_k| \times |\cX_{k - 1}|}$ and vector $\tilde\vb_{k - 1} \in \RR^{|\cX_{k - 1}| \times 1}$ satisfying
$$
\bLambda_k \vb_{k - 1}
=
\tilde\bLambda_k \tilde\vb_{k - 1}
,\quad
\tilde\vb_{k - 1}^\top \one
=
1
,\quad
\tilde\bLambda_k^\top \one
=
\one
,
$$
there exists a permutation matrix $\Pb \in \RR^{|\cX_{k - 1}| \times |\cX_{k - 1}|}$ such that
$$
\tilde\bLambda_k
=
\bLambda_k \Pb
,\quad
\tilde\vb_{k - 1}
=
\Pb^\top \vb_{k - 1}
.
$$
\end{lemma}

The proof of Lemma \ref{lemm:strict_tensor} is based on transforming $\vb_k = \bLambda_k \vb_{k - 1}$ into a three-way tensor and viewing $\bLambda_k$ and $\vb_{k - 1}$ as a realization of its CP decomposition.
The uniqueness of $\bLambda_k$ and $\vb_{k - 1}$ follows from Theorem \ref{theo:kruskal} after showing the condition on the Kruskal ranks to hold.
Similar proof techniques can be found in \citet{gu2021bayesian} for their Proposition 2.
We provide the proof of Lemma \ref{lemm:strict_tensor} in Section \ref{ssec:iden_strict_aux}.

Given the vector $\bLambda_k \vb_{k - 1}$, Lemma \ref{lemm:strict_tensor} states that the choice of $\bLambda_k$ and $\vb_{k - 1}$ are unique up to shuffles of columns/entries.
Importantly, there are many more ways to shuffle the columns in $\bLambda_k$ and entries in $\vb_{k - 1}$ than to permute the columns of $\Ab_k$ and the rows and columns of $\Xb_{k - 1}, \bGamma_k$.
The former corresponds to permutations over $\cX_{k - 1}$, whereas the latter corresponds to permutations over the $p_{k - 1}$ nodes in the $(k - 1)$th layer of our Bayesian network.
Given a shuffled probability vector $\vb_{k - 1}$ and matrix $\bLambda_k$ from Lemma \ref{lemm:strict_tensor}, we further want to know the $\Xb_{k - 1}$ that each entry of $\vb_{k - 1}$ and each column of $\bLambda_k$ corresponds to in \eqref{eq:def_vk} and \eqref{eq:def_lambdakij}.
That is, we want to recover the bijective map $\sX_{k - 1}$.
The following lemma establishes that $\sX_{k - 1}$ can be uniquely recovered from a column shuffled $\bLambda_k$ up to permutations over the $p_{k - 1}$ nodes.

\begin{lemma}\label{lemm:strict_sX}
Let parameters $(\Ab, \bTheta), (\tilde\Ab, \tilde\bTheta) \in \cA_1 \times \cT_0$ and layer $k \in [K]$.
For any two bijective maps $\sX_{k - 1}, \tilde\sX_{k - 1}$ from $\cX_{k - 1}$ to $[|\cX_{k - 1}|]$, if for all $x \in [|\cX_{k - 1}|]$
$$
\law\big(
\Xb_k | \Xb_{k - 1} = \sX_{k - 1}^{-1}(x), \Ab_k, \bTheta_k
\big)
=
\law\big(
\Xb_k | \Xb_{k - 1} = \tilde\sX_{k - 1}^{-1}(x), \tilde\Ab_k, \tilde\bTheta_k
\big)
,
$$
then there exists a permutation matrix $\Pb \in \RR^{p_{k - 1} \times p_{k - 1}}$ such that for all $x \in [|\cX_{k - 1}|]$,
$$
\tilde\sX_{k - 1}^{-1}(x)
=
\Pb \sX_{k - 1}^{-1}(x) \Pb^\top
.
$$
\end{lemma}

Lemma \ref{lemm:strict_sX} states that, for any two bijective maps $\sX_{k - 1}, \tilde\sX_{k - 1}$ with possibly different model parameters, if they lead to the same probability matrix $\bLambda_k$, then the adjacency matrices $\Xb_{k - 1} = \sX_{k - 1}^{-1}(x)$ and $\tilde\Xb_{k - 1} = \tilde\sX_{k - 1}^{-1}(x)$ corresponding to the $x$th column of $\bLambda_k$ must be identical up to node permutations.
This implies that, given a column shuffled $\bLambda_k$ with model parameter $(\Ab, \bTheta)$ unknown, we can uniquely recover the bijective map $\sX_{k - 1}$ up to node permutations.
The proof is provided in Section \ref{ssec:iden_strict_aux}.

After recovering $\sX_{k - 1}$ from the probability matrix $\bLambda_k$, we know everything of the conditional distribution $\law(\Xb_k | \Xb_{k - 1}, \Ab_k, \bTheta_k)$.
It then remains to recover the model parameters $\Ab_k$ and $\bTheta_k$ from this conditional distribution.
The following lemma shows that the recovery of $\Ab_k, \bTheta_k$ is unique for $(\Ab, \bTheta)$ within the parameter space $\Ab_1 \times \cT_0$, with proof provided in Section \ref{ssec:iden_strict_aux}.

\begin{lemma}\label{lemm:strict_param}
Let parameters $(\Ab, \bTheta), (\tilde\Ab, \tilde\bTheta) \in \cA_1 \times \cT_0$ and layer $k \in [K]$.
If for all $\Xb_{k - 1} \in \cX_{k - 1}$
$$
\law(\Xb_k | \Xb_{k - 1}, \Ab_k, \bTheta_k)
=
\law(\Xb_k | \Xb_{k - 1}, \tilde\Ab_k, \tilde\bTheta_k)
,
$$
then
$$
\tilde\Ab_k
=
\Ab_k
,\quad
\tilde{C}_k
=
C_k
,\quad
\tilde\bGamma_k
=
\bGamma_k
.
$$
\end{lemma}

Lemmas \ref{lemm:strict_tensor}, \ref{lemm:strict_sX}, and \ref{lemm:strict_param} together provide a three-step procedure for recovering the parameters $\Ab_k, \bTheta_k$ and the probability vector $\vb_{k - 1}$ from $\vb_k = \bLambda_k \vb_{k - 1}$.
We now use an induction argument for $k = K, K - 1, \ldots, 1$ to sequentially recover all parameters $\Ab, \bTheta$ from the marginal distribution $\law(\Xb_K | \Ab, \bTheta)$ (or equivalently from the vector $\vb_K$) and prove Theorem \ref{theo:strict}.

\begin{proof}[Proof of Theorem \ref{theo:strict}]
We let parameter $(\Ab, \bTheta) \in \cA_1 \times \cT_0$ and pick an arbitrary $(\tilde\Ab, \tilde\bTheta) \in \marg_{\cA_1 \times \cT_0}(\Ab, \bTheta)$, then by definition \eqref{eq:pi_def},
\begin{equation}\label{eq:strict_theo_marg}
\law(\Xb_K | \Ab, \bTheta)
=
\law(\Xb_K | \tilde\Ab, \tilde\bTheta)
.
\end{equation}

We prove by induction for $k \in [K]$ in descending order.
Let $k \in [K]$, $\Pb_k \in \RR^{p_k \times p_k}$ be a permutation matrix, and $\sX_k$ be a bijective map from $\cX_k$ to $[|\cX_k|]$.
We also fix an arbitrary bijective map $\sX_{k - 1}$ from $\cX_{k - 1}$ to $[|\cX_{k - 1}|]$.
We let $\sN_k$ be the permutation map on the integers $[|\cX_k|]$ such that for all $x \in [|\cX_k|]$,
\begin{equation}\label{eq:strict_theo_sNk}
(\sX_k^{-1} \circ \sN_k)(x)
:=
\Pb_k \sX_k^{-1}(x) \Pb_k^\top
.
\end{equation}
Accordingly, we let $\Nb_k$ denote the permutation matrix corresponding to $\sN_k$, such that
$$
\big(
1 ~ 2 ~ \cdots ~ |\cX_k|
\big) \Nb_k
=
\big(
\sN_k(1) ~ \sN_k(2) ~ \cdots ~ \sN_k(|\cX_k|)
\big)
.
$$
Similar to \eqref{eq:def_vk}, \eqref{eq:def_lambdakij}, and \eqref{eq:def_lambdak}, we define the probability vectors $\vb_{k - 1}, \tilde\vb_{k - 1} \in \cS^{|\cX_{k - 1}| - 1}$ to characterize the marginal distributions $\law(\Xb_{k - 1} | \Ab, \bTheta)$ and $\law(\Xb_{k - 1} | \tilde\Ab, \tilde\bTheta)$ such that their entries satisfy
$$
v_{k - 1, \sX_{k - 1}(\Xb_{k - 1})}
=
\PP(\Xb_{k - 1} | \Ab, \bTheta)
,\quad
\tilde{v}_{k - 1, \sX_{k - 1}(\Xb_{k - 1})}
=
\PP(\Xb_{k - 1} | \tilde\Ab, \tilde\bTheta)
$$
for all $\Xb_{k - 1} \in \cX_{k - 1}$;
we define the probability matrices $\bLambda_k, \tilde\bLambda_k \in \RR^{|\cX_k| \times |\cX_{k - 1}|}$ to characterize the conditional distributions $\law(\Xb_k | \Xb_{k - 1}, \Ab_k, \bTheta_k)$ and $\law(\Xb_k | \Xb_{k - 1}, \tilde\Ab_k, \tilde\bTheta_k)$ such that their entries satisfy
$$
\Lambda_{k, \sX_k(\Xb_k), \sX_{k - 1}(\Xb_{k - 1})}
=
\PP(\Xb_k | \Xb_{k - 1}, \Ab, \bTheta)
,\quad
\tilde\Lambda_{k, \sX_k(\Xb_k), \sX_{k - 1}(\Xb_{k - 1})}
=
\PP(\Xb_k | \Xb_{k - 1}, \tilde\Ab, \tilde\bTheta)
$$
for all $\Xb_k \in \cX_k$ and $\Xb_{k - 1} \in \cX_{k - 1}$;
and we define the probability vectors $\vb_k, \tilde\vb_k \in \cS^{|\cX_k| - 1}$ to characterize the marginal distributions $\law(\Xb_k | \Ab, \bTheta)$ and $\law(\Xb_k | \tilde\Ab, \tilde\bTheta)$ such that their entries satisfy
$$
v_{k, \sX_k(\Xb_k)}
=
\PP(\Xb_k | \Ab, \bTheta)
,\quad
\tilde{v}_{k, \sX_k(\Xb_k)}
=
\PP(\Xb_k | \tilde\Ab, \tilde\bTheta)
$$
for all $\Xb_k \in \cX_k$.
Similar to \eqref{eq:vk_lambdak_vk1}, we have
$$
\vb_k = \bLambda_k \vb_{k - 1}
,\quad
\tilde\vb_k = \tilde\bLambda_k \tilde\vb_{k - 1}
.
$$

To do the induction, we suppose for layer $k$ that $\vb_k = \Nb_k^\top \tilde\vb_k$.
By Lemma \ref{lemm:strict_tensor}, there exists a permutation matrix $\Nb_{k - 1} \in \RR^{|\cX_{k - 1}| \times |\cX_{k - 1}|}$ such that
\begin{equation}\label{eq:strict_theo_step1}
\bLambda_k
=
\Nb_k^\top \tilde\bLambda_k \Nb_{k - 1}
,\quad
\vb_{k - 1}
=
\Nb_{k - 1}^\top \tilde\vb_{k - 1}
.
\end{equation}
We let $\sN_{k - 1}$ denote the permutation map on $[|\cX_{k - 1}|]$ corresponding to the permutation matrix $\Nb_{k - 1}$, such that
$$
\big(
1 ~ 2 ~ \cdots ~ |\cX_{k - 1}|
\big) \Nb_{k - 1}
=
\big(
\sN_{k - 1}(1) ~ \sN_{k - 1}(2) ~ \cdots ~ \sN_{k - 1}(|\cX_{k - 1}|)
\big)
.
$$
For simplicity, we slightly abuse notation and also let $\sN_{k - 1}$ denote the permutation map $\sX_{k - 1}^{-1} \circ \sN_{k - 1} \circ \sX_{k - 1}$ on $\cX_{k - 1}$ and $\sN_k$ denote the permutation map $\sX_k^{-1} \circ \sN_k \circ \sX_k$ on $\cX_k$.
Then \eqref{eq:strict_theo_step1} implies, for all $\Xb_k \in \cX_k$ and $\Xb_{k - 1} \in \cX_{k - 1}$,
\begin{equation}\label{eq:strict_theo_PsNkk1}
\PP(\Xb_k | \Xb_{k - 1}, \Ab_k, \bTheta_k)
=
\PP(\sN_k(\Xb_k) | \sN_{k - 1}(\Xb_{k - 1}), \tilde\Ab_k, \tilde\bTheta_k)
.
\end{equation}
By definition of $\sN_k$ in \eqref{eq:strict_theo_sNk} and our model formulation in \eqref{eq:model_entry} and \eqref{eq:model}, we have
\begin{align*}
\PP(\sN_k(\Xb_k) | \sN_{k - 1}(\Xb_{k - 1}), \tilde\Ab_k, \tilde\bTheta_k)
&=
\PP(\Pb_k \Xb_k \Pb_k^\top | \sN_{k - 1}(\Xb_{k - 1}), \tilde\Ab_k, \tilde\bTheta_k)
\\&=
\PP(\Xb_k | \sN_{k - 1}(\Xb_{k - 1}), \Pb_k^\top \tilde\Ab_k, \tilde\bTheta_k)
,
\end{align*}
which together with \eqref{eq:strict_theo_PsNkk1} gives for all $x \in [|\cX_{k - 1}|]$
$$
\law(\Xb_k | \sX_{k - 1}^{-1}(x), \Ab_k, \bTheta_k)
=
\law(\Xb_k | (\sX_{k - 1}^{-1} \circ \sN_{k - 1})(x), \Pb_k^\top \tilde\Ab_k, \tilde\bTheta_k)
.
$$
By Lemma \ref{lemm:strict_sX}, there exists a permutation matrix $\Pb_{k - 1} \in \RR^{p_{k - 1} \times p_{k - 1}}$, such that for all $x \in [|\cX_{k - 1}|]$,
\begin{equation}\label{eq:strict_theo_sNk1}
(\sX_{k - 1}^{-1} \circ \sN_{k - 1})(x)
=
\Pb_{k - 1} \sX_{k - 1}^{-1}(x) \Pb_{k - 1}^\top
.
\end{equation}
This implies, for all $\Xb_{k - 1} \in \cX_{k - 1}$,
\begin{equation}\label{eq:strict_theo_step2}
\law\big(
\Xb_k | \Xb_{k - 1}, \Ab_k, \bTheta_k
\big)
=
\law\big(
\Xb_k | \Pb_{k - 1} \Xb_{k - 1} \Pb_{k - 1}^\top, \Pb_k^\top \tilde\Ab_k, \tilde\bTheta_k
\big)
.
\end{equation}
Using our model formulation in \eqref{eq:model_entry}, for each $(i, j) \in \langle p_k \rangle$ we have
\begin{align*}
&\quad~
\PP(X_{k, i, j} | \Xb_{k - 1}, \Ab_k, \bTheta_k)
\\&=
\frac{\exp\left( X_{k, i, j} \left(
C_k + \ab_{k, i}^\top (\bGamma_k * \Xb_{k - 1}) \ab_{k, j}
\right) \right)}{1 + \exp\left(
C_k + \ab_{k, i}^\top (\bGamma_k * \Xb_{k - 1}) \ab_{k, j}
\right)}
\\&=
\frac{\exp\left( X_{k, i, j} \left(
C_k + (\ab_{k, i}^\top \Pb_{k - 1}^\top) \big((\Pb_{k - 1} \bGamma_k \Pb_{k - 1}^\top) * (\Pb_{k - 1} \Xb_{k - 1} \Pb_{k - 1}^\top)\big) (\Pb_{k - 1} \ab_{k, j})
\right) \right)}{1 + \exp\left(
C_k + (\ab_{k, i}^\top \Pb_{k - 1}^\top) \big((\Pb_{k - 1} \bGamma_k \Pb_{k - 1}^\top) * (\Pb_{k - 1} \Xb_{k - 1} \Pb_{k - 1}^\top)\big) (\Pb_{k - 1} \ab_{k, j})
\right)}
\\&=
\PP(X_{k, i, j} | \Pb_{k - 1} \Xb_{k - 1} \Pb_{k - 1}^\top, \Ab_k \Pb_{k - 1}^\top,~ \{C_k, \Pb_{k - 1} \bGamma_k \Pb_{k - 1}^\top\})
.
\end{align*}
Combined with \eqref{eq:strict_theo_step2}, we obtain for all $\Xb_{k - 1} \in \cX_{k - 1}$
$$
\law(\Xb_k | \Xb_{k - 1},~ \Ab_k \Pb_{k - 1}^\top,~ \{C_k, \Pb_{k - 1} \bGamma_k \Pb_{k - 1}^\top\})
=
\law(\Xb_k | \Xb_{k - 1}, \Pb_k^\top \tilde\Ab_k, \tilde\bTheta_k)
.
$$
By Lemma \ref{lemm:strict_param}, this implies
\begin{equation}\label{eq:strict_theo_step3}
\tilde\Ab_k
=
\Pb_k \Ab_k \Pb_{k - 1}^\top
,\quad
\tilde{C}_k
=
C_k
,\quad
\tilde\bGamma_k
=
\Pb_{k - 1} \bGamma_k \Pb_{k - 1}^\top
.
\end{equation}

The induction passes on from layer $k$ to layer $k - 1$, since $\vb_{k - 1} = \Nb_{k - 1}^\top \tilde\vb_{k - 1}$ follows from \eqref{eq:strict_theo_step1} and the permutation matrix $\Nb_{k - 1}$ satisfies \eqref{eq:strict_theo_sNk1}.
It remains to show that the induction can start at layer $k = K$.
We recall from \eqref{eq:sigma_def} that we have specified $\Pb_K := \Ib_{p_K}$.
For the permutation matrix $\Nb_K := \Ib_{|\cX_K|}$ and the permutation map $\sN_K := \mathrm{Id}_{[|\cX_K|]}$ being the identity map, we find that
$$
(\sX_K^{-1} \circ \sN_K)(x)
=
\Pb_K \sX_K^{-1}(x) \Pb_K^\top
$$
holds for all $x \in [|\cX_K|]$ and \eqref{eq:strict_theo_marg} is equivalent to
$$
\vb_K
=
\tilde\vb_K
=
\Nb_K^\top \tilde\vb_K
.
$$
This shows that the induction holds for $k = K$.
Therefore, through induction we have shown that \eqref{eq:strict_theo_step3} holds for all $k \in [K]$.
Additionally, we have $\vb_0 = \Nb_0^\top \tilde\vb_0$ for $\Nb_0$ and its associated $\sN_0$ satisfying \eqref{eq:strict_theo_sNk1} with $k = 1$, which implies
$$
\PP(\Xb_0 | \bnu)
=
v_{0, \sX_0(\Xb_0)}
=
\tilde{v}_{0, (\sN_0 \circ \sX_0)(\Xb_0)}
=
\tilde{v}_{0, (\sX_0^{-1} \circ \sN_0)(\Xb_0)}
=
\tilde{v}_{0, \Pb_0 \Xb_0 \Pb_0^\top}
=
\PP(\Pb_0 \Xb_0 \Pb_0^\top | \tilde\bnu)
.
$$
All above results together establish that $(\tilde\Ab, \tilde\bTheta) \in \perm_{\cA_1 \times \cT_0}(\Ab, \bTheta)$.
This suggests
$$
\perm_{\cA_1 \times \cT_0}(\Ab, \bTheta)
\supset
\marg_{\cA_1 \times \cT_0}(\Ab, \bTheta)
$$
for all $(\Ab, \bTheta) \in \cA_1 \times \cT_0$.
Since by definition $\perm_{\cA_1 \times \cT_0}(\Ab, \bTheta) \subset \marg_{\cA_1 \times \cT_0}(\Ab, \bTheta)$ holds, we have proven that our model with parameter space $\cA_1 \times \cT_0$ is strictly identifiable.
\end{proof}

\subsection{Proof of Auxiliary Lemmas for Theorem \ref{theo:strict}}\label{ssec:iden_strict_aux}

\begin{proof}[Proof of Lemma \ref{lemm:strict_tensor}]
Recalling the definition of parameter space $\cA_1$ in \eqref{eq:A1_def}, the matrix $\Ab_k$ satisfies the blockwise condition
\begin{equation}\label{eq:A_block}
\Pb_k \Ab_k
=
\big(
\Ib_{p_{k - 1}} ~ \Ib_{p_{k - 1}} ~ \Bb_k
\big)^\top
\end{equation}
for $\Pb_k \in \RR^{p_k \times p_k}$ a permutation matrix and $\Bb_k \in \RR^{p_{k - 1} \times (p_k - 2p_{k - 1})}$ an arbitrary matrix.
Without loss of generality, we can let $\Pb_k = \Ib_{p_k}$.
For a general layer $k$, the only probability vector and matrix involved in Lemma \ref{lemm:strict_tensor} are $\vb_{k - 1}$ and $\bLambda_k$, so we abbreviate them by $\vb$ and $\bLambda$ for notational simplicity.
Similarly, we abbreviate $\tilde\bLambda_k$ by $\tilde\bLambda$, $\tilde\vb_{k - 1}$ by $\tilde\vb$, and the matrix $\bLambda_{k, (i, j)}$ defined in \eqref{eq:def_lambdakij} by $\bLambda_{(i, j)}$.

We define four sets of index pairs in $\langle p_k \rangle$ as
\begin{equation}\label{eq:def_I1234}\begin{aligned}
&
\cI_1
:=
\big\{
(i, j):~
1 \le i < j \le p_{k - 1}
\big\}
,\\&
\cI_2
:=
\big\{
(i, j):~
p_{k - 1} + 1 \le i < j \le 2p_{k - 1}
\big\}
,\\&
\cI_3
:=
\big\{
(i, j):~
1 \le i \le p_{k - 1}
,~
p_{k - 1} + 1 \le j \le 2p_{k - 1}
\big\}
,\\&
\cI_4
:=
\big\{
(i, j):~
2p_{k - 1} + 1 \le i < j \le p_k
\big\}
,
\end{aligned}\end{equation}
such that together they form a partition of $\langle p_k \rangle$.
This partition of $\langle p_k \rangle$ corresponds to a partition of the upper-triangular off-diagonal entries of the adjacency matrix $\Xb_k$ into four blocks.
Similar to \eqref{eq:def_lambdak}, we characterize the conditional distributions of each block as
\begin{equation}\label{eq:def_lambda1234}\begin{aligned}
&
\bLambda_{[1]}
:=
\odot_{(i, j) \in \cI_1} \bLambda_{(i, j)}
\in
\RR^{|\cX_{k - 1}| \times |\cX_{k - 1}|}
,\\&
\bLambda_{[2]}
:=
\odot_{(i, j) \in \cI_2} \bLambda_{(i, j)}
\in
\RR^{|\cX_{k - 1}| \times |\cX_{k - 1}|}
,\\&
\bLambda_{[3]}
:=
\odot_{(i, j) \in \cI_3} \bLambda_{(i, j)}
\in
\RR^{2^{p_{k - 1}^2} \times |\cX_{k - 1}|}
,\\&
\bLambda_{[4]}
:=
\odot_{(i, j) \in \cI_4} \bLambda_{(i, j)}
\in
\RR^{2^{\big(\frac{p_k (p_k - 1)}{2} - p_{k - 1}(2p_{k - 1} - 1)\big)} \times |\cX_{k - 1}|}
.
\end{aligned}\end{equation}
In the case $p_k = 2p_{k - 1}$, the number of entries in the fourth block is $\frac{p_k (p_k - 1)}{2} - p_{k - 1}(2p_{k - 1} - 1) = 0$, so we replace the definition of $\bLambda_{[4]}$ by $\bLambda_{[4]} = \one^\top \in \RR^{1 \times |\cX_{k - 1}|}$.

Following the definition \eqref{eq:def_lambda1234}, by appropriately reordering the rows of $\bLambda$, we have $\bLambda = \odot_{1 \le \ell \le 4} \bLambda_{[\ell]}$.
Recall that $\vec(\cdot)$ denotes the vectorization of a tensor or matrix and $\llbracket \cdot, \cdot, \cdot \rrbracket$ is the three-way operator defined in \eqref{eq:cp_decom}.
By appropriately reordering the rows in $\bLambda$, we further have
\begin{equation}\label{eq:strict_tensor_cp}
\bLambda \vb
=
\vec\big( \big\llbracket
\bLambda_{[1]}, \bLambda_{[2]}, (\bLambda_{[3]} \odot \bLambda_{[4]}) \diag(\vb)
\big\rrbracket \big)
,
\end{equation}
which bridges the connection between Lemma \ref{lemm:strict_tensor} and the CP decomposition of three-way tensors.
To apply Theorem \ref{theo:kruskal} on the CP decomposition uniqueness to the three-way tensor $\big\llbracket \bLambda_{[1]}, \bLambda_{[2]}, (\bLambda_{[3]} \odot \bLambda_{[4]}) \diag(\vb) \big\rrbracket$, we need the Kruskal's condition
\begin{equation}\label{eq:strict_tensor_kruskal}
\rank_K\big(
\bLambda_{[1]}
\big)
+
\rank_K\big(
\bLambda_{[2]}
\big)
+
\rank_K\big(
(\bLambda_{[3]} \odot \bLambda_{[4]}) \diag(\vb)
\big)
\ge
2 |\cX_{k - 1}| + 2
.
\end{equation}
We now study the invertibility of $\bLambda_{[1]}, \bLambda_{[2]}$ and the pairwise linear independence of the columns in $\bLambda_{[3]}$, which together could ensure \eqref{eq:strict_tensor_kruskal}, as shown in the following.

We first show that $\bLambda_{[1]}$ is invertible, so that its Kruskal rank is $|\cX_{k - 1}|$.
We denote $\lambda^\star := e^{C_k} / (1 + e^{C_k})$ and define for each $(i, j) \in \cI_1$ the matrix
$$
\bLambda_{(i, j)}'
:=
\left(\begin{matrix}
1 & -\lambda^\star \\
0 & 1
\end{matrix}\right)
\left(\begin{matrix}
1 & 0 \\
1 & 1
\end{matrix}\right)
\bLambda_{(i, j)}
,
$$
which is a $2 \times |\cX_{k - 1}|$ matrix obtained by subtracting the first row of $\bLambda_{(i, j)}$ entrywise by $\lambda^\star$ and making the second row all-ones.
Similar to \eqref{eq:def_lambda1234}, we define
$$
\bLambda_{[1]}'
:=
\odot_{(i, j) \in \cI_1} \bLambda_{(i, j)}'
.
$$
Using the formula $\left(\otimes_{i = 1}^c \Cb_i\right) \left(\otimes_{i = 1}^c \Db_i\right) \left(\odot_{i = 1}^c \Eb_i\right) = \odot_{i = 1}^c (\Cb_i \Db_i \Eb_i)$ for matrices $\Cb_i, \Db_i, \Eb_i$ of compatible dimensions, we obtain
\begin{equation}\label{eq:Lambda_prime}
\bLambda_{[1]}'
=
\left(\begin{matrix}
1 & -\lambda^\star \\
0 & 1
\end{matrix}\right)^{\otimes \frac{p_{k - 1} (p_{k - 1} - 1)}{2}}
\left(\begin{matrix}
1 & 0 \\
1 & 1
\end{matrix}\right)^{\otimes \frac{p_{k - 1} (p_{k - 1} - 1)}{2}}
\bLambda_{[1]}
,
\end{equation}
where $\Db^{\otimes c}$ of a matrix $\Db$ denotes the Kronecker product $\otimes_{i \in [c]} \Db$.
On the right hand side of \eqref{eq:Lambda_prime}, we note that the first term is a upper-triangular square matrix with diagonal all ones and the second term is a lower-triangular square matrix with diagonal all ones, which are both invertible.
Therefore, $\bLambda_{[1]}$ is invertible if and only if $\bLambda_{[1]}'$ is invertible.

We define the \emph{lexicographical order} between any two distinct adjacency matrices $\Xb_{k - 1}$ and $\tilde\Xb_{k - 1}$ as $\Xb_{k - 1} \succ_{lex} \tilde\Xb_{k - 1}$ if and only if there exists $(s, t) \in \langle p_{k - 1} \rangle$ such that
\begin{equation}\label{eq:lex}
X_{k - 1, s, t}
>
\tilde{X}_{k - 1, s, t}
,\quad
X_{k - 1, i, j} = \tilde{X}_{k - 1, i, j}
~~
\forall (i, j) \in \big\{
(i, j) \in \langle p_{k - 1} \rangle:~
i < s \text{ or } (i = s, j < t)
\big\}
.
\end{equation}
If we vectorize the upper triangular off-diagonal entries $(i, j)$ of $\Xb_{k - 1}$ in ascending order of $p_{k - 1} i + j$ and view the resulting binary vector as a binary number, then the lexicographical order is exactly the order of this binary number.
Recall that we have previously let $\sX_{k - 1}$ be an arbitrary bijective map from $\cX_{k - 1}$ to $[|\cX_{k - 1}|]$.
Within this proof, we can specify it as the unique bijective map that is monotonically decreasing, such that
$$
\sX_{k - 1}(\Xb_{k - 1})
<
\sX_{k - 1}(\tilde\Xb_{k - 1})
~
\iff
~
\Xb_{k - 1}
\succ_{lex}
\tilde\Xb_{k - 1}
,\quad
\forall \Xb_{k - 1}, \tilde\Xb_{k - 1} \in \cX_{k - 1}
.
$$
Note that any other choice of bijective map $\cX_{k - 1}$ will only complicate our proof statement but not affect our result.
We also define a permutation matrix $\Nb \in \RR^{|\cX_{p - 1}| \times |\cX_{p - 1}|}$ such that
$$
(\Nb \bLambda_{[1]})_s
=
\odot_{(i, j) \in \langle p_{k - 1} \rangle} (\bLambda_{i, j})_{2 - \sX_{k - 1}^{-1}(s)_{i, j}}
$$
for $(\Bb)_s$ denoting the $s$th row of matrix $\Bb$.
Then the $(s, t)$th entry of $\Nb \bLambda_{[1]}'$ is given by
\begin{equation}\label{eq:NLp}
(\Nb \bLambda_{[1]}')_{s, t}
=
\prod_{(i, j) \in \langle p_{k - 1} \rangle} \left(
(\lambda_{(i, j), 1, t} - \lambda^\star) 1_{\sX_{k - 1}^{-1}(s)_{i, j} = 1} + 1_{\sX_{k - 1}^{-1}(s)_{i, j} = 0}
\right)
,
\end{equation}
where $\lambda_{(i, j), 1, t}$ denotes the $(1, t)$th entry of $\bLambda_{(i, j)}$.

Since $\bTheta \in \cT_0$ and $\Ab_k$ takes the blockwise form \eqref{eq:A_block} with $\Pb_k = \Ib_{p_k}$, we have each $\Gamma_{k, i, j} > 0$ and $\ab_{k, i} = \eb_i$ for $i \in [p_{k - 1}]$, which implies
$$
\lambda_{(i, j), 1, t}
=
\frac{\exp\left(
C_k + \ab_{k, i}^\top (\bGamma_k * \sX_{k - 1}^{-1}(t)) \ab_{k, j}
\right)}{1 + \exp\left(
C_k + \ab_{k, i}^\top (\bGamma_k * \sX_{k - 1}^{-1}(t)) \ab_{k, j}
\right)}
=
\frac{\exp\left(
C_k + \Gamma_{k, i, j} \sX_{k - 1}^{-1}(t)_{i, j}
\right)}{1 + \exp\left(
C_k + \Gamma_{k, i, j} \sX_{k - 1}^{-1}(t)_{i, j}
\right)}
>
\lambda^\star
$$
whenever $\sX_{k - 1}^{-1}(t)_{i, j} = 1$.
From \eqref{eq:NLp}, all the diagonal entries of $\Nb \bLambda_{[1]}'$ are therefore positive.
For $1 \le s < t \le |\cX_{k - 1}|$, since $\sX_{k - 1}^{-1}(s) \succ_{lex} \sX_{k - 1}^{-1}(t)$, there exists $(i, j) \in \langle p_{k - 1} \rangle$ such that $\sX_{k - 1}^{-1}(s)_{i, j} = 1$ and $\sX_{k - 1}^{-1}(t)_{i, j} = 0$, which implies
$$
(\lambda_{(i, j), 1, t} - \lambda^\star) 1_{\sX_{k - 1}^{-1}(s)_{i, j} = 1} + 1_{\sX_{k - 1}^{-1}(s)_{i, j} = 0}
=
\frac{\exp\left(
C_k + \Gamma_{k, i, j} \sX_{k - 1}^{-1}(t)_{i, j}
\right)}{1 + \exp\left(
C_k + \Gamma_{k, i, j} \sX_{k - 1}^{-1}(t)_{i, j}
\right)} - \lambda^\star
=
0
.
$$
From \eqref{eq:NLp}, all the upper-triangular off-diagonal entries $(\Nb \bLambda_{[1]}')_{s, t}$ are therefore zero.
As a result, the $|\cX_{k - 1}| \times |\cX_{k - 1}|$ matrix $\Nb \bLambda_{[1]}'$ is lower-triangular with positive diagonal entries.
This proves that $\Nb \bLambda_{[1]}'$ is invertible, as is $\bLambda_{[1]}$, and $\rank_K(\bLambda_{[1]}) = |\cX_{k - 1}|$.
With the same procedure, we can also prove that $\rank_K(\bLambda_{[2]}) = |\cX_{k - 1}|$.

We next show that $(\bLambda_{[3]} \odot \bLambda_{[4]}) \diag(\vb)$ has Kruskal rank $\ge 2$.
Since each column of $\bLambda_{[3]} \odot \bLambda_{[4]}$ sums to one and each entry of $\vb$ is positive by $\bTheta \in \cT_0$ and our model formulation \eqref{eq:model_entry}, it suffices to show that $\bLambda_{[3]} \odot \bLambda_{[4]}$ does not have two identical columns.
We prove by contradiction.
Suppose the two distinct columns $s, t \in [|\cX_{k - 1}|]$ of $\bLambda_{[3]} \odot \bLambda_{[4]}$ are identical.
By taking partial sums, we obtain that the $s$th and $t$th columns of $\bLambda_{(i, j)}$ are also identical for all $(i, j) \in \cI_3$.
Let $s < t$, then $\sX_{k - 1}^{-1}(s) \succ_{lex} \sX_{k - 1}^{-1}(t)$, so there exists $(i, j) \in \langle p_{k - 1} \rangle$ such that $\sX_{k - 1}^{-1}(s)_{i, j} = 1$ and $\sX_{k - 1}^{-1}(t)_{(i, j)} = 0$.
Since $\Ab_k$ takes the blockwise form of \eqref{eq:A_block} with $\Pb_k = \Ib_{p_k}$, we have $\ab_{k, i} = \eb_i$ and $\ab_{k, p_{k - 1} + j} = \eb_j$, which implies
$$
\lambda_{(i, p_{k - 1} + j), 1, s}
=
\frac{\exp\left(
C_k + \ab_{k, i}^\top (\bGamma_k * \sX_{k - 1}^{-1}(s)) \ab_{k, p_{k - 1} + j}
\right)}{1 + \exp\left(
C_k + \ab_{k, i}^\top (\bGamma_k * \sX_{k - 1}^{-1}(s)) \ab_{k, p_{k - 1} + j}
\right)}
=
\frac{
e^{C_k + \Gamma_{k, i, j}}
}{
1 + e^{C_k + \Gamma_{k, i, j}}
}
$$
and
$$
\lambda_{(i, p_{k - 1} + j), 1, t}
=
\frac{\exp\left(
C_k + \ab_{k, i}^\top (\bGamma_k * \sX_{k - 1}^{-1}(t)) \ab_{k, p_{k - 1} + j}
\right)}{1 + \exp\left(
C_k + \ab_{k, i}^\top (\bGamma_k * \sX_{k - 1}^{-1}(t)) \ab_{k, p_{k - 1} + j}
\right)}
=
\frac{
e^{C_k}
}{
1 + e^{C_k}
}
.
$$
We need to have $\lambda_{(i, p_{k - 1} + j), 1, s} = \lambda_{(i, p_{k - 1} + j), 1, t}$ due to the two identical columns, which implies $\Gamma_{k, i, j} = 0$, contradicting with $\bTheta \in \cT_0$.
Therefore, we have proven
$$
\rank_K((\bLambda_{[3]} \odot \bLambda_{[4]}) \diag(\vb))
\ge
2
.
$$

For simplicity, we denote $\bLambda_{[34]} := \bLambda_{[3]} \odot \bLambda_{[4]}$.
Now that
$$
\rank_K(\bLambda_{[1]}) + \rank_K(\bLambda_{[2]}) + \rank_K(\bLambda_{[34]} \diag(\vb))
\ge
2|\cX_{k - 1}| + 2
,
$$
we can apply Theorem \ref{theo:kruskal} and obtain that for any $\tilde\bLambda_{[1]}, \tilde\bLambda_{[2]}, \tilde\bLambda_{[34]}, \tilde\vb$ of the same dimensions as $\bLambda_{[1]}, \bLambda_{[2]}, \bLambda_{[34]}, \vb$ satisfying
$$
\big\llbracket
\tilde\bLambda_{[1]}, \tilde\bLambda_{[2]}, \tilde\bLambda_{[34]} \diag(\tilde\vb)
\big\rrbracket
=
\big\llbracket
\bLambda_{[1]}, \bLambda_{[2]}, \bLambda_{[34]} \diag(\vb)
\big\rrbracket
,
$$
there must exist diagonal matrices $\Qb_1, \Qb_2, \Qb_{34} \in \RR^{|\cX_{k - 1}| \times |\cX_{k - 1}|}$ satisfying $\Qb_1 \Qb_2 \Qb_{34} = \Ib_{|\cX_{k - 1}|}$ and permutation matrix $\Pb \in \RR^{|\cX_{k - 1}| \times |\cX_{k - 1}|}$ such that
$$
\tilde\bLambda_{[1]} = \bLambda_{[1]} \Qb_1 \Pb
,\quad
\tilde\bLambda_{[2]} = \bLambda_{[2]} \Qb_2 \Pb
,\quad
\tilde\bLambda_{[34]} \diag(\tilde\vb) = \bLambda_{[34]} \diag(\vb) \Qb_{34} \Pb
.
$$
Since each column of $\bLambda_{[i]}, \tilde\bLambda_{[i]}$ $(i \in \{1, 2, 34\})$ needs to sum to one, we must have $\Qb_1 = \Qb_2 = \Qb_{34} = \Ib$.
Therefore, we have
$$
\tilde\bLambda_{[1]}
=
\bLambda_{[1]} \Pb
,\quad
\tilde\bLambda_{[2]}
=
\bLambda_{[2]} \Pb
,
$$
$$
\tilde\vb
=
(\tilde\bLambda_{[34]} \diag(\tilde\vb))^\top \one
=
\Pb^\top \diag(\vb) \bLambda_{[34]}^\top \one
=
\Pb^\top \vb
,
$$
and
$$
\tilde\bLambda_{[34]}
=
\bLambda_{[34]} \diag(\vb) \Pb \diag(\tilde\vb)^{-1}
=
\bLambda_{[34]} \diag(\vb) \Pb (\Pb^\top \diag(\vb)^{-1} \Pb)
=
\bLambda_{[34]} \Pb
.
$$
Together they imply
$$
\tilde\bLambda
:=
\tilde\bLambda_{[1]} \odot \tilde\bLambda_{[2]} \odot \tilde\bLambda_{[34]}
=
\bLambda \Pb
,\quad
\tilde\vb
=
\Pb^\top \vb
.
$$
\end{proof}

\begin{proof}[Proof of Lemma \ref{lemm:strict_sX}]
We define $\sP := \tilde\sX_{k - 1}^{-1} \circ \sX_{k - 1}$, then $\sP$ is a permutation map on $\cX_{k - 1}$.
Lemma \ref{lemm:strict_sX} can be restated using $\sP$ as follows:
If for all $\Xb_{k - 1} \in \cX_{k - 1}$,
\begin{equation}\label{eq:strict_sX_equiv_cond}
\law(\Xb_k | \Xb_{k - 1}, \Ab_k, \bTheta_k)
=
\law(\Xb_k | \sP(\Xb_{k - 1}), \tilde\Ab_k, \tilde\bTheta_k)
,
\end{equation}
then there exists a permutation matrix $\Pb \in \RR^{p_{k - 1} \times p_{k - 1}}$ such that $\sP(\Xb_{k - 1}) = \Pb \Xb_{k - 1} \Pb^\top$ for all $\Xb_{k - 1} \in \cX_{k - 1}$.

To prove Lemma \ref{lemm:strict_sX}, we first discuss some properties of $\law(\Xb_k | \Xb_{k - 1}, \Ab_k, \bTheta_k)$ that hold for arbitrary $(\Ab, \bTheta) \in \cA_1 \times \cT_0$.
Notice that knowing $\law(\Xb_k | \Xb_{k - 1}, \Ab_k, \bTheta_k)$ is equivalent to knowing
$$
\law(X_{k, i, j} | \Xb_{k - 1}, \Ab_k, \bTheta_k)
\quad
\forall (i, j) \in \langle p_k \rangle
,
$$
which is further equivalent to knowing
$$
\logit~ \PP(X_{k, i, j} = 1 | \Xb_{k - 1}, \Ab_k, \bTheta_k)
=
C_k + \ab_{k, i}^\top (\bGamma_k * \Xb_{k - 1}) \ab_{k, j}
\quad
\forall (i, j) \in \langle p_k \rangle
,
$$
due to our model formulation \eqref{eq:model_entry} and \eqref{eq:model}.
We let $\gb(\Xb_{k - 1}, \Ab_k, \bTheta_k)$ denote the $\frac{p_k (p_k - 1)}{2}$-dimensional vector
\begin{equation}\label{eq:gb}
\gb(\Xb_{k - 1}, \Ab_k, \bTheta_k)
:=
\big(
\logit~ \PP(X_{k, i, j} = 1 | \Xb_{k - 1}, \Ab_k, \bTheta_k)
\big)_{(i, j) \in \langle p_k \rangle}
\end{equation}
with an arbitrarily fixed order of $(i, j)$ pairs in $\langle p_k \rangle$.
For two vectors $\gb, \tilde\gb \in \RR^d$ of the same dimension $d$, we define the partial order $\succeq$ by
$$
\gb \succeq \tilde\gb
\quad\iff\quad
g_i \ge \tilde{g}_i
\quad
\forall i \in [d]
,
$$
where $g_i, \tilde{g}_i$ denote the $i$th entries of $\gb, \tilde\gb$.
Similarly, we also define a partial order $\succeq$ on $\cX_{k - 1}$ by
$$
\Xb_{k - 1} \succeq \tilde\Xb_{k - 1}
\quad\iff\quad
X_{k - 1, i, j} \ge \tilde{X}_{k - 1, i, j}
\quad
\forall (i, j) \in \langle p_{k - 1} \rangle
$$
for any $\Xb_{k - 1}, \tilde\Xb_{k - 1} \in \cX_{k - 1}$.
We let $\cG(\Ab_k, \bTheta_k)$ denote the collection of all $\gb(\Xb_{k - 1}, \Ab_k, \bTheta_k)$ with $\Xb_{k - 1} \in \cX_{k - 1}$, i.e. the set of vectors
$$
\cG(\Ab_k, \bTheta_k)
:=
\big\{
\gb(\Xb_{k - 1}, \Ab_k, \bTheta_k):~
\Xb_{k - 1} \in \cX_{k - 1}
\big\}
.
$$

Let $(\Ab, \bTheta) \in \cA_1 \times \cT_0$ be an arbitrary parameter and $k \in [K]$.
Since each entry $\Gamma_{k, i, j}$ of $\bGamma_k$ is positive, the set $\cG(\Ab_k, \bTheta_k)$ obtains a unique minimal element $\gb(\Ib_{p_{k - 1}}, \Ab_k, \bTheta_k)$ under the partial order $\succeq$.
We denote the set after removing this minimal element by
$$
\cG_0(\Ab_k, \bTheta_k)
:=
\cG(\Ab_k, \bTheta_k) \cap \{\gb(\Ib_{p_{k - 1}}, \Ab_k, \bTheta_k)\}^c
.
$$
Since $\Ab \in \cA_1$, for each $i \in [p_{k - 1}]$ there exists row $r_i \in [p_k]$ such that $\ab_{k, r_i} = \eb_i$.
Therefore, for any $\Xb_{k - 1}, \tilde\Xb_{k - 1} \in \cX_{k - 1}$, we have the equivalence
\begin{equation}\label{eq:g_to_X}
\gb(\Xb_{k - 1}, \Ab_k, \bTheta_k)
\succeq
\gb(\tilde\Xb_{k - 1}, \Ab_k, \bTheta_k)
\quad
\iff
\quad
\Xb_{k - 1}
\succeq
\tilde\Xb_{k - 1}
.
\end{equation}
For $(i, j) \in \langle p_{k - 1} \rangle$, let $\Eb_{i, j} := \eb_i \eb_j^\top + \eb_j \eb_i^\top \in \RR^{p_{k - 1} \times p_{k - 1}}$.
We notice that under the partial order $\succeq$, the set $\cG_0(\Ab_k, \bTheta_k)$ admits $\frac{p_{k - 1} (p_{k - 1} - 1)}{2}$ minimal elements, forming the subset
$$
\cG_m(\Ab_k, \bTheta_k)
:=
\left\{
\gb(\Ib_{p_{k - 1}} + \Eb_{i, j}, \Ab_k, \bTheta_k)
:~
(i, j) \in \langle p_{k - 1} \rangle
\right\}
.
$$
For $(r, s) \in \langle p_k \rangle$, let $\gb(\Xb_{k - 1}, \Ab_k, \bTheta_k)_{(r, s)}$ denote the $(r, s)$th entry of the vector $\gb(\Xb_{k - 1}, \Ab_k, \bTheta_k)$ defined in \eqref{eq:gb}.
For each row $r \in [p_k]$ of $\Ab_k$, we further define a subset of $\cG_m(\Ab_k, \bTheta_k)$ as
\begin{equation}\label{eq:Gmr}
\cG_{m, r}(\Ab_k, \bTheta_k)
:=
\big\{
\gb \in \cG_m(\Ab_k, \bTheta_k):~
\exists s \in [p_k] \cap \{r\}^c \text{ s.t. }
\gb_{(r, s)} > \gb(\Ib_{p_{k - 1}}, \Ab_k, \bTheta_k)_{(r, s)}
\big\}
,
\end{equation}
which essentially contains all vectors $\gb(\Ib_{p_{k - 1}} + \Eb_{i, j}, \Ab_k, \bTheta_k)$'s with $(i, j) \in \langle p_{k - 1} \rangle$ satisfying $A_{k, r, i} = 1$ or $A_{k, r, j} = 1$.
For any $(r, s) \in \langle p_k \rangle$, we notice that the following equivalence holds:
\begin{equation}\label{eq:Gmr_equiv_ab}
\cG_{m, r}(\Ab_k, \bTheta_k)
\supset
\cG_{m, s}(\Ab_k, \bTheta_k)
\quad\iff\quad
\ab_{k, r}
\succeq
\ab_{k, s}
\quad\text{or}\quad
\one^\top \ab_{k, r} \ge p_{k - 1} - 1
.
\end{equation}

We consider the case $p_{k - 1} \ge 3$ for now.
The collection of all rows in $\Ab_k$ after removing the zero vector is the set $\big\{ \ab_{k, r}:~ r \in [p_k] \big\} \cap \{\zero\}^c$, and since $\Ab \in \cA_1$, it contains $p_{k - 1}$ minimal elements $\eb_1, \ldots, \eb_{p_{k - 1}}$ under the partial order $\succeq$.
Correspondingly, \eqref{eq:Gmr_equiv_ab} suggests that the collection of sets $\big\{ \cG_{m, r}(\Ab_k, \bTheta_k):~ r \in [p_k] \big\} \cap \{\varnothing\}^c$ also contains $p_{k - 1}$ minimal elements under the partial order of set inclusion, which we denote by $\cG_m^{(1)}(\Ab_k, \bTheta_k), \ldots, \cG_m^{(p_{k - 1})}(\Ab_k, \bTheta_k)$.
For each $\cG_m^{(\ell)}(\Ab_k, \bTheta_k)$, there exists some row $r_\ell \in [p_k]$ and some $i_\ell \in [p_{k - 1}]$ satisfying $\ab_{k, r_\ell} = \eb_{i_\ell}$, such that
$$
\cG_m^{(\ell)}(\Ab_k, \bTheta_k)
=
\cG_{m, r_\ell}(\Ab_k, \bTheta_k)
=
\big\{
\gb(\Ib_{p_{k - 1}} + \Eb_{i_\ell, j}, \Ab_k, \bTheta_k):~
j \in [p_{k - 1}] \cap \{i_\ell\}^c
\big\}
,
$$
which contains exactly $p_{k - 1} - 1$ elements.
Additionally, we also note that each $\gb(\Ib_{p_{k - 1}} + \Eb_{i, j}, \Ab_k, \bTheta_k)$ belongs to exactly two of the sets $\cG_m^{(\ell)}(\Ab_k, \bTheta_k)$'s.
Therefore, there exists a permutation map $\sQ$ on $[p_{k - 1}]$ satisfying
\begin{equation}\label{eq:perm_Q}
\cG_m^{(\sQ(i))}(\Ab_k, \bTheta_k) \cap \cG_m^{(\sQ(j))}(\Ab_k, \bTheta_k)
=
\{\gb(\Ib_{p_{k - 1}} + \Eb_{i, j}, \Ab_k, \bTheta_k)\}
\end{equation}
for all $(i, j) \in \langle p_{k - 1} \rangle$.

We note that all above definitions of $\gb, \cG, \cG_0, \cG_m, \cG_m^{(\ell)}$ are made based on the distributions $\law(\Xb_k | \Xb_{k - 1}, \Ab_k, \bTheta_k)$'s for $\Xb_{k - 1} \in \cX_{k - 1}$, without explicitly knowing $\Ab_k, \bTheta_k$. 
Therefore, for any $(\Ab, \bTheta), (\tilde\Ab, \tilde\bTheta) \in \cA_1 \times \cT_0$, if \eqref{eq:strict_sX_equiv_cond} holds for all $\Xb_{k - 1} \in \cX_{k - 1}$, then we have
$$
\cG(\Ab_k, \bTheta_k)
=
\cG(\tilde\Ab_k, \tilde\bTheta_k)
,\quad
\cG_0(\Ab_k, \bTheta_k)
=
\cG_0(\tilde\Ab_k, \tilde\bTheta_k)
,\quad
\cG_m(\Ab_k, \bTheta_k)
=
\cG_m(\tilde\Ab_k, \tilde\bTheta_k)
,
$$
and if $p_{k - 1} \ge 3$, we further have
$$
\big\{
\cG_m^{(\ell)}(\Ab_k, \bTheta_k):~
\ell \in [p_{k - 1}]
\big\}
=
\big\{
\cG_m^{(\ell)}(\tilde\Ab_k, \tilde\bTheta_k):~
\ell \in [p_{k - 1}]
\big\}
.
$$
When $p_{k - 1} \ge 3$, from \eqref{eq:perm_Q} we know there further exists a permutation map $\tilde\sQ$ on $[p_{k - 1}]$ such that
$$
\gb(\Ib_{p_{k - 1}} + \Eb_{i, j}, \Ab_k, \bTheta_k)
=
\gb(\Ib_{p_{k - 1}} + \Eb_{\tilde\sQ(i), \tilde\sQ(j)}, \tilde\Ab_k, \tilde\bTheta_k)
.
$$
The above results together suggest
$$
\sP(\Ib_{p_{k - 1}})
=
\Ib_{p_{k - 1}}
,\quad
\sP(\Ib_{p_{k - 1}} + \Eb_{i, j})
=
\Ib_{p_{k - 1}} + \Eb_{\tilde\sQ(i), \tilde\sQ(j)}
\quad
\forall (i, j) \in \langle p_{k - 1} \rangle
.
$$
Combining with \eqref{eq:strict_sX_equiv_cond} and \eqref{eq:g_to_X}, for all $\Xb_{k - 1} \in \cX_{k - 1}$, we know the following equivalence holds for $\sP(\Xb_{k - 1})$:
\begin{align*}
\sP(\Xb_{k - 1})
\succeq
\Ib_{p_{k - 1}} + \Eb_{\tilde\sQ(i), \tilde\sQ(j)}
&\iff
\sP(\Xb_{k - 1})
\succeq
\sP(\Ib_{p_{k - 1}} + \Eb_{i, j})
\\&\iff
\gb(\sP(\Xb_{k - 1}), \tilde\Ab_k, \tilde\bTheta_k)
\succeq
\gb(\sP(\Ib_{p_{k - 1}} + \Eb_{i, j}), \tilde\Ab_k, \tilde\bTheta_k)
\\&\iff
\gb(\Xb_{k - 1}, \Ab_k, \bTheta_k)
\succeq
\gb(\Ib_{p_{k - 1}} + \Eb_{i, j}, \Ab_k, \bTheta_k)
\\&\iff
\Xb_{k - 1}
\succeq
\Ib_{p_{k - 1}} + \Eb_{i, j}
.
\end{align*}
This suggests that the permutation map $\sP$ on $\cX_{k - 1}$ is uniquely determined through the permutation map $\tilde\sQ$ on $[p_{k - 1}]$, 
and the entries of $\sP(\Xb_{k - 1})$ satisfy the equivalence
\begin{equation}\label{eq:sP_sQ}
(\sP(\Xb_{k - 1}))_{\tilde\sQ(i), \tilde\sQ(j)} = 1
\quad
\iff
X_{k - 1, i, j} = 1
,\quad
\forall (i, j) \in \langle p_{k - 1} \rangle
.
\end{equation}
When $p_{k - 1} = 2$, the set $\cG_0(\Ab_k, \bTheta_k)$ contains the single element $\gb(\one \one^\top, \Ab_k, \bTheta_k)$ and the set $\cG_0(\tilde\Ab_k, \tilde\bTheta_k)$ contains the single element $\gb(\one \one^\top, \tilde\Ab_k, \tilde\bTheta_k)$, so we must have
$$
\sP(\Ib_{p_{k - 1}})
=
\Ib_{p_{k - 1}}
,\quad
\sP(\one \one^\top)
=
\one \one^\top
,
$$
suggesting that \eqref{eq:sP_sQ} holds with the permutation map $\tilde\sQ$ being the identity map on $\cX_{k - 1} = \{\Ib_{p_{k - 1}}, \one \one^\top\}$.

In either case of $p_{k - 1}$, letting $\Pb \in \RR^{p_{k - 1} \times p_{k - 1}}$ denote the permutation matrix associated with the permutation map $\tilde\sQ$ in \eqref{eq:sP_sQ}, i.e.
$$
\big(
1 ~ 2 ~ \cdots p_{k - 1}
\big) \Pb
=
\big(
\tilde\sQ(1) ~ \tilde\sQ(2) ~ \cdots \tilde\sQ(p_{k - 1})
\big)
,
$$
then we have
$$
\sP(\Xb_{k - 1})
=
\Pb \Xb_{k - 1} \Pb^\top
,\quad
\forall \Xb_{k - 1} \in \cX_{k - 1}
.
$$
\end{proof}

\begin{proof}[Proof of Lemma \ref{lemm:strict_param}]
We notice that
$$
\law(\Xb_k | \Xb_{k - 1}, \Ab_k, \bTheta_k)
=
\law(\Xb_k | \Xb_{k - 1}, \tilde\Ab_k, \tilde\bTheta_k)
,\quad
\forall \Xb_{k - 1} \in \cX_{k - 1}
$$
is equivalent to
$$
\law(X_{k, i, j} | \Xb_{k - 1}, \Ab_k, \bTheta_k)
=
\law(X_{k, i, j} | \Xb_{k - 1}, \tilde\Ab_k, \tilde\bTheta_k)
,\quad
\forall \Xb_{k - 1} \in \cX_{k - 1},~ (i, j) \in \langle p_k \rangle
,
$$
which is further equivalent to
\begin{equation}\label{eq:strict_param_equiv}
C_k + \ab_{k, i}^\top (\bGamma_k * \Xb_{k - 1}) \ab_{k, j}
=
\tilde{C}_k + \tilde\ab_{k, i}^\top (\tilde\bGamma_k * \Xb_{k - 1}) \tilde\ab_{k, j}
,\quad
\forall \Xb_{k - 1} \in \cX_{k - 1},~ (i, j) \in \langle p_k \rangle
.
\end{equation}
We let $\Xb_{k - 1} = \Ib_{p_{k - 1}}$ and take minimums
$$
\min_{(i, j) \in \langle p_k \rangle} \big(
C_k + \ab_{k, i}^\top (\bGamma_k * \Ib_{p_{k - 1}}) \ab_{k, j}
\big)
,\quad
\min_{(i, j) \in \langle p_k \rangle} \big(
\tilde{C}_k + \tilde\ab_{k, i}^\top (\tilde\bGamma_k * \Ib_{p_{k - 1}}) \tilde\ab_{k, j}
\big)
.
$$
Since $\Ab, \tilde\Ab \in \cA_1$, the minimums are respectively $C_k, \tilde{C}_k$ and needs to equal by \eqref{eq:strict_param_equiv}.
Hence we obtain $\tilde{C}_k = C_k$.

We define a partial order $\succeq$ on $\{0, 1\}^{p_{k - 1}}$ as
$$
\ab_{k, r}
\succeq
\ab_{k, s}
\quad\iff\quad
A_{k, r, j}
\ge
A_{k, s, j}
\quad
\forall j \in [p_{k - 1}]
.
$$
We notice for any $(r, s) \in \langle p_k \rangle$, since $\Ab, \tilde\Ab \in \cA_1$ and $\bTheta, \tilde\bTheta \in \cT_0$, using \eqref{eq:strict_param_equiv} we have the equivalence
\begin{equation}\label{eq:a_partord}\begin{aligned}
&\qquad\quad
\tilde\ab_{k, r}
\succeq
\tilde\ab_{k, s}
\\&\iff
\tilde{C}_k + \tilde\ab_{k, r}^\top (\tilde\bGamma_k * \Xb_{k - 1}) \tilde\ab_{k, t}
\ge
\tilde{C}_k + \tilde\ab_{k, s}^\top (\tilde\bGamma_k * \Xb_{k - 1}) \tilde\ab_{k, t}
\quad
\forall \Xb_{k - 1} \in \cX_{k - 1}, t \in [p_k] \cap \{r, s\}^c
\\&\iff
C_k + \ab_{k, r}^\top (\bGamma_k * \Xb_{k - 1}) \ab_{k, t}
\ge
C_k + \ab_{k, s}^\top (\bGamma_k * \Xb_{k - 1}) \ab_{k, t}
\quad
\forall \Xb_{k - 1} \in \cX_{k - 1}, t \in [p_k] \cap \{r, s\}^c
\\&\iff
\ab_{k, r}
\succeq
\ab_{k, s}
.
\end{aligned}\end{equation}
This suggests that $\tilde\ab_{k, r} = \zero$ if and only if $\ab_{k, r} = \zero$.
For the collections of non-zero rows in $\Ab_k$ and $\tilde\Ab_k$ given by
$$
\cR_k
:=
\big\{
\ab_{k, r}:~ r \in [p_k], \ab_{k, r} \ne \zero
\big\}
,\quad
\tilde\cR_k
:=
\big\{
\tilde\ab_{k, r}:~ r \in [p_k], \tilde\ab_{k, r} \ne \zero
\big\}
,
$$
each collection has the same $p_{k - 1}$ minimal elements $\eb_1, \ldots, \eb_{p_{k - 1}}$.
From \eqref{eq:a_partord} we know any non-zero $\ab_{k, r}$ is a minimal element of $\cR_k$ if and only if $\tilde\ab_{k, r}$ is a minimal element of $\tilde\cR_k$.
This implies that for all $r \in [p_k]$, we have $\tilde\ab_{k, r} = \eb_i$ for some $i \in [p_{k - 1}]$ if and only if $\ab_{k, r} = \eb_j$ for some $j \in [p_{k - 1}]$.
We let $\cE$ denote the collection of row index $r$'s satisfying $\one^\top \ab_{k, r} = \one^\top \tilde\ab_{k, r} = 1$, then for $r \in \cE$ and $i \in [p_{k - 1}]$ we have the equivalence
\begin{align*}
\tilde{A}_{k, r, i} = 1
&\iff
\forall s \in \cE \cap \{r\}^c,~
\exists j \in [p_{k - 1}] \cap \{i\}^c \text{ s.t. }
\tilde{C}_k + \tilde\ab_{k, r}^\top (\tilde\bGamma_k * (\Ib_{p_{k - 1}} + \Eb_{i, j})) \tilde\ab_{k, s}
>
\tilde{C}_k
\\&\iff
\forall s \in \cE \cap \{r\}^c,~
\exists j \in [p_{k - 1}] \cap \{i\}^c \text{ s.t. }
C_k + \ab_{k, r}^\top (\bGamma_k * (\Ib_{p_{k - 1}} + \Eb_{i, j})) \ab_{k, s}
>
C_k
\\&\iff
A_{k, r, i} = 1
.
\end{align*}
This further suggests that for all $r \in [p_k]$ and any $i \in [p_{k - 1}]$, $\tilde\ab_{k, r} = \eb_i$ if and only if $\ab_{k, r} = \eb_i$.
Since $\Ab, \tilde\Ab \in \cA_1$, we can find $2p_{k - 1}$ distinct rows $r_1, \ldots, r_{p_{k - 1}}, r_1', \ldots, r_{p_{k - 1}}'$ of $\Ab_k$ and $\tilde\Ab_k$ satisfying
$$
\ab_{k, r_j}
=
\tilde\ab_{k, r_j}
=
\ab_{k, r_j'}
=
\tilde\ab_{k, r_j'}
=
\eb_j
,\quad
\forall j \in [p_{k - 1}]
.
$$
Then for all $r \in [p_k]$ and $i \in [p_{k - 1}]$, we have the equivalence
\begin{align*}
\tilde{A}_{k, r, i} = 1
&\iff
\tilde\ab_{k, r} \succeq \tilde\ab_{k, r_i}
\\&\iff
\ab_{k, r} \succeq \ab_{k, r_i}
\\&\iff
A_{k, r, i} = 1
.
\end{align*}
This shows that $\tilde\Ab_k = \Ab_k$.
We notice that \eqref{eq:a_partord} has implicitly assumed $p_k \ge 3$.
In the corner case of $p_k = 2$, we automatically have $\tilde\Ab_k = (1 ~ 1)^\top = \Ab_k$ due to the definition of $\cA_1$.

By letting $\Xb_{k - 1} = \one \one^\top$ and using \eqref{eq:strict_param_equiv}, we have for all $i \in [p_k]$
$$
\tilde{C}_k + \tilde\Gamma_{k, i, i}
=
\tilde{C}_k + \tilde\ab_{k, r_i}^\top (\tilde\bGamma * (\one \one^\top)) \tilde\ab_{k, r_i'}
=
C_k + \ab_{k, r_i}^\top (\bGamma * (\one \one^\top)) \ab_{k, r_i'}
=
C_k + \Gamma_{k, i, i}
,
$$
which suggests $\diag(\tilde\bGamma_k) = \diag(\bGamma_k)$.
We also have for all $(i, j) \in \langle p_k \rangle$
$$
\tilde{C}_k + \tilde\Gamma_{k, i, j}
=
\tilde{C}_k + \tilde\ab_{k, r_i}^\top (\tilde\bGamma * (\one \one^\top)) \tilde\ab_{k, r_j}
=
C_k + \ab_{k, r_i}^\top (\bGamma * (\one \one^\top)) \ab_{k, r_j}
=
C_k + \Gamma_{k, i, j}
,
$$
which suggests $\offdiag(\tilde\bGamma_k) = \offdiag(\bGamma_k)$.
Together they imply $\tilde\bGamma_k = \bGamma_k$.
\end{proof}

\subsection{Holomorphic Function and Its Zero Set}\label{ssec:holo}

Theory of generic identifiability has been established for latent class models \citep{allman2009identifiability} and constrained latent class models \citep{gu2021bayesian}, where the notion of algebraic variety from algebraic geometry is used.
An \emph{algebraic variety} refers to the set of simultaneous zeros of a finite number of polynomial functions \citep{cox2013ideals}.
In both models, generic identifiability is defined in terms of continuous parameters and established in two steps:
We first show that the set of non-identifiable parameters lies in an affine algebraic variety, then find the existence of an identifiable parameter within the parameter space.
Using the fact that the Lebesgue measure of an algebraic variety is positive if and only if all the polynomial functions defining the algebraic variety are constant zero functions, we obtain that the parameters are generically identifiable.

Our model contains both discrete 
$\Ab$ and continuous $\bTheta$ parameters, so we have extended the definition of generic identifiability as in Definition \ref{defi:generic}.
We aim at showing for any fixed parameter $\Ab \in \cA_2$, the set of $\bTheta$ making parameter $(\Ab, \bTheta)$ non-identifiable has Lebesgue measure zero in $\cT_0$.
These sets of $\bTheta$, however, are hard to include in algebraic varieties.
To handle them properly, we use the notion of a holomorphic function and the identity theorem from complex analysis \citep{stein2010complex}, which studies the zeros of holomorphic functions.

A function of several complex variables is said to be \emph{holomorphic} if it is complex-differentiable within its domain, or equivalently, if it satisfies Cauchy-Riemann equations in each variable within its domain.
The following theorem, widely known as the \emph{identity theorem}, states that a holomorphic function of one complex variable is identically zero if its zero set contains a limit point.

\begin{theorem}[Identity Theorem]\label{theo:holo}
Let $f(z)$ be a holomorphic function in a connected region $\Omega \subset \CC$.
If the zero set of $f$ has a limit point in $\Omega$, then $f \equiv 0$ in $\Omega$.
\end{theorem}

For proof of Theorem \ref{theo:holo}, see \citet{stein2010complex}.
Using Theorem \ref{theo:holo}, we can obtain a theorem on the Lebesgue measure of the zero set, which can be applied to the theory of generic identifiability.
While Theorem \ref{theo:holo} only holds for holomorphic functions of one complex variable (e.g. $f(z_1, z_2) = z_1 z_2$ is a counterexample with two complex variables), our theorem can be established for holomorphic functions of several complex variables.

\begin{theorem}\label{theo:holo_mz}
Let $\Omega_1, \ldots, \Omega_n \subset \CC$ be connected regions and $f(z_1, \ldots, z_n)$ be a holomorphic function in the region $\Omega = \Omega_1 \times \cdots \times \Omega_n \subset \CC^n$.
If its zero set in $\RR^n \cap \Omega$ has Lebesgue measure
$$
\uplambda\big( \big\{
x \in \RR^n \cap \Omega:~ f(x) = 0
\big\} \big)
>
0
,
$$
then $f \equiv 0$ in $\Omega$.
\end{theorem}

Importantly, Theorem \ref{theo:holo_mz} suggests a general technique for showing generic identifiability, which involves two steps:
We first show that the set of non-identifiable parameters are contained within the zero set of a holomorphic function.
We then prove the existence of an identifiable parameter within the parameter space.
Together they would imply that the set of non-identifiable parameters are contained within the zero set of a non-constant holomorphic function, which by Theorem \ref{theo:holo_mz} has Lebesgue measure zero.

Since holomorphic functions not only include polynomial functions, but also rational functions, exponential functions, trigonometric functions, as well as their compositions, this method is a generalization of the algebraic variety method described earlier.
Specifically, the simultaneous zero set of polynomials $f_1, \ldots, f_n$ in $\RR^n$ can be written as $\left\{x \in \RR^n:~ (f_1^2 + \ldots + f_n^2)(x) = 0\right\}$, where $f_1^2 + \ldots + f_n^2$ is a holomorphic function over $\CC^n$.

We note that many results exist on the zero set of holomorphic functions and Theorem \ref{theo:holo_mz} should not be anything new.
However, these results typically come from complex analysis and do not appear in the form of our theorem.
For completeness, we present an elementary proof as below.
We also note that other results from complex analysis, e.g. the Hadamard factorization theorem, also have applications in statistical theory \citep{gassiat2020identifiability, gassiat2022deconvolution}.

\begin{proof}[Proof of Theorem \ref{theo:holo_mz}]
We prove by induction.
For $n = 1$, let $f$ be a holomorphic function in $\Omega \subset \CC$ that is not constantly zero.
Using Theorem \ref{theo:holo} we know the zero set of $f$ in $\RR \cap \Omega$, denoted by $E := \big\{ x \in \RR \cap \Omega:~ f(x) = 0 \big\}$, does not have a limit point in $\RR \cap \Omega$.
Therefore, for any $x \in E$, there exists $r_x > 0$ such that $E \cap (x - r_x, x + r_x) = \{x\}$.
This implies that the intervals $\big\{ \big( x - \frac{r_x}{2}, x + \frac{r_x}{2} \big):~ x \in E \big\}$ are disjoint.
Since each interval contains a distinct rational number, we know $E$ is countable.
Let $\{e_n\}_{n = 1}^N$ be an enumeration of $E$ with $N \in \NN \cup \{\infty\}$.
For any $\epsilon > 0$, we notice that the open intervals $\big\{ \big( e_n - \frac{\epsilon}{2^{n + 1}}, e_n + \frac{\epsilon}{2^{n + 1}} \big) \big\}_{n = 1}^N$ cover $E$, which implies $\uplambda(E) \le \sum_{n = 1}^N \frac{\epsilon}{2^n} \le \epsilon$.
Therefore, the Lebesgue measure of $E$ is $\uplambda(E) = 0$, suggesting that Theorem \ref{theo:holo_mz} holds for $n = 1$.

We now suppose that Theorem \ref{theo:holo_mz} holds for $n - 1$ and consider a holomorphic function $f$ of $n$ complex variables in $\Omega \subset \CC^n$ satisfying
\begin{equation}\label{eq:holo_n}
\uplambda\big( \big\{
x \in \RR^n \cap \Omega:~ f(x) = 0
\big\} \big)
>
0
.
\end{equation}
For notational simplicity, we denote $x_{-n} := (x_1, \ldots, x_{n - 1}) \in \RR^{n - 1}$, $z_{-n} := (z_1, \ldots, z_{n - 1}) \in \CC^{n - 1}$, and $\Omega_{-n} := \Omega_1 \times \cdots \Omega_{n - 1} \subset \CC^{n - 1}$.
We define the set
\begin{equation}\label{eq:holo_B}
B
:=
\Big\{
x_{-n} \in \RR^{n - 1} \cap \Omega_{-n}:~
\uplambda\big(\big\{
x_n \in \RR \cap \Omega_n:~
f(x_{-n}, x_n) = 0
\big\}\big) > 0
\Big\}
,
\end{equation}
which contains all $x_{-n} \in \RR^{n - 1} \cap \Omega_{-n}$ such that the function $g(x_n) := f(x_{-n}, x_n)$ has zero set of positive Lebesgue measure in the region $\RR \cap \Omega_n$.
By Fubini's theorem,
\begin{align*}
\uplambda\big( \big\{
x \in \RR^n \cap \Omega:~ f(x) = 0
\big\} \big)
&=
\int_{\RR^n \cap \Omega} 1_{f(x) = 0} ~\ud x
\\&=
\int_{\RR^{n - 1} \cap \Omega_{-n}}
\int_{\RR \cap \Omega_n}
1_{f(x_{-n}, x_n) = 0}
~\ud x_n ~\ud x_{-n}
\\&=
\int_{\RR^{n - 1} \cap \Omega_{-n}}
\uplambda\big(\big\{
x_n \in \RR \cap \Omega_n:~
f(x_{-n}, x_n) = 0
\big\}\big)
~\ud x_{-n}
\\&=
\int_B
\uplambda\big(\big\{
x_n \in \RR \cap \Omega_n:~
f(x_{-n}, x_n) = 0
\big\}\big)
~\ud x_{-n}
.
\end{align*}
Therefore, \eqref{eq:holo_n} implies that $B$ must have positive Lebesgue measure in $\RR^{n - 1}$.

For each $x_{-n} \in B$, the function $g(z_n) := f(x_{-n}, z_n)$ is a holomorphic function of $z_n$ in region $\Omega_n$.
By \eqref{eq:holo_B}, we have
$$
\uplambda\big( \big\{
x_n \in \RR \cap \Omega_n:~
g(x_n) = 0
\big\} \big)
>
0
.
$$
Applying Theorem \ref{theo:holo_mz} to $g$ in region $\Omega_n$ implies that $g(z_n) \equiv 0$ in $\Omega_n$, that is,
$$
f(x_{-n}, z_n)
=
0
,\quad
\forall
x_{-n} \in B,
z_n \in \Omega_n
.
$$

Now for each $z_n \in \Omega_n$, the function $h(z_{-n}) := f(z_{-n}, z_n)$ is a holomorphic function of $z_{-n}$ in region $\Omega_{-n}$.
Since
$$
\uplambda\big( \big\{
x_{-n} \in \RR^{n - 1} \cap \Omega_{-n}:~
h(x_{-n}) = 0
\big\} \big)
\ge
\uplambda(B)
>
0
,
$$
applying Theorem \ref{theo:holo_mz} to $h$ in region $\Omega_{-n}$, we obtain $h(z_{-n}) \equiv 0$ on $\Omega_{-n}$, that is,
$$
f(z)
=
0
,\quad
\forall z \in \Omega
.
$$
Hence Theorem \ref{theo:holo_mz} holds for $n$ complex variables as well.
\end{proof}

\subsection{An Application of Geometric Measure Theory}\label{ssec:geom}

Following Section \ref{ssec:holo} on the Lebesgue measure of the zero sets of holomorphic functions, we further develop some theory along this direction using tools from \emph{geometric measure theory} \citep{federer2014geometric, nicolaescu2011coarea}, which will be useful in the later proof of generic identifiability.

To present the results, we need suitable measures on fibers of smooth maps between spaces of arbitrary dimensions.
We hereby recall the definition of the \emph{Hausdorff measure}.
Let $(X, \rho)$ be a separable metric space and $d \ge 0$ be a fixed dimension.
For any $\delta > 0$ and set $E \subset X$, we define
$$
\uH_\delta^d(U)
:=
\frac{\pi^{\frac{d}{2}}}{2^d \Gamma\big(\frac{d}{2} + 1\big)} \inf\left\{
\sum_{j \ge 1} (\diam B_j)^d:~
U \subset \bigcup_{j \ge 1} B_j,~
\diam B_j < \delta
\right\}
,
$$
where $B_j$ are arbitrary subsets of $X$.
Since $\uH_\delta^d(U)$ is non-increasing in $\delta$, we can define the $d$-dimensional Hausdorff measure of $U$ as
$$
\uH^d(U)
:=
\lim_{\delta \to 0} \uH_\delta^d(U)
.
$$
We note that any Borel set of $X$ is measurable with respect to $\uH^d$.
When $X$ is a $d$-dimensional Riemannian manifold with $\rho$ induced by its metric tensor $g$, $\uH^d$ coincides with the Lebesgue measure $\uplambda$ on $(X, g)$.
When $d = 0$, $\uH^0(U)$ equals the cardinality of the set $U$, or equivalently, $\uH^0$ coincides with the counting measure $\upmu$.

We also recall that a function on $\RR^n$ is an \emph{analytic function} if it has convergent Taylor series at every point in its domain.
An analytic function on $\RR^n$ is smooth and can be locally extended to a holomorphic function of $n$ complex variables.
Following Theorem \ref{theo:holo_mz}, we further have the following result for analytic functions on $\RR^n$.

\begin{theorem}\label{theo:geom}
Let $f, g:~ \RR^n \to \RR^m$ $(m \ge n)$ be analytic functions.
Let $B \subset \RR^n$ be a set satisfying
$$
\sup_{y \in \RR^m} \upmu(f^{-1}(y) \cap B)
<
\infty
.
$$
Then for any set $A \subset \RR^n$ with measure $\uplambda(A) = 0$, we have
$$
\uplambda\left(
f^{-1}(g(A)) \cap B
\right)
=
0
.
$$
\end{theorem}

Previously, Theorem \ref{theo:holo_mz} suggests that we can establish the generic identifiability of our model with parameter space $\cR$ by showing that all parameters $r \in \cR$ with $\perm_\cR(r) \subsetneq \marg_\cR(r)$ lie in zero sets of holomorphic functions that are not constantly zero.
By appropriately choosing a set $E \subset \cR$ of measure zero, comparing $\perm_{\cR \cap E^c}(r)$ with $\marg_{\cR \cap E^c}(r)$ is often a lot easier than comparing $\perm_\cR(r)$ with $\marg_\cR(r)$, as seen in Section \ref{ssec:iden_generic}.
However, it is important to note that
$$
\perm_{\cR \cap E^c}(r)
=
\marg_{\cR \cap E^c}(r)
\quad
\forall r \in \cR \cap E^c
$$
does not necessarily imply
$$
\perm_\cR(r)
=
\marg_\cR(r)
\quad
\forall r \in \cR \cap E^c
.
$$
For instance, the set of parameters in $R \cap E^c$ that have marginal distributions coinciding with the parameters in $E$ could have positive measure despite $E$ having measure zero (e.g. examples constructed with the Cantor function).
Significantly, under the conditions on analyticity of functions and bounded cardinality of fibers ($\sup_{y \in \RR^m} \upmu(f^{-1}(y) \cap B) < \infty$), Theorem \ref{theo:geom} rules out these possibilities and suggests that
$$
\perm_{\cR \cap E^c}(r)
=
\marg_{\cR \cap E^c}(r)
\quad
\forall r \in \cR \cap E^c
$$
can imply
$$
\perm_\cR(r)
=
\marg_\cR(r)
\quad
\forall r \in \cR \cap E^c \cap \tilde{E}^c
$$
for some additional set $\tilde{E} \subset \cR$ of measure zero.
This provides a simple strategy of establishing generic identifiability by proving generic identifiability in a smaller parameter space that excludes an undesirable set of measure zero.

The proof of Theorem \ref{theo:geom} utilizes Theorem \ref{theo:holo_mz} and the \emph{area formula} for Lipschitz functions.
Formal introductions and proofs of the area formula can be found in standard geometric measure theory references, e.g. \citet{federer2014geometric, krantz2008geometric, simon2014introduction}.

\begin{proof}[Proof of Theorem \ref{theo:geom}]
We start by considering the case when the set $B$ is bounded.
We denote the Jacobians of functions $f$ and $g$ by $J_f$ and $J_g$.
Since $f, g$ are analytic functions on $\RR^n$, by the Jacobian definition we know $J_f^2$ and $J_g^2$ are also analytic functions on $\RR^n$, which can be locally extended to holomorphic functions of $n$ complex variables.
Applying Theorem \ref{theo:holo_mz}, we know
$$
\lambda(\{J_f = 0\})
=
\lambda(\{J_g = 0\})
=
0
.
$$
Along with the boundedness of $B$, this implies that for an arbitrary $\epsilon > 0$, there exists $\delta > 0$ and a set $E \subset B$ satisfying $\lambda(E) < \epsilon$ and
$$
J_f(x) > \delta
,\quad
J_g(x) > \delta
,\quad
\forall x \in B \cap E^c
.
$$
We denote $\tilde{B} := B \cap E^c$ for simplicity.
Again using the boundedness of $B$ and the smoothness of $f, g$, we can find $D > 0$ such that
$$
J_f(x) < D
,\quad
J_g(x) < D
,\quad
\forall x \in B \cap E^c
.
$$
We denote the constant
$$
C
:=
\sup_{y \in \RR^m} \upmu\big(
f^{-1}(y) \cap B
\big)
<
\infty
.
$$
Additionally, we notice that the functions $f$ and $g$ are Lipschitz on $B$.

Now by applying the area formula to the function $f$, we have
\begin{align*}
\int_{f^{-1}(g(A)) \cap \tilde{B}} J_f(x) ~\uH^n(\ud x)
&=
\int_{\RR^m} \uH^0\left(
(f^{-1}(g(A)) \cap \tilde{B}) \cap f^{-1}(y)
\right) ~\uH^n(\ud y)
\\&=
\int_{g(A)} \uH^0\left(
\tilde{B} \cap f^{-1}(y)
\right) \uH^n(\ud y)
\\&\le
C \uH^n(g(A))
,
\end{align*}
since $\uH^0(\tilde{B} \cap f^{-1}(y)) \le \upmu(f^{-1}(y) \cap B) \le C$.
By applying the area formula to the function $g$, we have
$$
\int_A J_g(x) ~\uH^n(\ud x)
=
\int_{\RR^m} \uH^0\left(
A \cap g^{-1}(y)
\right) ~\uH^n(\ud y)
\ge
\uH^n(g(A))
,
$$
since $\uH^0(A \cap g^{-1}(y)) \ge 1$ for all $y \in g(A)$.

Putting everything together, we obtain
\begin{align*}
\uplambda\left(
f^{-1}(g(A)) \cap B
\right)
&\le
\uplambda\left(
f^{-1}(g(A)) \cap \tilde{B}
\right)
+
\epsilon
\\&\le
\frac{1}{\delta} \int_{f^{-1}(g(A)) \cap \tilde{B}} J_f(x) ~\uH^n(\ud x)
+
\epsilon
\\&\le
\frac{C}{\delta} \uH^n(g(A)) + \epsilon
\\&\le
\frac{C}{\delta} \int_A J_g(x) ~\uH^n(\ud x) + \epsilon
\\&\le
\frac{C D}{\delta} \uplambda(A) + \epsilon
\\&=
\epsilon
.
\end{align*}
Since $\epsilon > 0$ can be chosen arbitrarily, this implies
$$
\uplambda\left(
f^{-1}(g(A)) \cap B
\right)
=
0
.
$$

We now consider the general case where $B \subset \RR^n$ is not necessarily bounded.
We can choose a monotonically increasing sequence of sets $\{\cO_i\}_{i = 1}^\infty \subset \RR^n$ such that each $\cO_i$ is bounded and $\lim_{i \to \infty} \cO_i = \RR^n$.
By defining $B_i := B \cap \cO_i$, previous arguments for bounded $B$ implies that
$$
\uplambda\left(
f^{-1}(g(A)) \cap B_i
\right)
=
0
,\quad
\forall i \in \NN
.
$$
We therefore obtain
$$
\uplambda\left(
f^{-1}(g(A)) \cap B
\right)
=
\lim_{i \to \infty} \uplambda\left(
f^{-1}(g(A)) \cap B_i
\right)
=
0
.
$$
\end{proof}

\subsection{Proof of Theorem \ref{theo:generic}}\label{ssec:iden_generic}

Similar to the proof of Theorem \ref{theo:strict}, transforming the problem of model identifiability to the uniqueness of three-way tensor decomposition is an essential step of our proof.
We recall our previously defined probability matrices $\bLambda_k$'s and probability vectors $\vb_k$'s from \eqref{eq:def_vk}, \eqref{eq:def_lambdakij}, \eqref{eq:def_lambdak}, and \eqref{eq:vk_lambdak_vk1}, which characterize the conditional distributions $\law(\Xb_k | \Xb_{k - 1}, \Ab_k, \bTheta_k)$'s and the marginal distributions $\law(\Xb_k | \Ab, \bTheta)$'s.

In studying generic identifiability, excluding sets of $\upmu \times \uplambda$ measure zero will not affect the result, so in the following we define a few such sets.
Let $\cT_1$ be a subset of $\cT_0$ defined as
\begin{equation}\label{eq:T1_def}
\cT_1
:=
\left\{
\bTheta \in \cT_0:~
\exists k \in [K],~ \zb \in \ZZ^{1 + \frac{p_{k - 1} (p_{k - 1} + 1)}{2}} \cap \{\zero\}^c \text{ s.t. }
z_C C_k + \sum_{1 \le i \le j \le p_{k - 1}} z_{(i, j)} \Gamma_{k, i, j} = 0
\right\}
,
\end{equation}
where $\zb$ is a non-zero $\big(1 + \frac{p_{k - 1} (p_{k - 1} + 1)}{2}\big)$-dimensional vector of integers with $z_C$ and $z_{(i, j)}$ $(1 \le i \le j \le p_{k - 1})$ denoting its entries.
Essentially, $\cT_1$ includes all the parameters $\bTheta$ satisfying, for some $k \in [K]$, a non-trivial linear combination with integer coefficients of $C_k$ and the upper-triangular entries in $\bGamma_k$ equals zero.
For each fixed layer $k$ and non-zero integer vector $\zb \in \ZZ^{1 + \frac{p_{k - 1} (p_{k - 1} + 1)}{2}}$, the collection of $C_k, \bGamma_k$ with
$$
z_C C_k + \sum_{1 \le i \le j \le p_{k - 1}} z_{(i, j)} \Gamma_{k, i, j}
=
0
$$
is a hyperplane in $\cT_0$ and has Lebesgue measure zero.
Since there are only countably many pairs of layers $k$ and integer vectors $\zb$, the Lebesgue measure of $\cT_1$ is $\uplambda(\cT_1) = 0$.

Recall the definition of $\cA_2$ from \eqref{eq:A2_def}.
For any $\Ab \in \cA_2$ and each layer $k \in [K]$, there exists a permutation matrix $\Pb_k \in \RR^{p_k \times p_k}$ and matrices $\Mb_k, \Mb_k' \in \cM_{p_{k - 1}}$, $\Bb_k \in \RR^{p_{k - 1} \times (p_k - 2p_{k - 1})}$ such that
\begin{equation}\label{eq:A_block_M}
\Pb_k \Ab_k
=
\big(
\Mb_k ~ \Mb_k' ~ \Bb_k
\big)^\top
.
\end{equation}
Without loss of generality, we can let $\Pb_k = \Ib_{p_k}$.
We also recall from Section \ref{ssec:iden_strict} that $\cX_k$ is an arbitrarily fixed bijective map from the collection of adjacency matrices $\cX_k$ to the integers $[|\cX_k|]$.
For each layer $k \in [K]$, we define two $|\cX_{k - 1}| \times |\cX_{k - 1}|$ matrices $\bPhi_k$ and $\bPhi_k'$ entrywise by
\begin{equation}\label{eq:H_def}\begin{aligned}
\Phi_{k, s, t}
&:=
\exp\left(
\sum_{(i, j) \in \langle p_{k - 1} \rangle} (\sX_{k - 1}^{-1}(s))_{i, j} \big(
C_k + \ab_{k, i}^\top (\bGamma_k * \sX_{k - 1}^{-1}(t)) \ab_{k, j}
\big) \right)
\quad
\forall s, t \in [|\cX_{k - 1}|]
,\\
\Phi_{k, s, t}'
&:=
\exp\left(
\sum_{(i, j) \in \langle p_{k - 1} \rangle} (\sX_{k - 1}^{-1}(s))_{i, j} \big(
C_k + \ab_{k, p_{k - 1} + i}^\top (\bGamma_k * \sX_{k - 1}^{-1}(t)) \ab_{k, p_{k - 1} + j}
\big) \right)
\quad
\forall s, t \in [|\cX_{k - 1}|]
,
\end{aligned}\end{equation}
where $\Phi_{k, s, t}, \Phi_{k, s, t}'$ denote the $(s, t)$th entry of $\bPhi_k, \bPhi_k'$ and $(\sX_{k - 1}^{-1}(s))_{i, j}$ denotes the $(i, j)$th entry of the matrix $\sX_{k - 1}^{-1}(s) \in \cX_{k - 1}$.
We notice that by definition, $\bPhi_k$ and $\bPhi_k'$ are closely connected to the matrices $\bLambda_{[1]}$ and $\bLambda_{[2]}$ defined in \eqref{eq:def_lambda1234}.
Specifically, $\bPhi_k$ is obtained by multiplying the $t$th column of $\bLambda_{[1]}$ by the positive value
\begin{equation}\label{eq:Phi_Lambda1_pos}
\prod_{(i, j) \in \langle p_{k - 1} \rangle} \Big( 1 + \exp\big(
C_k + \ab_{k, i}^\top (\bGamma_k * \sX_{k - 1}^{-1}(t)) \ab_{k, j}
\big) \Big)
\end{equation}
for $t \in [|\cX_{k - 1}|]$, and $\bPhi_k'$ is obtained by multiplying the $t$th column of $\bLambda_{[2]}$ by the positive value
\begin{equation}\label{eq:Phip_Lambda2_pos}
\prod_{(i, j) \in \langle p_{k - 1} \rangle} \Big( 1 + \exp\big(
C_k + \ab_{k, p_{k - 1} + i}^\top (\bGamma_k * \sX_{k - 1}^{-1}(t)) \ab_{k, p_{k - 1} + j}
\big) \Big)
\end{equation}
for $t \in [|\cX_{k - 1}|]$.
As seen earlier in the proof of Lemma \ref{lemm:strict_tensor}, to apply Kruskal's theorem we want to establish that the matrices $\bLambda_{[1]}$ and $\bLambda_{[2]}$ are invertible.
Therefore, for each fixed $\Ab \in \cA_2$, we define the functions
\begin{equation}\label{eq:func_psi_phi}
\psi_{\Ab, k}(\bTheta)
:=
\det(\bPhi_k)
,\quad
\phi_{\Ab, k}(\bTheta)
:=
\det(\bPhi_k')
,
\end{equation}
which are functions of all active entries $\nu_i$ ($1 \le i \le |\cX_0| - 1$) in $\bnu$, all $C_k$, and all active entries $\Gamma_{k, i, j}$ ($1 \le i \le j \le p_{k - 1}$) in $\bGamma_k$ for $k \in [K]$.
By their definitions in \eqref{eq:func_psi_phi}, we notice that each $\psi_{\Ab, k}$ and $\phi_{\Ab, k}$ for $k \in [K]$ are holomorphic functions in the complex domain $\CC^d$ of dimension
$$
d
=
(|\cX_0| - 1)
+
\sum_{k = 1}^K \big(
1 + \frac{p_{k - 1}(p_{k - 1} + 1)}{2}
\big)
,
$$
where $d$ is the dimension of all active parameters in $\bTheta$, or equivalently the dimension of $\cT_0$ when viewed as a Riemannian manifold.
By excluding the zero sets of these holomorphic functions from $\cT_0 \cap \cT_1^c$ for each $\Ab \in \cA_2$, we have the following lemma.

\begin{lemma}\label{lemm:generic_tensor}
Let parameter $(\Ab, \bTheta) \in \cA_2 \times (\cT_0 \cap \cT_1^c)$ and layer $k \in [K]$.
For any matrix $\tilde\bLambda_k \in \RR^{|\cX_k| \times |\cX_{k - 1}|}$ and vector $\tilde\vb_{k - 1} \in \RR^{|\cX_{k - 1}| \times 1}$ satisfying
$$
\bLambda_k \vb_{k - 1}
=
\tilde\bLambda_k \tilde\vb_{k - 1}
,\quad
\tilde\vb_{k - 1}^\top \one
=
1
,\quad
\tilde\bLambda_k^\top \one
=
\one
,
$$
if $\psi_{\Ab, k}(\bTheta) \ne 0$ and $\phi_{\Ab, k}(\bTheta) \ne 0$, then there exists a permutation matrix $\Pb \in \RR^{|\cX_{k - 1}| \times |\cX_{k - 1}|}$ such that
$$
\tilde\bLambda_k
=
\bLambda_k \Pb
,\quad
\tilde\vb_{k - 1}
=
\Pb^\top \vb_{k - 1}
.
$$
\end{lemma}

We notice that Lemma \ref{lemm:generic_tensor} is almost identical to Lemma \ref{lemm:strict_tensor}, except that the parameter space is changed from $\cA_1 \times \cT_0$ to $\cA_2 \times (\cT_0 \cap \cT_1^c)$ and we have additionally imposed the condition on each of $\psi_{\Ab, k}(\bTheta)$ and $\phi_{\Ab, k}(\bTheta)$ being non-zero.
Its proof similarly uses the Kruskal Theorem \ref{theo:kruskal} and is provided in Section \ref{ssec:iden_generic_aux}.

From Lemma \ref{lemm:generic_tensor}, for some parameter $(\Ab, \bTheta) \in \cA_2 \times (\cT_0 \cap \cT_1^c)$, if $\bLambda_k$ and $\vb_{k - 1}$ can not be uniquely recovered up to shuffles from the product $\bLambda_k \vb_{k - 1}$, then $\bTheta$ must lie in the zero sets of the holomorphic functions $\psi_{\Ab, k}(\bTheta)$ or $\phi_{\Ab, k}(\bTheta)$ for some $k \in [K]$.
We want to further establish that these zero sets have Lebesgue measure zero using Theorem \ref{theo:holo_mz}.
The following lemma shows that these holomorphic functions $\psi_{\Ab, k}(\bTheta)$ and $\phi_{\Ab, k}(\bTheta)$ are non-constant.

\begin{lemma}\label{lemm:generic_holo}
Let parameter $\Ab \in \cA_2$ and layer $k \in [K]$.
Then there exist $\bTheta, \bTheta' \in \cT_0 \cap \cT_1^c$ such that $\psi_{\Ab, k}(\bTheta) \ne 0$ and $\phi_{\Ab, k}(\bTheta') \ne 0$.
\end{lemma}

The proof of Lemma \ref{lemm:generic_holo} is deferred to Section \ref{ssec:iden_generic_aux}.
In the proof, we exploit the structure of $\bPhi_k, \bPhi_k'$, utilize the Leibniz formula of their determinants, and explicitly construct valid choices of $\bTheta, \bTheta' \notin \cT_1$ using properties of prime numbers.

Using Theorem \ref{theo:holo_mz} with Lemmas \ref{lemm:generic_tensor} and \ref{lemm:generic_holo}, we can establish that by excluding sets of measure zero from $\cA_2 \times \cT_0$, we can uniquely recover $\bLambda_k$ and $\vb_{k - 1}$ from their product $\bLambda_k \vb_{k - 1}$ up to shuffling of columns in $\bLambda_k$ and entries in $\vb_{k - 1}$.
It is desirable to mimic the proof of Theorem \ref{theo:strict} and use Lemmas \ref{lemm:strict_sX} and \ref{lemm:strict_param}, so in the following we provide similar lemmas that work with the parameter space $\cA_2 \times (\cT_0 \cap \cT_1^c)$.

\begin{lemma}\label{lemm:generic_sX}
Let parameters $(\Ab, \bTheta), (\tilde\Ab, \tilde\bTheta) \in \cA_2 \times (\cT_0 \cap \cT_1^c)$ and layer $k \in [K]$.
For any two bijective maps $\sX_{k - 1}, \tilde\sX_{k - 1}$ from $\cX_{k - 1}$ to $[|\cX_{k - 1}|]$, if for all $x \in [|\cX_{k - 1}|]$
$$
\law\big(
\Xb_k | \Xb_{k - 1} = \sX_{k - 1}^{-1}(x), \Ab_k, \bTheta_k
\big)
=
\law\big(
\Xb_k | \Xb_{k - 1} = \tilde\sX_{k - 1}^{-1}(x), \tilde\Ab_k, \tilde\bTheta_k
\big)
,
$$
then there exists a permutation matrix $\Pb \in \RR^{p_{k - 1} \times p_{k - 1}}$ such that for all $x \in [|\cX_{k - 1}|]$,
$$
\tilde\sX_{k - 1}^{-1}(x)
=
\Pb \sX_{k - 1}^{-1}(x) \Pb^\top
.
$$
\end{lemma}

\begin{lemma}\label{lemm:generic_param}
Let parameters $(\Ab, \bTheta), (\tilde\Ab, \tilde\bTheta) \in \cA_2 \times (\cT_0 \cap \cT_1^c)$ and layer $k \in [K]$.
If for all $\Xb_{k - 1} \in \cX_{k - 1}$
$$
\law(\Xb_k | \Xb_{k - 1}, \Ab_k, \bTheta_k)
=
\law(\Xb_k | \Xb_{k - 1}, \tilde\Ab_k, \tilde\bTheta_k)
,
$$
then
$$
\tilde\Ab_k
=
\Ab_k
,\quad
\tilde{C}_k
=
C_k
,\quad
\tilde\bGamma_k
=
\bGamma_k
.
$$
\end{lemma}

The proofs of Lemmas \ref{lemm:generic_sX} and \ref{lemm:generic_param} are closely connected to the proofs of Lemmas \ref{lemm:strict_sX} and \ref{lemm:strict_param}.
In generalization of space $\cA_1$ to $\cA_2$, each connection matrix $\Ab_k$ no longer necessarily contains rows of indicator vectors.
The newly imposed condition $\bTheta \notin \cT_1$ helps us get around this obstacle in the proofs.
We defer the proofs to Section \ref{ssec:iden_generic_aux}.

Putting Lemmas \ref{lemm:generic_tensor}, \ref{lemm:generic_holo}, \ref{lemm:generic_sX}, and \ref{lemm:generic_param} together, we now provide a proof of Theorem \ref{theo:generic}.
Similar to the proof of Theorem \ref{theo:strict}, for each layer $k$ we take a three-step procedure to recover the parameters $\Ab_k, \bTheta_k$ and the probability vector $\vb_{k - 1}$ from $\vb_k$.
An induction argument then allows us to sequentially recover all parameters $\Ab, \bTheta$ from the marginal distribution $\law(\Xb_K | \Ab, \bTheta)$ for $k = K, K - 1, \ldots, 1$.
Compared to the proof of Theorem \ref{theo:strict}, we need to take additional care for the zero measure sets excluded from the parameter space $\cA_2 \times \cT_0$ using Theorems \ref{theo:holo_mz} and \ref{theo:geom}.

\begin{proof}[Proof of Theorem \ref{theo:generic}]
By excluding some sets of measure zero from $\cA_2 \times \cT_0$, we define a smaller parameter space
$$
\cR
:=
\big\{
(\Ab, \bTheta) \in \cA_2 \times (\cT_0 \cap \cT_1^c):~
\forall k \in [K],~
\psi_{\Ab, k}(\bTheta) \ne 0,~
\phi_{\Ab, k}(\bTheta) \ne 0
\big\}
$$
and denote for each $\Ab \in \cA_2$
$$
\cR_\Ab
:=
\big\{
\bTheta \in \cT_0 \cap \cT_1^c:~
\forall k \in [K],~
\psi_{\Ab, k}(\bTheta) \ne 0,~
\phi_{\Ab, k}(\bTheta) \ne 0
\big\}
.
$$
For the difference between $\cA_2 \times \cT_0$ and $\cR$, we have
$$
(\upmu \times \uplambda)\big(
(\cA_2 \times \cT_0) \cap \cR^c
\big)
\le
\upmu(\cA_2) \uplambda(\cT_1)
+
\sum_{\Ab \in \cA_2} \sum_{k = 1}^K \big(
\uplambda(\psi_{\Ab, k}^{-1}(0)) + \uplambda(\phi_{\Ab, k}^{-1}(0))
\big)
.
$$
By definition of $\cT_1$ in \eqref{eq:T1_def}, we have $\uplambda(\cT_1) = 0$.
By Lemma \ref{lemm:generic_holo} and the definition of $\psi_{\Ab, k}, \phi_{\Ab, k}$ in \eqref{eq:func_psi_phi}, for any $\Ab \in \cA_2$ and $k \in [K]$, the functions $\psi_{\Ab, k}$ and $\phi_{\Ab, k}$ are holomorphic functions in domain $\CC^d$ that are not constantly zero, with $d$ being the dimension of $\cT_0$.
Therefore, we can apply Theorem \ref{theo:holo_mz} to each $\psi_{\Ab, k}$ and $\phi_{\Ab, k}$ to obtain that $\uplambda(\psi_{\Ab, k}^{-1}(0)) = \uplambda(\phi_{\Ab, k}^{-1}(0)) = 0$ for all $\Ab \in \cA_2$ and $k \in [K]$.
Since $\cA_2$ has finite cardinality, we therefore obtain
\begin{equation}\label{eq:measure_A2T0Rc}
(\upmu \times \uplambda)\big(
(\cA_2 \times \cT_0) \cap \cR^c
\big)
=
0
.
\end{equation}

We notice that the proof of Theorem \ref{theo:strict} does not explicitly use the properties of the parameter space $\cA_1 \times \cT_0$ other than Lemmas \ref{lemm:strict_tensor}, \ref{lemm:strict_sX}, \ref{lemm:strict_param}.
Therefore, by working with the parameter space $\cR$ and applying Lemmas \ref{lemm:generic_tensor}, \ref{lemm:generic_sX}, \ref{lemm:generic_param}, using the very same proof we can obtain
\begin{equation}\label{eq:generic_mid_proof_R}
\perm_\cR(\Ab, \bTheta)
=
\marg_\cR(\Ab, \bTheta)
,\quad
\forall (\Ab, \bTheta) \in \cR
.
\end{equation}

To extend \eqref{eq:generic_mid_proof_R} to the parameter space $\cA_2 \times \cT_0$, we note that the definitions of $\cT_1$ in \eqref{eq:T1_def} and $\psi_{\Ab, k}(\bTheta), \phi_{\Ab, k}(\bTheta)$ in \eqref{eq:func_psi_phi} are invariant to permutation of nodes in each layer of our Bayesian network.
This implies that
$$
\perm_{\cA_2 \times \cT_0}(\Ab, \bTheta)
\subset
(\cA_2 \times \cT_0) \cap \cR^c
,\quad
\forall (\Ab, \bTheta) \in (\cA_2 \times \cT_0) \cap \cR^c
,
$$
which suggests
\begin{equation}\label{eq:generic_perm_same}
\perm_{\cA_2 \times \cT_0}(\Ab, \bTheta)
=
\perm_\cR(\Ab, \bTheta)
,\quad
\forall (\Ab, \bTheta) \in \cR
.
\end{equation}
However, it is important to note that we do not necessarily have $\marg_{\cA_2 \times \cT_0}(\Ab, \bTheta) = \marg_\cR(\Ab, \bTheta)$ for all $(\Ab, \bTheta) \in \cR$.
We instead establish in the following the weaker result that the collection of $(\Ab, \bTheta) \in \cR$ for which $\marg_\cR(\Ab, \bTheta) \subsetneq \marg_{\cA_2 \times \cT_0}(\Ab, \bTheta)$ has measure zero, i.e.
\begin{equation}\label{eq:need_geom}
(\upmu \times \uplambda)\left\{
(\Ab, \bTheta) \in \cR:~
\marg_\cR(\Ab, \bTheta) \subsetneq \marg_{\cA_2 \times \cT_0}(\Ab, \bTheta)
\right\}
=
0
.
\end{equation}

We notice that for any $\Ab \in \cA_2$, the mapping of parameter $\bTheta \in \cT_0$ to its corresponding probability vector $\vb_K$ (i.e. the marginal probabilities of the observed adjacency matrix $\Xb_K$) is a mapping from a lower-dimensional Euclidean space to a higher-dimensional Euclidean space, which we denote by $L_\Ab$.
By our model formulations in \eqref{eq:model_entry} and \eqref{eq:model}, we can easily verify that $L_\Ab$ is an analytic function in $\bTheta$.
We notice that the nodes at layer $k$ of our Bayesian network can be permuted in at most $p_k!$ ways, which suggests the cardinalities
$$
\left|
\marg_\cR(\Ab, \bTheta)
\right|
=
\left|
\perm_\cR(\Ab, \bTheta)
\right|
\le
\prod_{k = 0}^{K - 1} (p_k!)
$$
for all parameters $(\Ab, \bTheta) \in \cR$, due to \eqref{eq:generic_mid_proof_R}.
This implies
$$
\sup_\ub \upmu\left(
L_\Ab^{-1}(\ub) \cap \cR_\Ab
\right)
\le
\prod_{k = 0}^{K - 1} (p_k!)
<
\infty
$$
for all vector $\ub$ having the same dimension as $\vb_K$.
Additionally, similar to our previous derivations we have
$$
\uplambda(\cT_0 \cap \cR_\Ab^c)
\le
\uplambda(\cT_1) + \sum_{k = 1}^K \left(
\uplambda(\psi_{\Ab, k}^{-1}(0)) + \uplambda(\phi_{\Ab, k}^{-1}(0))
\right)
=
0
.
$$
Therefore, we have verified all conditions of Theorem \ref{theo:geom}.
For any $\Ab, \Ab' \in \cA_2$, by applying Theorem \ref{theo:geom} to the analytic functions $L_\Ab, L_{\Ab'}$ and the sets $\cT_0 \cap \cR_{\Ab'}^c, \cR_\Ab$, we obtain
$$
\uplambda\left(
L_\Ab^{-1}(L_{\Ab'}(\cT_0 \cap \cR_{\Ab'}^c)) \cap \cR_\Ab
\right)
=
0
.
$$
Taking summation over all $\Ab, \Ab' \in \cA_2$, we obtain
\begin{align*}
&\qquad~
(\upmu \times \uplambda)\left\{
(\Ab, \bTheta) \in \cR:~
\marg_\cR(\Ab, \bTheta) \subsetneq \marg_{\cA_2 \times \cT_0}(\Ab, \bTheta)
\right\}
\\&\le
\sum_{\Ab \in \cA_2} \sum_{\Ab' \in \cA_2} \uplambda\left\{
\bTheta \in \cR_\Ab:~
\exists \bTheta' \in \cT_0 \cap \cR_{\Ab'}^c~
\text{s.t. } (\Ab', \bTheta') \in \marg_{\cA_2 \times \cT_0}(\Ab, \bTheta)
\right\}
\\&=
\sum_{\Ab \in \cA_2} \sum_{\Ab' \in \cA_2} \uplambda\left(
L_\Ab^{-1}(L_{\Ab'}(\cT_0 \cap \cR_{\Ab'}^c)) \cap \cR_\Ab
\right)
\\&=
0
,
\end{align*}
which proves our desired \eqref{eq:need_geom}.

By combining \eqref{eq:measure_A2T0Rc}, \eqref{eq:generic_mid_proof_R}, \eqref{eq:generic_perm_same}, and \eqref{eq:need_geom} together, we obtain
\begin{align*}
&\quad~
(\upmu \times \uplambda)\left\{
(\Ab, \bTheta) \in \cA_2 \times \cT_0:~
\perm_{\cA_2 \times \cT_0}(\Ab, \bTheta) \subsetneq \marg_{\cA_2 \times \cT_0}(\Ab, \bTheta)
\right\}
\\&\le
(\upmu \times \uplambda)\left\{
(\Ab, \bTheta) \in \cA_2 \times \cT_0:~
\perm_\cR(\Ab, \bTheta) \subsetneq \marg_\cR(\Ab, \bTheta)
\right\}
\\&\quad+
(\upmu \times \uplambda)\left\{
(\Ab, \bTheta) \in \cA_2 \times \cT_0:~
\perm_\cR(\Ab, \bTheta) \subsetneq \perm_{\cA_2 \times \cT_0}(\Ab, \bTheta)
\right\}
\\&\quad+
(\upmu \times \uplambda)\left\{
(\Ab, \bTheta) \in \cA_2 \times \cT_0:~
\marg_\cR(\Ab, \bTheta) \subsetneq \marg_{\cA_2 \times \cT_0}(\Ab, \bTheta)
\right\}
\\&=
(\upmu \times \uplambda)\left\{
(\Ab, \bTheta) \in \cA_2 \times \cT_0:~
\marg_\cR(\Ab, \bTheta) \subsetneq \marg_{\cA_2 \times \cT_0}(\Ab, \bTheta)
\right\}
\\&\le
(\upmu \times \uplambda)\left(
(\cA_2 \times \cT_0) \cap \cR^c
\right)
\\&\quad+
(\upmu \times \uplambda)\left\{
(\Ab, \bTheta) \in \cR:~
\marg_\cR(\Ab, \bTheta) \subsetneq \marg_{\cA_2 \times \cT_0}(\Ab, \bTheta)
\right\}
\\&=
0
,
\end{align*}
which proves that our model with parameter space $\cA_2 \times \cT_0$ is generically identifiable.
\end{proof}

\subsection{Proof of Auxiliary Lemmas for Theorem \ref{theo:generic}}\label{ssec:iden_generic_aux}

\begin{proof}[Proof of Lemma \ref{lemm:generic_tensor}]
For a general layer $k \in [K]$, since the only probability matrix and vector involved in Lemma \ref{lemm:generic_tensor} are $\bLambda_k, \vb_{k - 1}$, for simplicity we abbreviate $\bLambda_k$ as $\bLambda$ and $\vb_{k - 1}$ as $\vb$.
Similarly, we also abbreviate $\tilde\bLambda_k$ as $\tilde\bLambda$, $\tilde\vb_{k - 1}$ as $\tilde\vb$, and $\bLambda_{k, (i, j)}$ as $\bLambda_{(i, j)}$ for any index pair $(i, j) \in \langle p_k \rangle$.
From the proof of Lemma \ref{lemm:strict_tensor}, we recall the definition of index pairs $\cI_1, \cI_2, \cI_3, \cI_4 \subset \langle p_k \rangle$ in \eqref{eq:def_I1234} and the definition of matrices $\bLambda_{[1]}, \bLambda_{[2]}, \bLambda_{[3]}, \bLambda_{[4]}$ in \eqref{eq:def_lambda1234}.
We have established in \eqref{eq:strict_tensor_cp} that, by appropriately reordering the rows in $\bLambda$, we have
$$
\bLambda \vb
=
\vec\big( \big\llbracket
\bLambda_{[1]}, \bLambda_{[2]}, (\bLambda_{[3]} \odot \bLambda_{[4]}) \diag(\vb)
\big\rrbracket \big)
,
$$
which bridges the connection between the uniqueness of $\bLambda, \vb$ and the uniqueness of three-way tensor CP decomposition.
A sufficient condition to apply Kruskal's Theorem \ref{theo:kruskal} has been given earlier in \eqref{eq:strict_tensor_kruskal}, which holds if we have
\begin{equation}\label{eq:generic_tensor_kruskal_suff}
\rank_K(\bLambda_{[1]})
=
\rank_K(\bLambda_{[2]})
=
|\cX_{k - 1}|
,\quad
\rank_K\big(
(\bLambda_{[3]} \odot \bLambda_{[4]}) \diag(\vb)
\big)
\ge
2
.
\end{equation}
Let $(\Ab, \bTheta) \in \cA_2 \times (\cT_0 \cap \cT_1^c)$ satisfy $\psi_{\Ab, k}(\bTheta) \ne 0$ and $\phi_{\Ab, k}(\bTheta) \ne 0$.
It is therefore sufficient to prove \eqref{eq:generic_tensor_kruskal_suff} for such an $(\Ab, \bTheta)$.

Recalling the definition of $\psi_{\Ab, k}(\bTheta)$ and $\phi_{\Ab, k}(\bTheta)$ from \eqref{eq:func_psi_phi}, we know that the matrices $\bPhi_k, \bPhi_k'$ are both invertible, which further implies that the matrices $\bLambda_{[1]}, \bLambda_{[2]}$ are invertible due to \eqref{eq:Phi_Lambda1_pos} and \eqref{eq:Phip_Lambda2_pos}.
Therefore, we have
$$
\rank_K(\bLambda_{[1]})
=
\rank(\bLambda_{[1]})
=
|\cX_{k - 1}|
,\quad
\rank_K(\bLambda_{[2]})
=
\rank(\bLambda_{[2]})
=
|\cX_{k - 1}|
$$
as a direct result of $\psi_{\Ab, k}(\bTheta) \ne 0, \phi_{\Ab, k}(\bTheta) \ne 0$.

Now since $\bTheta \in \cT_0 \cap \cT_1^c$, we know each entry of $\vb$ is positive.
We also notice that each column of $\bLambda_{[3]} \odot \bLambda_{[4]}$ sums to one by the definitions of $\bLambda_{[3]}$ and $\bLambda_{[4]}$ in \eqref{eq:def_lambda1234}.
Therefore, showing
$$
\rank_K\big(
(\bLambda_{[3]} \odot \bLambda_{[4]}) \diag(\vb)
\big)
\ge
2
$$
is equivalent to showing that the columns of $\bLambda_{[3]} \odot \bLambda_{[4]}$ are distinct.
We recall the lexicographic order of adjacency matrices in $\cX_{k - 1}$ defined in \eqref{eq:lex}.
While we have previously let $\cX_{k - 1}$ be an arbitrary bijective map from $\cX_{k - 1}$ to $[|\cX_{k - 1}|]$, within this proof we specify it to be the unique monotonically decreasing map such that
$$
\cX_{k - 1}(\Xb_{k - 1})
<
\cX_{k - 1}(\tilde\Xb_{k - 1})
\quad\iff\quad
\Xb_{k - 1}
\succ_{lex}
\tilde\Xb_{k - 1}
,\quad
\forall \Xb_{k - 1}, \tilde\Xb_{k - 1} \in \cX_{k - 1}
.
$$
This specification of the map $\cX_{k - 1}$ works solely to simplify our proof statement, while using any other choice of bijective map will not affect our result.

Suppose that the $s$th and $t$th columns of $(\bLambda_{[3]} \odot \bLambda_{[4]}) \diag(\vb)$ are identical, with $s < t$.
Then there must exist $(i, j) \in \langle p_{k - 1} \rangle$ such that the entries $(\sX_{k - 1}^{-1}(s))_{i, j} = 1$ and $(\sX_{k - 1}^{-1}(t))_{i, j} = 0$.
Recall that $\langle \Xb, \Yb \rangle := \tr(\Xb \Yb^\top)$ is the inner product between matrices of the same size.
By the definition in \eqref{eq:def_lambdakij}, the $(1, s)$th and $(1, t)$th entries of the probability matrix $\bLambda_{(i, p_{k - 1} + j)}$ are
\begin{align*}
&\qquad~
\lambda_{(i, p_{k - 1} + j), 1, s}
\\&=
\frac{\exp\left(
C_k + \ab_{k, i}^\top (\bGamma_k * \sX_{k - 1}^{-1}(s)) \ab_{k, p_{k - 1} + j}
\right)}{1 + \exp\left(
C_k + \ab_{k, i}^\top (\bGamma_k * \sX_{k - 1}^{-1}(s)) \ab_{k, p_{k - 1} + j}
\right)}
\\&=
\frac{\exp\left(
C_k + \tr\big(
\bGamma_k * (\ab_{k, p_{k - 1} + j} \ab_{k, i}^\top)
\big)
+
\left\langle
\bGamma_k * \sX_{k - 1}^{-1}(s)
,~
\offdiag(\ab_{k, p_{k - 1} + j} \ab_{k, i}^\top)
\right\rangle
\right)}{1 + \exp\left(
C_k + \tr\big(
\bGamma_k * (\ab_{k, p_{k - 1} + j} \ab_{k, i}^\top)
\big)
+
\left\langle
\bGamma_k * \sX_{k - 1}^{-1}(s)
,
\offdiag(\ab_{k, p_{k - 1} + j} \ab_{k, i}^\top)
\right\rangle
\right)}
\end{align*}
and
\begin{align*}
&\qquad~
\lambda_{(i, p_{k - 1} + j), 1, t}
\\&=
\frac{\exp\left(
C_k + \ab_{k, i}^\top (\bGamma_k * \sX_{k - 1}^{-1}(t)) \ab_{k, p_{k - 1} + j}
\right)}{1 + \exp\left(
C_k + \ab_{k, i}^\top (\bGamma_k * \sX_{k - 1}^{-1}(t)) \ab_{k, p_{k - 1} + j}
\right)}
\\&=
\frac{\exp\left(
C_k + \tr\big(
\bGamma_k * (\ab_{k, p_{k - 1} + j} \ab_{k, i}^\top)
\big)
+
\left\langle
\bGamma_k * \sX_{k - 1}^{-1}(t)
,~
\offdiag(\ab_{k, p_{k - 1} + j} \ab_{k, i}^\top)
\right\rangle
\right)}{1 + \exp\left(
C_k + \tr\big(
\bGamma_k * (\ab_{k, p_{k - 1} + j} \ab_{k, i}^\top)
\big)
+
\left\langle
\bGamma_k * \sX_{k - 1}^{-1}(t)
,~
\offdiag(\ab_{k, p_{k - 1} + j} \ab_{k, i}^\top)
\right\rangle
\right)}
.
\end{align*}
Since the $s$th and $t$th columns of $\bLambda_{[3]} \odot \bLambda_{[4]}$ are identical, the $s$th and $t$th columns of $\bLambda_{(i, j)}$ must also be identical.
This implies $\lambda_{(i, p_{k - 1} + j), 1, s} = \lambda_{(i, p_{k - 1} + j), 1, t}$, which gives
\begin{equation}\label{eq:generic_tensor_rank2_contra}
\left\langle
\bGamma_k * \sX_{k - 1}^{-1}(s)
,~
\offdiag(\ab_{k, p_{k - 1} + j} \ab_{k, i}^\top)
\right\rangle
=
\left\langle
\bGamma_k * \sX_{k - 1}^{-1}(t)
,~
\offdiag(\ab_{k, p_{k - 1} + j} \ab_{k, i}^\top)
\right\rangle
.
\end{equation}
We notice that both the left hand side and the right hand side of \eqref{eq:generic_tensor_rank2_contra} are linear combinations with integer coefficients of the upper-triangular entries in $\bGamma_k$.
Since $\Ab \in \cA_2$ has $\Ab_k$ taking the form of \eqref{eq:A_block_M} with $\Pb_k$ assumed to be $\Ib_{p_k}$ for simplicity, we know the entries $A_{k, i, i} = A_{k, p_{k - 1} + j, j} = 1$.
Along with $(\sX_{k - 1}^{-1}(s))_{i, j} = 1$ and $(\sX_{k - 1}^{-1}(t))_{i, j} = 0$, this implies that the left hand side of \eqref{eq:generic_tensor_rank2_contra} involves a term of $\Gamma_{k, i, j}$ with positive integer coefficient whereas the right hand side of \eqref{eq:generic_tensor_rank2_contra} does not.
Therefore, \eqref{eq:generic_tensor_rank2_contra} suggests that a non-trivial linear combination with integer coefficients of the upper-triangular entries in $\bGamma_k$ is zero, which contradicts with $\bTheta \notin \cT_1$.
Therefore, the columns of $(\bLambda_{[3]} \odot \bLambda_{[4]}) \diag(\vb)$ must be distinct, implying
$$
\rank_K\big(
(\bLambda_{[3]} \odot \bLambda_{[4]}) \diag(\vb)
\big)
\ge
2
.
$$

We have now established the conditions in \eqref{eq:generic_tensor_kruskal_suff}, under which Kruskal's Theorem \ref{theo:kruskal} can be applied to the three-way tensor CP decomposition $\big\llbracket \bLambda_{[1]}, \bLambda_{[3]}, (\bLambda_{[3]} \odot \bLambda_{[4]}) \diag(\vb) \big\rrbracket$.
The remaining proof is the same as in the proof of Lemma \ref{lemm:strict_tensor} and hence omitted here.
\end{proof}

\begin{proof}[Proof of Lemma \ref{lemm:generic_holo}]
We fix an arbitrary $\Ab \in \cA_2$ and layer $k \in [K]$.
Since each $\Ab_k$ takes the form of \eqref{eq:A_block_M} with $\Pb_k = \Ib_{p_k}$, the proof for existence of $\bTheta'$ satisfying $\phi_{\Ab, k}(\bTheta') \ne 0$ is similar to the proof for existence of $\bTheta$ satisfying $\psi_{\Ab, k}(\bTheta) \ne 0$.
We therefore only prove for $\bTheta$ in the following.

We recall the definition of $\bLambda_{[1]}$ from \eqref{eq:def_lambda1234} and $\sX_{k - 1}$ as a bijective map from $\cX_{k - 1}$ to $[|\cX_{k - 1}|]$.
Within this proof, we denote the $p_{k - 1} \times p_{k - 1}$ upper-left principal submatrix of $\Xb_k$ by $\tilde\Xb_k \in \cX_{k - 1}$.
By appropriately reordering the rows of $\bLambda_{[1]}$, we can have the $(s, t)$th entry of $\bLambda_{[1]}$ given by
\begin{equation}\label{eq:generic_beta_delta}\begin{aligned}
\lambda_{[1], s, t}
&=
\PP\big(
\tilde\Xb_k = \sX_{k - 1}^{-1}(s) \big| \Xb_{k - 1} = \sX_{k - 1}^{-1}(t), \Ab_k, \bTheta_k
\big)
\\&=
\prod_{(i, j) \in \langle p_{k - 1} \rangle} \frac{\exp\big(
(\sX_{k - 1}^{-1}(s))_{i, j} (C_k + \ab_{k, i}^\top (\bGamma_k * \sX_{k - 1}^{-1}(t)) \ab_{k, j})
\big)}{1 + \exp\big(
C_k + \ab_{k, i}^\top (\bGamma_k * \sX_{k - 1}^{-1}(t)) \ab_{k, j}
\big)}
\\&=
\exp\left( \sum_{(i, j) \in \langle p_{k - 1} \rangle} (\sX_{k - 1}^{-1}(s))_{i, j} \left(
C_k + \left\langle \diag(\ab_{k, i} \ab_{k, j}^\top),~ \bGamma_k \right\rangle
\right) \right)
\\&\quad\cdot
\exp\left( \sum_{(i, j) \in \langle p_{k - 1} \rangle} \frac{(\sX_{k - 1}^{-1}(s))_{i, j}}{2} \left\langle
\offdiag(\ab_{k, i} \ab_{k, j}^\top + \ab_{k, j} \ab_{k, i}^\top)
,~
\bGamma_k * \sX_{k - 1}^{-1}(t)
\right\rangle \right)
\\&\quad\cdot
\prod_{(i, j) \in \langle p_{k - 1} \rangle} \frac{1}{1 + \exp\left(
C_k + \ab_{k, i}^\top (\bGamma_k * \sX_{k - 1}^{-1}(t)) \ab_{k, j}
\right)}
\\&:=
\beta_s^{(1)} \cdot \delta_{s, t} \cdot \beta_t^{(3)}
,
\end{aligned}\end{equation}
which decomposes into the product of three terms $\beta_s^{(1)}, \delta_{s, t}, \beta_t^{(3)}$, with the first term $\beta_s^{(1)}$ being invariant of column $t$ and the third term $\beta_t^{(3)}$ being invariant of row $s$.
Accordingly, we define the $|\cX_{k - 1}| \times |\cX_{k - 1}|$ diagonal matrices 
$$
\Bb_1
:=
\diag(\beta_1^{(1)}, \ldots, \beta_{|\cX_{k - 1}|}^{(1)})
,\quad
\Bb_3
:=
\diag(\beta_1^{(3)}, \ldots, \beta_{|\cX_{k - 1}|}^{(3)})
,
$$
and let $\bDelta$ be the $|\cX_{k - 1}| \times |\cX_{k - 1}|$ matrix with its $(s, t)$th entry being $\delta_{s, t}$, then
$$
\bLambda_{[1]}
=
\Bb_1 \bDelta \Bb_3
.
$$
Since all values of $\beta_s^{(1)}$ and $\beta_t^{(3)}$ are positive, both matrices $\Bb_1, \Bb_3$ are invertible, which implies that $\psi_\Ab(\bTheta) = \det(\bLambda_{[1]}) \ne 0$ is equivalent to $\det\bDelta \ne 0$.

Let $\iota$ denote any permutation on the integers $[|\cX_{k - 1}|]$ and $\sI$ denote the collection of all $|\cX_{k - 1}|!$ permutations on $[|\cX_{k - 1}|]$.
The Leibniz formula for determinants gives the expansion
\begin{equation}\label{eq:Delta_expan}
\det \bDelta
=
\sum_{\iota \in \sI} \sgn(\iota) \prod_{s = 1}^{|\cX_{k - 1}|} \delta_{s, \iota(s)}
=
\sum_{\iota \in \sI} \sgn(\iota) \exp\left( \frac12 \left\langle
\bGamma_k
,~
\sum_{s = 1}^{|\cX_{k - 1}|} \big(
\Hb_s *  (\sX_{k - 1}^{-1} \circ \iota)(s)
\big) \right\rangle \right)
,
\end{equation}
where $\sgn(\iota)$ is the parity (oddness or evenness) of permutation $\iota$ and $\Hb_s$ is the $p_{k - 1} \times p_{k - 1}$ matrix
\begin{equation}\label{eq:Hs}
\Hb_s
:=
\offdiag\left(
\sum_{(i, j) \in \langle p_{k - 1} \rangle} (\sX_{k - 1}^{-1}(s))_{i, j} (\ab_{k, i} \ab_{k, j}^\top + \ab_{k, j} \ab_{k, i}^\top)
\right)
.
\end{equation}
Recalling the definition of matrices in class $\cM_d$ from Definition \ref{defi:M} and the blockwise form of connection matrix $\Ab_k$ from \eqref{eq:A_block_M}, we know
$$
\Hb_s
\ne
\Hb_r
,\quad
\forall s, r \in [|\cX_{k - 1}|],~
s \ne r
.
$$

We also recall the definition of lexicographical order between adjacency matrices from \eqref{eq:lex}.
Note that this definition extends to a complete order on any collection of symmetric matrices with zero diagonals, e.g. the collection
$$
\big\{
\Hb_s:~
s \in [|\cX_{k - 1}|]
\big\}
.
$$
For any two distinct adjacency matrices $\Xb_{k - 1}, \tilde\Xb_{k - 1} \in \cX_{k - 1}$, we have either $\Xb_{k - 1} \prec_{lex} \tilde\Xb_{k - 1}$ or $\Xb_{k - 1} \succ_{lex} \tilde\Xb_{k - 1}$.
For any two distinct integers $s \ne r$, similarly we have either $\Hb_s \prec_{lex} \Hb_r$ or $\Hb_s \succ_{lex} \Hb_r$.
Therefore, among all the permutations in $\sI$, there exists a unique permutation $\iota^*$ such that the following equivalence holds:
$$
\Hb_s
\prec_{lex}
\Hb_t
\quad\iff\quad
(\sX_{k - 1}^{-1} \circ \iota^*)(s)
\prec_{lex}
(\sX_{k - 1}^{-1} \circ \iota^*)(t)
,\quad
\forall s, t \in [|\cX_{k - 1}|]
.
$$
For any permutation $\iota \in \sI$, we define the $|\cX_{k - 1}| \times |\cX_{k - 1}|$ matrix
\begin{equation}\label{eq:Giota}
\Gb_\iota
:=
\sum_{s = 1}^{|\cX_{k - 1}|} \big(
\Hb_s * (\sX_{k - 1}^{-1} \circ \iota)(s)
\big)
.
\end{equation}
We notice for any non-decreasing sequences of non-negative values $\{x_i\}_{i = 1}^n$ and $\{y_i\}_{i = 1}^n$, we have
$$
\sum_{i = 1}^n x_i y_{\pi(i)}
\le
\sum_{i = 1}^n x_i y_{i}
$$
for all permutations $\pi$ on the integers $[n]$.
When the sequences $\{x_i\}_{i = 1}^n, \{y_i\}_{i = 1}^n$ are strictly increasing, the inequality becomes equality if and only if $\pi$ is the identity permutation on $[n]$.
Since all entries of $\Hb_s$'s and $(\sX_{k - 1}^{-1} \circ \iota)(s)$'s are non-negative, using this elementary inequality above, we can obtain
\begin{equation}\label{eq:Giota_lex}
\Gb_\iota
\prec_{lex}
\Gb_{\iota^*}
,\quad
\forall \iota \in \sI,~
\iota \ne \iota^*
.
\end{equation}

We now explicitly construct the parameter $\bTheta \in \cT_0 \cap \cT_1^c$ such that its corresponding $\det\bDelta \ne 0$.
From the definition of $\delta_{s, t}$ in \eqref{eq:generic_beta_delta}, we know that the parameter $C_k$, the diagonal entries of $\bGamma_k$, the parameter $\bnu$, and the parameters of other layers $\bTheta_\ell$ for $\ell \ne k$ do not affect the invertibility of $\bDelta$.
Therefore, we can choose them to be any valid value such that, together with the off-diagonal entries of $\bGamma_k$ specified in the following, we can have $\bTheta \in \cT_0 \cap \cT_1^c$.

For the off-diagonal entries of $\bGamma_k$, we choose them sequentially by pairs $(i, j) \in \langle p_{k - 1} \rangle$ in descending order of $p_{k - 1} i + j$, that is, in the order
$$
(p_{k - 1} - 1, p_{k - 1})
,~
(p_{k - 1} - 2, p_{k - 1})
,~
(p_{k - 1} - 2, p_{k - 1} - 1)
,~
\ldots
,~
(2, 3)
,~
(1, p_{k - 1})
,~
\ldots
,~
(1, 2)
.
$$
We first construct a sequence of prime numbers
$$
q_1
,~
\ldots
,~
q_{\frac{p_{k - 1} (p_{k - 1} - 1)}{2}}
$$
satisfying that $q_1 \ge 3$ and
\begin{equation}\label{eq:def_q}
q_{t + 1}
>
|\cX_{k - 1}|! \left(
\prod_{s = 1}^t q_s
\right)^{|\cX_{k - 1}| p_{k - 1} (p_{k - 1} - 1)}
\quad
\forall t \in \left[ \frac{p_{k - 1} (p_{k - 1} - 1)}{2} - 1 \right]
,
\end{equation}
which always exists since there are infinitely many prime numbers.
The off-diagonal entries of $\bGamma_k$ are then assigned in descending order of $p_{k - 1} i + j$ to the sequence
\begin{equation}\label{eq:def_theta_q}
\big\{
\log q_t
\big\}_{t = 1}^{\frac{p_{k - 1} (p_{k - 1} - 1)}{2}}
.
\end{equation}

We claim that the our constructed parameter $\bTheta$ satisfies $\det\bDelta \ne 0$.
To see this, we notice that the cardinality of the collection of permutations $\sI$ is $|\cX_{k - 1}|!$ and each entry of the matrix $\Gb_\iota$ satisfies
$$
0
\le
G_{\iota, s, t}
\le
|\cX_{k - 1}| p_{k - 1} (p_{k - 1} - 1)
\quad
\forall s, t \in [|\cX_{k - 1}|]
$$
by the definitions of $\Hb_s$ in \eqref{eq:Hs} and $\Gb_\iota$ in \eqref{eq:Giota}.
Using \eqref{eq:Giota_lex}, for each permutation $\iota \ne \iota^*$ there exists an index pair $(i, j) \in \langle p_{k - 1} \rangle$ such that entrywise
$$
G_{\iota, s, t}
=
G_{\iota^*, s, t}
\quad
\forall (s, t) \in \{(s, t) < (i, j)\}
,\quad
G_{\iota, i, j}
<
G_{\iota^*, i, j}
,
$$
where for simplicity we denote the set of index pairs
$$
\{(s, t) < (i, j)\}
:=
\big\{
(s, t) \in \langle p_{k - 1} \rangle:~
s < i \text{ or } (s = i, t < j)
\big\}
.
$$
Similarly, we also denote the set of index pairs
$$
\{(s, t) > (i, j)\}
:=
\langle p_{k - 1} \rangle \cap \{(s, t) < (i, j)\}^c \cap \{(i, j)\}^c
,
$$
then our choice of $\bGamma_k$ satisfies
\begin{align*}
&\quad~
\exp\left(
\frac12 \langle \bGamma_k, \Gb_{\iota^*} \rangle - \frac12 \langle \bGamma_k, \Gb_\iota \rangle
\right)
\\&=
\exp\big(
\Gamma_{k, i, j} (G_{\iota^*, i, j} - G_{\iota, i, j})
\big)
\cdot
\exp\left(
\sum_{(s, t) \in \{(s, t) > (i, j)\}} \Gamma_{k, s, t} (G_{\iota^*, s, t} - G_{\iota, s, t})
\right)
\\&\ge
\exp(\Gamma_{k, i, j})
\cdot
\exp\left(
- \sum_{\{(s, t) > (i, j)\}} \Gamma_{k, s, t} |\cX_{k - 1}| p_{k - 1} (p_{k - 1} - 1)
\right)
\\&>
|\cX_{k - 1}|!
,
\end{align*}
where we have used the definitions in \eqref{eq:def_q} and \eqref{eq:def_theta_q}.
This implies that among the terms in the expansion \eqref{eq:Delta_expan} of $\det\bDelta$, we have
$$
\exp\left(
\frac12 \langle \bGamma_k, \Gb_{\iota^*} \rangle
\right)
>
\sum_{\iota \in \cJ \cap \{\iota^*\}^c} \exp\left(
\frac12 \langle \bGamma_k, \Gb_\iota \rangle
\right)
,
$$
which implies that $\det\bDelta \ne 0$.

As a last step of our proof, we show that $\bTheta \notin \cT_1$.
Since we can choose the parameter $C_k$ and the diagonal entries of $\bGamma_k$ arbitrarily, it suffices to prove that any non-trivial linear combination with integer coefficients of the upper-triangular, off-diagonal entries of $\bGamma_k$ is non-zero.
Suppose there exists an integer vector $\zb \in \ZZ^{\frac{p_{k - 1} (p_{k - 1} - 1)}{2}}$ satisfying
\begin{equation}\label{eq:generic_tensor_T1}
\sum_{(i, j) \in \langle p_{k - 1} \rangle} z_{(i, j)} \Gamma_{k, i, j}
=
0
.
\end{equation}
In \eqref{eq:def_theta_q}, we have associated the sequence $\{q_t\}_{t = 1}^{\frac{p_{k - 1} (p_{k - 1} - 1)}{2}}$ with the sequence of index pairs in $\langle p_{k - 1} \rangle$ in descending order of $p_{k - 1} i + j$, so we will also denote the corresponding $q_t$ by $q_{(i, j)}$ for simplicity.
Then \eqref{eq:generic_tensor_T1} implies
$$
\prod_{(i, j) \in \langle p_{k - 1} \rangle} q_{(i, j)}^{z_{(i, j)}}
=
1
.
$$
Since $q_{(i, j)}$ for $(i, j) \in \langle p_{k - 1} \rangle$ are chosen to be distinct prime numbers in \eqref{eq:def_q}, we must have
$$
z_{(i, j)}
=
0
,\quad
\forall (i, j) \in \langle p_{k - 1} \rangle
.
$$
Therefore, by choosing the other parameters in $\bTheta$ appropriately (e.g. as logarithms of other distinct prime numbers), we have constructed a choice of $\bTheta \in \cT_0 \cap \cT_1^c$.

We note that our proof above works for all $\Ab \in \cA_2$.
That is, for any $\Ab \in \cA_2$, through the above procedure we can construct a choice of parameter $\bTheta \in \cT_0 \cap \cT_1^c$ such that $\psi_{\Ab, k}(\bTheta) \ne 0$.
\end{proof}

\begin{proof}[Proof of Lemma \ref{lemm:generic_sX}]
Similar to the proof of Lemma \ref{lemm:strict_sX}, we define $\sP := \tilde\sX_{k - 1}^{-1} \circ \sX_{k - 1}$, then $\sP$ is a permutation map on $\cX_{k - 1}$.
Lemma \ref{lemm:generic_sX} can be restated using $\sP$ as follows:
If for all $\Xb_{k - 1} \in \cX_{k - 1}$,
\begin{equation}\label{eq:generic_sX_equiv_cond}
\law(\Xb_k | \Xb_{k - 1}, \Ab_k, \bTheta_k)
=
\law(\Xb_k | \sP(\Xb_{k - 1}), \tilde\Ab_k, \tilde\bTheta_k)
,
\end{equation}
then there exists a permutation matrix $\Pb \in \RR^{p_{k - 1} \times p_{k - 1}}$ such that $\sP(\Xb_{k - 1}) = \Pb \Xb_{k - 1} \Pb^\top$ for all $\Xb_{k - 1} \in \cX_{k - 1}$.

To prove Lemma \ref{lemm:generic_sX}, we first recall some properties of $\law(\Xb_k | \Xb_{k - 1}, \Ab_k, \bTheta_k)$ discussed in the proof of Lemma \ref{lemm:strict_sX}, which extends to hold for arbitrary $(\Ab, \bTheta) \in \cA_2 \times (\cT_0 \cap \cT_1^c)$.
Notice that knowing $\law(\Xb_k | \Xb_{k - 1}, \Ab_k, \bTheta_k)$ is equivalent to knowing
$$
\law(X_{k, i, j} | \Xb_{k - 1}, \Ab_k, \bTheta_k)
\quad
\forall (i, j) \in \langle p_k \rangle
,
$$
which is further equivalent to knowing
$$
\logit~ \PP(X_{k, i, j} = 1 | \Xb_{k - 1}, \Ab_k, \bTheta_k)
=
C_k + \ab_{k, i}^\top (\bGamma_k * \Xb_{k - 1}) \ab_{k, j}
\quad
\forall (i, j) \in \langle p_k \rangle
$$
due to our model formulation \eqref{eq:model_entry}.
Recall that we have defined $\gb(\Xb_{k - 1}, \Ab_k, \bTheta_k)$ in \eqref{eq:gb} to be the $2^{\frac{p_k (p_k - 1)}{2}}$-dimensional vector
$$
\gb(\Xb_{k - 1}, \Ab_k, \bTheta_k)
:=
\big(
\logit~ \PP(X_{k, i, j} = 1 | \Xb_{k - 1}, \Ab_k, \bTheta_k)
\big)_{(i, j) \in \langle p_k \rangle}
$$
with an arbitrary fixed order of $(i, j)$ pairs in $\langle p_k \rangle$.
For two vectors $\gb, \tilde\gb \in \RR^d$ of the same dimension $d$, a partial order $\succeq$ between them is defined by
$$
\gb \succeq \tilde\gb
\quad\iff\quad
g_i \ge \tilde{g}_i
\quad
\forall i \in [d]
,
$$
where $g_i, \tilde{g}_i$ denote the $i$th entries of $\gb, \tilde\gb$.
A partial order $\succeq$ on $\cX_{k - 1}$ is similarly defined by
$$
\Xb_{k - 1} \succeq \tilde\Xb_{k - 1}
\quad\iff\quad
X_{k - 1, i, j} \ge \tilde{X}_{k - 1, i, j}
\quad
\forall (i, j) \in \langle p_{k - 1} \rangle
$$
for any $\Xb_{k - 1}, \tilde\Xb_{k - 1} \in \cX_{k - 1}$.
We also recall that $\cG(\Ab_k, \bTheta_k)$ is defined to be the collection of all $\gb(\Xb_{k - 1}, \Ab_k, \bTheta_k)$ with $\Xb_{k - 1} \in \cX_{k - 1}$, i.e. the set of vectors
$$
\cG(\Ab_k, \bTheta_k)
:=
\big\{
\gb(\Xb_{k - 1}, \Ab_k, \bTheta_k):~
\Xb_{k - 1} \in \cX_{k - 1}
\big\}
.
$$

Let $(\Ab, \bTheta) \in \cA_2 \times (\cT_0 \cap \cT_1^c)$ be an arbitrary parameter and $k \in [K]$.
Since each entry $\Gamma_{k, i, j}$ of $\bGamma_k$ is positive, the set $\cG(\Ab_k, \bTheta_k)$ obtains a unique minimal element $\gb(\Ib_{p_{k - 1}}, \Ab_k, \bTheta_k)$ under the partial order $\succeq$.
We denote the set after removing this minimal element by
$$
\cG_0(\Ab_k, \bTheta_k)
:=
\cG(\Ab_k, \bTheta_k) \cap \{\gb(\Ib_{p_{k - 1}}, \Ab_k, \bTheta_k)\}^c
.
$$
If $p_{k - 1} = 2$, then the proof of Lemma \ref{lemm:generic_sX} is the same as the proof of Lemma \ref{lemm:strict_sX}, so we only consider the case $p_{k - 1} \ge 3$ in the following.
Since $\Ab \in \cA_2 \subset \cA_{2, 1}$, there exists distinct row indices
\begin{equation}\label{eq:generic_rows}
r_1, \ldots, r_{p_{k - 1}}, r_1', \ldots, r_{p_{k - 1}}'
\in
[p_k]
\end{equation}
such that $\ab_{k, r_i} \succeq \eb_i$ and $\ab_{k, r_i'} \succeq \eb_i$ for each $i \in [p_{k - 1}]$.
We denote $\Eb_{i, j} := \eb_i \eb_j^\top + \eb_j \eb_i^\top \in \RR^{p_{k - 1} \times p_{k - 1}}$ for $(i, j) \in \langle p_{k - 1} \rangle$.
Importantly, under the partial order $\succeq$, the $\frac{p_{k - 1} (p_{k - 1} - 1)}{2}$ elements
$$
\gb(\Ib_{p_{k - 1}} + \Eb_{i, j}, \Ab_k, \bTheta_k)
,\quad
(i, j) \in \langle p_{k - 1} \rangle
$$
are no longer necessarily the minimal elements of the set $\cG_0(\Ab_k, \bTheta_k)$, contrary to the proof of Lemma \ref{lemm:strict_sX} with $p_{k - 1} \ge 3$ and $\Ab \in \cA_1$.
However, we are still able to pick out the set of these $\frac{p_{k - 1} (p_{k - 1} - 1)}{2}$ elements from the set $\cG_0(\Ab_k, \bTheta_k)$, as discussed in the following.

We consider the collection $\cH$ of all entries in the vectors
$$
\gb - \gb(\Ib_{p_{k - 1}}, \Ab_k, \bTheta_k)
,\quad
\gb \in \cG_0(\Ab_k, \bTheta_k)
,
$$
then each element of $\cH$ is a linear combination with non-negative integer coefficients of the values $\Gamma_{k, i, j}$ for $(i, j) \in \langle p_{k - 1} \rangle$.
Recall the definition of rows $r_i, r_i'$ from \eqref{eq:generic_rows}.
For any $(i, j) \in \langle p_{k - 1} \rangle$, we have $\ab_{k, r_i} \succeq \eb_i$ and $\ab_{k, r_j} \succeq \eb_j$.
If either $\ab_{k, r_i} \nsucceq \eb_j$ or $\ab_{k, r_j} \nsucceq \eb_i$, then we have
\begin{equation}\label{eq:cH_Gammakij_case1}
\left(
\gb(\Ib_{p_{k - 1}} + \Eb_{i, j}, \Ab_k, \bTheta_k) - \gb(\Ib_{p_{k - 1}}, \Ab_k, \bTheta_k)
\right)_{r_i, r_j}
=
\Gamma_{k, i, j}
.
\end{equation}
Otherwise, we have $\ab_{k, r_i} \succeq \eb_i + \eb_j$ and $\ab_{k, r_j} \succeq \eb_i + \eb_j$.
Since $\Ab \in \cA_2 \subset \cA_{2, 2}$, the columns of $\Ab_k$ are distinct, which implies the existence of a third row $s \in [p_k]$ such that either $A_{k, s, i} = 1, A_{k, s, j} = 0$ or $A_{k, s, i} = 0, A_{k, s, j} = 1$.
Without loss of generality, we let the former case hold, then we have
\begin{equation}\label{eq:cH_Gammakij_case2}
\left(
\gb(\Ib_{p_{k - 1}} + \Eb_{i, j}, \Ab_k, \bTheta_k) - \gb(\Ib_{p_{k - 1}}, \Ab_k, \bTheta_k)
\right)_{s, r_j}
=
\Gamma_{k, i, j}
.
\end{equation}
By combining the two cases \eqref{eq:cH_Gammakij_case1} and \eqref{eq:cH_Gammakij_case2}, we obtain that each upper-triangular, off-diagonal entry $\Gamma_{k, i, j}$ is an element of the set $\cH$, i.e.
\begin{equation}\label{eq:cH_Gammakij}
\cH
\supset
\big\{
\Gamma_{k, i, j}:~
(i, j) \in \langle p_{k - 1} \rangle
\big\}
.
\end{equation}
Since $\bTheta \notin \cT_1$, each $\Gamma_{k, i, j}$ can not be expressed as a linear combination with integer coefficients of other $\Gamma_{k, i, j}$'s, which implies that $\big\{ \Gamma_{k, i, j}:~ (i, j) \in \langle p_{k - 1} \rangle \big\}$ is the unique subset $\cH_\gamma$ of $\cH$ of size $\frac{p_{k - 1} (p_{k - 1} - 1)}{2}$ satisfying
$$
\forall h \in \cH,~
\exists 0 \le m \le \frac{p_{k - 1} (p_{k - 1} - 1)}{2},~
h_1, \ldots, h_m \in \cH_\gamma,~
z_1, \ldots, z_m \in \ZZ_+
\text{ s.t. }
h
=
\sum_{\ell = 1}^m h_\ell z_\ell
,
$$
where $\ZZ_+$ denotes the set of all positive integers.
This essentially says that
\begin{equation}\label{eq:cH_gamma}
\cH_\gamma
:=
\big\{
\Gamma_{k, i, j}:~
(i, j) \in \langle p_{k - 1} \rangle
\big\}
\end{equation}
is the only subset of $\cH$ of its size such that each element of $\cH$ can be written as a linear combination with non-negative integer coefficients of the subset elements.
Importantly, this implies that $\cH_\gamma$ is well-defined given $\cH$, without knowing the parameters $\Ab_k, \bTheta_k$ or the matrix $\Xb_{k - 1} \in \cX_{k - 1}$ to which each $\gb \in \cG_0$ corresponds.
To proceed, notice that among all elements of $\cG_0$, $\gb(\Ib_{p_{k - 1}} + \Eb_{i, j}, \Ab_k, \bTheta_k)$ is the unique element $\gb$ satisfying that each entry of $\gb - \gb(\Ib_{p_{k - 1}}, \Ab_k, \bTheta_k)$ is an integer multiple of $\Gamma_{k, i, j}$.
Therefore, by defining the set
\begin{equation}\label{eq:generic_cGm}
\cG_m(\Ab_k, \bTheta_k)
:=
\big\{
\gb \in \cG_0:~
\exists \gamma \in \cH_\gamma,~
\zb \in \ZZ^{\frac{p_{k - 1} (p_{k - 1} - 1)}{2}}
\text{ s.t. }
\gb - \gb(\Ib_{p_{k - 1}}, \Ab_k, \bTheta_k)
=
\gamma \zb
\big\},
\end{equation}
we have
$$
\cG_m(\Ab_k, \bTheta_k)
=
\big\{
\gb(\Ib_{p_{k - 1}} + \Eb_{i, j}, \Ab_k, \bTheta_k):~
(i, j) \in \langle p_{k - 1} \rangle
\big\}
.
$$

Based on the well-defined set $\cG_m(\Ab_k, \bTheta_k)$, we can define its subset $\cG_{m, r}(\Ab_k, \bTheta_k)$ for each row $r \in [p_k]$ of $\Ab_k$ as in \eqref{eq:Gmr}, given by
$$
\cG_{m, r}(\Ab_k, \bTheta_k)
:=
\big\{
\gb \in \cG_m:~
\exists s \in [p_k] \cap \{r\}^c \text{ s.t. }
\big(\gb - \gb(\Ib_{p_{k - 1}}, \Ab_k, \bTheta_k)\big)_{(r, s)} > 0
\big\}
.
$$
By definition, we find that $\cG_{m, r}(\Ab_k, \bTheta_k)$ consists of all vectors $\gb(\Ib_{p_{k - 1}} + \Eb_{i, j}, \Ab_k, \bTheta_k)$'s with $(i, j) \in \langle p_{k - 1} \rangle$ satisfying $A_{k, r, i} = 1$ or $A_{k, r, j} = 1$, i.e.
$$
\cG_{m, r}(\Ab_k, \bTheta_k)
=
\big\{
\gb(\Ib_{p_{k - 1}} + \Eb_{i, j}, \Ab_k, \bTheta_k):~
(i, j) \in \langle p_{k - 1} \rangle,~
(A_{k, r, i} = 1 \text{ or } A_{k, r, j} = 1)
\big\}
.
$$
For any $(r, s) \in \langle p_k \rangle$, we notice that
$$
\cG_{m, r}(\Ab_k, \bTheta_k)
\supset
\cG_{m, s}(\Ab_k, \bTheta_k)
\quad\iff\quad
\ab_{k, r}
\succeq
\ab_{k, s}
\quad\text{or}\quad
\one^\top \ab_{k, r} \ge p_{k - 1} - 1
.
$$
Since $\Ab \in \cA \subset \cA_{2, 1} \cap \cA_{2, 2} \cap \cA_{2, 3}$, we know that the columns of $\Ab_k$ are distinct and each column contains both an entry of zero and an entry of one.
We let $\cG_m^{(\ell)}(\Ab_k, \bTheta_k)$ for $\ell \in [p_{k - 1}]$ denote the subset of $\cG_m$ given by
$$
\big\{
\gb(\Ib_{p_{k - 1}} + \Eb_{i, j}, \Ab_k, \bTheta_k):~
(i, j) \in \langle p_{k - 1} \rangle,~
(i = \ell \text{ or } j = \ell)
\big\}
,
$$
then it satisfies that for all $\ell \in [p_{k - 1}]$, there exists $r, s \in [p_k]$ such that
$$
\cG_m^{(\ell)}(\Ab_k, \bTheta_k)
\subset
\cG_{m, r}(\Ab_k, \bTheta_k)
,\quad
\cG_m^{(\ell)}(\Ab_k, \bTheta_k)
\not\subset
\cG_{m, s}(\Ab_k, \bTheta_k)
,
$$
and for any two distinct $\ell, \ell' \in [p_{k - 1}]$, there exists $r \in [p_k]$ such that either
$$
\cG_m^{(\ell)}(\Ab_k, \bTheta_k)
\subset
\cG_{m, r}(\Ab_k, \bTheta_k)
,\quad
\cG_m^{(\ell')}(\Ab_k, \bTheta_k)
\not\subset
\cG_{m, r}(\Ab_k, \bTheta_k)
$$
or
$$
\cG_m^{(\ell')}(\Ab_k, \bTheta_k)
\subset
\cG_{m, r}(\Ab_k, \bTheta_k)
,\quad
\cG_m^{(\ell)}(\Ab_k, \bTheta_k)
\not\subset
\cG_{m, r}(\Ab_k, \bTheta_k)
.
$$
Additionally, we notice for each $r \in [p_k]$ that
$$
\cG_{m, r}(\Ab_k, \bTheta_k)
=
\bigcup_{\ell \in [p_{k - 1}]:~ A_{k, r, \ell} = 1} \cG_m^{(\ell)}(\Ab_k, \bTheta_k)
.
$$
Therefore, the collection of sets $\big\{ \cG_m^{(\ell)}(\Ab_k, \bTheta_k):~ \ell \in [p_{k - 1}] \big\}$ is the unique subset of size $p_{k - 1}$ of the power set $2^{\cG_m(\Ab_k, \bTheta_k)}$ satisfying that each $\cG_{m, r}$ can be written as a finite union of the sets in this collection.
This implies that the collection $\big\{ \cG_m^{(\ell)}(\Ab_k, \bTheta_k):~ \ell \in [p_{k - 1}] \big\}$ is also well-defined given $\cG_{m, r}$, without knowing $\Ab_k$ or which $\Xb_{k - 1}$ each $\gb \in \cG_m$ corresponds to.
While $\cG_{m, r}(\Ab_k, \bTheta_k)$ differs for different $\Ab_k$, by its definition, the collection $\big\{ \cG_m^{(\ell)}(\Ab_k, \bTheta_k):~ \ell \in [p_{k - 1}] \big\}$ is identical for all $\Ab_k, \bTheta_k$ with $(\Ab, \bTheta) \in \cA_1 \times \cT_0$.
Additionally, we also note that each $\gb(\Ib_{p_{k - 1}} + \Eb_{i, j}, \Ab_k, \bTheta_k)$ belongs to exactly two of the sets $\cG_m^{(\ell)}(\Ab_k, \bTheta_k)$'s.
Therefore, there exists a permutation map $\sQ$ on $[p_{k - 1}]$ satisfying
\begin{equation}\label{eq:perm_Q2}
\{\gb(\Ib_{p_{k - 1}} + \Eb_{i, j}, \Ab_k, \bTheta_k)\}
=
\cG_m^{(\sQ(i))}(\Ab_k, \bTheta_k) \cap \cG_m^{(\sQ(j))}(\Ab_k, \bTheta_k)
\end{equation}
for all $(i, j) \in \langle p_{k - 1} \rangle$.

Now for any $(\Ab, \bTheta), (\tilde\Ab, \tilde\bTheta) \in \cA_2 \times (\cT_0 \cap \cT_1^c)$, if \eqref{eq:generic_sX_equiv_cond} holds for all $\Xb_{k - 1} \in \cX_{k - 1}$, then we have
$$
\cG(\Ab_k, \bTheta_k)
=
\cG(\tilde\Ab_k, \tilde\bTheta_k)
,\quad
\cG_0(\Ab_k, \bTheta_k)
=
\cG_0(\tilde\Ab_k, \tilde\bTheta_k)
,\quad
\cG_m(\Ab_k, \bTheta_k)
=
\cG_m(\tilde\Ab_k, \tilde\bTheta_k)
,
$$
and additionally,
$$
\big\{
\cG_m^{(\ell)}(\Ab_k, \bTheta_k):~
\ell \in [p_{k - 1}]
\big\}
=
\big\{
\cG_m^{(\ell)}(\tilde\Ab_k, \tilde\bTheta_k):~
\ell \in [p_{k - 1}]
\big\}
.
$$
From \eqref{eq:perm_Q2} we know there further exists a permutation map $\tilde\sQ$ on $[p_{k - 1}]$ such that
$$
\gb(\Ib_{p_{k - 1}} + \Eb_{i, j}, \Ab_k, \bTheta_k)
=
\gb(\Ib_{p_{k - 1}} + \Eb_{\tilde\sQ(i), \tilde\sQ(j)}, \tilde\Ab_k, \tilde\bTheta_k)
,
$$
which suggests
$$
\sP(\Ib_{p_{k - 1}})
=
\Ib_{p_{k - 1}}
,\quad
\sP(\Ib_{p_{k - 1}} + \Eb_{i, j})
=
\Ib_{p_{k - 1}} + \Eb_{\tilde\sQ(i), \tilde\sQ(j)}
\quad
\forall (i, j) \in \langle p_{k - 1} \rangle
.
$$
From \eqref{eq:cH_gamma} and $\bTheta \notin \cT_1$, for every vector $\gb \in \cG_0(\Ab_k, \bTheta_k)$, we can decompose the entries of $\gb - \gb(\Ib_{p_{k - 1}}, \Ab_k, \bTheta_k)$ into linear combinations with non-negative integer coefficients of $\Gamma_{k, i, j}$'s and denote the collection of all $\Gamma_{k, i, j}$'s with at least one positive coefficient by $\cH(\gb)$, i.e.
\begin{align*}
\cH(\gb)
:=
\Bigg\{
\Gamma_{k, i, j}:~
(i, j) \in \langle p_{k - 1} \rangle,~
\exists (s, t) \in \langle p_k \rangle,~
\zb \in \ZZ_+^{\frac{p_{k - 1} (p_{k - 1} - 1)}{2}}
\text{ s.t. }
z_{(i, j)} > 0,~
\\
(\gb - \gb(\Ib_{p_{k - 1}}, \Ab_k, \bTheta_k))_{(s, t)}
=
\sum_{(u, v) \in \langle p_{k - 1} \rangle} z_{(u, v)} \Gamma_{k, u, v}
\Bigg\}
.
\end{align*}
Similarly, for every vector $\gb \in \cG_0(\tilde\Ab_k, \tilde\bTheta_k)$, we can decompose its entries into linear combinations with non-negative integer coefficients of $\tilde\Gamma_{k, i, j}$'s and denote the collection of all $\tilde\Gamma_{k, i, j}$'s with at least one positive coefficient by $\tilde\cH(\gb)$.
Then for all $\Xb_{k - 1} \in \cX_{k - 1}$, we have the following equivalence:
\begin{align*}
\sP(\Xb_{k - 1})
\succeq
\Ib_{p_{k - 1}} + \Eb_{\tilde\sQ(i), \tilde\sQ(j)}
&\iff
\sP(\Xb_{k - 1})
\succeq
\sP(\Ib_{p_{k - 1}} + \Eb_{i, j})
\\&\iff
\tilde\cH\left(
\gb(\sP(\Xb_{k - 1}), \tilde\Ab_k, \tilde\bTheta_k)
\right)
\supset
\tilde\cH\left(
\gb(\sP(\Ib_{p_{k - 1}} + \Eb_{i, j}), \tilde\Ab_k, \tilde\bTheta_k)
\right)
\\&\iff
\cH\left(
\gb(\Xb_{k - 1}, \Ab_k, \bTheta_k)
\right)
\supset
\cH\left(
\gb(\Ib_{p_{k - 1}} + \Eb_{i, j}, \Ab_k, \bTheta_k)
\right)
\\&\iff
\Xb_{k - 1}
\succeq
\Ib_{p_{k - 1}} + \Eb_{i, j}
.
\end{align*}
This suggests that the permutation map $\sP$ on $\cX_{k - 1}$ is uniquely determined through the permutation map $\tilde\sQ$ on $[p_{k - 1}]$, 
and the entries of $\sP(\Xb_{k - 1})$ satisfy the equivalence
$$
(\sP(\Xb_{k - 1}))_{\tilde\sQ(i), \tilde\sQ(j)} = 1
\quad
\iff
X_{k - 1, i, j} = 1
,\quad
\forall (i, j) \in \langle p_{k - 1} \rangle
.
$$
Letting $\Pb \in \RR^{p_{k - 1} \times p_{k - 1}}$ denote the permutation matrix associated with the permutation map $\tilde\sQ$, i.e.
$$
\big(
1 ~ 2 ~ \cdots p_{k - 1}
\big) \Pb
=
\big(
\tilde\sQ(1) ~ \tilde\sQ(2) ~ \cdots \tilde\sQ(p_{k - 1})
\big)
,
$$
then we have
$$
\sP(\Xb_{k - 1})
=
\Pb \Xb_{k - 1} \Pb^\top
,\quad
\forall \Xb_{k - 1} \in \cX_{k - 1}
.
$$
\end{proof}

\begin{proof}[Proof of Lemma \ref{lemm:generic_param}]
Part of this proof is the same as in the proof of Lemma \ref{lemm:strict_param}, which we repeat in the following for completeness.
We notice that
$$
\law(\Xb_k | \Xb_{k - 1}, \Ab_k, \bTheta_k)
=
\law(\Xb_k | \Xb_{k - 1}, \tilde\Ab_k, \tilde\bTheta_k)
,\quad
\forall \Xb_{k - 1} \in \cX_{k - 1}
$$
is equivalent to
$$
\law(X_{k, i, j} | \Xb_{k - 1}, \Ab_k, \bTheta_k)
=
\law(X_{k, i, j} | \Xb_{k - 1}, \tilde\Ab_k, \tilde\bTheta_k)
,\quad
\forall \Xb_{k - 1} \in \cX_{k - 1},~ (i, j) \in \langle p_k \rangle
,
$$
which is further equivalent to
\begin{equation}\label{eq:generic_param_equiv}
C_k + \ab_{k, i}^\top (\bGamma_k * \Xb_{k - 1}) \ab_{k, j}
=
\tilde{C}_k + \tilde\ab_{k, i}^\top (\tilde\bGamma_k * \Xb_{k - 1}) \tilde\ab_{k, j}
,\quad
\forall \Xb_{k - 1} \in \cX_{k - 1},~ (i, j) \in \langle p_k \rangle
.
\end{equation}
We define the two sets of real values
$$
\cH
:=
\Big\{
(C_k + \ab_{k, i}^\top (\bGamma_k * \Xb_{k - 1}) \ab_{k, j}) - (C_k + \ab_{k, i}^\top (\bGamma_k * \Ib_{p_{k - 1}}) \ab_{k, j}):~
\Xb_{k - 1} \in \cX_{k - 1},~
(i, j) \in \langle p_{k} \rangle
\Big\}
$$
and
$$
\tilde\cH
:=
\Big\{
(\tilde{C}_k + \tilde\ab_{k, i}^\top (\tilde\bGamma_k * \Xb_{k - 1}) \tilde\ab_{k, j}) - (\tilde{C}_k + \tilde\ab_{k, i}^\top (\tilde\bGamma_k * \Ib_{p_{k - 1}}) \tilde\ab_{k, j}):~
\Xb_{k - 1} \in \cX_{k - 1},~
(i, j) \in \langle p_{k} \rangle
\Big\}
,
$$
then by \eqref{eq:generic_param_equiv} we have $\cH = \tilde\cH$.
As discussed earlier in \eqref{eq:cH_Gammakij} and \eqref{eq:cH_gamma} in the proof of Lemma \ref{lemm:generic_sX}, since $\Ab \in \cA_2$ and $\bTheta \in \cT_0 \cap \cT_1^c$, the collection
$$
\big\{
\Gamma_{k, s, t}:~
(s, t) \in \langle p_{k - 1} \rangle
\big\}
$$
is the unique subset of $\cH$ of size $\frac{p_{k - 1} (p_{k - 1} - 1)}{2}$ such that every element of $\cH$ is a linear combination with non-negative integers of the subset elements.
Similarly, this relation holds for $\tilde\cH$ and its subset $\{\tilde\Gamma_{k, s, t}:~ (s, t) \in \langle p_{k - 1} \rangle\}$, which implies that the two subsets
\begin{equation}\label{eq:generic_gamma_set}
\big\{
\Gamma_{k, s, t}:~
(s, t) \in \langle p_{k - 1} \rangle
\big\}
=
\big\{
\tilde\Gamma_{k, s, t}:~
(s, t) \in \langle p_{k - 1} \rangle
\big\}
.
\end{equation}

We further notice that for any $(s, t) \in \langle p_{k - 1} \rangle$,
$$
(C_k + \ab_{k, i}^\top (\bGamma_k * (\Ib_{p_{k - 1}} + \Eb_{s, t})) \ab_{k, j}) - (C_k + \ab_{k, i}^\top (\bGamma_k * \Ib_{p_{k - 1}}) \ab_{k, j})
=
(A_{k, i, s} A_{k, j, t} + A_{k, i, t} A_{k, j, s}) \Gamma_{k, s, t}
$$
holds for all $(i, j) \in \langle p_k \rangle$.
This implies that the set
$$
\big\{
(C_k + \ab_{k, i}^\top (\bGamma_k * (\Ib_{p_{k - 1}} + \Eb_{s, t})) \ab_{k, j}) - (C_k + \ab_{k, i}^\top (\bGamma_k * \Ib_{p_{k - 1}}) \ab_{k, j}):~
(i, j) \in \langle p_k \rangle
\big\}
$$
is a subset of $\{0, \Gamma_{k, s, t}, 2\Gamma_{k, s, t}\}$.
Similarly, the set
$$
\big\{
(\tilde{C}_k + \tilde\ab_{k, i}^\top (\tilde\bGamma_k * (\Ib_{p_{k - 1}} + \Eb_{s, t})) \tilde\ab_{k, j}) - (\tilde{C}_k + \tilde\ab_{k, i}^\top (\tilde\bGamma_k * \Ib_{p_{k - 1}}) \tilde\ab_{k, j}):~
(i, j) \in \langle p_k \rangle
\big\}
$$
is a subset of $\{0, \tilde\Gamma_{k, s, t}, 2\tilde\Gamma_{k, s, t}\}$.
Along with \eqref{eq:generic_param_equiv}, \eqref{eq:generic_gamma_set}, and $\bTheta \notin \cT_1$, we must have
$$
\Gamma_{k, s, t}
=
\tilde\Gamma_{k, s, t}
,\quad
\forall (s, t) \in \langle p_{k - 1} \rangle
.
$$

We define a partial order $\succeq$ on $\{0, 1\}^{p_{k - 1}}$ as
$$
\ab_{k, r}
\succeq
\ab_{k, s}
\quad\iff\quad
A_{k, r, j}
\ge
A_{k, s, j}
\quad
\forall j \in [p_{k - 1}]
.
$$
Now for each $(i, j) \in \langle p_k \rangle$ and $\Xb_{k - 1} \in \cX_{k - 1}$, we let $\zb_{(i, j), \Xb_{k - 1}}$ denote the integer coefficients of $\Gamma_{k, s, t}$'s with which the value
$$
(C_k + \ab_{k, i}^\top (\bGamma_k * \Xb_{k - 1}) \ab_{k, j}) - (C_k + \ab_{k, i}^\top (\bGamma_k * \Ib_{p_{k - 1}}) \ab_{k, j})
$$
can be uniquely expressed as a linear combination of $\Gamma_{k, s, t}$'s.
For instance, if $(C_k + \ab_{k, i}^\top (\bGamma_k * \Xb_{k - 1}) \ab_{k, j}) - (C_k + \ab_{k, i}^\top (\bGamma_k * \Ib_{p_{k - 1}}) \ab_{k, j}) = \Gamma_{k, 2, 1} + 2\Gamma_{k, 3, 1}$, then $\zb_{(i, j), \Xb_{k - 1}}$ is the $\frac{p_{k - 1}(p_{k - 1} - 1)}{2}$-dimensional vector with its entry corresponding to $\Gamma_{k, 2, 1}$ being 1, its entry corresponding to $\Gamma_{k, 3, 1}$ being 2, and all other entries being 0.
Similarly, we let $\tilde\zb_{(i, j), \Xb_{k - 1}}$ denote the integer coefficients of $\tilde\Gamma_{k, s, t}$'s of which the value
$$
(\tilde{C}_k + \tilde\ab_{k, i}^\top (\tilde\bGamma_k * \Xb_{k - 1}) \tilde\ab_{k, j}) - (\tilde{C}_k + \tilde\ab_{k, i}^\top (\tilde\bGamma_k * \Ib_{p_{k - 1}}) \tilde\ab_{k, j})
$$
can be uniquely expressed as a linear combination of $\tilde\Gamma_{k, s, t}$'s.
From \eqref{eq:generic_param_equiv}, we have
\begin{equation}\label{eq:generic_zb_equal}
\zb_{(i, j), \Xb_{k - 1}}
=
\tilde\zb_{(i, j), \Xb_{k - 1}}
,\quad
\forall (i, j) \in \langle p_k \rangle,~
\Xb_{k - 1} \in \cX_{k - 1}
.
\end{equation}
Then for each $i \in [p_k]$, we have the equivalence
\begin{align*}
\ab_{k, i} = \zero
&\iff
\forall j \in [p_k], j \ne i,~
\zb_{(i, j), \one \one^\top} = \zero
\\&\iff
\forall j \in [p_k], j \ne i,~
\tilde\zb_{(i, j), \one \one^\top} = \zero
\\&\iff
\tilde\ab_{k, i} = \zero
,
\end{align*}
which suggests that $\Ab_k$ and $\tilde\Ab_k$ have the same all-zero rows.

We exclude these all-zero rows and let $\cJ$ and $\tilde\cJ$ denote the collection of indices corresponding to the not all-zero rows in $\Ab_k$ and $\tilde\Ab_k$, then $\cJ = \tilde\cJ$.
For each row $i \in \cJ$ and each column $s \in [p_{k - 1}]$, we let $y_{i, s}$ be the zero-one indicator of the event
$$
\forall t \in [p_{k - 1}], t \ne s,~
\exists j \in \cJ, j \ne i,~
\text{s.t. }
(\zb_{(i, j), \Ib_{p_{k - 1}} + \Eb_{s, t}})_{(s, t)} > 0
,
$$
where $(\zb_{(i, j), \Ib_{p_{k - 1}} + \Eb_{s, t}})_{(s, t)}$ denotes the $(s, t)$th entry of the vector $\zb_{(i, j), \Ib_{p_{k - 1}} + \Eb_{s, t}}$, i.e. the integer coefficient of $\Gamma_{k, s, t}$ when writing the value
$$
(C_k + \ab_{k, i}^\top (\bGamma_k * (\Ib_{p_{k - 1}} + \Eb_{s, t})) \ab_{k, j}) - (C_k + \ab_{k, i}^\top (\bGamma_k * \Ib_{p_{k - 1}}) \ab_{k, j})
$$
as a linear combination of the entries in $\bGamma_k$.
Similarly, for each row $i \in \tilde\cJ$ and each column $s \in [p_{k - 1}]$, we let $\tilde{y}_{i, s}$ be the zero-one indicator of the event
$$
\forall t \in [p_{k - 1}], t \ne s,~
\exists j \in \tilde\cJ, j \ne i,~
\text{s.t. }
(\tilde\zb_{(i, j), \Ib_{p_{k - 1}} + \Eb_{s, t}})_{(s, t)} > 0
.
$$
From \eqref{eq:generic_zb_equal} we know that
$$
y_{i, s}
=
\tilde{y}_{i, s}
,\quad
\forall i \in \cJ = \tilde\cJ,~
s \in [p_{k - 1}]
.
$$
Importantly, we notice that $y_{i, s} = 1$ is equivalent to saying that for any other row $j \in \cJ, j \ne i$ at least one of the two rows $\ab_{k, i}, \ab_{k, j}$ has its $s$th entry equal to one.
That is, we have $y_{i, s} = 1$ if and only if either
\begin{equation}\label{eq:generic_yis1_cond1}
\ab_{k, i}
\succeq
\eb_s
\end{equation}
or
\begin{equation}\label{eq:generic_yis1_cond2}
\forall j \in \cJ, j \ne i,~
\ab_{k, j} \succeq \eb_s
.
\end{equation}

For each column $s \in [p_{k - 1}]$, we let $\yb_s$ and $\tilde\yb_s$ be the $|\cJ|$-dimensional vectors with its entries being $y_{i, s}$ and $\tilde{y}_{i, s}$ respectively, with $i \in \cJ$, then $\yb_s = \tilde\yb_s$.
We notice that if \eqref{eq:generic_yis1_cond1} holds for at least $|\cJ| - 1$ rows in $\cJ$, then we have \eqref{eq:generic_yis1_cond2} holding for the remaining row in $\cJ$ (if there is one remaining), such that all entries $y_{i, s} = 1$ and $\one^\top \yb_s = |\cJ|$.
If \eqref{eq:generic_yis1_cond1} holds for no more than $|\cJ| - 2$ rows in $\cJ$, then \eqref{eq:generic_yis1_cond2} can not hold for any row in $\cJ$, such that we have $\one^\top \yb_s \le |\cJ| - 2$.
In the latter case, for each column $s \in [p_{k - 1}]$ of $\Ab_k$ with $\one^\top \yb_s \le |\cJ| - 2$, we have for all $i \in \cJ = \tilde\cJ$ that
$$
\ab_{k, i} \succeq \eb_s
\quad\iff\quad
y_{i, s} = 1
\quad\iff\quad
\tilde{y}_{i, s} = 1
\quad\iff\quad
\tilde\ab_{k, i} \succeq \eb_s
,
$$
which implies that the matrices $\Ab_k$ and $\tilde\Ab_k$ are identical for all columns $s$ with $\one^\top \yb_s \le |\cJ| - 2$.
In the former case, for a column $s$ with $\one^\top \yb_s = |\cJ|$, restricting to the not all-zero rows in $\cJ$ of $\Ab_k$, there is at most one entry of zero in the column $s$ of $\Ab_k$.
In such a column $s$, we notice for all $i \in \cJ = \tilde\cJ$ that
\begin{align*}
\ab_{k, i} \succeq \eb_s
&\iff
\exists j, \ell \in \cJ \cap \{i\}^c, j \ne \ell,~
\exists t \in [p_{k - 1}], t \ne s,~
\text{s.t. }
(\zb_{(i, j), \one \one^\top})_{(s, t)}
\ne
(\zb_{(i, \ell), \one \one^\top})_{(s, t)}
\\&\iff
\exists j, \ell \in \tilde\cJ \cap \{i\}^c, j \ne \ell,~
\exists t \in [p_{k - 1}], t \ne s,~
\text{s.t. }
(\tilde\zb_{(i, j), \one \one^\top})_{(s, t)}
\ne
(\tilde\zb_{(i, \ell), \one \one^\top})_{(s, t)}
\\&\iff
\tilde\ab_{k, i} \succeq \eb_s
.
\end{align*}
This is because for a column $s$ with $\one^\top \yb_s = |\cJ|$, if for some row $i \in \cJ$ we have $y_{i, s} = 0$, then for any other row $j \in \cJ, j \ne i$, the integer coefficients of $\Gamma_{k, s, t}$'s for $t \in [p_{k - 1}], t \ne s$ in $\zb_{(i, j), \one \one^\top}$ are always the same.
If $y_{i, s} = 1$, these integer coefficients will not always be the same because the columns of $\Ab_k$ are distinct due to $\Ab \in \cA_{2, 2}$, which suggests that we can always find two distinct rows $j, \ell \in \cJ$ other than $i$ such that the rows $\ab_{k, j}$ and $\ab_{k, \ell}$ differ in some column $t \in [p_{k - 1}], t \ne s$, resulting in different coefficients of $\Gamma_{k, s, t}$ in $\zb_{(i, j), \one \one^\top}$ and $\zb_{(i, \ell), \one \one^\top}$ since $\ab_{k, i} \succeq \eb_s$.
Therefore, the above shows that the matrices $\Ab_k$ and $\tilde\Ab_k$ are also identical for all columns $s$ with $\one^\top \yb_s = |\cJ|$.
Combining the two arguments above that separately consider the columns with $\one^\top \yb_s \le |\cJ| - 2$ and the columns with $\one^\top \yb_s = |\cJ|$, we therefore have established that
$$
\Ab_k
=
\tilde\Ab_k
.
$$

It now remains to verify that $C_k = \tilde{C}_k$ and the diagonal entries $\diag(\bGamma_k) = \diag(\tilde\bGamma_k)$.
We recall the definition of the $\frac{p_k (p_k - 1)}{2} \times (1 + p_{k - 1})$ matrix $\Db_k$ from \eqref{eq:A23_def} as the stack of row vectors $\big( 1 ~ \ab_{k, i}^\top * \ab_{k, j}^\top \big)$ for $(i, j) \in \langle p_k \rangle$.
From \eqref{eq:generic_param_equiv}, for all index pairs $(i, j) \in \langle p_k \rangle$ we have
\begin{align*}
C_k + \sum_{s = 1}^{p_{k - 1}} A_{k, i, s} A_{k, j, s} \Gamma_{k, s, s}
&=
C_k + \ab_{k, i}^\top (\bGamma_k * \Ib_{p_{k - 1}}) \ab_{k, j}
\\&=
\tilde{C}_k + \tilde\ab_{k, i}^\top (\tilde\bGamma_k * \Ib_{p_{k - 1}}) \tilde\ab_{k, j}
\\&=
\tilde{C}_k + \sum_{s = 1}^{p_{k - 1}} \tilde{A}_{k, i, s} \tilde{A}_{k, j, s} \tilde\Gamma_{k, s, s}
.
\end{align*}
Letting $\ub_k$ denote the $(1 + p_{k - 1})$-dimensional vector $(C_k ~ \Gamma_{k, 1, 1} ~ \cdots ~ \Gamma_{k, p_{k - 1}, p_{k - 1}})$ and $\tilde\ub_k$ denote the $(1 + p_{k - 1})$-dimensional vector $(\tilde{C}_k ~ \tilde\Gamma_{k, 1, 1} ~ \cdots ~ \tilde\Gamma_{k, p_{k - 1}, p_{k - 1}})$, then the above implies that $\Db_k \ub_k = \Db_k \tilde\ub_k$, where we have used the earlier result of $\Ab_k = \tilde\Ab_k$.
Since $\Ab \in \cA_{2, 3}$, we have $\rank(\Db_k) = 1 + p_{k - 1}$, which implies that we must have $\ub_k = \tilde\ub_k$.
This implies
$$
C_k
=
\tilde{C}_k
,\quad
\Gamma_{k, s, s}
=
\tilde\Gamma_{k, s, s}
\quad
\forall
s \in [p_{k - 1}]
.
$$
This completes the proof of Lemma \ref{lemm:generic_param}.
\end{proof}

\subsection{Necessary Conditions for Identifiability}\label{ssec:iden_nece}

In the previous subsections, we have proven the strict identifiability of our model with parameter space $\cA_1 \times \cT_0$ and the generic identifiability of our model with parameter space $\cA_2 \times \cT_0$.
We now discuss the necessity of restricting the connection matrices $\Ab$ to $\cA_1$ and $\cA_2$.

An essential step of our proofs has been transforming the model identifiability problem into the uniqueness of three-way tensor CP decomposition.
While the Kruskal Theorem \ref{theo:kruskal} provides only sufficient but not necessary conditions for CP decomposition uniqueness, the conditions are relatively weak.
The parameter space $\cA_1$ is chosen to guarantee the Kruskal's conditions at every parameter $(\Ab, \bTheta)$, and the parameter space $\cA_{2, 1}$ is chosen to guarantee the Kruskal's conditions at almost every parameter $(\Ab, \bTheta)$ excluding a null set of zero measure under $\upmu \times \uplambda$.
While restricting the connection matrices $\Ab$ to $\cA_1$ or $\cA_{2, 1}$ is not exactly a necessary condition, it is hard to find any weaker conditions within our current proof framework utilizing the three-way tensor CP decomposition.

We recall from \eqref{eq:A2_def} the parameter space $\cA_2 = \cA_{2, 1} \cap \cA_{2, 2} \cap \cA_{2, 3}$ and $\cA_1 \subset \cA_2$.
In the following proposition, we show that restricting the connection matrices $\Ab$ to $\cA_{2, 3}$ is a necessary condition for both the strict and generic model identifiability.

\begin{proposition}\label{prop:iden_nece_A23}
Let space $\cA_* \subset \cA_0$ satisfy $\cA_* \cap \cA_{2, 3}^c \ne \varnothing$.
For the model with parameter space $\cA_* \times \cT_0$, we have
$$
(\upmu \times \uplambda)\big( \big\{
(\Ab, \bTheta) \in \cA_* \times \cT_0:~
\perm_{\cA_* \times \cT_0}(\Ab, \bTheta)
\subsetneq
\marg_{\cA_* \times \cT_0}(\Ab, \bTheta)
\big\} \big)
>
0
.
$$
\end{proposition}

Proposition \ref{prop:iden_nece_A23} suggests that a positive measure set of parameter $(\Ab, \bTheta)$ with $\Ab \notin \cA_{2, 3}$ is not identifiable, therefore showing the necessity of restricting $\Ab$ to $\cA_{2, 3}$.
Its proof is provided at the end of this subsection.

We also provide the following proposition, which establishes a condition slightly weaker than $\Ab \in \cA_{2, 2}$ that is necessary for strict and generic model identifiability.
We define the space
\begin{equation}\label{eq:tA22_def}
\tilde\cA_{2, 2}
:=
\Big\{
\Ab \in \cA_0:~
\forall k \in [K],~
\forall \Pb_{k - 1},~
\Ab_k \Pb_{k - 1} \ne \Ab_k
\text{ or }
\Pb_{k - 1} \Ab_{k - 1} \ne \Ab_{k - 1} (k \ge 2)
\Big\}
,
\end{equation}
where each $\Pb_{k - 1} \in \RR^{p_{k - 1} \times p_{k - 1}}$ is an arbitrary permutation matrix.
Note for $k = 1$, there is no matrix $\Ab_0$ and the condition of $\Ab_k \Pb_{k - 1} \ne \Ab_k$ or $\Pb_{k - 1} \Ab_{k - 1} \ne \Ab_{k - 1}$ reduces to just $\Ab_k \Pb_{k - 1} \ne \Ab_k$.
Recalling the definition of $\cA_{2, 2}$ from \eqref{eq:A22_def}, we notice that $\cA_{2, 2} \subseteq \tilde\cA_{2, 2} \subseteq \cA_0$.
For $K = 1$, we further have $\cA_{2, 2} = \tilde\cA_{2, 2}$.

\begin{proposition}\label{prop:iden_nece_tA22}
Let space $\cA_* \subset \cA_0$ satisfy $\cA_* \cap \tilde\cA_{2, 2}^c \ne \varnothing$.
For the model with parameter space $\cA_* \times \cT_0$, we have
$$
(\upmu \times \uplambda)\big( \big\{
(\Ab, \bTheta) \in \cA_* \times \cT_0:~
\perm_{\cA_* \times \cT_0}(\Ab, \bTheta)
\subsetneq
\marg_{\cA_* \times \cT_0}(\Ab, \bTheta)
\big\} \big)
>
0
.
$$
\end{proposition}

Proposition \ref{prop:iden_nece_tA22} suggests that a positive measure set of parameters $(\Ab, \bTheta)$ with $\Ab \notin \tilde\cA_{2, 2}$ is not identifiable, therefore showing the necessity of restricting $\Ab$ to $\tilde\cA_{2, 2}$.
Comparing the definitions of $\cA_{2, 2}$ and $\tilde\cA_{2, 2}$ in \eqref{eq:A22_def} and \eqref{eq:tA22_def}, conditions $\Ab \in \cA_{2, 2}$ and $\Ab \in \tilde\cA_{2, 2}$ have the same restriction on $\Ab_1$ of having distinct columns.
However, for $K \ge 2$, while the condition $\Ab \in \cA_{2, 2}$ still requires each $\Ab_k (k \ge 2)$ to have distinct columns, the condition $\Ab \in \tilde\cA_{2, 2}$ only requires for each pair of indices $(s, t) \in \langle p_{k - 1} \rangle (k \ge 2)$ that either the $s$th and $t$th column of $\Ab_k$ are distinct or the $s$th and $t$th row of $\Ab_{k - 1}$ are distinct.
Intuitively, the weaker requirement for layers $k \ge 2$ in $\tilde\cA_{2, 2}$ is due to the difference that $\Xb_{k - 1} (k = 1)$ is modeled by an unstructured marginal distribution with parameter $\bnu$, whereas $\Xb_{k - 1} (k \ge 2)$ is modeled by a structured marginal distribution through our model formulations \eqref{eq:model_entry} and \eqref{eq:model}.
While it is possible that $\tilde\cA_{2, 2}$ could replace $\cA_{2, 2}$ as part of the sufficient conditions for model identifiability, this is hard to establish within our current proof framework.

We now provide the proofs of Propositions \ref{prop:iden_nece_A23} and \ref{prop:iden_nece_tA22}.

\begin{proof}[Proof of Proposition \ref{prop:iden_nece_A23}]
For any $(\Ab, \bTheta) \in \cA_0 \times \cT_0$, there are only finitely many ways of permuting the nodes at each layer of our Bayesian network.
Specifically, we have
\begin{equation}\label{eq:iden_nece_A23_perm_fi}
\left|
\perm_{\cA_* \times \cT_0}(\Ab, \bTheta)
\right|
\le
\left|
\perm_{\cA_0 \times \cT_0}(\Ab, \bTheta)
\right|
\le
\prod_{k = 1}^K (p_{k - 1}!)
,\quad
\forall (\Ab, \bTheta) \in \cA_0 \times \cT_0
.
\end{equation}
Since $\cA_* \cap \cA_{2, 3}^c \ne \varnothing$, we can find $\Ab \in \cA_*$ such that $\Ab \notin \cA_{2, 3}$.
We focus on this choice of $\Ab$ for the remaining proof.

We recall from \eqref{eq:A23_def} that $\Db_k$ denotes the $\frac{p_k (p_k - 1)}{2} \times (1 + p_{k - 1})$ matrix obtained by stacking the rows $\big( 1 ~ \ab_{k, i}^\top * \ab_{k, j}^\top \big)$ for $(i, j) \in \langle p_k \rangle$ together.
By definition of $\cA_{2, 3}$, there exists some $k \in [K]$ such that $\rank(\Db_k) < p_{k - 1} + 1$.
We denote the $p_{k - 1}$-dimensional vector consisting of the diagonal entries of $\bGamma_k$ by $\bgamma_k$ and define the extended $(1 + p_{k - 1})$-dimensional vector $\tilde\bgamma_k^\top := \big( C_k ~ \bgamma_k^\top \big)$.
For the rows of the matrix $\Db_k$, we denote them by the vectors $\db_{k, (i, j)}^\top := \big( 1 ~ \ab_{k, i}^\top * \ab_{k, j}^\top \big)$ for $(i, j) \in \langle p_k \rangle$.

From our model formulation \eqref{eq:model_entry}, we can rewrite the entrywise conditional distribution of $X_{k, i, j}$ given $\Xb_{k - 1}, \Ab_k, \bTheta_k$ as
$$
\PP(X_{k, i, j} | \Xb_{k - 1}, \Ab_k, \bTheta_k)
=
\frac{\exp\left( X_{k, i, j} \left(
\db_{k, (i, j)}^\top \tilde\bgamma_k + \left\langle
\offdiag(\bGamma_k) * \Xb_{k - 1}
,~
\ab_{k, i} \ab_{k, j}^\top
\right\rangle \right) \right)}{1 + \exp\left(
\db_{k, (i, j)}^\top \tilde\bgamma_k + \left\langle
\offdiag(\bGamma_k) * \Xb_{k - 1}
,~
\ab_{k, i} \ab_{k, j}^\top
\right\rangle \right)}
.
$$
We notice that the conditional distribution of $X_{k, i, j} | \Xb_{k - 1}, \Ab_k, \bTheta_k$ has dependency on $C_k$ and the diagonal entries of $\bGamma_k$ solely through the term $\db_{k, (i, j)}^\top \tilde\bgamma_k$.
Let $\bTheta' \in \cT_0$ be any parameter identical to $\bTheta$ except for $C_k$ and the diagonal entries of $\bGamma_k$, i.e. satisfying
\begin{equation}\label{eq:nece_iden_except_k}
\bnu'
=
\bnu
,\quad
\offdiag(\bGamma_k')
=
\offdiag(\bGamma_k)
,\quad
\forall \ell \in [K], \ell \ne k,\quad
C_\ell'
=
C_\ell
,\quad
\bGamma_\ell'
=
\bGamma_\ell
.
\end{equation}
For such a parameter $\bTheta'$, we automatically have
$$
\law(\Xb_\ell | \Xb_{\ell - 1}, \Ab_\ell, \bTheta_\ell')
=
\law(\Xb_\ell | \Xb_{\ell - 1}, \Ab_\ell, \bTheta_\ell)
$$
for all layers $\ell \in [K]$ other than $k$.
If $\bTheta_k'$ further satisfies for some $(i, j) \in \langle p_k \rangle$ that $\db_{k, (i, j)}^\top \tilde\bgamma_k' = \db_{k, (i, j)}^\top \tilde\bgamma_k$, then we have
$$
\law(X_{k, i, j} | \Xb_{k - 1}, \Ab_k, \bTheta_k')
=
\law(X_{k, i, j} | \Xb_{k - 1}, \Ab_k, \bTheta_k)
.
$$
Therefore, if parameter $\bTheta' \in \cT_0$ satisfies \eqref{eq:nece_iden_except_k} and the condition $\Db_k \tilde\bgamma_k' = \Db_k \tilde\bgamma_k$, we will have
$$
\law(\Xb_\ell | \Xb_{\ell - 1}, \Ab_\ell, \bTheta_\ell')
=
\law(\Xb_\ell | \Xb_{\ell - 1}, \Ab_\ell, \bTheta_\ell)
,\quad
\forall \ell \in [K]
,
$$
which suggests $(\Ab, \bTheta') \in \marg_{\cA_0 \times \cT_0}(\Ab, \bTheta)$.

Since the matrix $\Db_k$ does not have full column rank for $\Ab \notin \cA_{2, 3}$, its null space is at least one-dimensional.
Letting $\epsilon = \frac12 \min\{\Gamma_{k, 1, 1}, \ldots, \Gamma_{k, p_{k - 1}, p_{k - 1}}\}$, then the intersection of the null space of $\Db_k$ and the $\epsilon$-neighborhood of $\tilde\bgamma_k$, denoted by
$$
\cO_{k, \epsilon}
:=
\text{Null-Space}\left(
\tilde\Db_k
\right) \cap \left\{
\tilde\bgamma_k' \in \RR^{(1 + p_{k - 1}) \times 1}:~
\|\tilde\bgamma_k' - \tilde\bgamma_k\|_1 < \epsilon
\right\}
,
$$
has infinitely many elements.
We notice that any parameter $\bTheta_k'$ satisfying \eqref{eq:nece_iden_except_k} and $\tilde\bgamma_k' \in \cO_{k, \epsilon}$ belongs to the space $\cT_0$.
Therefore, $\marg_{\cA_* \times \cT_0}(\Ab, \bTheta)$ also contains infinitely many elements.

Combined with $\perm_{\cA_* \times \cT_0}(\Ab, \bTheta)$ having only finitely many elements due to \eqref{eq:iden_nece_A23_perm_fi}, we know
$$
\perm_{\cA_* \times \cT_0}(\Ab, \bTheta)
\subsetneq
\marg_{\cA_* \times \cT_0}(\Ab, \bTheta)
,\quad
\forall \Ab \in \cA_* \cap \cA_{2, 3}^c, \bTheta \in \cT_0
.
$$
This implies
\begin{align*}
&\quad~
(\upmu \times \uplambda)\big( \big\{
(\Ab, \bTheta) \in \cA_* \times \cT_0:~
\perm_{\cA_* \times \cT_0}(\Ab, \bTheta)
\subsetneq
\marg_{\cA_* \times \cT_0}(\Ab, \bTheta)
\big\} \big)
\\&\ge
(\upmu \times \uplambda)((\cA_* \cap \cA_{2, 3}^c) \times \cT_0)
\ge
\uplambda(\cT_0)
>
0
.
\end{align*}
\end{proof}

\begin{proof}[Proof of Proposition \ref{prop:iden_nece_tA22}]
Recalling the definition of space $\cT_1$ from \eqref{eq:T1_def}, we consider any parameter $\bTheta \in \cT_0 \cap \cT_1^c$.
Since $\cA_* \cap \tilde\cA_{2, 2}^c \ne \varnothing$, we can find $\Ab \in \cA_*$ such that $\Ab \notin \tilde\cA_{2, 2}$.
By definition of $\tilde\cA_{2, 2}$ from \eqref{eq:tA22_def}, we know there exists some layer $k \in [K]$ and two columns $s, t \in [p_{k - 1}]$ such that one of the following two conditions hold:
\begin{enumerate}
\item[(C1)]
For layer $k = 1$, the $s$th and $t$th columns in $\Ab_1$ are identical;
\item[(C2)]
For a layer $k \ge 2$, the $s$th and $t$th columns in $\Ab_k$ are identical and the $s$th and $t$th rows in $\Ab_{k - 1}$ are identical.
\end{enumerate}
We further notice that if $p_{k - 1} = 2$, then each of the vectors $\ab_{k, i} * \ab_{k, j}$ for $(i, j) \in \langle p_k \rangle$ can only be either the all-one vector $\one$ or the all-zero vector $\zero$, due to the two identical columns in $\Ab_k$.
This suggests that the matrix $\Db_k$ from \eqref{eq:A23_def} satisfies $\rank(\Db_k) \le 2 < p_{k - 1} + 1$, implying that $\Ab \notin \cA_{2, 3}$.
This is already covered by Proposition \ref{prop:iden_nece_A23}, so we restrict attention to $p_{k - 1} \ge 3$ in the following.
We now separately consider the cases of (C1) and (C2).

Suppose (C1) is true and consider layer $k = 1$.
Since $p_0 \ge 3$, we can find a third index $r \in [p_0] \cap \{s, t\}^c$.
We now construct a new $\bTheta' = \{\bnu'\} \cup \{C_\ell', \bGamma_\ell'\}_{\ell = 1}^K \in \cT_0$ based on $\bTheta$.
For all layers $\ell \ge 2$, we use the same parameters as in $\bTheta$, i.e.
$$
C_\ell'
:=
C_\ell
,\quad
\bGamma_\ell'
:=
\bGamma_\ell
,\quad
\forall \ell \in [2, K]
.
$$
For the layer $k = 1$, we use the same $C_1$ and swap the $(r, s)$th entry of $\bGamma_1$ with its $(r, t)$th entry, i.e.
$$
C_1'
:=
C_1
,\quad
\bGamma_1'
:=
\bGamma_1 + (\Gamma_{1, r, t} - \Gamma_{1, r, s})(\Eb_{r, s} - \Eb_{r, t})
,
$$
where $\Eb_{r, s} := \eb_r \eb_s^\top + \eb_s \eb_r^\top$ denotes the $p_0 \times p_0$ symmetric indicator matrix for entry $(r, s)$.
We let the probability vector $\bnu'$ be a shuffle of entries in $\bnu$ given by
$$
\nu_{\sP_0(\Xb_0)}'
:=
\nu_{\Xb_0}
,\quad
\forall \Xb_0 \in \cX_0
,
$$
where $\sP_0$ denotes the permutation on $\cX_0$ such that $\sP_0(\Xb_0)$ is obtained from $\Xb_0$ by exchanging its $(r, s)$th and $(r, t)$th entries (and also its $(s, r)$th and $(t, r)$th entries by symmetry), i.e.
$$
\sP_0(\Xb_0)
:=
\Xb_0 + (X_{0, r, t} - X_{0, r, s}) (\Eb_{r, s} - \Eb_{r, t})
,\quad
\forall \Xb_0 \in \cX_0
.
$$
Note that if we view $\sP_0$ as a general function on the space of symmetric matrices, then we also have $\bGamma_1' = \sP_0(\bGamma_1)$.
Through the above definitions, we have $\bTheta' \notin \cT_1$ as well, which suggests that any upper-triangular off-diagonal entries of $\bGamma_1$ are distinct.
Therefore, $\bTheta'$ can not be obtained from $\bTheta$ through permutation of nodes at each layer in our Bayesian network, which implies
\begin{equation}\label{eq:iden_nece_tA22_perm_C1}
(\Ab, \bTheta')
\notin
\perm_{\cA_* \times \cT_0}(\Ab, \bTheta)
,\quad
\forall \bTheta \in \cT_0 \cap \cT_1^c
.
\end{equation}
Since the $s$th and $t$th columns of $\Ab_1$ are identical, for any $(i, j) \in \langle p_1 \rangle$ we have
\begin{align*}
&\quad~
\left(
C_1' + \ab_{1, i}^\top (\sP_0(\Xb_0) * \bGamma_1') \ab_{1, j}
\right)
-
\left(
C_1 + \ab_{1, i}^\top (\Xb_0 * \bGamma_1) \ab_{1, j}
\right)
\\&=
\left\langle
\big(
\sP_0(\Xb_0 * \bGamma_1) - \Xb_0 * \bGamma_1
\big)
,~
\ab_{1, i} \ab_{1, j}^\top
\right\rangle
\\&=
\left(
(\Xb_0 * \bGamma_1)_{r, t} - (\Xb_0 * \bGamma_1)_{r, s}
\right) \left(
\offdiag\big( \ab_{1, i} \ab_{1, j}^\top \big)_{r, s}
-
\offdiag\big( \ab_{1, i} \ab_{1, j}^\top \big)_{r, t}
\right)
\\&=
(\Gamma_{1, r, t} X_{0, r, t} - \Gamma_{1, r, s} X_{0, r, s}) (A_{1, i, r} A_{1, j, s} + A_{1, j, r} A_{1, i, s} - A_{1, i, r} A_{1, j, t} - A_{1, j, r} A_{1, i, t})
\\&=
0
,
\end{align*}
which suggests the conditional distributions
\begin{align*}
\PP(X_{1, i, j} | \sP(\Xb_0), \Ab_1, \bTheta_1')
&=
\frac{\exp\left( X_{1, i, j} \left(
C_1' + \ab_{1, i}^\top (\sP_0(\Xb_0) * \bGamma_1') \ab_{1, j}
\right) \right)}{1 + \exp\left(
C_1' + \ab_{1, i}^\top (\sP_0(\Xb_0) * \bGamma_1') \ab_{1, j}
\right)}
\\&=
\frac{\exp\left( X_{1, i, j} \left(
C_1 + \ab_{1, i}^\top (\Xb_0 * \bGamma_1) \ab_{1, j}
\right) \right)}{1 + \exp\left(
C_1 + \ab_{1, i}^\top (\Xb_0 * \bGamma_1) \ab_{1, j}
\right)}
\\&=
\PP(X_{1, i, j} | \Xb_0, \Ab_1, \bTheta_1)
.
\end{align*}
Since $\PP(\sP_0(\Xb_0) | \bnu') = \PP(\Xb_0 | \bnu)$ is implied by the definition of $\bnu'$, taking summations we have
\begin{align*}
\PP(\Xb_1 | \Ab, \bTheta')
&=
\sum_{\Xb_0 \in \cX_0} \left(
\PP(\sP(\Xb_0) | \bnu') \prod_{(i, j) \in \langle p_1 \rangle} \PP(X_{1, i, j} | \sP(\Xb_0), \Ab_1, \bTheta_1')
\right)
\\&=
\sum_{\Xb_0 \in \cX_0} \left(
\PP(\Xb_0 | \bnu) \prod_{(i, j) \in \langle p_1 \rangle} \PP(X_{1, i, j} | \Xb_0, \Ab_1, \bTheta_1)
\right)
\\&=
\PP(\Xb_1 | \Ab, \bTheta)
.
\end{align*}
Since $\bTheta'$ and $\bTheta$ agree on all parameters for layer $\ell \ge 2$, the marginal distributions of $\Xb_K$ therefore satisfy $\law(\Xb_K | \Ab, \bTheta') = \law(\Xb_K | \Ab, \bTheta)$, which implies
$$
(\Ab, \bTheta')
\in
\marg_{\cA_* \times \cT_0}(\Ab, \bTheta)
,\quad
\forall \bTheta \in \cT_0 \cap \cT_1^c
.
$$
Combined with \eqref{eq:iden_nece_tA22_perm_C1}, we know for all $\Ab$ satisfying (C1) with $p_0 \ge 3$ and $\bTheta \in \cT_0 \cap \cT_1^c$, we have $\perm_{\cA_* \times \cT_0}(\Ab, \bTheta) \subsetneq \marg_{\cA_* \times \cT_0}(\Ab, \bTheta)$.

Now suppose (C2) is true for some layer $k \ge 2$.
Since $p_{k - 1} \ge 3$, we can find a third index $r \in [p_{k - 1}] \cap \{s, t\}^c$.
Contrary to the case of (C1), we now construct a new $\bTheta' = \{\bnu'\} \cup \{C_\ell', \bGamma_\ell'\}_{\ell = 1}^K \in \cT_0$ based on $\bTheta$ such that only the $(r, s)$th entry and $(r, t)$th entry in $\bGamma_k$ of layer $k$ are swapped (also swapping the entries $(s, r)$ with $(t, r)$ in $\bGamma_k$ by symmetry), while all other parameters including $\bnu$, $C_\ell (\ell \in [K])$, and $\bGamma_\ell (\ell \in [K], \ell \ne k)$ are kept the same.
Letting $\sP_{k - 1}$ be a function defined on the space of $p_{k - 1} \times p_{k - 1}$ symmetric matrices given by
$$
\sP_{k - 1}(\Bb)
:=
\Bb + (B_{r, t} - B_{r, s}) (\Eb_{r, s} - \Eb_{r, t})
,
$$
then our definition of $\bTheta'$ can be formally written as
$$
\bnu'
:=
\bnu
,\quad
C_\ell'
:=
C_\ell
\quad
\forall \ell \in [K]
,\quad
\bGamma_k'
:=
\sP_{k - 1}(\bGamma_k)
,\quad
\bGamma_\ell'
:=
\bGamma_\ell
\quad
\forall \ell \in [K] \cap \{k\}^c
.
$$
Since $\bTheta, \bTheta' \in \cT_0 \cap \cT_1^c$, we know $\bTheta'$ can not be obtained from $\bTheta$ through permutations, suggesting
\begin{equation}\label{eq:iden_nece_tA22_perm_C2}
(\Ab, \bTheta')
\notin
\perm_{\cA_* \times \cT_0}(\Ab, \bTheta)
,\quad
\forall \bTheta \in \cT_0 \cap \cT_1^c
.
\end{equation}
Since the $s$th and $t$th rows in $\Ab_{k - 1}$ are identical, we have
\begin{align*}
\PP\left(
X_{k - 1, r, s} = 1
~\big|~
\Xb_{k - 2}, \Ab_{k - 1}, \bTheta_{k - 1}'
\right)
&=
\frac{\exp\left(
C_{k - 1}' + \ab_{k - 1, r}^\top (\Xb_{k - 2} * \bGamma_{k - 2}') \ab_{k - 1, s}
\right)}{1 + \exp\left(
C_{k - 1}' + \ab_{k - 1, r}^\top (\Xb_{k - 2} * \bGamma_{k - 2}') \ab_{k - 1, s}
\right)}
\\&=
\frac{\exp\left(
C_{k - 1} + \ab_{k - 1, r}^\top (\Xb_{k - 2} * \bGamma_{k - 2}) \ab_{k - 1, t}
\right)}{1 + \exp\left(
C_{k - 1} + \ab_{k - 1, r}^\top (\Xb_{k - 2} * \bGamma_{k - 2}) \ab_{k - 1, t}
\right)}
\\&=
\PP\left(
X_{k - 1, r, t} = 1
~\big|~
\Xb_{k - 2}, \Ab_{k - 1}, \bTheta_{k - 1}
\right)
.
\end{align*}
Since $\bTheta'$ and $\bTheta$ agree for all layers $\ell \le k - 1$, we have marginal distributions $\law(\Xb_{k - 2} | \Ab, \bTheta') = \law(\Xb_{k - 2} | \Ab, \bTheta)$ and conditional distributions $\law(\Xb_{k - 1} | \Xb_{k - 2}, \Ab, \bTheta') = \law(\Xb_{k - 1} | \Xb_{k - 2}, \Ab, \bTheta)$.
Along with the above, this further suggests
$$
\PP(\sP_{k - 1}(\Xb_{k - 1}) | \Ab, \bTheta')
=
\PP(\Xb_{k - 1} | \Ab, \bTheta)
,\quad
\forall \Xb_{k - 1} \in \cX_{k - 1}
.
$$
Since the $s$th and $t$th columns in $\Ab_k$ are identical, using the same arguments as in the case of (C1), we can establish
$$
\PP(\Xb_k | \Ab, \bTheta')
=
\PP(\Xb_k | \Ab, \bTheta)
,\quad
\forall \Xb_k \in \cX_k
.
$$
Since $\bTheta'$ and $\bTheta$ agree for all layers $\ell \ge k + 1$, we can obtain $\law(\Xb_K | \Ab, \bTheta') = \law(\Xb_K | \Ab, \bTheta)$, which implies
$$
(\Ab, \bTheta')
\in
\marg_{\cA_* \times \cT_0}(\Ab, \bTheta)
,\quad
\forall \bTheta \in \cT_0 \cap \cT_1^c
.
$$
Combined with \eqref{eq:iden_nece_tA22_perm_C1}, we know for all $\Ab$ satisfying (C1) with $p_{k - 1} \ge 3$ and $\bTheta \in \cT_0 \cap \cT_1^c$, we have $\perm_{\cA_* \times \cT_0}(\Ab, \bTheta) \subsetneq \marg_{\cA_* \times \cT_0}(\Ab, \bTheta)$.

Putting together the cases of (C1) and (C2), we obtain
\begin{align*}
&\quad~
(\upmu \times \uplambda)\big( \big\{
(\Ab, \bTheta) \in \cA_* \times \cT_0:~
\perm_{\cA_* \times \cT_0}(\Ab, \bTheta)
\subsetneq
\marg_{\cA_* \times \cT_0}(\Ab, \bTheta)
\big\} \big)
\\&\ge
(\upmu \times \uplambda)((\cA_* \cap \tilde\cA_{2, 2}^c) \times (\cT_0 \cap \cT_1^c))
\\&\ge
\uplambda(\cT_0) - \uplambda(\cT_1)
>
0
.
\end{align*}

\end{proof}

\section{Matrices of Class \texorpdfstring{$\cM_d$}{}}\label{sec:class_M}

As introduced in Section \ref{sec:theo}, the binary square matrices from class $\cM_d$ serve as the appropriate extension of identity submatrices in each connection matrix $\Ab_k$ when relaxing the conditions of strict identifiability \eqref{eq:A1_def} to generic identifiability \eqref{eq:A2_def}.
In this section, we provide the intuition behind its Definition \ref{defi:M} in Section \ref{ssec:M_int}, discuss its connections with commonly used linear algebra properties in Section \ref{ssec:M_linalg}, and present numerous examples of matrices that belong to and do not belong to $\cM_d$ in Section \ref{ssec:M_eg}.
The proof of Proposition \ref{prop:simp_A2} is provided Section \ref{ssec:proof_prop_simp_A2}, which is largely based on properties of matrices in class $\cM_d$.

\subsection{Interpretations and Intuitions}\label{ssec:M_int}

We start by interpreting the definition of class $\cM_d$ through the space $\cA_{2, 1}$ of adjacency matrices.
For any adjacency matrix $\Ab \in \cA_{2, 1}$, the condition (i) in Definition \ref{defi:M} ensures that each community of our Bayesian network (node in the upper/deeper layer) contains at least two nodes in the lower/shallower layer.
This is a natural relaxation from $\Ab \in \cA_1$, where each $\Mb_k, \Mb_k' = \Ib_{p_{k - 1}}$ in \eqref{eq:A21_def} and each community contains at least two pure nodes.
The existence of nodes or pure nodes belonging to each community is a standard requirement for identifiability of community detection models, see e.g. \citet{jin2015fast}.

We now delve into the intuitive understanding of condition (ii) in Definition \ref{defi:M}.
Without loss of generality, we let $\Pb_k = \Ib_{p_k}$ in \eqref{eq:A21_def} and denote the $p_{k - 1} \times p_{k - 1}$ leading principal submatrix of $\Xb_k$ by $\tilde\Xb_k$.
For each $\Xb_{k - 1}, \tilde\Xb_k \in \cX_{k - 1}$, we notice that the likelihood of $\tilde\Xb_k | \Xb_{k - 1}, \Ab_k, \bTheta_k$ can be factorized as
\begin{align*}
\PP(\tilde\Xb_k | \Xb_{k - 1}, \Ab_k, \bTheta_k)
&=
\prod_{1 \le i < j \le p_{k - 1}} \frac{\exp\left( X_{k, i, j} \left(
C_k + \ab_{k, i}^\top (\bGamma_k * \Xb_{k - 1}) \ab_{k, j}
\right) \right)}{1 + \exp\left(
C_k + \ab_{k, i}^\top (\bGamma_k * \Xb_{k - 1}) \ab_{k, j}
\right)}
\\&=
\prod_{1 \le i < j \le p_{k - 1}} \frac{\exp\left( X_{k, i, j} \left(
C_k + \frac12 \left\langle
\bGamma_k * \Xb_{k - 1}
,~
(\ab_{k, j} \ab_{k, i}^\top + \ab_{k, i} \ab_{k, j}^\top)
\right\rangle \right) \right)}{1 + \exp\left(
C_k + \ab_{k, i}^\top (\bGamma_k * \Xb_{k - 1}) \ab_{k, j}
\right)}
\\&=
\exp\left(
\sum_{1 \le i < j \le p_{k - 1}} X_{k, i, j} \left(
C_k + \tr\left(
\diag(\ab_{k, i} \ab_{k, j}^\top) \bGamma_k
\right) \right) \right)
\\&\quad \times
\exp\left( \frac12 \left\langle
\bGamma_k * \Xb_{k - 1}
,
\sum_{(i, j) \in \cS(\tilde\Xb_k)} \offdiag(\ab_{k, i} \ab_{k, j}^\top + \ab_{k, j} \ab_{k, i}^\top)
\right\rangle \right)
\\&\quad \times
\prod_{1 \le i < j \le p_{k - 1}} \left( 1 + \exp\left(
C_k + \ab_{k, i}^\top (\bGamma_k * \Xb_{k - 1}) \ab_{k, j}
\right) \right)^{-1}
\\&:=
\kT_1(\tilde\Xb_k) \cdot \kT_2(\Xb_{k - 1}, \tilde\Xb_k) \cdot \kT_3(\Xb_{k - 1})
,
\end{align*}
where $\cS(\tilde\Xb_k) \in 2^{\langle p_{k - 1} \rangle}$ denotes the collection of all index pairs $(i, j) \in \langle p_{k - 1} \rangle$ with $X_{k, i, j} = 1$, i.e.
$$
\cS(\tilde\Xb_k)
:=
\big\{
(i, j) \in \langle p_{k - 1} \rangle:~
X_{k, i, j} = 1
\big\}
.
$$
We consider the $|\cX_{k - 1}| \times |\cX_{k - 1}|$ probability matrix $\tilde\bLambda_k$, with each column corresponding to a distinct $\Xb_{k - 1} \in \cX_{k - 1}$, each column corresponding to a distinct $\tilde\Xb_k \in \cX_{k - 1}$, and each entry corresponding to the probability $\PP(\tilde\Xb_k | \Xb_{k - 1}, \Ab_k, \bTheta_k)$.
We similarly define a matrix $\tilde\bDelta_k$ with each column corresponding to a distinct $\Xb_{k - 1} \in \cX_{k - 1}$, each column corresponding to a distinct $\tilde\Xb_k \in \cX_{k - 1}$, and each entry corresponding to the term $\kT_2(\Xb_{k - 1}, \tilde\Xb_k)$.
For instance, let $\sX$ be a bijection from $\cX_{k - 1}$ to $[|\cX_{k - 1}|]$, then we can let the entries
$$
(\tilde\bLambda_k)_{\sX(\tilde\Xb_k), \sX(\Xb_{k - 1})}
:=
\PP(\tilde\Xb_k | \Xb_{k - 1}, \Ab_k, \bTheta_k)
,\quad
(\tilde\bDelta_k)_{\sX(\tilde\Xb_k), \sX(\Xb_{k - 1})}
:=
\kT_2(\Xb_{k - 1}, \tilde\Xb_k)
.
$$
To apply Kruskal Theorem \ref{theo:kruskal}, it is desirable to establish $\rank_K(\tilde\bLambda_k) = |\cX_{k - 1}|$, that is, $\tilde\bLambda_k$ being invertible.
Since the term $\kT_1(\tilde\Xb_k)$ does not depend on the $\Xb_{k - 1}$ corresponding to each column and the term $\kT_3(\Xb_{k - 1})$ does not depend on the $\tilde\Xb_k$ corresponding to each row, we find that $\tilde\bLambda_k$ is invertible if and only if $\tilde\bDelta_k$ is invertible.
We further notice that the matrix $\tilde\bDelta_k$ has distinct rows if and only if the submatrix $\Mb_k$ of $\Ab_k$ in \eqref{eq:A21_def} satisfies the condition (ii) in Definition \ref{defi:M}.
While having distinct rows is in general weaker than being invertible, it allows us to establish invertibility under at least one specific choice of continuous parameters $C_k, \bGamma_k$.
This is enough for proving generic identifiability using the theory of holomorphic function and its zero sets.
We refer to Section \ref{sec:iden} for the details.

\subsection{Connections with Other Linear Algebra Properties}\label{ssec:M_linalg}

We now discuss the connections between binary square matrices in class $\cM_d$ and other linear algebra properties, including distinct rows, distinct columns, symmetry, invertibility, rank, and Kruskal rank.
For $d \le 2$, since the set $2^{\langle d \rangle}$ contains at most one pair of indices, the condition (ii) in Definition \ref{defi:M} holds automatically.
We restrict attention to $d \ge 3$ in the following.

\paragraph{Distinct columns}
Any binary square matrix $\Mb \in \cM_d$ must have distinct columns.

\begin{proof}
We prove by contradiction.
Let $s, t \in [d]$ be the indices of two identical columns $\mb_s, \mb_t$ in $\Mb \in \cM_d$ and let $\ell \in [d]$ be a distinct third index.
Then for the two distinct subsets $\cS_1, \cS_2 \in 2^{\langle d \rangle}$ defined as $\cS_1 := \{(s, \ell)\}$, $\cS_2 := \{(t, \ell)\}$, we have
$$
\offdiag\Big(
\sum_{(i, j) \in \cS_1} (\mb_i \mb_j^\top + \mb_j \mb_i^\top)
\Big)
=
\offdiag\Big(
\sum_{(i, j) \in \cS_2} (\mb_i \mb_j^\top + \mb_j \mb_i^\top)
\Big)
,
$$
contradicting the condition (ii) of class $\cM_d$ in Definition \ref{defi:M}.
\end{proof}

Interestingly, the property of having distinct columns is also not too much weaker than the condition (ii) of class $\cM_d$.
For binary square matrices with an all-ones diagonal and distinct columns, a large proportion of them belong to class $\cM_d$.
For $d = 3, 4, 5$, the proportions are approximately 53\%, 88\%, and 86\%.

\paragraph{Distinct rows}
Neither the conditions of class $\cM_d$ nor the property of having distinct rows leads to the other.
The following are four examples of binary square matrices, where only $\Mb_1, \Mb_2$ belong to class $\cM_d$ and only $\Mb_1, \Mb_3$ have distinct rows.
$$
\Mb_1
:=
\left(\begin{matrix}
1 & 0 & 0 \\
0 & 1 & 0 \\
0 & 0 & 1
\end{matrix}\right)
,\quad
\Mb_2
:=
\left(\begin{matrix}
1 & 1 & 0 & 0 & 1 \\
0 & 1 & 1 & 1 & 0 \\
1 & 0 & 1 & 0 & 0 \\
1 & 0 & 1 & 1 & 0 \\
1 & 1 & 0 & 0 & 1
\end{matrix}\right)
,\quad
\Mb_3
:=
\left(\begin{matrix}
1 & 1 & 0 \\
0 & 1 & 1 \\
1 & 0 & 1
\end{matrix}\right)
,\quad
\Mb_4
:=
\left(\begin{matrix}
1 & 1 & 0 \\
1 & 1 & 0 \\
0 & 0 & 1
\end{matrix}\right)
.
$$
Specifically for $d = 3$ and $4$, we find through exhaustive enumeration that all matrices in $\cM_3$ and $\cM_4$ have distinct rows.
However, this no longer holds for $d \ge 5$, with $\Mb_2$ being a counterexample.

\paragraph{Symmetry}
Neither the conditions of class $\cM_d$ nor the property of symmetry leads to the other.
The following are four examples of binary square matrices, where only $\Mb_1, \Mb_2$ belong to class $\cM_d$ and only $\Mb_1, \Mb_3$ are symmetric.
$$
\Mb_1
:=
\left(\begin{matrix}
1 & 1 & 1 \\
1 & 1 & 0 \\
1 & 0 & 1
\end{matrix}\right)
,\quad
\Mb_2
:=
\left(\begin{matrix}
1 & 1 & 0 \\
0 & 1 & 0 \\
0 & 0 & 1
\end{matrix}\right)
,\quad
\Mb_3
:=
\left(\begin{matrix}
1 & 1 & 0 \\
1 & 1 & 0 \\
0 & 0 & 1
\end{matrix}\right)
,\quad
\Mb_4
:=
\left(\begin{matrix}
1 & 1 & 0 \\
1 & 1 & 0 \\
0 & 1 & 1
\end{matrix}\right)
.
$$

\paragraph{Invertibility}
Neither the conditions of class $\cM_d$ nor the property of invertibility leads to the other.
The following are four examples of binary square matrices, where only $\Mb_1, \Mb_2$ belong to class $\cM_d$ and only $\Mb_1, \Mb_3$ are invertible.
$$
\Mb_1
:=
\left(\begin{matrix}
1 & 0 & 0 \\
0 & 1 & 0 \\
0 & 0 & 1
\end{matrix}\right)
,\quad
\Mb_2
:=
\left(\begin{matrix}
1 & 0 & 1 & 0 \\
0 & 1 & 0 & 1 \\
0 & 1 & 1 & 0 \\
1 & 0 & 0 & 1
\end{matrix}\right)
,\quad
\Mb_3
:=
\left(\begin{matrix}
1 & 1 & 0 \\
0 & 1 & 1 \\
1 & 0 & 1
\end{matrix}\right)
,\quad
\Mb_4
:=
\left(\begin{matrix}
1 & 1 & 0 \\
1 & 1 & 0 \\
0 & 0 & 1
\end{matrix}\right)
.
$$
Specifically for $d = 3$, we find through exhaustive enumeration that all matrices in $\cM_3$ are invertible.
However, this no longer holds for $d \ge 4$, with $\Mb_2$ being a counterexample.

\paragraph{Rank and Kruskal Rank}
Any binary square matrix $\Mb \in \cM_d$ must have Kruskal rank $\rank_K(\Mb) \ge 3$.
This also implies $\rank(\Mb) \ge 3$, since the rank is never small than the Kruskal rank.
See Section \ref{ssec:cp_kruskal} for the definition of Kruskal rank.

\begin{proof}
The case $d = 3$ follows since all matrices in $\cM_3$ are invertible.
We prove for $d \ge 4$ by contradiction.
Suppose there exists distinct indices $i, j, k \in [d]$ such that the columns $\mb_i, \mb_j, \mb_k$ of $\Mb \in \cM_d$ are linearly dependent, i.e.
$$
a \mb_i + b \mb_j + c \mb_k
=
0
$$
for some coefficients $a, b, c \in \RR$ not all being zero.
From the connection between class $\cM_d$ and having distinct columns, we know $\mb_i, \mb_j, \mb_k$ must be distinct.
Therefore, there exists $r, s, t \in [d]$ such that the binary entries of $\Mb$ satisfy
$$
M_{r, i} \ne M_{r, j}
,\quad
M_{s, i} \ne M_{s, k}
,\quad
M_{t, j} \ne M_{t, k}
.
$$
Through permuting indices $i, j, k$ and indices $r, s, t$, without loss of generality we can let
$$
M_{r, i} = 1
,\quad
M_{r, j} = 0
,\quad
M_{s, i} = 0
,\quad
M_{s, k} = 1
,
$$
which implies
$$
a + M_{r, k} c
=
0
,\quad
M_{s, j} b + c
=
0
.
$$
We cannot have $a = 0$, otherwise $M_{t, i} a + M_{t, j} b + M_{t, k} c = 0$ implies one of $b, c$ is 0 and then all $a, b, c$ will become 0.
Similarly, we cannot have $c = 0$.
Hence we have $M_{r, k} = M_{s, j} = 1$ and $a = b = -c$ with $c \ne 0$, which implies
$$
\mb_k
=
\mb_i + \mb_j
.
$$
Let $\ell \in [d]$ be a fourth index distinct from $i, j, k$.
For the two distinct subsets $\cS_1, \cS_2 \in 2^{\langle d \rangle}$ defined as $\cS_1 := \{(i, \ell), (j, \ell)\}$ and $\cS_2 := \{(k, \ell)\}$, we have
$$
\offdiag\Big(
\sum_{(i, j) \in \cS_1} (\mb_i \mb_j^\top + \mb_j \mb_i^\top)
\Big)
=
\offdiag\Big(
\sum_{(i, j) \in \cS_2} (\mb_i \mb_j^\top + \mb_j \mb_i^\top)
\Big)
,
$$
which contradicts with the definition of class $\cM_d$.
\end{proof}

We note that the connections we have discussed so far are for general binary square matrices in class $\cM_d$.
When sparsity constraints are imposed such that the number of ones in each column of $\Mb \in \cM_d$ are bounded by some constant $S$ (e.g. $S = 2$), there exist additional connections between class $\cM_d$ and the common linear algebra properties.
Some of these are explored and used in the proof of Proposition \ref{prop:simp_A2} later this section.

\subsection{Examples}\label{ssec:M_eg}

A trivial example of class $\cM_d$ is the identity matrix $\Ib_d$.
For binary square matrices not in class $\cM_d$, trivial examples are those with diagonal entries containing 0.
In the following, we provide non-trivial examples of binary square matrices belonging to $\cM_3, \cM_4, \cM_5$ and non-trivial examples of binary square matrices not belonging to $\cM_3, \cM_4, \cM_5$.

\begin{example}
The following matrices belong to $\cM_3$:
$$
\left(\begin{matrix}
1 & 0 & 0 \\
0 & 1 & 0 \\
0 & 0 & 1
\end{matrix}\right)
,\quad
\left(\begin{matrix}
1 & 1 & 0 \\
0 & 1 & 0 \\
0 & 0 & 1
\end{matrix}\right)
,\quad
\left(\begin{matrix}
1 & 1 & 1 \\
0 & 1 & 0 \\
0 & 0 & 1
\end{matrix}\right)
,\quad
\left(\begin{matrix}
1 & 0 & 1 \\
1 & 1 & 0 \\
0 & 0 & 1
\end{matrix}\right)
,\quad
\left(\begin{matrix}
1 & 1 & 1 \\
0 & 1 & 1 \\
0 & 0 & 1
\end{matrix}\right)
.
$$
\end{example}

\begin{example}
The following matrices do not belong to $\cM_3$:
$$
\left(\begin{matrix}
1 & 1 & 0 \\
1 & 1 & 0 \\
1 & 0 & 1
\end{matrix}\right)
,\quad
\left(\begin{matrix}
1 & 0 & 1 \\
1 & 1 & 0 \\
1 & 0 & 1
\end{matrix}\right)
,\quad
\left(\begin{matrix}
1 & 1 & 0 \\
1 & 1 & 0 \\
0 & 0 & 1
\end{matrix}\right)
,\quad
\left(\begin{matrix}
1 & 1 & 1 \\
1 & 1 & 1 \\
1 & 1 & 1
\end{matrix}\right)
,\quad
\left(\begin{matrix}
1 & 1 & 0 \\
0 & 1 & 1 \\
1 & 0 & 1
\end{matrix}\right)
,
$$
\end{example}

\begin{example}
The following matrices belong to $\cM_4$:
$$
\left(\begin{matrix}
1 & 1 & 1 & 0 \\
0 & 1 & 1 & 0 \\
1 & 0 & 1 & 1 \\
0 & 1 & 1 & 1
\end{matrix}\right)
,\quad
\left(\begin{matrix}
1 & 1 & 0 & 0 \\
0 & 1 & 0 & 0 \\
0 & 0 & 1 & 0 \\
0 & 0 & 0 & 1
\end{matrix}\right)
,\quad
\left(\begin{matrix}
1 & 1 & 1 & 0 \\
0 & 1 & 0 & 0 \\
0 & 0 & 1 & 0 \\
0 & 0 & 0 & 1
\end{matrix}\right)
,\quad
\left(\begin{matrix}
1 & 0 & 1 & 1 \\
1 & 1 & 0 & 0 \\
0 & 0 & 1 & 0 \\
0 & 0 & 0 & 1
\end{matrix}\right)
,\quad
\left(\begin{matrix}
1 & 1 & 1 & 1 \\
0 & 1 & 1 & 0 \\
0 & 0 & 1 & 0 \\
0 & 0 & 0 & 1
\end{matrix}\right)
.
$$
\end{example}

\begin{example}
The following matrices do not belong to $\cM_4$:
$$
\left(\begin{matrix}
1 & 0 & 1 & 0 \\
0 & 1 & 1 & 0 \\
1 & 0 & 1 & 0 \\
0 & 0 & 0 & 1
\end{matrix}\right)
,\quad
\left(\begin{matrix}
1 & 1 & 0 & 0 \\
1 & 1 & 0 & 0 \\
0 & 0 & 1 & 0 \\
0 & 0 & 0 & 1
\end{matrix}\right)
,\quad
\left(\begin{matrix}
1 & 1 & 0 & 0 \\
1 & 1 & 0 & 1 \\
0 & 1 & 1 & 1 \\
0 & 0 & 0 & 1
\end{matrix}\right)
,\quad
\left(\begin{matrix}
1 & 1 & 1 & 0 \\
1 & 1 & 0 & 1 \\
0 & 1 & 1 & 0 \\
0 & 0 & 1 & 1
\end{matrix}\right)
,\quad
\left(\begin{matrix}
1 & 1 & 1 & 0 \\
1 & 1 & 0 & 1 \\
1 & 1 & 1 & 0 \\
1 & 0 & 1 & 1
\end{matrix}\right)
.
$$
\end{example}

\begin{example}
The following matrices belong to $\cM_5$:
$$
\left(\begin{matrix}
1 & 1 & 1 & 0 & 0 \\
1 & 1 & 0 & 0 & 0 \\
0 & 1 & 1 & 1 & 0 \\
0 & 0 & 1 & 1 & 0 \\
0 & 0 & 1 & 1 & 1
\end{matrix}\right)
,\quad
\left(\begin{matrix}
1 & 1 & 1 & 1 & 1 \\
1 & 1 & 1 & 0 & 1 \\
1 & 1 & 1 & 1 & 0 \\
0 & 1 & 0 & 1 & 0 \\
1 & 1 & 0 & 0 & 1
\end{matrix}\right)
,\quad
\left(\begin{matrix}
1 & 1 & 0 & 0 & 0 \\
1 & 1 & 0 & 0 & 1 \\
1 & 1 & 1 & 1 & 0 \\
0 & 1 & 0 & 1 & 1 \\
0 & 1 & 1 & 1 & 1
\end{matrix}\right)
,\quad
\left(\begin{matrix}
1 & 1 & 0 & 0 & 0 \\
0 & 1 & 1 & 0 & 0 \\
0 & 0 & 1 & 1 & 0 \\
0 & 0 & 0 & 1 & 1 \\
0 & 0 & 0 & 0 & 1
\end{matrix}\right)
.
$$
\end{example}

\begin{example}
The following matrices do not belong to $\cM_5$:
$$
\left(\begin{matrix}
1 & 1 & 0 & 0 & 0 \\
0 & 1 & 1 & 1 & 0 \\
1 & 0 & 1 & 0 & 1 \\
0 & 1 & 0 & 1 & 0 \\
0 & 1 & 0 & 0 & 1
\end{matrix}\right)
,\quad
\left(\begin{matrix}
1 & 1 & 1 & 0 & 1 \\
1 & 1 & 1 & 1 & 0 \\
0 & 0 & 1 & 1 & 1 \\
1 & 0 & 1 & 1 & 1 \\
0 & 1 & 0 & 0 & 1
\end{matrix}\right)
,\quad
\left(\begin{matrix}
1 & 0 & 0 & 1 & 0 \\
0 & 1 & 0 & 1 & 1 \\
1 & 0 & 1 & 0 & 1 \\
1 & 0 & 0 & 1 & 1 \\
0 & 1 & 1 & 0 & 1
\end{matrix}\right)
,\quad
\left(\begin{matrix}
1 & 1 & 0 & 0 & 0 \\
0 & 1 & 1 & 0 & 1 \\
0 & 0 & 1 & 0 & 0 \\
0 & 1 & 0 & 1 & 1 \\
1 & 0 & 0 & 1 & 1
\end{matrix}\right)
.
$$
\end{example}

\subsection{Proof of Proposition \ref{prop:simp_A2}}\label{ssec:proof_prop_simp_A2}

We divide the proof of Proposition \ref{prop:simp_A2} into the following lemmas.

\begin{lemma}\label{lemm:A21_A3_to_A22}
Let $S \in \{1, 2\}$.
For any $\Ab \in \cA_{2, 1} \cap \cA_3$ and $k \in [K]$ with $p_{k - 1} \ge 3$, $\Ab_k$ must have distinct columns.
\end{lemma}

\begin{proof}[Proof of Lemma \ref{lemm:A21_A3_to_A22}]
We prove by contradiction.
Suppose for some $\Ab \in \cA_{2, 1} \cap \cA_3$, the $i$th and $j$th columns of $\Ab_k$ are identical.
Without loss of generality, we can let $\Pb_k = \Ib_{p_k}$ in \eqref{eq:A21_def}.
Letting $\tilde\Ab_k \in \RR^{p_{k - 1} \times p_{k - 1}}$ denote the submatrix of $\Ab_k$ consisting of its first $p_{k - 1}$ rows and first $p_{k - 1}$ columns, by definition in \eqref{eq:A21_def} we have $\tilde\Ab_k^\top \in \cM_{p_{k - 1}}$.

By definition of class $\cM_{p_{k - 1}}$, we know the entries $A_{k, i, i} = A_{k, j, j} = 1$, which implies $A_{k, i, j} = A_{k, j, i} = 1$ due to the identical columns.
For $\Ab \in \cA_3$ with $S \le 2$, each row of $\tilde\Ab_k$ contains at most two non-zero entries, so the $i$th and $j$th rows of $\tilde\Ab_k$ are also identical.
It is shown in Section \ref{ssec:M_linalg} that any $d \times d$ binary square matrix $(d \ge 3)$ containing identical columns can not belong to the class $\cM_d$, which implies $\tilde\Ab_k^\top \notin \cM_{p_{k - 1}}$.
This raises a contradiction.
\end{proof}

To state the following lemmas, we recall that the $i$th row of a connection matrix $\Ab_k$ is denoted by $\ab_{k, i}^\top$, for $\ab_{k, i}$ being a $p_{k - 1}$-dimensional column vector.
We also recall that $*$ denotes the Hadamard product of matrices or vectors of the same dimensions.

\begin{lemma}\label{lemm:A21_A3_to_A23_part1}
Let $S \in \{1, 2\}$.
For any $\Ab \in \cA_{2, 1} \cap \cA_3$ and $k \in [K]$ with $p_{k - 1} \ge 4$, if $\Ab_k$ does not contain any all-ones column, then there exists $(i_0, j_0) \in \langle p_k \rangle$ such that $\ab_{k, i_0} * \ab_{k, j_0} = \zero$.
\end{lemma}

\begin{proof}[Proof of Lemma \ref{lemm:A21_A3_to_A23_part1}]
We prove by contradiction.
Suppose $\ab_{k, i} * \ab_{k, j} \ne \zero$ for all $(i, j) \in \langle p_k \rangle$.

If one of the rows $\ab_{k, i}^\top$ has only one non-zero entry $A_{k, i, \ell} = 1$, then any other row $\ab_{k, j}^\top$ also contains the non-zero entry $A_{k, j, \ell} = 1$.
This implies that $\Ab_k$ contains the all-ones column $\ell$, raising a contradiction.

If all the rows each have two non-zero entries, without loss of generality we can assume $A_{k, 1, 1} = A_{k, 1, 2} = 1$ for the first row and $\Pb_k = \Ib_{p_k}$ in \eqref{eq:A21_def}.
From the definition in \eqref{eq:A21_def}, we know that $A_{k, 3, 3} = A_{k, 4, 4} = 1$.
In order to satisfy $\ab_{k, 1} * \ab_{k, 3} \ne \zero$, $\ab_{k, 1} * \ab_{k, 4} \ne \zero$, $\ab_{k, 3} * \ab_{k, 4} \ne \zero$, there must exist $r \in \{1, 2\}$ such that the entries $A_{k, 3, r} = A_{k, 4, r} = 1$.
Since $\Ab_k$ does not contain any all-ones column, there must also exist $s \in [p_k] \cap \{1, 3, 4\}^c$ such that the entry $A_{k, s, r} = 0$.
Since $\ab_{k, 1} * \ab_{k, s} \ne \zero$, $\ab_{k, 3} * \ab_{k, s} \ne \zero$, $\ab_{k, 4} * \ab_{k, s} \ne \zero$, we know the entries $A_{k, s, 3 - r} = A_{k, s, 3} = A_{k, s, 4} = 1$.
This shows that the row $\ab_{k, s}^\top$ has at least three non-zero entries, again raising a contradiction.
\end{proof}

The following lemma addresses the case of $p_{k - 1} = 3$.

\begin{lemma}\label{lemm:A3_to_A23_p3}
Let $S \in \{1, 2\}$.
For any $\Ab \in \cA_3$ and $k \in [K]$ with $p_{k - 1} = 3$, if $\Ab_k$ does not contain any all-ones column, then at least one of the following two results holds:
\begin{enumerate}[(i)]
\item
there exists $(i_0, j_0) \in \langle p_k \rangle$ such that $\ab_{k, i_0} * \ab_{k, j_0} = 0$;
\item
each row of $\Ab_k$ has two non-zero entries and $\Pb_k \Ab_k = \big( \Mb_k ~ \Bb_k \big)^\top$ for some permutation matrix $\Pb_k \in \RR^{p_k \times p_k}$, some arbitrary matrix $\Bb_k \in \RR^{p_{k - 1} \times (p_k - p_{k - 1})}$, and
$$
\Mb_k
:=
\left(\begin{matrix}
1 & 1 & 0 \\
0 & 1 & 1 \\
1 & 0 & 1
\end{matrix}\right)
.
$$
\end{enumerate}
\end{lemma}

\begin{proof}[Proof of Lemma \ref{lemm:A3_to_A23_p3}]
We assume that result (i) does not hold and prove that result (ii) must hold then.
As discussed in the proof of Lemma \ref{lemm:A21_A3_to_A23_part1}, contradiction arises if some row of $\Ab_k$ has only one non-zero entry.
Since $\Ab_k$ does not contain an all-ones column, there exists rows $i, j, \ell \in [p_k]$ such that the entries $A_{k, i, 1} = A_{k, j, 2} = A_{k, \ell, 3} = 0$, which implies that the remaining two entries in each of these rows are both ones.
This gives the form of $\Mb_k$ through appropriately permuting the rows of $\Ab_k$.
\end{proof}

We recall that $\eb_r$ denotes the indicator vector with its $r$th entry equal to one and all other entries equal to zero.

\begin{lemma}\label{lemm:A21_A22_A3_to_A23_part2}
Let $S \in \{1, 2\}$.
For any $\Ab \in \cA_{2, 1} \cap \cA_{2, 2} \cap \cA_3$ and $k \in [K]$, at least one of the following results holds for each column $r \in [p_{k - 1}]$:
\begin{enumerate}[(i)]
\item
there exists $(i, j) \in \langle p_k \rangle$ such that $\ab_{k, i} * \ab_{k, j} = \eb_r$;
\item
there exists a different column $s \in [p_{k - 1}]$ and three distinct indices $i, j, \ell \in [p_k]$ such that $\ab_{k, i} * \ab_{k, j} = \eb_r + \eb_s$ and $\ab_{k, i} * \ab_{k, \ell} = \eb_s$.
\end{enumerate}
\end{lemma}

\begin{proof}[Proof of Lemma \ref{lemm:A21_A22_A3_to_A23_part2}]
Without loss of generality, we let $\Pb_k = \Ib_{p_k}$ in \eqref{eq:A21_def}, then we have entries $A_{k, t, t} = A_{k, p_{k - 1} + t, t} = 1$ for all $t \in [p_{k - 1}]$.
If $\ab_{k, r} * \ab_{k, p_{k - 1} + r} = \eb_r$, then result (i) holds with $i = r$ and $j = p_{k - 1} + r$.
Otherwise, we must have $\ab_{k, r} = \ab_{k, p_{k - 1} + r} = \eb_r + \eb_s$ for a second column $s \in [p_{k - 1}]$.
Since $\Ab \in \cA_{2, 2}$, the $r$th and $s$th columns of $\Ab_k$ are distinct, implying that there exists a row $w \in [p_k]$ such that $A_{k, w, r} \ne A_{k, w, s}$.
If $A_{k, w, r} = 1$ and $A_{k, w, s} = 0$, then we have $\ab_{k, r} * \ab_{k, w} = \eb_r$, so result (i) holds with $i = r$ and $j = w$.
If $A_{k, w, r} = 0$ and $A_{k, w, s} = 1$, then we have $\ab_{k, r} * \ab_{k, w} = \eb_s$, so result (ii) holds with $i = r$, $j = p_{k - 1} + r$, and $\ell = w$.
\end{proof}

Lemma \ref{lemm:A21_A22_A3_to_A23_part2} leads to the following lemma.

\begin{lemma}\label{lemm:A21_A22_A3_to_A23_part3}
Let $S \in \{1, 2\}$.
For any $\Ab \in \cA_{2, 1} \cap \cA_{2, 2} \cap \cA_3$ and $k \in [K]$, there exist $p_{k - 1}$ pairs of indices $(i_1, j_1), \ldots, (i_{p_{k - 1}}, j_{p_{k - 1}}) \in \langle p_k \rangle$ such that the $p_{k - 1} \times p_{k - 1}$ matrix
$$
\Gb_k
:=
\left(\begin{matrix}
\ab_{k, i_1} * \ab_{k, j_1} & \ab_{k, i_2} * \ab_{k, j_2} & \ldots & \ab_{k, i_{p_{k - 1}}} * \ab_{k, j_{p_{k - 1}}}
\end{matrix}\right)
$$
is invertible.
\end{lemma}

\begin{proof}[Proof of Lemma \ref{lemm:A21_A22_A3_to_A23_part3}]
We classify the columns $1, \ldots, p_{k - 1}$ of $\Ab_k$ into two groups $\cI_1$ and $\cI_2$ defined as
$$
\cI_1
:=
\big\{
r \in [p_{k - 1}]:~
\exists (i, j) \in \langle p_k \rangle \text{ s.t. } \ab_{k, i} * \ab_{k, j} = \eb_r
\big\}
,\quad
\cI_2
=
[p_{k - 1}] \cap \cI_1^c
.
$$
For each $r \in \cI_1$, we choose $(i_r, j_r) \in \langle p_k \rangle$ to satisfy $\ab_{k, i_r} * \ab_{k, j_r} = \eb_r$.
For each $r \in \cI_2$, Lemma \ref{lemm:A21_A22_A3_to_A23_part2} suggests that we can choose $(i_r, j_r) \in \langle p_k \rangle$ such that $\ab_{k, i_r} * \ab_{k, j_r} = \eb_r + \eb_s$ for some $s \in \cI_1$.
By permuting the rows and columns of $\Gb_k$, without loss of generality we can assume that each index in $\cI_1$ is smaller than each index in $\cI_2$.
This implies that $\Gb_k$ is an upper triangular matrix with all diagonal entries being ones.
Hence $\Gb_k$ is invertible.
\end{proof}

The following lemma combines Lemmas \ref{lemm:A21_A3_to_A23_part1}, \ref{lemm:A3_to_A23_p3}, and \ref{lemm:A21_A22_A3_to_A23_part3}.
To state it, we recall from \eqref{eq:A23_def} that the $\frac{p_k (p_k - 1)}{2} \times (1 + p_{k - 1})$ matrix $\Db_k$ denotes the stack of all row vectors $\big( 1 ~ \ab_{k, i}^\top * \ab_{k, j}^\top \big)$ with $(i, j) \in \langle p_k \rangle$.

\begin{lemma}\label{lemm:A21_A22_A3_to_A23}
Let $S \in \{1, 2\}$.
For any $\Ab \in \cA_{2, 1} \cap \cA_{2, 2} \cap \cA_3$ and $k \in [K]$ with $p_{k - 1} \ge 3$, if $\Ab_k$ does not contain any all-ones column, then $\rank(\Db_k) = p_{k - 1} + 1$.
\end{lemma}

\begin{proof}[Proof of Lemma \ref{lemm:A21_A22_A3_to_A23}]
We start by considering the case of $p_{k - 1} = 3$ with result (ii) in Lemma \ref{lemm:A3_to_A23_p3} holding.
By permuting the rows and columns of $\Ab_k$ (and $\Db_k$ accordingly), without loss of generality we can assume the first four rows of $\Ab_k$ to be
$$
\left(\begin{matrix}
1 & 1 & 0 \\
0 & 1 & 1 \\
1 & 0 & 1 \\
1 & 1 & 0
\end{matrix}\right)
.
$$
We notice that the four vectors
$$
\big( 1 ~ \ab_{k, 1}^\top * \ab_{k, 4}^\top \big)
=
(1 ~ 1 ~ 1 ~ 0)
,~
\big( 1 ~ \ab_{k, 1}^\top * \ab_{k, 2}^\top \big)
=
(1 ~ 0 ~ 1 ~ 0)
,~
\big( 1 ~ \ab_{k, 1}^\top * \ab_{k, 3}^\top \big)
=
(1 ~ 1 ~ 0 ~ 0)
,~
\big( 1 ~ \ab_{k, 2}^\top * \ab_{k, 3}^\top \big)
=
(1 ~ 0 ~ 0 ~ 1)
$$
are linearly independent, hence $\rank(\Db_k) = 4 = p_{k - 1} + 1$.

We now consider the case of $p_{k - 1} \ge 4$ and the case of $p_{k - 1} = 3$ with result (i) in Lemma \ref{lemm:A3_to_A23_p3} holding.
By definition of the pairs of indices $(i_0, j_0)$ and $(i_1, j_1), \ldots, (i_{p_{k - 1}}, j_{p_{k - 1}})$ in Lemmas \ref{lemm:A21_A3_to_A23_part1} and \ref{lemm:A21_A22_A3_to_A23_part3}, we know these pairs must be distinct and the $(p_{k - 1} + 1) \times (p_{k - 1} + 1)$ submatrix of $\Db_k$ formed by the corresponding rows takes the form
$$
\left(\begin{matrix}
1 & \zero^\top \\
\one & \Gb_k^\top
\end{matrix}\right)
.
$$
Since $\Gb_k$ is invertible as shown in Lemma \ref{lemm:A21_A22_A3_to_A23_part3}, the submatrix above is also invertible, and hence $\rank(\Db_k) = p_{k - 1} + 1$.
\end{proof}

Finally, we combine Lemmas \ref{lemm:A21_A3_to_A22} and \ref{lemm:A21_A22_A3_to_A23} to prove Proposition \ref{prop:simp_A2}.

\begin{proof}[Proof of Proposition \ref{prop:simp_A2}]
For $S \in \{1, 2\}$ and $p_0 \ge 3$, let $\Ab$ be any adjacency matrix without an all-ones column.
By definition of $\cA_2$ in \eqref{eq:A2_def}, we have
$$
\Ab \in \cA_2 \cap \cA_3
\implies
\Ab \in \cA_{2, 1} \cap \cA_3
.
$$
From Lemma \ref{lemm:A21_A3_to_A22}, we have
$$
\Ab \in \cA_{2, 1} \cap \cA_3
\implies
\Ab \in \cA_{2, 1} \cap \cA_{2, 2} \cap \cA_3
.
$$
From Lemma \ref{lemm:A21_A22_A3_to_A23_part3}, we have
$$
\Ab \in \cA_{2, 1} \cap \cA_{2, 2} \cap \cA_3
\implies
\Ab \in \cA_{2, 1} \cap \cA_{2, 2} \cap \cA_{2, 3} \cap \cA_3 = \cA_2 \cap \cA_3
.
$$
Together they complete the proof.
\end{proof}

\section{Posterior Consistency}\label{sec:post_theo}

In this section, we present the proofs of our posterior consistency theorems.
We start by providing background on standard techniques in Section \ref{ssec:cite_schwartz_doob}, including Schwartz' theorem and Doob's theorem for (strong) posterior consistency.
The Schwartz' theorem works with the Wasserstein space of probability measures.
By establishing certain topological properties of the parameter space, this consistency of probability measures can be transformed into consistency of model parameters, leading to the proof of Theorem \ref{theo:schwartz} in Section \ref{ssec:prove_schwartz}.
Proofs of auxiliary lemmas on these topological properties are provided in Section \ref{ssec:post_aux}.
The Doob's theorem works directly with model parameters, but it requires identifiability of the model.
It does not automatically apply here due to the non-trivial type of non-identifiability problem of our model, so we introduce the idea of quotient parameters and study the Borel $\sigma$-algebra of the quotient parameter space in Section \ref{ssec:post_quotient}.
Theorem \ref{theo:doob} is then proven in Section \ref{ssec:proof_doob} by applying intermediate steps of the Doob's theorem to the quotient parameters.

\subsection{Schwartz' Theorem and Doob's Theorem}\label{ssec:cite_schwartz_doob}

Before we present proofs of posterior consistency for our model, we first provide an overview of the general posterior consistency theories, including Schwartz' Theorem and Doob's Theorem.
We consider a general sample space $(\cX, \sX)$ and a family of models with distributions $P_\theta$ indexed by $\theta \in \Theta$ with $\Theta$ the parameter space.

Schwartz's theory of posterior consistency \citep{schwartz1965bayes} directly works with the probability densities.
Letting $\nu$ be a dominating measure on $\cX$, we view each distribution $P$ as parameterized by its density $p$ with respect to $\nu$.
The class of densities $\cP$ is the parameter space, equipped with the weak topology and the corresponding Borel $\sigma$-algebra $\sP$.
Let observations $\{X^{(n)}\}_{n = 1}^\infty$ i.i.d. follow the true distribution with density $p^* \in \cP$, such that $(p^*)^\infty$ denotes the joint distribution of $X^{(1:\infty)}$.
We impose the prior distribution $\Pi$ on $\cP$ and denote the posterior distribution given $X^{(1:n)}$ by $\Pi_n$.

Let $d$ be a metric on $\cP$ that induces the weak topology, e.g. the Wasserstein distance.
An $\epsilon$-neighborhood of $p^*$ in $\cP$ is naturally defined as $\big\{ p \in \cP:~ d(p, p^*) < \epsilon \big\}$.
Using the Kullback-Leibler divergence $D_{KL}(q, p) := \int_\cX q \log \frac{q}{p} \ud \nu$, we also define an $\epsilon$-KL neighborhood of $p^*$ to be the set $\big\{ p \in \cP:~ D_{KL}(p^*, p) < \epsilon \big\}$.
We notice that by Pinsker's inequality, every open neighborhood of $p^*$ contains an $\epsilon$-KL neighborhood of $p^*$ for some $\epsilon > 0$.
For the prior distribution $\Pi$, we say $p^*$ belongs to the \emph{KL support} of $\Pi$ if any KL neighborhood of $p^*$ has positive prior probability, i.e. for all $\epsilon > 0$, $\Pi\big\{ p \in \cP:~ D_{KL}(p^*, p) < \epsilon \big\} > 0$.

\begin{theorem}[Schwartz' theorem]\label{theo:cite_schwartz}
If $p^*$ is in the KL support of $\Pi$, then for any open neighborhood $\cO$ of $p^*$, we have
$$
\lim_{n \to \infty} \Pi_n\left( \cO^c \left| X^{(1:n)} \right. \right)
=
0
\quad
(p^*)^\infty \text{-a.s.}
$$
\end{theorem}

Theorem \ref{theo:cite_schwartz} states that the posterior distribution $\Pi_n$ is (strongly) consistent for any true density $p^*$ in the KL support of the prior distribution $\Pi$.
A proof can be found in \citet{ghosal2017fundamentals} Section 6.4.
To apply Theorem \ref{theo:cite_schwartz} to our model, we notice that a gap exists between the consistency of the probability density $p$ and the consistency of our model parameters $(\Ab, \bTheta)$.
These two notions of consistency can be bridged by utilizing the continuity of $p$ in $(\Ab, \bTheta)$ and enforcing constraints on the topological properties of our parameter space $\cR$. This will be formally shown in the proof of Theorem \ref{theo:schwartz}.

In contrast, Doob's theory of posterior consistency works with an arbitrary parameter space $(\Theta, \sB)$ that is a separable metric space.
It establishes the (strong) consistency of posterior distribution $\Pi_n$ under all true parameters $\theta^*$ in $\Theta$ except for a set that is negligible under the prior distribution $\Pi$.
Let $Q$ be a probability measure on the product space $(\Theta \times \cX^\infty, \sB \times \sX^\infty)$ defined by
\begin{equation}\label{eq:jpmQ}
Q(B \times X)
=
\int_B P_\theta^\infty(E) \Pi(\ud \theta)
,\quad
\forall B \in \sB,~ E \in \sX
.
\end{equation}

\begin{theorem}[Doob's theorem]\label{theo:cite_doob}
If $\theta$ is measurable with respect to the $Q$-completion of $\sigma(X^{(1:\infty)})$, then there exists $\Theta_0 \subset \Theta$ with $\Pi(\Theta_0) = 0$ such that for all $\theta^* \in \Theta \cap \Theta_0^c$ and any open neighborhood $\cO$ of $\theta^*$, we have
$$
\lim_{n \to \infty} \Pi_n\left(
\cO^c \Big| X^{(1:n)}
\right)
=
0
\quad
P_{\theta^*}^\infty \text{-a.s.}
$$
\end{theorem}

Theorem \ref{theo:cite_doob} is based on an application of Doob's martingale convergence theorem \citep{durrett2019probability}, see \citet{ghosal2017fundamentals} Section 6.2 for a proof.
The measurability of $\theta$ with respect to the $Q$-completion of $\sigma(X^{(1:\infty)})$ in Theorem \ref{theo:cite_doob} can be guaranteed by the slightly stronger condition of $\theta = f(X^{(1:\infty)})$ $P_\theta$-a.s. for every $\theta \in \Theta$ and some measurable function $f: \cX^{\infty} \to \Theta$.
A sufficient condition is model identifiability, as indicated in the following proposition.

\begin{proposition}\label{prop:doob_meas}
Let $\bTheta$ be a Borel subset of a Polish space.
If $\theta \to P_\theta(E)$ is Borel measurable for every $E \in \sX$ and the model $\{P_\theta:~ \theta \in \Theta\}$ is identifiable, then there exists a Borel measurable function $f:~ \cX^\infty \to \Theta$ satisfying for every $\theta \in \Theta$,
$$
\theta
=
f(X^{(1:\infty)})
\quad
P_\theta \text{-a.s.}
$$
\end{proposition}

The proof of Proposition \ref{prop:doob_meas} can also be found in \citet{ghosal2017fundamentals} Section 6.2.
We note that the Borel measurability of $\theta \to P_\theta(X)$ and the parameter space $\Theta$ being the Borel subset of a Polish space is important, because the proof utilizes the Kuratowski's theorem \citep{srivastava2008course} to transform measurability of $P_\theta$ with respect to $\sigma(X^{(1:\infty)})$ to measurability of $\theta$ with respect to $\sigma(X^{(1:\infty)})$.
Kuratowski's theorem states that the inverse of a one-to-one Borel map between Polish spaces is also Borel measurable, which does not hold in general for other measurable maps or general spaces.

For our model with some parameter space $\cA \times \cT$, we notice that neither Theorem \ref{theo:cite_doob} nor Proposition \ref{prop:doob_meas} directly apply, due to the trivial type of non-identifiability in $\perm_{\cA \times \cT}(\Ab, \bTheta)$.
Recall that in Definition \ref{defi:post}, we have loosened the requirement of posterior consistency in our model from the posterior distribution asymptotically concentrating at $(\Ab, \bTheta)$ to asymptotically having support within $\perm_{\cA \times \cT}(\Ab, \bTheta)$.
We eliminate the trivial type of non-identifiability by considering the quotient parameter space given by the equivalence relation of $(\Ab, \bTheta) \sim (\Ab', \bTheta')$ for all $(\Ab', \bTheta') \in \perm_{\cA \times \cT}(\Ab, \bTheta)$.
We will show properties of the quotient Borel-algebra and prove Theorem \ref{theo:doob} using Proposition \ref{prop:doob_meas} later this section.

\subsection{Proof of Theorem \ref{theo:schwartz}}\label{ssec:prove_schwartz}

To avoid confusion with notations, within this section we will use $\Pi(\Ab, \bTheta) := \PP(\Ab, \bTheta)$ to denote the prior distribution of $(\Ab, \bTheta)$ and $\Pi_N(\Ab, \bTheta) := \PP(\Ab, \bTheta | \Xb_K^{(1:N)})$ to denote the posterior distribution of $(\Ab, \bTheta)$ given observed adjacency matrices $\Xb_K^{(1:N)}$.

We consider the counting measure $\upmu$ as the base measure on the space $\cX_K$ of observed adjacency matrices, then a distribution's density with respect to $\upmu$ on $\cX_K$ is just its probability mass function.
For a parameter $(\Ab, \bTheta) \in \cR$, we denote the probability vector of $\PP(\Xb_K | \Ab, \bTheta)$ under our model formulations \eqref{eq:model_entry} and \eqref{eq:model} by
$$
\cL(\Ab, \bTheta)
:=
\big\{
\PP(\Xb_K | \Ab, \bTheta)
\big\}_{\Xb_K \in \cX_K}
,
$$
which belongs to the space of probability measures $\cP := \cS^{|\cX_K| - 1}$ on $\cX_K$, i.e. the probability simplex.
Recalling that the space $\cA_0 \times \cT_0$ is equipped with the $L_1$ distance, we let its subspace $\cR \subset \cA_0 \times \cT_0$ inherit the same metric.
Since $\cX_K$ is discrete and finite-dimensional, the weak topology on $\cP$ is the same as the topology induced by the $L_1$ distance inherited as a subspace of $\RR^{|\cX_K|}$.
Our model specifications in \eqref{eq:model_entry} and \eqref{eq:model} ensure for every $\Xb_K \in \cX_K$ that the probability $\PP(X_K | \Ab, \bTheta)$ is a continuous function of $\Ab$ and $\bTheta$.
We therefore find that the map $\cL:~ \cR \to \cP$ is continuous.

We notice that Theorem \ref{theo:schwartz} considers true parameter $(\Ab^*, \bTheta^*)$ in the support of the prior distribution $\Pi(\Ab, \bTheta)$, whereas Theorem \ref{theo:cite_schwartz} applies to true probability density $\cL(\Ab^*, \bTheta^*)$ within the KL support of the prior.
We bridge their connections in the following lemma.
Note that we denote the pushforward measure of $\Pi$ through the map $\cL$ on $\cP$ by $\cL_\# \Pi$.

\begin{lemma}\label{lemm:supp_kl_supp}
If $(\Ab, \bTheta) \in \cR$ is in the support of $\Pi$, then $\cL(\Ab, \bTheta)$ is in the KL support of $\cL_\# \Pi$.
\end{lemma}

The proof of Lemma \ref{lemm:supp_kl_supp} is deferred to Section \ref{ssec:post_aux}.
The continuity of the map $\cL$, the Pinsker's inequality, and the reverse Pinsker's inequality \citep{binette2019note} are used to bridge the connection between neighborhoods in $\cR$, neighborhoods in $\cP$, and KL neighborhoods in $\cP$.

In Theorem \ref{theo:schwartz}, the condition (i) requires the parameter space $\cR$ to be compact, while the condition (ii) only requires the parameter space $\cR \subset \cA_1 \times \cT_0$.
The following lemma shows that the space $\cA_1 \times \cT_0$ is actually not far from being compact for our posterior consistency purposes.
For a probability measure $p \in \cP$, we let $\cO_\epsilon(p)$ denote the $\epsilon$-neighborhood of $p$ in $\cP$.

\begin{lemma}\label{lemm:A1T0_almost_compact}
Let true parameter $(\Ab^*, \bTheta^*) \in \cA_1 \times \cT_0$.
Then there exists $\epsilon > 0$ and a compact set $B \subset \cA_1 \times \cT_0$ such that
$$
\cL^{-1}\left(
\cO_\epsilon(\cL(\Ab^*, \bTheta^*))
\right)
\cap
(\cA_1 \times \cT_0)
\subset
B
.
$$
\end{lemma}

Lemma \ref{lemm:A1T0_almost_compact} is obtained by exploiting the topological properties of the parameter space $\cA_1 \times \cT_0$, the space of probability measures $\cP$, and the space of probability vectors and matrices $\vb_k, \bLambda_k$ earlier defined in Section \ref{sec:iden}.
We defer the detailed proof to Section \ref{ssec:post_aux}.

With Lemmas \ref{lemm:supp_kl_supp} and \ref{lemm:A1T0_almost_compact}, we can verify the KL support condition in Theorem \ref{theo:cite_schwartz} and work with a compact or almost compact parameter space.
Theorem \ref{theo:strict} is also used to ensure identifiability of the true parameter under condition (ii) of Theorem \ref{theo:schwartz}.
We prove Theorem \ref{theo:schwartz} as follows.

\begin{proof}[Proof of Theorem \ref{theo:schwartz}]
For true parameter $(\Ab^*, \bTheta^*)$ in the support of prior distribution $\Pi$, Lemma \ref{lemm:supp_kl_supp} suggests that the true probability measure $\cL(\Ab^*, \bTheta^*)$ is in the KL-support of $\cL_\# \Pi$.
Therefore, we can apply Theorem \ref{theo:cite_schwartz} and obtain that for any $\epsilon > 0$,
$$
\lim_{n \to \infty} \cL_\# \Pi_n\left(
\cO_\epsilon(\cL(\Ab^*, \bTheta^*))^c
\right)
=
0
\quad
\PP\left(
\Xb_K^{(1:\infty)} \Big| \Ab^*, \bTheta^*
\right) \text{-a.s.}
$$

From Lemma \ref{lemm:A1T0_almost_compact}, we know under either condition (i) or (ii) of Theorem \ref{theo:schwartz}, there exists $\kappa > 0$ and a compact set $B \subset \cA_0 \times \cT_0$ such that $\cL^{-1}\big( \cO_\kappa(\cL(\Ab^*, \bTheta^*)) \big) \subset B$, which implies
$$
\cL(B^c)
\subset
\cO_\kappa\left(
\cL(\Ab^*, \bTheta^*)
\right)^c
.
$$

Let $\delta > 0$ be arbitrary.
Recalling the definition of $\tilde\cO_\delta(\Ab^*, \bTheta^*)$ from \eqref{eq:eps_nbhd} and using the continuity of $\cL$, we know that $\cL\big( \tilde\cO_\delta(\Ab^*, \bTheta^*)^c \cap B \big)$ is a compact set in $\cP$.
Under either condition (i) or condition (ii), using Theorem \ref{theo:strict} we know that the true parameter $(\Ab^*, \bTheta^*) \in \cR$ is identifiable, i.e. $\perm_\cR(\Ab^*, \bTheta^*) = \marg_\cR(\Ab^*, \bTheta^*)$.
This suggests that the true probability measure $\cL(\Ab^*, \bTheta^*) \notin \cL\big( \tilde\cO_\delta(\Ab^*, \bTheta^*)^c \cap B \big)$, so its $L_1$-distance to the compact set $\cL\big( \tilde\cO_\delta(\Ab^*, \bTheta^*)^c \cap B \big)$ can be lower bounded by some $\epsilon > 0$, i.e.
$$
\cL\left(
\tilde\cO_\delta(\Ab^*, \bTheta^*)^c \cap B
\right)
\subset
\cO_\epsilon(\cL(\Ab^*, \bTheta^*))^c
.
$$

Putting everything together, we obtain
$$
\cL\big( \tilde\cO_\delta(\Ab^*, \bTheta^*)^c \big)
\subset
\cL\left(
\tilde\cO_\delta(\Ab^*, \bTheta^*)^c \cap B
\right)
\cup
\cL(B^c)
\subset
\cO_{\min\{\epsilon, \kappa\}}(\cL(\Ab^*, \bTheta^*))^c
$$
with $\min\{\epsilon, \kappa\} > 0$, which gives
\begin{align*}
\lim_{n \to \infty} \Pi_n\left(
\tilde\cO_\delta(\Ab^*, \bTheta^*)^c
\right)
&\le
\lim_{n \to \infty} \cL_\# \Pi_n\left(
\cL\big( \tilde\cO_\delta(\Ab^*, \bTheta^*)^c \big)
\right)
\\&\le
\lim_{n \to \infty} \cL_\# \Pi_n\left(
\cO_{\min\{\epsilon, \kappa\}}(\cL(\Ab^*, \bTheta^*))^c
\right)
=
0
\end{align*}
$\PP(\Xb_K^{(1:\infty)} | \Ab^*, \bTheta^*)$-almost surely.
\end{proof}

\subsection{Proof of Auxiliary Lemmas}\label{ssec:post_aux}

\begin{proof}[Proof of Lemma \ref{lemm:supp_kl_supp}]
For $\epsilon > 0$ and $p \in \cP$, we let $\cO_\epsilon(p)$ denote the $\epsilon$-neighborhood of density $p$ in $\cP$ and let $\cO^{KL}_\epsilon(p)$ denote the $\epsilon$-KL neighborhood of $p$ in $\cP$, i.e.
$$
\cO_\epsilon(p)
:=
\big\{
q \in \cP:~
\|p - q\|_1 < \epsilon
\big\}
,\quad
\cO_\epsilon^{KL}(p)
:=
\big\{
q \in \cP:~
D_{KL}(p, q) < \epsilon
\big\}
.
$$
Similarly, for $(\Ab, \bTheta) \in \cR$, we let $\cO_\epsilon(\Ab, \bTheta)$ denote the $\epsilon$-neighborhood of parameter $(\Ab, \bTheta)$ in $\cR$, i.e.
$$
\cO_\epsilon(\Ab, \bTheta)
:=
\big\{
(\Ab', \bTheta') \in \cR:~
\|(\Ab', \bTheta') - (\Ab, \bTheta)\|_1 < \epsilon
\big\}
.
$$

Let $(\Ab, \bTheta) \in \cR$ be an arbitrary parameter in the support of $\Pi$.
Our model assigns positive probability $\PP(\Xb_K | \Ab, \bTheta)$ to any $\Xb_K \in \cX_K$, which suggests that the probability measure $p := \cL(\Ab, \bTheta)$ belongs to the (combinatorial) interior of the probability simplex $\cS_+^{|\cX_K| - 1}$.
Using Pinsker's inequality, for any $q \in \cP$ we have $D_{KL}(p, q) \ge \frac12 \|p - q\|_1^2$, which suggests for any $\epsilon > 0$,  the neighborhoods $\cO_\epsilon^{KL}(p) \subset \cO_{\sqrt{2\epsilon}}(p)$.
We let $\kappa := \inf_{q \in \partial \cS^{|\cX_K| - 1}} \|p - q\|_1 > 0$ denote the $L_1$-distance between $p$ and the (combinatorial) boundary of $\cS^{|\cX_K| - 1}$, then we have
$$
\cO_\epsilon^{KL}(p)
\subset
\overline\cO_{\frac{\kappa}{2}}(p)
\subsetneq
\cS_+^{|\cX_K| - 1}
,\quad
\forall 0 < \epsilon \le \frac{\kappa^2}{8}
,
$$
where $\overline{\cO}$ denotes the closure of $\cO$.

By compactness of $\overline\cO_{\frac{\kappa}{2}}(p)$, there exists finite positive values $a, b > 0$ such that for all $q \in \cO_\epsilon^{KL}(p)$ we have $a \le \frac{\ud q}{\ud p} \le b$, where $\frac{\ud q}{\ud p}$ denotes the Radon-Nikodym derivative of the probability measures $q$ over $p$, i.e. the entrywise division of $q$ over $p$.
Using reverse Pinsker's inequality \citep{binette2019note}, we obtain for $0 < \epsilon \le \frac{\kappa^2}{8}$
$$
D_{KL}(p, q)
\le
\frac12 \left(
\frac{\log a}{1 - a} + \frac{\log b}{b - 1}
\right) \|p - q\|_1
,\quad
\forall q \in \cO_\epsilon^{KL}(p)
.
$$
Denoting $m := \frac12 \big( \frac{\log a}{1 - a} + \frac{\log b}{b - 1} \big) > 0$, the above suggests
$$
\cO_{\frac{\epsilon}{m}}(p)
\subset
\cO_\epsilon^{KL}(p)
,\quad
\forall 0 < \epsilon \le \frac{\kappa^2}{8}
.
$$

Since the map $\cL:~ \cR \to \cP$ is continuous, for any $\epsilon > 0$, the pre-image $\cL^{-1}\big( \cO_{\min\{\frac{\epsilon}{m}, \frac{\kappa^2}{8m}\}}(p) \big)$ is open in $\cR$ and contains the $\delta$-neighborhood $\cO_\delta(\Ab, \bTheta)$ for some $\delta > 0$.
Therefore, since $(\Ab, \bTheta)$ is in the support of $\Pi$, for any $\epsilon > 0$ we have
\begin{align*}
(\cL_\# \Pi)\left(
\cO_\epsilon^{KL}(p)
\right)
&\ge
(\cL_\# \Pi)\left(
\cO_{\min\{\epsilon, \frac{\kappa^2}{8}\}}^{KL}(p)
\right)
\\&\ge
(\cL_\# \Pi)\left(
\cO_{\min\{\frac{\epsilon}{m}, \frac{\kappa^2}{8m}\}}(p)
\right)
\\&=
\Pi\left( \cL^{-1}\left(
\cO_{\min\{\frac{\epsilon}{m}, \frac{\kappa^2}{8m}\}}(p)
\right) \right)
\\&\ge
\Pi(\cO_\delta(\Ab, \bTheta))
\\&>
0
.
\end{align*}
This suggests that $p = \cL(\Ab, \bTheta)$ is in the KL support of $\cL_\# \Pi$.
\end{proof}

\begin{proof}[Proof of Lemma \ref{lemm:A1T0_almost_compact}]
We prove by contradiction.
Suppose
\begin{equation}\label{eq:A1T0_almost_compact_supp}
\cL^{-1}\big( \cO_\epsilon(\cL(\Ab^*, \bTheta^*)) \big)
\cap
(\cA_1 \times \cT_0)
\cap
B^c
\ne
\varnothing
\end{equation}
for any $\epsilon > 0$ and any compact set $B \subset \cA_1 \times \cT_0$.
We consider the sequence of positive values $\{\delta^{(m)}\}_{m = 1}^\infty$ and the sequence of compact sets $\{B^{(m)}\}_{m = 1}^\infty$ defined as as $\delta^{(m)} := \frac{1}{m}$ and
$$
B^{(m)}
:=
\left\{
(\Ab, \bTheta) \in \cA_1 \times \cT_0:~
\forall k \in [K], 1 \le i \le j \le p_{k - 1},~
-m \le C_k \le m,~
\frac{1}{m} \le \Gamma_{k, i, j} \le m
\right\}
.
$$
By \eqref{eq:A1T0_almost_compact_supp}, we can find a sequence of parameters $\{(\Ab^{(m)}, \bTheta^{(m)})\}_{m = 1}^\infty \subset \cA_1 \times \cT_0$ such that
$$
\cL(\Ab^{(m)}, \bTheta^{(m)}) \in \cO_\epsilon(\cL(\Ab^*, \bTheta^*))
,\quad
(\Ab^{(m)}, \bTheta^{(m)}) \notin B^{(m)}
,\quad
\forall m \in \NN
.
$$
Since the space of probability measure $\cP$ is compact, we can take a subsequence of $\{(\Ab^{(m)}, \bTheta^{(m)})\}_{m = 1}^\infty$ such that the corresponding subsequence of probability measures $\{\cL(\Ab^{(m)}, \bTheta^{(m)})\}_{m = 1}^\infty$ converges in $\cP$.
We can further take a subsequence, which we still denote by $\{(\Ab^{(m)}, \bTheta^{(m)})\}_{m = 1}^\infty$ for notational simplicity, such that the sequence of $\Ab^{(m)}$ is constant and each sequence of $C_k^{(m)}$ and $\Gamma_{k, i, j}^{(m)}$ converges in the extended real line $\RR \cup \{\pm\infty\}$.

We recall the definitions of probability vectors and matrices $\vb_k, \bLambda_k$ from \eqref{eq:def_vk}, \eqref{eq:def_lambdakij}, \eqref{eq:def_lambdak}.
For each parameter $(\Ab^{(m)}, \bTheta^{(m)})$, we denote the corresponding probability vectors and matrices by $\vb_k^{(m)} (k \in [0, K])$ and $\bLambda_k^{(m)} (k \in [K])$.
For simplicity, we denote the collection of $\vb_k^{(m)} (k \in [0, K])$ and $\bLambda_k^{(m)} (k \in [K])$ by $\Wb^{(m)}$ and let $\cW$ denote the space of all valid $\Wb$.
We also denote the collection of probability vectors and matrices corresponding to the true parameter $(\Ab^*, \bTheta^*)$ by $\Wb^*$.
Since each $\vb_k$ and each column of $\bLambda_k$ are constrained to the probability simplex, we know $\cW$ is homeomorphic to a product of probability simplexes and hence compact.
We further notice that $\cL(\Ab^{(m)}, \bTheta^{(m)}) = \bLambda_{K - 1}^{(m)} \vb_{K - 1}^{(m)}$, which represents a continuous map $\cL_W(\Wb) := \bLambda_{K - 1} \vb_{K - 1}$ from $\cW$ to $\cP$.
For any $\Wb = \{\bLambda_k\}_{k = 1}^K \cup \{\vb_k\}_{k = 0}^K \in \cW$, we denote the collection of all its row and column shuffled versions by
\begin{align*}
\shuff(\Wb)
:=
\Big\{
\Wb' \in \cW:~
&
\exists~ \Pb_k (k \in [0, K])
\text{ such that }
\vb_K' = \vb_K,
\\&
\forall k \in [K],~
\bLambda_k' = \Pb_k^\top \bLambda_k \Pb_{k - 1}
,~
\vb_{k - 1}' = \Pb_{k - 1}^\top \vb_{k - 1}
\Big\}
,
\end{align*}
where each $\Pb_k (k \in [0, K - 1])$ is an arbitrary $|\cX_k| \times |\cX_k|$ permutation matrix and $\Pb_K := \Ib_{|\cX_k|}$.
Letting the space $\cW$ be endowed with the $L_1$-distance and $\cO_\delta(\Wb)$ denoting the $\delta$-neighborhood of $\Wb$ in $\cW$, we consider the open set
$$
\tilde\cO_\delta(\Wb)
:=
\bigcup_{\Wb' \in \shuff(\Wb)} \cO_\delta(\Wb')
.
$$
We notice that its complement $\tilde\cO_\delta(\Wb)^c$ in $\cW$ is compact, which implies that its image $\cL_W\big( \tilde\cO_\delta(\Wb)^c \big)$ under the continuous map $\cL_W$ is also compact.
Lemma \ref{lemm:strict_tensor} suggests for any $\delta > 0$
$$
\cL_W^{-1}(\cL(\Ab^*, \bTheta^*))
=
\shuff(\Wb^*)
\subset
\tilde\cO_\delta(\Wb)
,
$$
which implies that the point $\cL(\Ab^*, \bTheta^*)$ in $\cP$ does not belong to the compact set $\cL_W\big( \tilde\cO_\delta(\Wb)^c \big)$ and hence there is a positive distance between them.
Therefore, for any $\delta > 0$, the sequence $\{\cL_W(\Wb^{(m)})\}_{m = 1}^\infty$ will be eventually outside $\cL_W\big( \tilde\cO_\delta(\Wb)^c \big)$, which implies that the sequence $\{\Wb^{(m)}\}_{m = 1}^\infty$ will be eventually within $\tilde\cO_\delta(\Wb)$.
Since $\{\Ab^{(m)}, \bTheta^{(m)}\}_{m = 1}^\infty$ converges, we know $\{\Wb^{(m)}\}_{m = 1}^\infty$ converges to a point of $\shuff(\Wb^*)$, so without loss of generality we have $\lim_{m \to \infty} \Wb^{(m)} = \Wb^*$.

Since each parameter $(\Ab^{(m)}, \bTheta^{(m)}) \notin B^{(m)}$, the convergence of the sequence $\{(\Ab^{(m)}, \bTheta^{(m)})\}$ implies that at least one of the following six statements is true for some $k \in [K]$:
\begin{itemize}
\item[($S_1$)]
$\lim_{m \to \infty} C_k^{(m)} = -\infty$;
\item[($S_2$)]
$\lim_{m \to \infty} C_k^{(m)} = \infty$;
\item[($S_3$)]
$\exists (i, j) \in \langle p_{k - 1} \rangle$ s.t. $\lim_{m \to \infty} \Gamma_{k, i, j}^{(m)} = 0$;
\item[($S_4$)]
$\exists (i, j) \in \langle p_{k - 1} \rangle$ s.t. $\lim_{m \to \infty} \Gamma_{k, i, j}^{(m)} = \infty$;
\item[($S_5$)]
$\exists i \in [p_{k - 1}]$ s.t. $\lim_{m \to \infty} \Gamma_{k, i, i}^{(m)} = 0$;
\item[($S_6$)]
$\exists i \in [p_{k - 1}]$ s.t. $\lim_{m \to \infty} \Gamma_{k, i, i}^{(m)} = \infty$.
\end{itemize}
We examine each of the six statements in the following and show that contradictions with $\lim_{m \to \infty} \bLambda_k^{(m)} = \bLambda_k^*$ can be found as long as one of the statements is true.
Note since the sequence $\{\Ab^{(m)}\}_{m = 1}^\infty$ is constant and in $\cA_1$, by definition in \eqref{eq:A1_def} we can assume without loss of generality that each $\Ab_k^{(m)}$ takes the blockwise form of $\big( \Ib_{p_{k - 1}} ~ \Ib_{p_{k - 1}} ~ \Bb \big)^\top$.

If $S_1$ is true, then for any $(s, t) \in \langle p_k \rangle$ we have
$$
\lim_{m \to \infty} \PP\left(
X_{k, s, t} = 1
~\left|~
\Xb_{k - 1} = \Ib_{p_{k - 1}}, \Ab_k^{(m)}, \bTheta_k^{(m)}
\right. \right)
=
\lim_{m \to \infty} \frac{e^{C_k^{(m)}}}{1 + e^{C_k^{(m)}}}
=
0
.
$$

If $S_2$ is true, then for any $(s, t) \in \langle p_k \rangle$ we have
$$
\lim_{m \to \infty} \PP\left(
X_{k, s, t} = 0 \left| \Xb_{k - 1} = \Ib_{p_{k - 1}}, \Ab_k^{(m)}, \bTheta_k^{(m)}
\right. \right)
=
\lim_{m \to \infty} \frac{1}{1 + e^{C_k^{(m)}}}
=
0
.
$$

If $S_3$ is true for some $(i, j) \in \langle p_{k - 1} \rangle$, then for any $(s, t) \in \langle p_k \rangle$ we have
\begin{align*}
&\quad~
\lim_{m \to \infty} \PP\left(
X_{k, s, t} = 0 \left| \Xb_{k - 1} = \Ib_{p_{k - 1}}, \Ab_k^{(m)}, \bTheta_k^{(m)}
\right. \right)
\\&=
\lim_{m \to \infty} \frac{1}{1 + e^{C_k^{(m)}}}
\\&=
\lim_{m \to \infty} \PP\left(
X_{k, s, t} = 0 \left| \Xb_{k - 1} = \Ib_{p_{k - 1}} + \Eb_{i, j}, \Ab_k^{(m)}, \bTheta_k^{(m)}
\right. \right)
,
\end{align*}
where we denote $\Eb_{i, j} := \eb_i \eb_j^\top + \eb_j \eb_i^\top$.

If $(\neg S_1) \wedge S_4$ is true for some $(i, j) \in \langle p_{k - 1} \rangle$, then we have
$$
\lim_{m \to \infty} \PP\left(
X_{k, i, j} = 0 \left| \Xb_{k - 1} = \Ib_{p_{k - 1}} + \Eb_{i, j}, \Ab_k^{(m)}, \bTheta_k^{(m)}
\right. \right)
=
\lim_{m \to \infty} \frac{1}{1 + e^{C_k^{(m)} + \Gamma_{k, i, j}^{(m)}}}
=
0
.
$$

If $S_5$ is true for some $i \in [p_{k - 1}]$, then we have
\begin{align*}
&\quad~
\lim_{m \to \infty} \PP\left(
X_{k, p_{k - 1} + i, i} = 0 \left| \Xb_{k - 1} = \Ib_{p_{k - 1}}, \Ab_k^{(m)}, \bTheta_k^{(m)}
\right. \right)
\\&=
\lim_{m \to \infty} \frac{1}{1 + e^{C_k^{(m)}}}
\\&=
\lim_{m \to \infty} \PP\left(
X_{k, 1, 2} = 0 \left| \Xb_{k - 1} = \Ib_{p_{k - 1}}, \Ab_k^{(m)}, \bTheta_k^{(m)}
\right. \right)
.
\end{align*}

If $(\neg S_1) \wedge S_6$ is true for some $i \in [p_{k - 1}]$, then we have
$$
\lim_{m \to \infty} \PP\left(
X_{k, p_{k - 1} + i, i} = 0 \left| \Xb_{k - 1} = \Ib_{p_{k - 1}}, \Ab_k^{(m)}, \bTheta_k^{(m)}
\right. \right)
=
\lim_{m \to \infty} \frac{1}{1 + e^{C_k^{(m)} + \Gamma_{k, i, i}^{(m)}}}
=
0
.
$$

Recalling the proof of Lemma \ref{lemm:strict_sX} from Section \ref{ssec:iden_strict_aux}, we know that the probability matrix $\bLambda^*$ corresponding to the parameter $(\Ab^*, \bTheta^*)$ must have properties including distinct columns, non-zero entries, etc.
This suggests that none of the above six equations can hold, meaning that $S_1 \vee S_2 \vee S_3 \vee S_4 \vee S_5 \vee S_6$ contradicts with $\lim_{m \to \infty} \Wb^{(m)} = \Wb^*$.
Therefore, \eqref{eq:A1T0_almost_compact_supp} can not hold.
\end{proof}

\subsection{Quotient Parameter Space and Its Properties}\label{ssec:post_quotient}

As discussed in Section \ref{ssec:cite_schwartz_doob}, to apply Proposition \ref{prop:doob_meas} and prove Doob's posterior consistency for our model, we need to get rid of the trivial non-identifiability issue coming from node permutations in each layer of our Bayesian network.
We will consider the quotient of parameter space $\cR$ defined by the equivalence relation of node permutations.
We show some properties of the quotient parameter spaces in the following and prove Theorem \ref{theo:doob} in the following subsection.

We consider a parameter space $\cR \subset \cA_0 \times \cT_0$ that is invariant to node permutations, i.e. satisfying $\perm_\cR(\Ab, \bTheta) = \perm_{\cA_0 \times \cT_0}(\Ab, \bTheta)$ for all $(\Ab, \bTheta) \in \cR$.
Recall that $\cR$ is endowed with the $L_1$-distance and the induced metric topology.
Let $\sR$ denote its Borel $\sigma$-algebra.
We notice that the binary relation of $(\Ab, \bTheta), (\Ab', \bTheta') \in \cR$ defined by
$$
(\Ab, \bTheta) \sim (\Ab', \bTheta')
\quad\iff\quad
(\Ab', \bTheta') \in \perm_\cR(\Ab, \bTheta)
$$
is an equivalence relation.
For notational simplicity, we also use $\perm_\cR(\Ab, \bTheta)$ to denote the equivalence class of $(\Ab, \bTheta)$ in $\cR$.
Letting $\cQ := \big\{ \perm_\cR(\Ab, \bTheta):~ (\Ab, \bTheta) \in \cR \big\}$ denote the collection of all the equivalence classes, we can define the map $\varphi:~ \cR \to \cQ$ as $\varphi(\Ab, \bTheta) := \perm_\cR(\Ab, \bTheta)$.
Since $\varphi$ is surjective, it induces the quotient topology on $\cQ$ such that a set $U \subset \cQ$ is open if and only if $\varphi^{-1}(U)$ is open in $\cR$.
We denote the Borel $\sigma$-algebra of the quotient topology on $\cQ$ by $\sQ$.

By definition of quotient topology, we know that the quotient map $\varphi$ is Borel measurable, i.e. for any $B \in \sQ$ we have $\varphi^{-1}(B) \in \sR$.
Furthermore, the following lemma shows that the quotient map $\varphi$ also maps Borel sets in $\cR$ to Borel sets in $\cQ$.
Note that this is a property specific for our model, which does not hold for general quotient maps.

\begin{lemma}\label{lemm:quotient_borel}
For any Borel set $B \in \sR$, its image $\varphi(B) \in \sQ$.
\end{lemma}

\begin{proof}[Proof of Lemma \ref{lemm:quotient_borel}]
We recall from \eqref{eq:sigma_def} that every parameter in $\perm_\cR(\Ab, \bTheta)$ corresponds to a permutation of the nodes in each layer of our model, that is, a combination of $K$ permutation matrices $\Pb_{0:(K - 1)} := \{\Pb_0, \ldots, \Pb_{K - 1}\}$.
Let $\kP$ denote the collection of all combinations $\Pb_{0:(K - 1)}$, which has finite cardinality $|\kP| = \prod_{k = 0}^{K - 1} (p_k!)$.
For $\Pb_{0:(K - 1)} \in \kP$, we let $M_{\Pb_{0:(K - 1)}}:~ \cR \to \cR$ be the mapping that takes a parameter $(\Ab, \bTheta) \in \cR$ to its permuted version $(\Ab', \bTheta')$ given by
$$
\forall k \in [K]
,\quad
\Ab_k' = \Pb_k \Ab_k \Pb_{k - 1}^\top
,\quad
C_k' = C_k
,\quad
\bGamma_k' = \Pb_{k - 1} \bGamma_k \Pb_{k - 1}^\top
,
$$
and
$$
\forall \Xb_0 \in \cX_0
,\quad
\PP(\Pb_0 \Xb_0 \Pb_0^\top | \bnu')
=
\PP(\Xb_0 | \bnu)
.
$$

Since $\cR$ is invariant to node permutations, for any $\Pb_{0:(K - 1)} \in \kP$ the mapping $M_{\Pb_{0:(K - 1)}}$ is a homeomorphism on $\cR$.
We also notice that for any $(\Ab, \bTheta) \in \cR$,
$$
\perm_\cR(\Ab, \bTheta)
=
\bigcup_{\Pb_{0:(K - 1)} \in \kP} \big\{
M_{\Pb_{0:(K - 1)}}(\Ab, \bTheta)
\big\}
.
$$
Therefore, for any open set $\cO \subset \cR$, we have
$$
\phi^{-1}(\phi(\cO))
=
\bigcup_{(\Ab, \bTheta) \in \cO} \perm_\cR(\Ab, \bTheta)
=
\bigcup_{\Pb_{0:(K - 1)} \in \kP} M_{\Pb_{0:(K - 1)}}(\cO)
,
$$
which is a finite union of open sets and therefore also open.
By definition of quotient topology, $\phi(\cO)$ is open in $\cQ$.
This completes the proof since the open sets generate the Borel $\sigma$-algebra.
\end{proof}

In addition to the topological properties, we find that $\cQ$ is also naturally a metric space.
For any equivalence classes $\perm_\cR(\Ab, \bTheta), \perm_\cR(\Ab', \bTheta') \in \cQ$, we define
$$
d_\cQ\left(
\perm_\cR(\Ab, \bTheta), \perm_\cR(\Ab', \bTheta')
\right)
:=
\inf_{(\tilde\Ab, \tilde\bTheta) \in \perm_\cR(\Ab, \bTheta)} \inf_{(\tilde\Ab', \tilde\bTheta') \in \perm_\cR(\Ab', \bTheta')} \left\|
(\tilde\Ab, \tilde\bTheta) - (\tilde\Ab', \tilde\bTheta')
\right\|_1
.
$$
Having an explicit metric simplifies the verification of the quotient parameter space being a Borel subset of Polish space, which is an important condition for applying Proposition \ref{prop:doob_meas}.

\begin{lemma}\label{lemm:quotient_metric}
$d_\cQ$ is a metric on $\cR$ that induces the quotient topology.
\end{lemma}

\begin{proof}[Proof of Lemma \ref{lemm:quotient_metric}]
The positivity and symmetry of $d_\cQ$ is direct.
We notice that the mappings $M_{P_{0:(K - 1)}}:~ \cR \to \cR$ defined in the proof of Lemma \ref{lemm:quotient_borel} are isometries under $L_1$-distance on $\cR$, which suggests the simplified formula
$$
d_\cQ\left(
\perm_\cR(\Ab, \bTheta), \perm_\cR(\Ab', \bTheta')
\right)
=
\inf_{(\tilde\Ab, \tilde\bTheta) \in \perm_\cR(\Ab, \bTheta)} \left\|
(\tilde\Ab, \tilde\bTheta) - (\Ab', \bTheta')
\right\|_1
.
$$
Therefore, for any equivalence classes $\perm_\cR(\Ab, \bTheta), \perm_\cR(\Ab', \bTheta'), \perm_\cR(\Ab^\dagger, \bTheta^\dagger) \in \cQ$, we have
\begin{align*}
&\quad~
d_\cQ\left(
\perm_\cR(\Ab, \bTheta), \perm_\cR(\Ab', \bTheta')
\right)
+
d_\cQ\left(
\perm_\cR(\Ab', \bTheta'), \perm_\cR(\Ab^\dagger, \bTheta^\dagger)
\right)
\\&=
\inf_{(\tilde\Ab, \tilde\bTheta) \in \perm_\cR(\Ab, \bTheta)} \left\|
(\tilde\Ab, \tilde\bTheta) - (\Ab', \bTheta')
\right\|_1
+
\inf_{(\tilde\Ab^\dagger, \tilde\bTheta^\dagger) \in \perm_\cR(\Ab^\dagger, \bTheta^\dagger)} \left\|
(\tilde\Ab^\dagger, \tilde\bTheta^\dagger) - (\Ab', \bTheta')
\right\|_1
\\&\ge
\inf_{(\tilde\Ab, \tilde\bTheta) \in \perm_\cR(\Ab, \bTheta)} \inf_{(\tilde\Ab^\dagger, \tilde\bTheta^\dagger) \in \perm_\cR(\Ab^\dagger, \bTheta^\dagger)} \left\|
(\tilde\Ab^\dagger, \tilde\bTheta^\dagger) - (\tilde\Ab, \tilde\bTheta)
\right\|_1
\\&=
d_\cQ\left(
\perm_\cR(\Ab, \bTheta), \perm_\cR(\Ab^\dagger, \bTheta^\dagger)
\right)
,
\end{align*}
which proves the triangle inequality for $d_\cQ$.
Hence $d_\cQ$ is a metric on $\cQ$.

For any $\epsilon > 0$, we let $\cO_\epsilon(\Ab, \bTheta)$ denote the $\epsilon$-neighborhood of $(\Ab, \bTheta)$ in $\cR$ under $L_1$ metric and $\cO_\epsilon(\perm_\cR(\Ab, \bTheta))$ denote the $\epsilon$-neighborhood of $\perm_\cR(\Ab, \bTheta)$ in $\cQ$ under $d_\cQ$ metric.
We notice that
\begin{equation}\begin{aligned}\label{eq:quotient_two_nbhd}
\varphi^{-1}\left(
\cO_\epsilon(\perm_\cR(\Ab, \bTheta))
\right)
&=
\left\{
(\Ab', \bTheta') \in \cR:~
\inf_{(\tilde\Ab, \tilde\bTheta) \in \perm_\cR(\Ab, \bTheta)} \left\|
(\tilde\Ab, \tilde\bTheta) - (\Ab', \bTheta')
\right\|_1 < \epsilon
\right\}
\\&=
\bigcup_{(\tilde\Ab, \tilde\bTheta) \in \perm_\cR(\Ab, \bTheta)} \cO_\epsilon(\tilde\Ab, \tilde\bTheta)
,
\end{aligned}\end{equation}
which is an open set in $\cR$.
This implies that $\cO_\epsilon(\perm_\cR(\Ab, \bTheta))$ is open under the quotient topology of $\cQ$ for any $\epsilon > 0$, which suggests that the metric topology induced by $d_\cQ$ is coarser than the quotient topology.

Now for any set $\cO \subset \cQ$ open under the quotient topology, we know $\varphi^{-1}(\cO)$ is an open set in $\cR$.
Hence there exist an index set $\cI$, parameters $\{r_i\}_{i \in \cI} \subset \cR$, and positive values $\{\epsilon_i\}_{i \in \cI}$ such that
$$
\varphi^{-1}(\cO)
=
\bigcup_{i \in \cI} \cO_{\epsilon_i}(r_i)
=
\bigcup_{i \in \cI} \bigcup_{r_i' \in \perm_\cR(r_i)} \cO_{\epsilon_i}(r_i')
=
\bigcup_{i \in \cI} \varphi^{-1}(\cO_{\epsilon_i}(\perm_\cR(r_i)))
,
$$
where the second equality is due to the definition of the quotient map $\varphi$ and the third equality is due to \eqref{eq:quotient_two_nbhd}.
This implies that $\cO$ is the union of open balls $\cO_{\epsilon_i}(\perm_\cR(r_i))$ and hence open under the metric topology, which suggests that the metric topology induced by $d_\cQ$ is finer than the quotient topology.

Putting everything together, we have shown that $(\cQ, d_\cQ)$ is a metric space and the induced metric topology is the same as the quotient topology.
\end{proof}

\subsection{Proof of Theorem \ref{theo:doob}}\label{ssec:proof_doob}

After introducing the quotient parameter space and its useful properties in Section \ref{ssec:post_quotient}, we are now ready to prove Theorem \ref{theo:doob} on Doob's posterior consistency of our model using Theorem \ref{theo:cite_doob}, Proposition \ref{prop:doob_meas}, and intermediate results from the proof of Theorem \ref{theo:generic} in Section \ref{ssec:iden_generic}.

\begin{proof}[Proof of Theorem \ref{theo:doob}]
In Section \ref{ssec:post_quotient}, we have assumed the parameter space $\cR \subset \cA_0 \times \cT_0$ to be invariant to node permutations, i.e. $\perm_\cR(\Ab, \bTheta) = \perm_{\cA_0 \times \cT_0}(\Ab, \bTheta)$ for all $(\Ab, \bTheta) \in \cR$.
While this is not assumed in Theorem \ref{theo:doob}, we can extend the parameter space $\cR$ to $\cup_{r \in \cR} \perm_\cR(r)$ and put zero prior mass on the extended regions.
Note that $\cup_{r \in \cR} \perm_\cR(r)$ is a Borel subset of $\cA_2 \times \cT_0$ as long as $\cR$ is.
Therefore, without loss of generality, we will assume the parameter space $\cR$ to be invariant to node permutations.

Recall the quotient map $\varphi$ defined in Section \ref{ssec:post_quotient}.
We let $\overline{\cA_2 \times \cT_0}$ denote the closure of $\cA_2 \times \cT_0$ and notice that it is a complete separable metric space under the $L_1$ distance.
Using the continuity of $\varphi$ and Lemma \ref{lemm:quotient_metric}, this suggests that the quotient space $\varphi(\overline{\cA_2 \times \cT_0})$ is a complete separable metric space under the $d_\cQ$ metric and hence a Polish space.

Recall \eqref{eq:measure_A2T0Rc} and \eqref{eq:generic_mid_proof_R} from the proof of Theorem \ref{theo:generic} in Section \ref{ssec:iden_generic}.
There exists a non-constant-zero holomorphic function $h_\Ab(\bTheta)$ for each $\Ab \in \cA_2$ such that the parameter space
$$
\cD
:=
\big\{
(\Ab, \bTheta) \in \cA_2 \times (\cT_0 \cap \cT_1^c):~
h_\Ab(\bTheta) \ne 0
\big\}
$$
satisfies the strict identifiability $\perm_\cD(\Ab, \bTheta) = \marg_\cD(\Ab, \bTheta)$ for all $(\Ab, \bTheta) \in \cD$ and the measure $(\upmu \times \uplambda)\left( (\cA_2 \times \cT_0) \cap \cD^c \right) = 0$.
Recalling the definition of $\cT_1$ from \eqref{eq:T1_def}, we notice that the set
$$
(\cA_2 \times \cT_0) \cap \cD^c
=
\bigcup_{\Ab \in \cA_2} \left(
\{\Ab\} \times \big( h_\Ab^{-1}(0) \cup \cT_1 \big)
\right)
$$
is Borel measurable in $\cA_2 \times \cT_0$.

Since the parameter space $\cR$ is a Borel subset of $\cA_2 \times \cT_0$, we know $\cR \cap \cD$ is also Borel measurable.
Applying Lemma \ref{lemm:quotient_borel} further suggests that $\cQ := \varphi(\cR \cap \cD)$ is a Borel subset of the Polish space $\varphi(\overline{\cA_2 \times \cT_0})$.
Recall from Section \ref{ssec:prove_schwartz} that the continuous map $\cL:~ \cR \cap \cD \to \cP$ denotes the map from parameter $(\Ab, \bTheta)$ to the probability vector $\cL(\Ab, \bTheta) := \{\PP(\Xb_K | \Ab, \bTheta)\}_{\Xb_K \in \cX_K}$.
Due to the trivial type of non-identifiability issue, $\cL$ agrees on all parameters $(\Ab, \bTheta)$ within the same equivalence class $\perm_{\cR \cap \cD}(\Ab, \bTheta)$, which suggests that $\cL$ passes to the quotient, i.e. there exists a continuous map $\cL_\cQ:~ \cQ \to \cP$ satisfying $\cL = \cL_\cQ \circ \varphi$.
Again applying Lemma \ref{lemm:quotient_borel}, we notice that $\cL_\cQ$ is Borel measurable.

Importantly, since $\perm_{\cR \cap \cD}(\Ab, \bTheta) = \marg_{\cR \cap \cD}(\Ab, \bTheta)$ for all $(\Ab, \bTheta) \in \cD$, the model $\{\cL_\cQ(q):~ q \in \cQ\}$ of quotient parameters is identifiable, i.e. $\cL_\cQ(\perm_{\cR \cap \cD}(\Ab, \bTheta)) \ne \cL_\cQ(\perm_{\cR \cap \cD}(\Ab', \bTheta'))$ for all $\perm_{\cR \cap \cD}(\Ab, \bTheta) \ne \perm_{\cR \cap \cD}(\Ab', \bTheta')$.
Therefore, all conditions of Proposition \ref{prop:doob_meas} are satisfied for the quotient parameter space $\cQ$, which suggests that there exists a Borel measurable function $f:~ \cX^\infty \to \cQ$ satisfying for every $\perm_{\cR \cap \cD}(\Ab, \bTheta) \in \cQ$,
$$
\perm_{\cR \cap \cD}(\Ab, \bTheta)
=
f(X^{(1:\infty)})
\quad
\cL_\cQ(\perm_{\cR \cap \cD}(\Ab, \bTheta)) \text{-a.s.}
.
$$
This indicates that $\perm_{\cR \cap \cD}(\Ab, \bTheta)$ is measurable with respect to the $Q$-completion of $\sigma(X^{(1:\infty)})$, where $Q$ is the joint probability measure on $\cQ \times \cX^\infty$ defined similar to \eqref{eq:jpmQ}.

Since the prior measure $\Pi$ is absolutely continuous with respect to $\upmu \times \uplambda$, we have
$$
\Pi\left(
\cR \cap (\cR \cap \cD)^c
\right)
\le
\Pi\left(
(\cA_2 \times \cT_0) \cap \cD^c
\right)
=
0
.
$$
We therefore neglect the set $\cR \cap (\cR \cap \cD)^c$ from the prior and posterior measures $\Pi, \Pi_n$ and let $\varphi_\# \Pi$ and $\varphi_\# \Pi_n$ denote their pushforward measures from $\cR$ to $\cQ$.
By applying Theorem \ref{theo:cite_doob}, we obtain for any $\epsilon > 0$,
$$
\lim_{n \to \infty} (\varphi_\# \Pi_n)\left(
\cO_\epsilon(\perm_\cR(\Ab^*, \bTheta^*))^c
\right)
=
0
\quad
\cL_\cQ\left(
\perm_\cR(\Ab^*, \bTheta^*)
\right)^\infty \text{-a.s.}
$$
for $\varphi_\# \Pi$-almost every true quotient parameter $\perm_\cR(\Ab^*, \bTheta^*)$.
Using \eqref{eq:quotient_two_nbhd} and recalling the definition of $\tilde\cO_\epsilon(\Ab, \bTheta)$ from \eqref{eq:eps_nbhd}, for any $\epsilon > 0$, the above result is equivalent to
$$
\lim_{n \to \infty} \Pi_n\left(
\tilde\cO_\epsilon(\Ab^*, \bTheta^*)^c
\right)
=
0
\quad
\cL(\Ab^*, \bTheta^*)^\infty \text{-a.s.}
$$
for $\Pi$-almost every true parameter $(\Ab^*, \bTheta^*)$ in $\cR$.
\end{proof}

\section{Posterior Computation: A Complete Guide}\label{sec:post_comp_supp}

In this section, we provide a complete guide of posterior computation for our model, including the details of our data augmented Gibbs sampler in Section \ref{ssec:stand_gibbs}, its subsampling version in Section \ref{ssec:sub_gibbs}, and the essentials of our spectral initialization algorithm in Section \ref{ssec:warm}.
We also discuss practical strategies for selecting the model structure, including the number of layers and layer widths, in Section \ref{ssec:post_comp_model_sel}.
Fo ease of reading, we defer a literature overview of subsampling MCMC and related discussions to Section \ref{sec:sub_gibbs}, and also defer the preliminaries and technical details of our spectral initialization algorithm to Section \ref{sec:spec_init}.

\subsection{Data Augmented Gibbs Sampler}\label{ssec:stand_gibbs}

To conduct posterior sampling, we have designed a Gibbs sampler using Polya-Gamma data augmentation \citep{polson2013bayesian}.
We provide its full details and derivations in this section.

With the model likelihood function specified in Section \ref{sec:model} and the prior distributions specified in Section \ref{sec:post}, the joint distribution of $\Xb^{(1:N)}$ and $\Ab, \bTheta$ is given by
\begin{align*}
&\quad~
\PP(\Xb^{(1:N)}, \Ab, \bTheta)
\\&=
\PP(\Ab, \bTheta) \prod_{n = 1}^N \PP(\Xb^{(n)} | \Ab, \bTheta)
\\&=
\PP(\bnu)
\left( \prod_{k = 1}^K \Big( \PP(C_k)
\prod_{1 \le i \le j \le p_{k - 1}} \PP(\Gamma_{k, i, j})
\Big) \right)
\left(
\prod_{k = 1}^K \prod_{i = 1}^{p_k} \PP(\ab_{k, i})
\right)
\\&\qquad
\prod_{n = 1}^N \left(
\PP(\Xb_0^{(n)} | \bnu) \prod_{k = 1}^K \PP(\Xb_k^{(n)} | \Xb_{k - 1}^{(n)}, \Ab_k, \bTheta_k)
\right)
\\&\propto
\left(
\prod_{\Xb_0 \in \cX_0} \nu_{\Xb_0}^{\alpha - 1}
\right)
\exp\left( \sum_{k = 1}^K \left(
- \frac{(C_k - \mu_C)^2}{2 \sigma_C^2}
-
\sum_{i = 1}^{p_{k - 1}} \frac{(\Gamma_{k, i, i} - \mu_\gamma)^2}{2\sigma_\gamma^2}
-
\sum_{1 \le i < j \le p_{k - 1}} \frac{(\Gamma_{k, i, j} - \mu_\delta)^2}{2\sigma_\delta^2}
\right) \right)
\\&\qquad
\left(
\prod_{k = 1}^K \prod_{1 \le i \le j \le p_{k - 1}} 1_{\Gamma_{k, i, j} > 0}
\right) \left(
\prod_{k = 1}^K \prod_{i = 1}^{p_k} 1_{1 \le \one^\top \ab_{k, i} \le S}
\right)
\\&\qquad
\prod_{n = 1}^N \left(
\nu_{\Xb_0^{(n)}} \prod_{k = 1}^K \prod_{1 \le i < j \le p_k} \frac{\exp\left( X_{k, i, j}^{(n)} \left(
C_k + \ab_{k, i}^\top (\bGamma_k * \Xb_{k - 1}^{(n)}) \ab_{k, j}
\right) \right)}{1 + \exp\left(
C_k + \ab_{k, i}^\top (\bGamma_k * \Xb_{k - 1}^{(n)}) \ab_{k, j}
\right)}
\right)
,
\end{align*}
where $\nu_{\Xb_0}$ denotes the entry in $\bnu$ corresponding to the probability $\PP(\Xb_0 | \bnu)$.

To obtain conditional conjugacy of parameters, for each $n \in [N], k \in [K], (i, j) \in \langle p_k \rangle$ we introduce an associated augmented variable $\omega_{k, i, j}^{(n)}$, conditionally following the distribution
$$
\omega_{k, i, j}^{(n)} \mid \Xb_{k - 1}^{(n)}, \Ab_k, \bTheta_k
\sim
\PG\left(
1
,
\psi_{k, i, j}^{(n)}
\right)
,
$$
where $\PG$ denotes the Polya-Gamma distribution \citep{polson2013bayesian} and
$$
\psi_{k, i, j}^{(n)}
:=
C_k + \ab_{k, i}^\top (\bGamma_k * \Xb_{k - 1}^{(n)}) \ab_{k, j}
.
$$
By properties of the Polya-Gamma distribution, we have the identity
$$
\frac{\exp\left( X_{k, i, j}^{(n)} \psi_{k, i, j}^{(n)} \right)}{1 + \exp\left(
\psi_{k, i, j}^{(n)}
\right)}
=
\EE\left[ \left. \frac12 \exp\left(
\big( X_{k, i, j}^{(n)} - \frac12 \big) \psi_{k, i, j}^{(n)} - \frac{\omega_{k, i, j}^{(n)}}{2} (\psi_{k, i, j}^{(n)})^2
\right) \right| \Xb_{k - 1}^{(n)}, \Xb_k^{(n)} \Ab_k, \bTheta_k \right]
,
$$
where the conditional expectation is taken over $\omega_{k, i, j}^{(n)}$.
Furthermore, the marginal distribution of each $\omega_{k, i, j}^{(n)}$
For each $k \in [K]$, we let
$$
\bomega_k
:=
\big\{
\omega_{k, i, j}^{(n)}:~
n \in [N], (i, j) \in \langle p_k \rangle
\big\}
$$
denote the collection of $\omega_{k, i, j}^{(n)}$'s for all $n, i, j$ and let $\bomega := \{ \bomega_k :~ k \in [K]\}$.
Letting $g$ denote the density of the Polya-Gamma distribution $\PG(1, 0)$, then the density of the Polya-Gamma distribution $\PG(1, c)$ for $c > 0$ is $\propto e^{- \frac{c^2}{2} \omega} g(\omega)$.
The joint distribution of $\Xb^{(1:N)}, \bomega, \Ab, \bTheta$ is given by
\begin{align*}
\\&\quad~
\PP(\Xb^{(1:N)}, \bomega, \Ab, \bTheta)
\\&=
\PP(\Ab, \bTheta)
\left(
\prod_{n = 1}^N \PP(\Xb^{(n)} | \Ab, \bTheta)
\right) \left(
\prod_{n = 1}^N \prod_{k = 1}^K \prod_{1 \le i < j \le p_k} \PP(\omega_{k, i, j}^{(n)} | \Xb_{k - 1}^{(n)}, \Ab_k, \bTheta_k)
\right)
\\&\propto
\left(
\prod_{\Xb_0 \in \cX_0} \nu_{\Xb_0}^{\alpha - 1}
\right)
\exp\left( \sum_{k = 1}^K \left(
- \frac{(C_k - \mu_C)^2}{2 \sigma_C^2}
-
\sum_{i = 1}^{p_{k - 1}} \frac{(\Gamma_{k, i, i} - \mu_\gamma)^2}{2\sigma_\gamma^2}
-
\sum_{1 \le i < j \le p_{k - 1}} \frac{(\Gamma_{k, i, j} - \mu_\delta)^2}{2\sigma_\delta^2}
\right) \right)
\\&\qquad
\left(
\prod_{k = 1}^K \prod_{1 \le i \le j \le p_{k - 1}} 1_{\Gamma_{k, i, j} > 0}
\right) \left(
\prod_{k = 1}^K \prod_{i = 1}^{p_k} 1_{1 \le \one^\top \ab_{k, i} \le S}
\right)
\\&\qquad
\left(
\prod_{n = 1}^N \nu_{\Xb_0^{(n)}}
\right) \left(
\prod_{n = 1}^N \prod_{k = 1}^K \prod_{1 \le i < j \le p_k} \left(
\exp\left(
\big( X_{k, i, j}^{(n)} - \frac12 \big) \psi_{k, i, j}^{(n)}
\right) \exp\left(
- \frac12 \omega_{k, i, j}^{(n)} (\psi_{k, i, j}^{(n)})^2
\right) g(\omega_{k, i, j}^{(n)})
\right) \right)
.
\end{align*}

We partition the parameters $\Ab, \bTheta$, latent variables $\Xb^{(1:N)}$, and augmented variables $\bomega$ into blocks.
These blocks include each row $\ab_{k, i}$ of $\Ab_k$ for $i \in [p_k], k \in [K]$, the parameter $\bnu$, each adjacency matrix $\Xb_0^{(n)}$ for $n \in [N]$, and each entry of all other parameters and variables.
The full conditional distributions of each block is given as follows.

For each $k \in [K]$ and each $i \in [p_k]$, letting $\Ab_{k, -i} := \big\{ \ab_{k, j}:~ j \in [p_k], j \ne i \big\}$ denote the collection of all row vectors in $\Ab_k$ other than $\ab_{i, k}$, we have
$$
\PP\left(
\ab_{k, i}
~\Big|~
\Xb_{k - 1}^{(1:N)}, \Xb_k^{(1:N)}, \Ab_{k, -i}, \bTheta_k
\right)
\propto
\prod_{n = 1}^N \prod_{j = 1, j \ne i}^{p_k} \frac{\exp\left(
X_{k, i, j}^{(n)} \psi_{k, i, j}^{(n)}
\right)}{1 + \exp\left(
\psi_{k, i, j}^{(n)}
\right)}
1_{1 \le \one^\top \ab_{k, i} \le S}
,
$$
which represents a categorical distribution over the space $\cV_k$ of all $p_k$-dimensional binary vectors satisfying our sparsity constraint.

For each $k \in [K]$, we have
$$
\PP\left(
C_k
~\Big|~
\Xb_{k - 1}^{(1:N)}, \Xb_k^{(1:N)}, \Ab_k, \bGamma_k, \bomega_k
\right)
\sim
N\left(
\tilde\mu_{C_k}
,
\tilde\sigma_{C_k}^2
\right)
,
$$
which is a normal distribution with mean and variance given by
$$
\tilde\mu_{C_k}
:=
\tilde\sigma_{C_k}^{-2} \left(
\sigma_C^{-2} \mu_C + \sum_{n = 1}^N \sum_{1 \le i < j \le p_k} \left(
X_{k, i, j}^{(n)} - \frac12 - \omega_{k, i, j}^{(n)} \left(
\psi_{k, i, j}^{(n)} - C_k
\right) \right) \right)
,
$$
$$
\tilde\sigma_{C_k}^2
:=
\left(
\sigma_C^{-2} + \sum_{n = 1}^N \sum_{1 \le i < j \le p_k} \omega_{k, i, j}^{(n)}
\right)^{-1}
.
$$

For each $k \in [K], s \in \langle p_{k - 1} \rangle$, letting $\bTheta_{k, -(s, t)} := \bTheta_k \cap \{\Gamma_{k, s, t}\}^c$ denote the collection of all parameters $C_k$ and $\Gamma_{k, i, j}$'s other than $\Gamma_{k, s, t}$, then we have
$$
\PP\left(
\Gamma_{k, s, t}
~\Big|~
\Xb_{k - 1}^{(1:N)}, \Xb_k^{(1:N)}, \Ab_k, \bTheta_{k, -(s, t)}, \bomega_k
\right)
\sim
N_+\left(
\tilde\mu_{\Gamma_{k, s, t}}
,
\tilde\sigma_{\Gamma_{k, s, t}}^2
\right)
,
$$
which is the normal distribution with mean $\tilde\mu_{\Gamma_{k, s, t}}$ and variance $\tilde\sigma_{\Gamma_{k, s, t}}^2$ truncated to the positive real line.
By denoting $\kappa_{i, j, s, t}^{(n)} := X_{k - 1, s, t}^{(n)} (A_{k, i, s} A_{k, j, t} + A_{k, i, t} A_{k, j, s})$, the mean and variance are given by
$$
\tilde\mu_{\Gamma_{k, s, t}}
:=
\tilde\sigma_{\Gamma_{k, s, t}}^{-2} \left(
\sigma_\delta^{-2} \mu_\delta + \sum_{n = 1}^N \sum_{1 \le i < j \le p_k} \left(
X_{k, i, j}^{(n)} - \frac12 - \omega_{k, i, j}^{(n)} \left(
\psi_{k, i, j}^{(n)} - \kappa_{i, j, s, t}^{(n)} \Gamma_{k, s, t}
\right) \right) \kappa_{i, j, s, t}^{(n)}
\right)
,
$$
$$
\tilde\sigma_{\Gamma_{k, s, t}}^2
:=
\left(
\sigma_\delta^{-2} + \sum_{n = 1}^N \sum_{1 \le i < j \le p_k} \omega_{k, i, j}^{(n)} (\kappa_{i, j, s, t}^{(n)})^2
\right)^{-1}
.
$$

Similarly, for each $k \in [K], s \in [p_{k - 1}]$, we have
$$
\PP\left(
\Gamma_{k, s, s}
~\Big|~
\Xb_{k - 1}^{(1:N)}, \Xb_k^{(1:N)}, \Ab_k, \bTheta_{k, -(s, s)}, \bomega_k
\right)
\sim
N_+\left(
\tilde\mu_{\Gamma_{k, s, s}}
,
\tilde\sigma_{\Gamma_{k, s, s}}^2
\right)
,
$$
which is the normal distribution with mean $\tilde\mu_{\Gamma_{k, s, s}}$ and variance $\tilde\sigma_{\Gamma_{k, s, s}}$ truncated to the positive real line.
By denoting $\kappa_{i, j, s} := A_{k, i, s} A_{k, j, s}$, the mean and variance are given by
$$
\tilde\mu_{\Gamma_{k, s, s}}
:=
\tilde\sigma_{\Gamma_{k, s, s}}^{-2} \left(
\sigma_\gamma^{-2} \mu_\gamma + \sum_{n = 1}^N \sum_{1 \le i < j \le p_k} \left(
X_{k, i, j}^{(n)} - \frac12 - \omega_{k, i, j}^{(n)} \left(
\psi_{k, i, j}^{(n)} - \kappa_{i, j, s} \Gamma_{k, s, s}
\right) \right) \kappa_{i, j, s}
\right)
,
$$
$$
\tilde\sigma_{\Gamma_{k, s, s}}^2
:=
\left(
\sigma_\gamma^{-2} + \sum_{n = 1}^N \sum_{1 \le i < j \le p_k} \omega_{k, i, j}^{(n)} \kappa_{i, j, s}
\right)^{-1}
.
$$

For the full conditional distribution of parameter $\bnu$, we have
$$
\PP\left(
\bnu
~\Big|~
\Xb_0^{(1:N)}
\right)
\sim
\mathrm{Dir}\big(
\tilde\balpha
\big)
,
$$
which is a Dirichlet distribution over the probability simplex $\cS^{2^{|\cX_0|} - 1}$.
Its parameter is defined entrywise by
$$
\tilde\alpha_{\Xb_0}
:=
\alpha + \sum_{n = 1}^N 1_{\Xb_0^{(n)} = \Xb_0}
$$
for all $\Xb_0 \in \cX_0$, with the entry $\tilde\alpha_{\Xb_0}$ in $\tilde\balpha$ corresponding to the entry $\nu_{\Xb_0}$ in $\bnu$.

For each $n \in [N], k \in [K], (i, j) \in \langle p_k \rangle$, the full conditional distribution of $\omega_{k, i, j}^{(n)}$ is by its definition the Polya-Gamma distribution
$$
\omega_{k, i, j}^{(n)}
~\Big|~
\Xb_{k - 1}^{(n)}, \Ab_k, \bTheta_k
\sim
\PG\left(
1
,
\psi_{k, i, j}^{(n)}
\right)
.
$$

As for the conditional distributions for blocks of $\Xb^{(1:N)}$, we partially collapse the augmented variables $\bomega$ and compute directly from the joint distribution of $\Ab, \bTheta, \Xb^{(1:N)}$.
For $k = 0$ and each $n \in [N]$, we have
$$
\PP\left(
\Xb_0^{(n)}
~\Big|~
\Xb_1^{(n)}, \bnu, \Ab_1, \bTheta_1
\right)
\propto
\nu_{\Xb_0^{(n)}}
\prod_{1 \le i < j \le p_1} \frac{\exp\left(
X_{1, i, j}^{(n)} \psi_{1, i, j}^{(n)}
\right)}{1 + \exp\left(
\psi_{1, i, j}^{(n)}
\right)}
,
$$
which is a categorical distribution over the space of all $p_0 \times p_0$ adjacency matrices $\cX_0$.

For each $k \in [K - 1], n \in [N], (i, j) \in \langle p_k \rangle$, letting $\Xb_{k, -(i, j)}^{(n)} := \Xb_k^{(n)} \cap \{X_{k, i, j}^{(n)}\}^c$ denote the collection of all entries in $\Xb_k^{(n)}$ other than $X_{k, i, j}^{(n)}$, we have
\begin{align*}
&\quad~
\PP\left(
X_{k, i, j}^{(n)}
~\Big|~
\Xb_{k - 1}^{(n)}, \Xb_{k, -(i, j)}^{(n)}, \Xb_{k + 1}^{(n)}, \Ab_{k:(k + 1)}, \bTheta_{k:(k + 1)}
\right)
\\&\propto
\frac{\exp\left(
X_{k, i, j}^{(n)} \psi_{k, i, j}^{(n)}
\right)}{1 + \exp\left(
\psi_{k, i, j}^{(n)}
\right)}
\left(
\prod_{1 \le i < j \le p_{k + 1}} \frac{\exp\left(
X_{k + 1, i, j}^{(n)} \psi_{k + 1, i, j}^{(n)}
\right)}{1 + \exp\left(
\psi_{k + 1, i, j}^{(n)}
\right)}
\right)
,
\end{align*}
which is a Bernoulli distribution over $\{0, 1\}$.
If we instead view the adjacency matrix $\Xb_k^{(n)}$ as a whole block, then its conditional distribution is
$$
\PP\left(
\Xb_k^{(n)}
~\Big|~
\Xb_{k - 1}^{(n)}, \Xb_{k + 1}^{(n)}, \Ab_{k:(k + 1)}, \bTheta_{k:(k + 1)}
\right)
\propto
\prod_{\ell \in \{k, k + 1\}} \prod_{1 \le i < j \le p_\ell} \frac{\exp\left( X_{\ell, i, j}^{(n)} \psi_{\ell, i, j}^{(n)}
\right)}{1 + \exp\left(
\psi_{\ell, i, j}^{(n)}
\right)}
,
$$
which is a categorical distribution over the space of all $p_k \times p_k$ adjacency matrices $\cX_k$.

\subsection{Subsampling Gibbs Sampler}\label{ssec:sub_gibbs}

As the sample size $N$ and the dimension $p$ get larger, the Gibbs sampler often experiences a high computational burden.
The main computational bottlenecks come from sampling the latent adjacency matrices $\Xb^{(1:N)}$ and the connection matrices $\Ab$, both of which have linear complexity on the sample size $N$.
The sampling of $\Xb^{(1:N)}$ can be accelerated by drawing the latent $\Xb^{(1)}, \ldots, \Xb^{(N)}$ in parallel, due to their conditional independence.
To boost the scalability in sampling the parameters $\Ab$ and $\bTheta$, it is worth considering whether there is a more efficient alternative way to evaluate their full conditional distributions.
A commonly used approach is subsampling MCMC, in which we aim at reducing the complexity of likelihood evaluation by designing a computationally efficient surrogate likelihood function \citep{quiroz2018speeding, quiroz2018subsampling}.
See Section \ref{ssec:overview_sub} for an overview of related literature.
We next adapt this idea to our Gibbs sampler.

After data augmentation, the latent variables of our model include $\Xb^{(1:N)}$ and $\bomega^{(1:N)}$, where $\bomega^{(n)} := \big\{\omega_{k, i, j}^{(n)}: k \in [K], (i, j) \in \langle p_k \rangle\big\}$ denotes the collection of augmented variables associated with sample $n$.
At each iteration of our subsampling Gibbs sampler, we start by drawing $\Xb^{(n)}, \bomega^{(n)}$ in parallel across $n$ from their standard full conditional distributions.
We then uniformly sample without replacement a subset $\cB$ of fixed cardinality $|\cB|$ from the index set $[N]$.
Using this subset of data and latent variables, we obtain an estimator for the data augmented log-likelihood
\begin{equation}\label{eq:surro}
\hat\ell(\Ab, \bTheta)
:=
\frac{N}{|\cB|} \sum_{n \in \cB} \log \PP(\Xb^{(n)}, \bomega^{(n)} | \Ab, \bTheta)
.
\end{equation}
With this estimator $\hat\ell(\Ab, \bTheta)$, we could derive the approximate full conditional distributions of $\Ab, \bTheta$ and sample them in blocks in the same way as our standard Gibbs sampler.

As a type of approximate MCMC algorithm \citep{johndrow2015optimal, johndrow2017error}, the stationary distribution of the subsampling Gibbs sampler is different from the standard Gibbs sampler.
For the subset ratio of $\frac{|\cB|}{N} = 1\%$, we numerically verify this distance to be small through simulation study in Section \ref{sec:sim}.
Our posterior computation involves iterations of the subsampling Gibbs sampler followed by iterations of the standard Gibbs sampler.
The subsampling Gibbs samples facilitate fast convergence to the posterior mode after initialization, while posterior inference is conducted based on the standard Gibbs samples.
See Section \ref{ssec:discuss_sub} for further discussions.

\subsection{Spectral Initialization}\label{ssec:warm}

We now present an additional technique for accelerating posterior computation.
The Degree Corrected Mixed Membership (DCMM) model \citep{jin2023mixed} is proposed to detect communities in a single sparse network, allowing mixed membership and node degree heterogeneity.
In the following, we establish a somewhat surprising link between the DCMM and our model.
Using a multilayer hierarchical extension of their Mixed-Score algorithm, we can efficiently find a spectral initialization of our connection matrices $\Ab$ and significantly reduce the burn-in period needed in our MCMC sampler.

For a single network with $p$ nodes and $q$ communities, DCMM models adjacency matrix $\Yb$ as
\begin{equation}\label{eq:dcmm}
\EE[\Yb | \Db, \bPi, \Zb]
=
\Db \bPi \Zb \bPi^\top \Db
,
\end{equation}
where $\Db \in \RR^{p \times p}$ is a positive diagonal matrix, $\bPi \in \RR^{p \times q}$ has each row $\bpi_i$ in the simplex $\cS^{q - 1}$, and $\Zb$ is a symmetric positive matrix.
Each diagonal entry $D_i$ of $\Db$ denotes the degree of node $i$, each row $\bpi_i$ in $\bPi$ denotes the mixed membership of node $i$, and each entry $Z_{s, t}$ in $\Zb$ represents the adjacency probability of two nodes with partial memberships in communities $s$ and $t$.
We find that any two consecutive layers in our model are related to the DCMM in a sense stated in the following theorem.
For notational simplicity, we denote $G_k := \max_{i, j \in \langle p_k \rangle} \ab_{k, i}^\top \bGamma_k \ab_{k, j}$ for each $k \in [K]$.

\begin{theorem}\label{theo:warm_init}
Let parameter $(\Ab, \bTheta) \in \cA_0 \times \cT_0$.
For each $k \in [K]$ and any $\beta, \epsilon > 0$, with probability at least $1 - \epsilon$ we have
$$
\left\|
\frac{1}{N} \sum_{n = 1}^N \Xb_k^{(n)}
-
\Db \bPi \Zb \bPi^\top \Db
\right\|_2
\le
p_k \sup_{0 \le x \le G_k} \left|
\frac{\exp(C_k + x)}{1 + \exp(C_k + x)} - \beta x
\right|
+
\frac{\sqrt{2p_k} \log(2p_k \epsilon^{-1})}{\sqrt{N}}
,
$$
where we let $\Db := \diag(\Ab_k \one)$, $\bPi := \Db^{-1} \Ab_k$, and $\Zb := \beta \bGamma_k * \frac{1}{N} \sum_{n = 1}^N \Xb_{k - 1}^{(n)}$.
\end{theorem}

In Theorem \ref{theo:warm_init}, we have upper bounded the difference between the sample average $\frac{1}{N} \sum_{n = 1}^N \Xb_k^{(n)}$ and $\Db \bPi \Zb \bPi^\top \Db$ by the sum of a constant term and a term decaying with $N$.
While the presence of the constant term implies that this bound does not vanish as $N \to \infty$, the goal of the theorem is not to establish an asymptotic equivalence between the DCMM and our model.
Rather, it highlights a structural proximity between the two formulations.

This connection motivates the use of the Mixed-SCORE algorithm \citep{jin2023mixed}, which is an efficient and accurate spectral method for recovering $\bPi$ from $\Yb$ under the DCMM, as an initialization tool for the connection matrices $\Ab$ in our model.
Although, in principle, the Gibbs sampler converges to the correct posterior distribution regardless of initialization, a well-chosen initialization, such as that provided by Mixed-SCORE, can substantially reduce burn-in and improve computational efficiency in practice.

In our multilayer extension of Mixed-SCORE algorithm, we will recover the connection matrices $\Ab$ in a bottom-up fashion, starting from $\Ab_K$ of the observed layer.
By applying the Mixed-SCORE algorithm to $\frac{1}{N} \Xb_K^{(n)}$, we obtain an estimate of $\bPi := \diag(\Ab_K \one)^{-1} \Ab_K$ and truncate it into an estimate of $\Ab_K$ through $A_{K, i, j} = 1_{\Pi_{i, j} \ge S^{-1}}$.
We then pick the $p_{K - 1}$ rows in $\bPi$ closest to the $p_{K - 1}$ vertices of simplex $\cS^{p_{K - 1} - 1}$ and select the adjacency sub-matrices $\tilde\Xb_{K - 1}^{(1:N)}$ by keeping only these $p_{k - 1}$ rows and columns of $\Xb_K^{(1:N)}$.
By repeatedly applying this approach to $\tilde\Xb_{K - 1}^{(1:N)}, \tilde\Xb_{K - 2}^{(1:N)}, \ldots$, we sequentially obtain initial estimates of $\Ab_{K - 1}, \Ab_{K - 2}, \ldots$, which together form the spectral initialization of our connection matrices $\Ab$.
We defer the details of this algorithm to Section \ref{sec:spec_init}, along with a thorough overview of the DCMM and the proof of Theorem \ref{theo:warm_init}.

\subsection{Model Selection Guide}\label{ssec:post_comp_model_sel}

The proposed Bayesian deep generative model requires specifying both the number of layers $K$ and the number of nodes $p_k$ in each latent layer of the Bayesian network.
These structural quantities can be determined in several ways. When domain knowledge or scientific context suggests a natural multi-resolution representation, such as known functional modules, interaction modes, or biological organization, the number of layers $K$ may simply be set a priori.
In many applications, however, it is desirable to choose $K$ and $(p_0,\dots,p_K)$ in a data-driven manner.
For this purpose, predictive model comparison tools such as the Widely Applicable Information Criterion (WAIC) \citep{watanabe2010asymptotic} provide a principled and convenient mechanism for selecting the model complexity.

At a high level, we treat the layer widths $(p_0, \dots, p_K)$ as unknown structural quantities and evaluate a range of plausible configurations subject to basic identifiability constraints (e.g., $p_k \ge 2 p_{k-1}$).
For each configuration, we perform posterior inference for the model parameters using the Gibbs sampling procedure described earlier and compute WAIC to assess predictive accuracy.
The configuration that minimizes WAIC is selected as the preferred model.
The same approach may be applied to selecting the number of layers $K$ by fitting the model under several candidate values and choosing the one with the lowest WAIC.

Conceptually, it is also possible to treat $K$ and each $p_k$ as a model parameter with its own prior distribution, and to infer it jointly with the other parameters using reversible jump MCMC or other trans-dimensional Bayesian methods.
While theoretically appealing, these approaches introduce substantial practical difficulties.
Reversible jump moves require carefully designed dimension-matching transformations, and in high-dimensional network models with many latent variables they tend to mix poorly.
As a result, fully Bayesian inference over $K$ is typically computationally expensive, unstable, and difficult to tune in practice. For these reasons, WAIC based model selection is more computationally tractable, easier to implement, and empirically stabler.

In practice, model selection must also balance interpretability and computational scalability.
For many applications involving community structure or hierarchical organization, a shallow hierarchy is sufficient to capture the relevant structure.
As a result, choosing $K = 2$, i.e. a three-layer latent hierarchy, often provides an effective trade-off between flexibility and interpretability while remaining computationally efficient for large networks.

A detailed illustration of the WAIC-based selection of $(p_k)$ appears in the simulation section, where we evaluate WAIC over a grid of candidate configurations and find that it is minimized at the true underlying values.
These results also show that WAIC varies smoothly across different choices of $(p_k)$, supporting the reliability of the proposed model-selection strategy.

\section{Subsampling Gibbs Sampler}\label{sec:sub_gibbs}

In this section, we provide a brief overview of subsampling MCMC methods in Section \ref{ssec:overview_sub} and discuss some further variants of our subsampling Gibbs sampler in Section \ref{ssec:discuss_sub}.

\subsection{Brief Overview of Subsampling MCMC}\label{ssec:overview_sub}

Many different versions of subsampling MCMC algorithms have been proposed in the past decade.
As a crucial step in any MCMC algorithm, the evaluation of model likelihood often has very high computational complexity when the sample size of data is large.
To speed up the MCMC algorithm, we could approximate the model log-likelihood with an estimator based on only a subset of data.
For MCMC algorithms based on stochastic differential equations, the approximated versions include stochastic gradient Langevin dynamics \citep{welling2011bayesian} and other variants \citep{ahn2012bayesian, ma2015complete, baker2019control}.
For Metropolis-Hastings algorithms that involve likelihoods in computation of acceptance ratios, efficient estimators for the acceptance ratio based on a subset of data are available \citep{korattikara2014austerity, bardenet2014towards, bardenet2017markov}.
To reduce the variance of the log-likelihood estimator, different control variate methods \citep{quiroz2018speeding, quiroz2018speeding2} and subsample weighting methods \citep{maire2019informed} have also been proposed.

Despite the numerous subsampling MCMC algorithm, it is important to note that all of them fall into the category of approximate MCMC methods, which do not exactly sample from the target distribution.
Theoretical tools have been developed to bound the distance between the stationary distribution of approximate MCMC method and the target distribution when the original MCMC method is uniformly ergodic \citep{johndrow2015optimal, johndrow2017coupling} or geometrically ergodic \citep{pillai2014ergodicity, johndrow2017error, rudolf2018perturbation, negrea2021approximations}.

\subsection{Discussions of Subsampling Gibbs Sampler}\label{ssec:discuss_sub}

We now discuss a few possible variants of our subsampling Gibbs sampler.

In the current subsampling Gibbs sampler, while a subset of samples $\cB$ is randomly selected at the start of each iteration, all the augmented variables $\omega_{k, i, j}^{(n)}$'s are drawn from their full conditional distributions.
These augmented variables can be sampled in parallel efficiently, but the computational complexity will be high when parallel computing is not available.
As a compromise, we could instead draw only the subset of augmented variables $\omega_{k, i, j}^{(n)}$ with $n \in \cB$.
While the resulting Markov chain has a stationary distribution that is theoretically further from the desired posterior distribution, the difference is observed to be small in practice.

As discussed in Section \ref{ssec:sub_gibbs}, in each iteration of the subsampling Gibbs sampler, we are estimating the data augmented model log-likelihood using a subset $\cB$ of data as
$$
\hat\ell(\Ab, \bTheta)
:=
\frac{N}{|\cB|} \sum_{n \in \cB} \log \PP(\Xb^{(n)}, \bomega^{(n)} | \Ab, \bTheta)
.
$$
Other surrogate functions of the log-likelihood could potentially also be designed, in order to further reduce the evaluation time or improve the approximation quality.
In current literature \citep{quiroz2018subsampling}, there are ideas of designing surrogate functions based on Taylor expansions or clustering of the data. We leave this for future work.

\section{Spectral Initialization}\label{sec:spec_init}

In this section, we provide the necessary background for the DCMM model and the Mixed-SCORE algorithm \citep{jin2023mixed} in Section \ref{ssec:prelim_dcmm}.
We then present the multi-layer extension of Mixed-SCORE algorithm to our model in Section \ref{ssec:spec_init_alg}, which can be used as the spectral initialization.
The proof of Theorem \ref{theo:warm_init} is provided in Section \ref{ssec:proof_warm_init}.

\subsection{Preliminaries of DCMM and Mixed-SCORE Algorithm}\label{ssec:prelim_dcmm}

The Degree Corrected Mixed Membership (DCMM) model \citep{jin2023mixed} is designed to detect communities in a single sparse network, allowing mixed membership and node degree heterogeneity.
An elementary introduction to the DCMM model and the Mixed-SCORE algorithm is provided here.
We refer readers to \citet{jin2023mixed} for a complete study.

For a single sparse network with a $p \times p$ adjacency matrix $\Yb$ and $\tilde{p}$ communities, the DCMM model assumes that the  upper-triangular entries of $\Yb$ are conditionally independent given Bernoulli parameters:
$$
\PP(Y_{i, j} = 1 ~|~ \Db, \bPi, \Zb)
=
D_i D_j \bpi_i^\top \Zb \bpi_j
,
$$
where $D_i > 0$ denotes the degree of node $i$, $\bpi_i \in \RR^{\tilde{p} \times 1}$ denotes the mixed membership vector of node $i$, and $\Zb \in \RR^{\tilde{p} \times \tilde{p}}$ is a symmetric matrix with positive entries.
The diagonal matrix $\Db = \diag(D_1, \ldots, D_p)$, the $p \times \tilde{p}$ membership matrix $\bPi$ with its $i$th row being $\bpi_i$, and the matrix $\Zb$ together compose the parameters of the DCMM model.
Intuitively, the entry $Z_{i, j}$ $(i, j \in [1, \tilde{p}])$ of $\Zb$ models the probability of connection between a node purely belonging to community $i$ and a node purely belonging to community $j$, as in stochastic block models \citep{holland1983stochastic}.
Each row vector $\bpi_i$ in the mixed membership matrix $\bPi$ belongs to the probability simplex, with all entries non-negative and summing up to one.
Therefore, the probability of connection between two nodes is a weighted average of $Z_{i, j}$'s over all pairs of communities $(i, j)$'s that the two nodes partially belong to.
Different nodes may have different propensity of being connected to other nodes regardless of their community memberships, which is known as the degree heterogeneity issue. The degree correction matrix $\Db$ is introduced to adjust for such degree heterogeneity.
In matrix form, we have
$$
\EE[\Yb | \Db, \bPi, \Zb]
=
\Db \bPi \Zb \bPi^\top \Db
.
$$

Given the observation $\Yb$, the Mixed-SCORE algorithm estimates $\bPi$ in three steps: spectral clustering on ratios-of-eigenvectors (SCORE), vertex hunting, and membership reconstruction.
\begin{enumerate}
\item
In the SCORE step, for the top $\tilde{p}$ eigenvectors $\bxi_1, \ldots, \bxi_{\tilde{p}}$ of $\Yb$, we divide $\bxi_2, \ldots, \bxi_{\tilde{p}}$ entrywise over $\bxi_1$ and stack the resulted $\tilde{p} - 1$ vectors by columns into a matrix $\Rb \in \RR^{p \times {\tilde{p} - 1}}$.
For intuition purposes, we also consider the matrix $\tilde\Rb$ obtained by applying this SCORE step to $\EE[\Yb | \Db, \bPi, \Zb]$ instead of $\Yb$, since the difference between $\Rb$ and $\tilde\Rb$ is proven to be small.
It is also shown that the rows of $\tilde\Rb$ form a simplex in $\RR^{\tilde{p} - 1}$.
When each community possesses at least one pure node (a node that belongs purely to this community), the vertices of this simplex corresponds to the rows in $\tilde\Rb$ of the pure nodes.
\item
In the vertex hunting step, we identify the vertices of this simplex using existing methods such as successive projection \citep{araujo2001successive, nascimento2005vertex}.
This allows us to find out the pure nodes of the model.
To improve the robustness of this step, it is often helpful to conduct $k$-means clustering of the rows in $\Rb$ before finding the vertices, know as sketched vertex search \citep[SVS,][]{jin2023mixed}.
\item
In the membership reconstruction step, given the pure nodes, the mixed memberships of each node can be recovered through solving linear equation systems.
\end{enumerate}
The mixed memberships $\bPi$ can be exactly recovered when the Mixed-SCORE algorithm is applied to $\EE[\Yb | \Db, \bPi, \Zb]$.
In practice, only $\Yb$ is known and the Mixed-SCORE algorithm is applied to $\Yb$, but the estimator of $\bPi$ obtained remains valid and consistent since the difference $\|\Yb - \EE[\Yb | \Db, \bPi, \Zb]\|_2$ can be bounded to be small.

\subsection{Our Spectral Initialization Algorithm}\label{ssec:spec_init_alg}

\begin{algorithm}
\caption{Multilayer Mixed-SCORE Algorithm}
\begin{algorithmic}
\STATE
Compute the sample mean of observed adjacency matrices
$$
\overline\Xb_K
:=
\frac{1}{N} \sum_{n = 1}^N \Xb_K^{(n)}
;
$$
\FOR{$k = K, K - 1, \ldots, 1$}
\STATE
Apply Mixed-SCORE algorithm to $\overline\Xb_k$ and obtain the mixed membership matrix $\bPi_k$;
\STATE
Truncate $\bPi_k$ to connection matrix $\Ab_k$ entrywise as
$$
A_{k, i, j}
:=
1_{\Pi_{k, i, j} \ge S^{-1}}
,\quad
\forall i \in [p_k], j \in [p_{k - 1}]
;
$$
\STATE
Find the $p_{k - 1}$ rows $\cI_k$ in $\bPi_k$ closest to distinct pure nodes through
$$
\cI_k
=
\argmin_{\cI_k \subset [p_k], |\cI_k| = p_{k - 1}} \|(\bPi_k)_{\cI_k, :} - \Ib_{p_{k - 1}}\|_F^2
;
$$
Keep only the $\cI_k$ rows and $\cI_k$ columns of $\overline\Xb_k$ and let the $p_{k - 1} \times p_{k - 1}$ submatrix be
$$
\overline\Xb_{k - 1}
:=
(\overline\Xb_k)_{\cI_k, \cI_k}
;
$$
\ENDFOR
\STATE
Output the initialized connection matrices $\Ab := \{\Ab_k\}_{k = 1}^K$.
\end{algorithmic}
\label{algo:warm_init}
\end{algorithm}

Theorem \ref{theo:warm_init} shows that the sample average of adjacency matrices $\frac{1}{N} \sum_{n = 1}^N \Xb_1^{(n)}$ is close to the expected adjacency matrix $\Db \bPi \Zb \bPi^\top \Db$ in a DCMM model, with the mixed membership matrix being $\bPi = \diag(\Ab_1 \one)^{-1} \Ab_1$.
Given $\bPi$, we can recover the connection matrix $\Ab_1$ through the simple truncation $\Ab_{1, i, j} = 1_{\Pi_{i, j} \ge S^{-1}}$.
This suggests an algorithm for warmly initializing $\Ab$ for $K = 1$.

When our Bayesian network has more than two layers (i.e. $K > 1$), we can take the following approach to iteratively apply the above algorithm for $K = 1$ and initialize the connection matrices $\Ab_k$ for $k = K, K - 1, \ldots, 1$ in a bottom-up fashion.
For each $k \in [K]$, after initializing the connection matrix $\Ab_k$, we can compute the corresponding DCMM mixed membership matrix $\bPi_k = \diag(\Ab_k \one)^{-1} \Ab_k$.
We pick the submatrix of $\bPi_k$ that is closest to an identity matrix, such that these selected $p_{k - 1}$ rows in $\bPi_k$ roughly corresponds to $p_{k - 1}$ distinct pure nodes.
We keep only these $p_{k - 1}$ rows and columns of the adjacency matrices $\Xb_k^{(1:N)}$, leading to the $p_{k - 1} \times p_{k - 1}$ adjacency sub-matrices $\tilde{\Xb}_k^{(1:N)}$.
The obtained adjacency sub-matrices will closely characterize the adjacency between the nodes in the $(k - 1)$th layer in the Bayesian network, to which the Mixed-SCORE algorithm and truncation can be applied again to initialize $\Ab_{k - 1}$.
Repeated application of this approach for $k = K, K - 1, \ldots, 1$ allows us to initialize all connection matrices $\Ab$.
We refer to this warm initialization algorithm as the Multilayer Mixed-SCORE algorithm and formalize it in Algorithm \ref{algo:warm_init}.

\subsection{Proof of Theorem \ref{theo:warm_init}}\label{ssec:proof_warm_init}

In the following, we first state the matrix Bernstein inequality.
A proof can be found in Theorem 6.1 of \citet{tropp2015introduction}.
As before, we let $\|\cdot\|_2$ denote the spectral norm of a matrix.

\begin{theorem}[Matrix Bernstein Inequality]\label{theo:mat_bern}
Let $\bZ_1, \ldots, \bZ_N$ be independent $d_1 \times d_2$ random matrices satisfying $\EE[\bZ_n] = 0$ and $\|\bZ_n\|_2 \le L'$ almost surely for some constant $L'$ and each $n \in [N]$.
Then for all $t \ge 0$,
$$
\PP\left( \left\|
\sum_{n = 1}^N \bZ_n
\right\|_2 \ge t \right)
\le
(d_1 + d_2) \exp\left(
- \frac{t^2 / 2}{L + L' t / 3}
\right)
,
$$
where constant $L := \max\left\{ \left\|\sum_{n = 1}^N \EE[\bZ_n \bZ_n^\top]\right\|_2, \left\|\sum_{n = 1}^N \EE[\bZ_n^\top \bZ_n]\right\|_2 \right\}$.
\end{theorem}

We now prove Theorem \ref{theo:warm_init} using Theorem \ref{theo:mat_bern}.

\begin{proof}[Proof of Theorem \ref{theo:warm_init}]
We note that it suffices to prove Theorem \ref{theo:warm_init} for $K = 1$.
Under $K = 1$, let $\Xb^{(1:N)}$ be independent realizations from our model with parameter $(\Ab, \bTheta) \in \cA_0 \times \cT_0$.
Conditioned on $\Xb_0^{(1:N)}$, for each entry $(i, j)\in \langle p_1 \rangle$ we have $X_{1, i, j}^{(n)}$ independently following the Bernoulli distribution with
$$
\PP\left(
X_{1, i, j}^{(n)} = 1
\Big|
\Ab, \bTheta, \Xb_0^{(n)}
\right)
=
\EE\left[
X_{1, i, j}^{(n)}
\Big|
\Ab, \bTheta, \Xb_0^{(n)}
\right]
=
\frac{\exp\left(
C + \ab_{1, i}^\top (\Xb_0^{(n)} * \bGamma_1) \ab_{1, j}
\right)}{1 + \exp\left(
C + \ab_{1, i}^\top (\Xb_0^{(n)} * \bGamma_1) \ab_{1, j}
\right)}
.
$$
By noticing
$$
0
\le
\ab_{1, i}^\top (\Xb_0^{(n)} * \bGamma_1) \ab_{1, j}
\le
G
=
\max_{1 \le i, j \le p_1} \ab_{1, i}^\top \bGamma_1 \ab_{1, j}
\quad
\forall (i, j) \in \langle p_1 \rangle, n \in [N]
,
$$
we have
$$
\left\|
\EE\left[
\Xb_1^{(n)}
\Big|
\Ab, \bTheta, \Xb_0^{(n)}
\right]
-
\beta \Ab_1^\top (\Xb_0^{(n)} * \bGamma_1) \Ab_1
\right\|_2
\le
p_1 \sup_{0 \le x \le G} \left|
\frac{\exp(C + x)}{1 + \exp(C + x)} - \beta x
\right|
,
$$
which by Jensen's inequality gives
$$
\left\|
\EE\left[
\Xb_1^{(n)}
\Big|
\Ab, \bTheta
\right]
-
\beta \Ab_1^\top (\bGamma_1 * \EE[\Xb_0 | \bnu]) \Ab_1
\right\|_2
\le
p_1 \sup_{0 \le x \le G} \left|
\frac{\exp(C + x)}{1 + \exp(C + x)} - \beta x
\right|
,
$$
where $\EE[\Xb_0 | \bnu]$ is the $p_0 \times p_0$ matrix denoting the expectation of $\Xb_0^{(n)}$ under the categorical distribution over its space $\cX_0$ parameterized by $\bnu$.

For each index pair $(i, j) \in \langle p_1 \rangle$ and sample $n \in [N]$, we define the random matrix
$$
\bZ_{n, i, j}
:=
\left(
X_{1, i, j}^{(n)} - \EE\left[
X_{1, i, j}^{(n)} \Big| \Ab, \bTheta, \Xb_0^{(n)}
\right] \right) \Eb_{i, j}
,
$$
where $\Eb_{i, j} := \eb_i \eb_j^\top + \eb_j \eb_i^\top 1_{i \ne j}$ denotes the indicator matrix with all entries zero except for entries $(i, j)$ and $(j, i)$.
We notice that each $\bZ_{n, i, j}$ satisfies
$$
\EE\left[
\bZ_{n, i, j} \Big| \Ab, \bTheta, \Xb_0^{(n)}
\right]
=
0
,\quad
\|\bZ_{n, i, j}\|_2
\le
1
.
$$
Furthermore, we have
$$
\EE\left[
\bZ_{n, i, j} \bZ_{n, i, j}^\top \Big| \Ab, \bTheta, \Xb_0^{(n)}
\right]
=
\VV\left[
X_{1, i, j} \Big| \Ab, \bTheta, \Xb_0^{(n)}
\right] (\Eb_{i, i} + \Eb_{j, j})
$$
and
$$
\VV\left[
X_{1, i, j} \Big| \Ab, \bTheta, \Xb_0^{(n)}
\right]
\le
1
,
$$
which together imply that
$$
\left\|
\sum_{n = 1}^N \sum_{(i, j) \in \langle p_1 \rangle} \EE\left[
\bZ_{n, i, j} \bZ_{n, i, j}^\top \Big| \Ab, \bTheta, \Xb_0^{(n)}
\right]
\right\|_2
\le
N \left\|
\sum_{(i, j) \in \langle p_1 \rangle} (\Eb_{i, i} + \Eb_{j, j})
\right\|_2
=
N (p_1 - 1)
.
$$
Notice that $\bZ_{n, i, j}$'s are conditionally independent given $\Ab, \bTheta, \Xb_0^{(1:N)}$.
By applying the matrix Bernstein inequality in Theorem \ref{theo:mat_bern} to the summation $\sum_{n = 1}^N \sum_{(i, j) \in \langle p_1 \rangle} \bZ_{n, i, j}$ conditioned on $\Ab, \bTheta, \Xb_0^{(1:N)}$, we obtain for any $t > 0$
$$
\PP\left( \left. \left\|
\sum_{n = 1}^N \sum_{(i, j) \in \langle p_1 \rangle} \bZ_{n, i, j}
\right\|_2 \ge t
~\right|~
\Ab, \bTheta, \Xb_0^{(1:N)}
\right)
\le
2p_1 \exp\left(
- \frac{t^2 / 2}{N(p_1 - 1) + t / 3}
\right)
.
$$

We recall that the DCMM model parameters corresponding to our model are $\Db := \diag(\Ab_1 \one), \bPi := \Db^{-1} \Ab_1, \Zb := \beta \bGamma_1 * \frac{1}{N} \sum_{n = 1}^N \Xb_0^{(n)}$.
Since the diagonal entries of adjacency matrices are fixed, by triangle inequality we have
\begin{align*}
\left\|
\frac{1}{N} \sum_{n = 1}^N \Xb_1^{(n)} - \Db \bPi \Zb \bPi^\top \Db
\right\|_2
&\le
\left\|
\frac{1}{N} \sum_{n = 1}^N \EE\left[
\Xb_1^{(n)} \Big| \Ab, \bTheta, \Xb_0^{(n)}
\right]
-
\beta \Ab_1^\top \left(
\bGamma_1 * \frac{1}{N} \sum_{n = 1}^N \Xb_0^{(n)}
\right) \Ab_1
\right\|_2
\\&\quad+
\left\|
\frac{1}{N} \sum_{n = 1}^N \sum_{(i, j) \in \langle p_1 \rangle} \bZ_{n, i, j}
\right\|_2
.
\end{align*}
Let $\epsilon > 0$ be arbitrary.
Putting everything together and letting $t := \sqrt{2p_1 N} \log(2p_1 \epsilon^{-1})$, with probability at least $1 - \epsilon$ we have
$$
\left\|
\frac{1}{N} \sum_{n = 1}^N \Xb_1^{(n)} - \Db \bPi \Zb \bPi^\top \Db
\right\|_2
\le
p_1 \sup_{0 \le x \le G} \left|
\frac{\exp(C + x)}{1 + \exp(C + x)} - \beta x
\right|
+
\frac{\sqrt{2p_1} \log(2p_1 \epsilon^{-1})}{\sqrt{N}}
.
$$
\end{proof}

\section{Simulation Study}\label{sec:sim}

In this section, we numerically evaluate our model consistency and posterior computation approaches through simulated examples.
We present the main results in Section \ref{ssec:main_sim} and additional results with discussions in Section \ref{ssec:more_sim}.
We also examine robustness to hyperparameter specification via sensitivity analysis in Section \ref{ssec:sens}, study performance in the large-$p_K$ small-$N$ regime in Section \ref{ssec:large_p_small_N}, and provide comparisons with the state-of-the-art hierarchical community detection methods proposed by \citet{li2022hierarchical} in Section \ref{ssec:compare_hcd}.

\subsection{Main Results}\label{ssec:main_sim}

To prepare for the application to brain connectomics, we design our simulation to mimic the sample size and the dimension of the brain connectome data.
Let there be $N = 1000$ samples of $68 \times 68$ observed adjacency matrix $\Xb_K^{(1:N)}$, generated from our model with the Bayesian network structure of $(K + 1) = 3$ layers and $p_0 = 4, p_1 = 16, p_2 = 68$ nodes.
In the data generating distribution, the continuous parameters of our model have true values specified as $\bnu = 2^{- p_0 (p_0 - 1) / 2} \one$, $C_k = -7$, and $\bGamma_k = 4 \one \one^\top + 6 \Ib$, for all $k \in [K]$.
The true value of each connection matrix $\Ab_k$ takes the form $\big( \Ib_{p_{k - 1}} ~ \Mb_k ~ \Bb_k \big)^\top$ for submatrices $\Mb_k \in \cM_{p_{k - 1}}$ and $\Bb_k \in \RR^{p_{k - 1} \times (p_k - 2 p_{k - 1})}$ (the exact value is given in Section \ref{ssec:more_sim}), with each row $\ab_{k, i}^\top$ of $\Ab_k$ containing one or two entries of ``1''.
In this way, we have $\Ab \in \cA_2 \cap \cA_1^c$ and $\bTheta \in \cT_0$, implying that $(\Ab, \bTheta)$ satisfies the condition for generic identifiability in Theorem \ref{theo:generic} but not the condition for strict identifiability in Theorem \ref{theo:strict}.
We illustrate through this design that generic identifiability suffices for practical purposes.
Additionally, the sparsity constraint under $S = 2$ holds for each $\Ab_k$.

We impose a prior distribution $\PP(\Ab, \bTheta)$ in the factorized form \eqref{eq:prior}, with each independent block of parameters following a weakly informative prior.
We let the prior of each $\ab_{k, i}$ be uniform over its space $\cV_k$ and the prior of $\bnu$ be uniform over the probability simplex $\cS_+^{\frac{p_0 (p_0 - 1)}{2} - 1}$.
For the continuous parameters $C$ and $\bGamma$, we follow \citet{gelman2006prior} and adopt the weakly informative prior proportional to the standard normal distribution $N(0, \Ib)$, such that each $C_k \sim N(0, 1)$ and each entry $\Gamma_{k, i, j} \sim N_+(0, 1)$.

We treat the number of nodes $p_k$ in each layer of the Bayesian network and the model parameters $(\Ab, \bTheta)$ as unknown.
To determine $p_k$, we conduct model selection based on the Watanabe-Akaike information criterion (WAIC) in \citet{watanabe2010asymptotic}.
Given fixed choices of $p_k$, we can then infer model parameters $(\Ab, \bTheta)$ using the posterior computation approaches in Section \ref{sec:post}.
Specifically, we consider all choices of $p_k$ satisfying $p_k \ge 2 p_{k - 1}$ for each $k \in [K]$, which is necessary for both the strict and generic notions of identifiability.
We further focus on $p_0 \in [2, 6]$ for better interpretability of the Bayesian network and also to prevent the parameter $\bnu$ from being too high-dimensional.
For every such choice of $p_k$, we use the spectral initialization approach in Section \ref{ssec:warm} to initialize connection matrices $\Ab$, then run the subsampling Gibbs sampler in Section \ref{ssec:sub_gibbs} with a 1\% subset ratio for 10,000 iterations, and follow by running the standard Gibbs sampler in Section \ref{ssec:sub_gibbs} for 100 iterations.
The 10,000 iterations of the subsampling Gibbs sampler is used for quick convergence from the spectral initialization, whereas the 100 iterations of the standard Gibbs sampler are used for posterior inference.
We use the Gelman-Rubin statistic \citep{gelman1992inference} and the Geweke statistic \citep{geweke1992evaluating} for convergence diagnosis of both the subsampling and the standard Gibbs samplers.
Detailed results reported in Section \ref{ssec:more_sim} provide evidence that the Markov chains of Gibbs samplers under each choice of $p_k$ have converged.

Given a Bayesian network structure specified by $p_k$, at the $t$th iteration of the standard Gibbs sampler, we have the posterior sample of parameters $(\Ab^{(t)}, \bTheta^{(t)})$ and latent variables $\Xb^{(1:N; t)}$.
Conditioned on the sampled parameters and latent variables, we denote the log likelihood of each observed adjacency matrix $\Xb_K^{(n)}$ as $\ell(t, n) := \log \PP(\Xb_K^{(n)} | \Ab^{(t)}, \bTheta^{(t)}, \Xb_{K - 1}^{(n; t)})$.
As recommended by \citet{gelman2014understanding} and \citet{merkle2019bayesian}, we use the following formula for computing WAIC
\begin{equation}\label{eq:waic}
\textsf{WAIC}
:=
-2 (\textsf{lppd} - p_{\textsf{WAIC}})
=
-2 \sum_{n = 1}^N \log \left(
\frac{1}{T} \sum_{t = 1}^T e^{\ell(t, n)}
\right)
+ 2 \sum_{n = 1}^N \var_{t = 1}^T\big(
\ell(t, n)
\big)
,
\end{equation}
where $\textsf{lppd}$ stands for the log pointwise predictive density and $p_{\textsf{WAIC}}$ is an approximation to the effective number of parameters, with $\var_{t = 1}^T(\ell(t, n))$ denoting the sample variance of $\{\ell(t, n)\}_{t = 1}^T$.
Among all choices of $p_k$, we find that WAIC is minimized at their true values $p_0 = 4, p_1 = 16, p_2 = 68$ and changes smoothly as we vary $p_k$.
This suggests that we are able to select the best model through the WAIC criterion.
A heatmap of WAIC with respect to $p_0$ and $p_1$ is provided in Section \ref{ssec:more_sim} for further illustration.

\begin{figure}[ht]
\centering
\includegraphics[width = \textwidth]{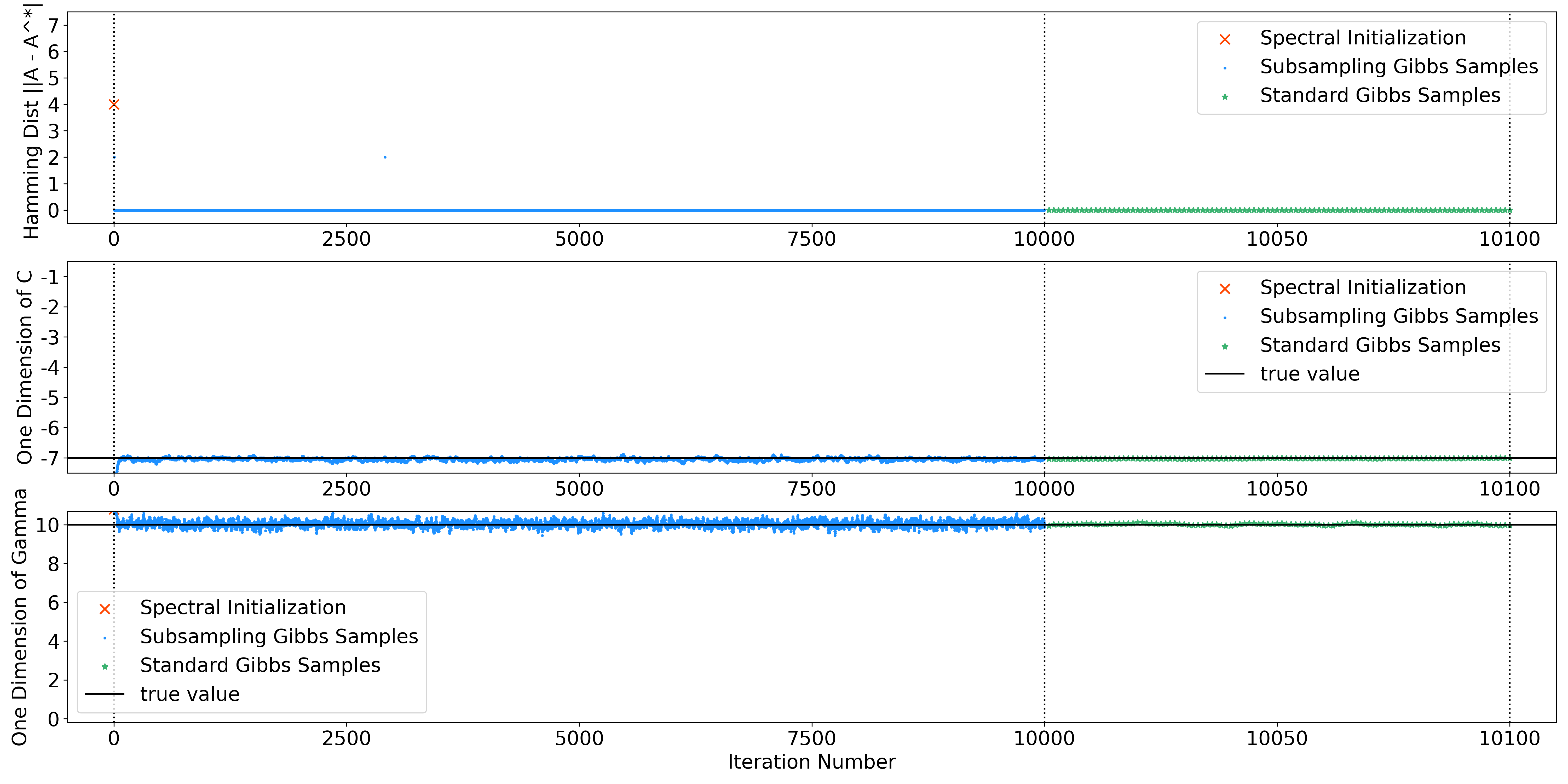}
\caption[Trace plots for Hamming distance, $C_2$, and $\Gamma_{2,1,1}$.]{
From top to bottom are the trace plot of the Hamming distance $\|\Ab - \Ab^*\|$ between the sample $\Ab$ and its true value $\Ab^*$, the trace plot of $C_2$, and the trace plot of entry $\Gamma_{2, 1, 1}$ in $\bGamma_2$.
We start at iteration 0 with spectral initialization and run the subsampling Gibbs sampler for 10,000 iterations, followed by 100 iterations of standard Gibbs sampler.
}
\label{fig:sim_trace}
\end{figure}

For the selected model with Bayesian network structure $p_0 = 4, p_1 = 16, p_2 = 68$, we examine the posterior distribution of model parameters $(\Ab, \bTheta)$ under both the subsampling Gibbs sampler and the standard Gibbs sampler.
Trace plots of the Hamming distance $\|\Ab - \Ab^*\|$ between sample $\Ab$ and its true value $\Ab^*$, the parameter $C_2$, and the entry $\Gamma_{2, 1, 1}$ of $\bGamma_2$ are provided in Figure \ref{fig:sim_trace}.
We observe that the spectral initialization gives a reasonably close estimate of the connection matrix $\Ab$, with only 4 out of 1152 entries being incorrect.
The initialization of continuous parameters are drawn randomly from their weakly informative prior distributions and hence distant from their true values.
In the 10,000 iterations of the subsampling Gibbs sampler, the parameters quickly converge towards their stationary distributions within the first 2,000 iterations and fluctuate for the next 8,000 iterations, exhibiting good mixing properties.
In the following 100 Gibbs sampler iterations, the parameters further converge to their new stationary distributions, which are the desired posterior distributions of parameters.
Despite the small subset ratio of $\frac{|\cB|}{N} = 1\%$ in the subsampling Gibbs sampler, it has stationary distributions very close to the standard Gibbs sampler.
Further comparison of stationary distributions in $W_1$ distance is provided in  Section \ref{ssec:more_sim}.
We also observe that the true values of parameters are recovered in the 100 iterations of the standard Gibbs sampler and uncertainty quantification can be conducted through Monte Carlo inference.

All the simulation results have justified (a) our posterior computation framework that combines spectral initialization, subsampling Gibbs sampler, and standard Gibbs sampler; and also (b) our model selection approach based on the WAIC criterion.
We note that based on our preliminary simulations, the presented results from one trial are typical and similar to observations from repeated experiments and different data generating distributions.
Therefore, we will be adopting the same computational framework in the following real data section when analyzing the brain connectome data, where no true values of model structure or model parameters are known.

\subsection{Additional Results and Discussions}\label{ssec:more_sim}

In this section, we include the deferred figures and tables of Section \ref{sec:sim} as well as some further discussions.

The true data generating mechanism in our simulation has $N = 1000, (K + 1) = 3$ and $p_0 = 4, p_1 = 16, p_2 = 68$.
The true values of continuous parameters are $\bnu = 2^{- p_0 (p_0 - 1) / 2} \one$, $C_k = -7$, and $\bGamma_k = 4 \one \one^\top + 6 \Ib$, for all $k \in [K]$.
The true value of each connection matrix $\Ab_k$ takes the form $\big( \Ib_{p_{k - 1}} ~ \Mb_k ~ \Bb_k \big)^\top$ for submatrices $\Mb_k \in \cM_{p_{k - 1}}$ and $\Bb_k \in \RR^{p_{k - 1} \times (p_k - 2 p_{k - 1})}$, with its exact value visualized in Figure \ref{fig:sim-true-A}.

\begin{figure}[ht]
\centering
\includegraphics[width = 0.7\textwidth]{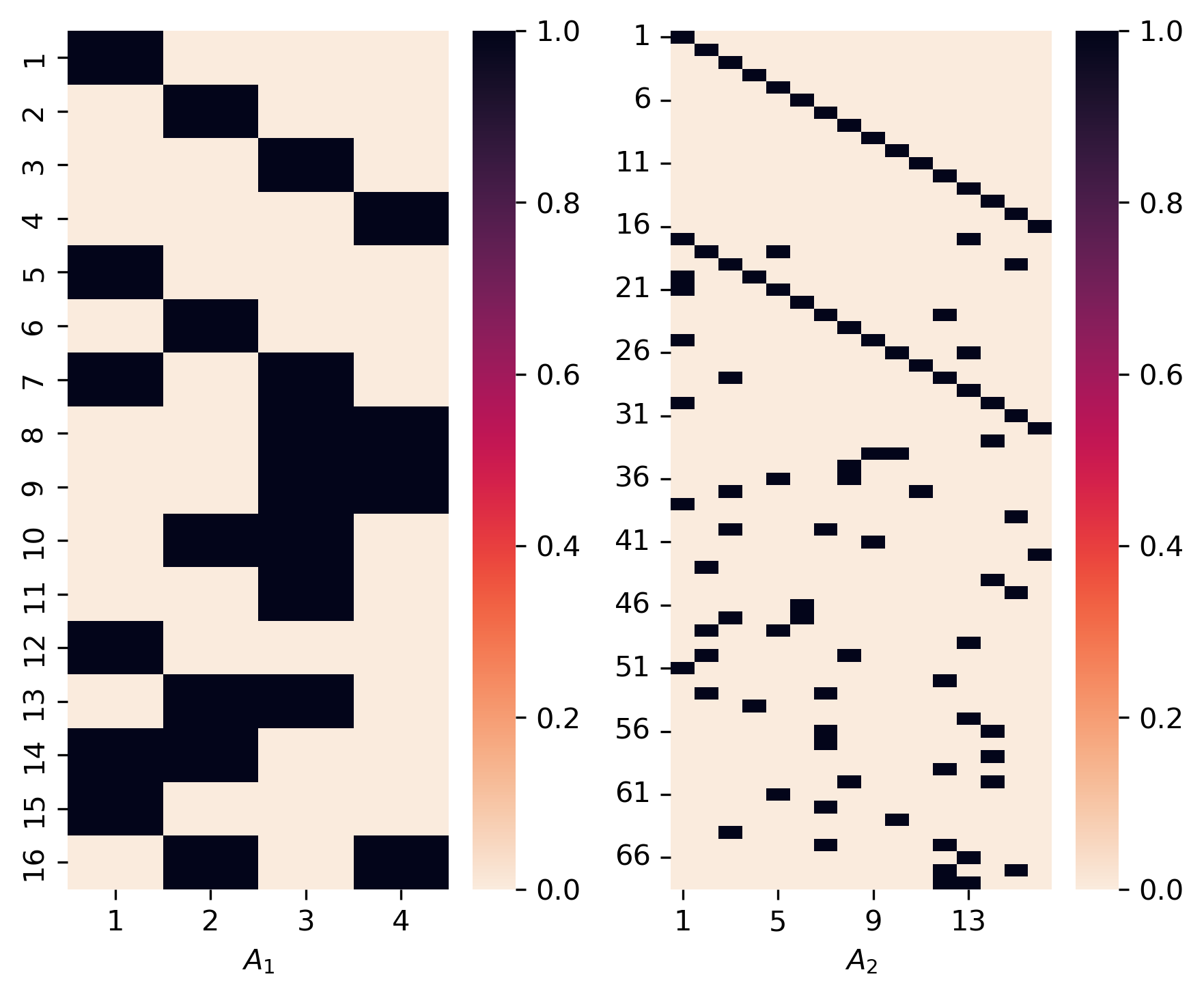}
\caption{True values of connection matrices $\Ab_1, \Ab_2$.}
\label{fig:sim-true-A}
\end{figure}

We have conducted model selection using WAIC as the information criterion.
Fixing $K = 2$, for each valid choice of $p_0$ and $p_1$, we start with spectral initialization, run 10,000 iterations of subsampling Gibbs sampler with 1\% subset ratio, and follow by 100 iterations of standard Gibbs sampler.
We compute the WAIC for each choice of $p_0, p_1$ and visualize in Figure \ref{fig:sim_waic}.
Note that for $p_1 \ge 23$ the computed WAICs are significantly larger and hence excluded from the plot.
We observe that the WAIC changes smoothly with $p_0, p_1$ and finds its mode at the true value $p_0 = 4, p_1 = 16$.

\begin{figure}[ht]
\centering
\includegraphics[width = \textwidth]{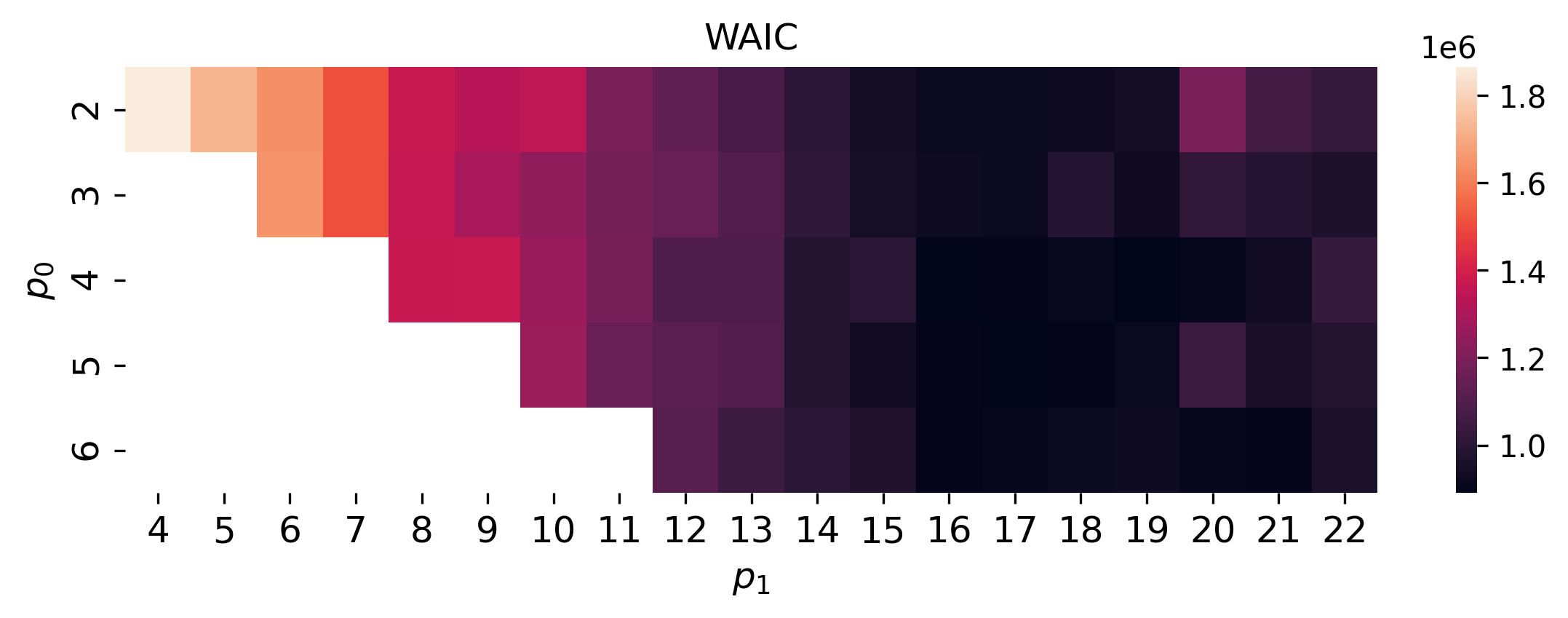}
\caption{WAIC for each choice of network structure with $K = 2, p_0 \in [2, 6], p_1 \in [2p_0, 22], p_2 = 68$.}
\label{fig:sim_waic}
\end{figure}

We diagnose the convergence and mixing of the Gibbs samplers using Gelman-Rubin statistic \citep{gelman1992inference} and Geweke statistic \citep{geweke1992evaluating}.
For every choice of $p_0, p_1$, we repeat the sampling process multiple times and obtain several Markov chains of samples from subsampling and standard Gibbs samplers for each dimension of the parameters $(\Ab, \bTheta)$.
We take averages of the computed Gelman-Rubin and Geweke statistics over all dimensions of $(\Ab, \bTheta)$ and all choices of $p_0, p_1$, as reported in Table \ref{tab:sim_cvg}.
All Gelman-Rubin statistics are below 1.1 and all Geweke statistics lie in the tails of the distribution $N(0, 1)$, which are good indicators for the convergence of our MCMC chains.

\begin{figure}
\centering
\begin{tabular}{c|cc}
\toprule
& Gelman-Rubin Statistic & Geweke Statistic \\
\midrule
Subsampling Gibbs Sampler & 1.0004 & 0.6975 \\
Standard Gibbs Sampler & 1.0473 & -0.2986 \\
\bottomrule
\end{tabular}
\renewcommand{\figurename}{Tab.}
\caption{
Convergence diagnosis for subsampling and standard Gibbs samplers.
}
\label{tab:sim_cvg}
\renewcommand{\figurename}{Fig.}
\end{figure}

For the connection matrices $\Ab_1$ and $\Ab_2$, we plot their true values, their spectral initializations, their posterior means under both Gibbs samplers in Figure \ref{fig:sim_A}.
We observe that the spectral initialization already gives an estimate that differs from the true value by only 7 out of 1152 entries.
The posterior means of $\Ab$ under both subsampling and standard Gibbs samplers are very close to the true value, with the posterior mode exactly being the true value.
This suggests that the MCMC chains have quickly converged to the true $\Ab$ within burn-in period and largely remain there afterwards for both subsampling Gibbs sampler and standard Gibbs sampler.
This agrees with the trace plot of the Hamming distance between samples of $\Ab$ and the true $\Ab$ shown in Figure \ref{fig:sim_trace} of the main paper.

\begin{figure}[ht]
\centering
\includegraphics[width = \textwidth]{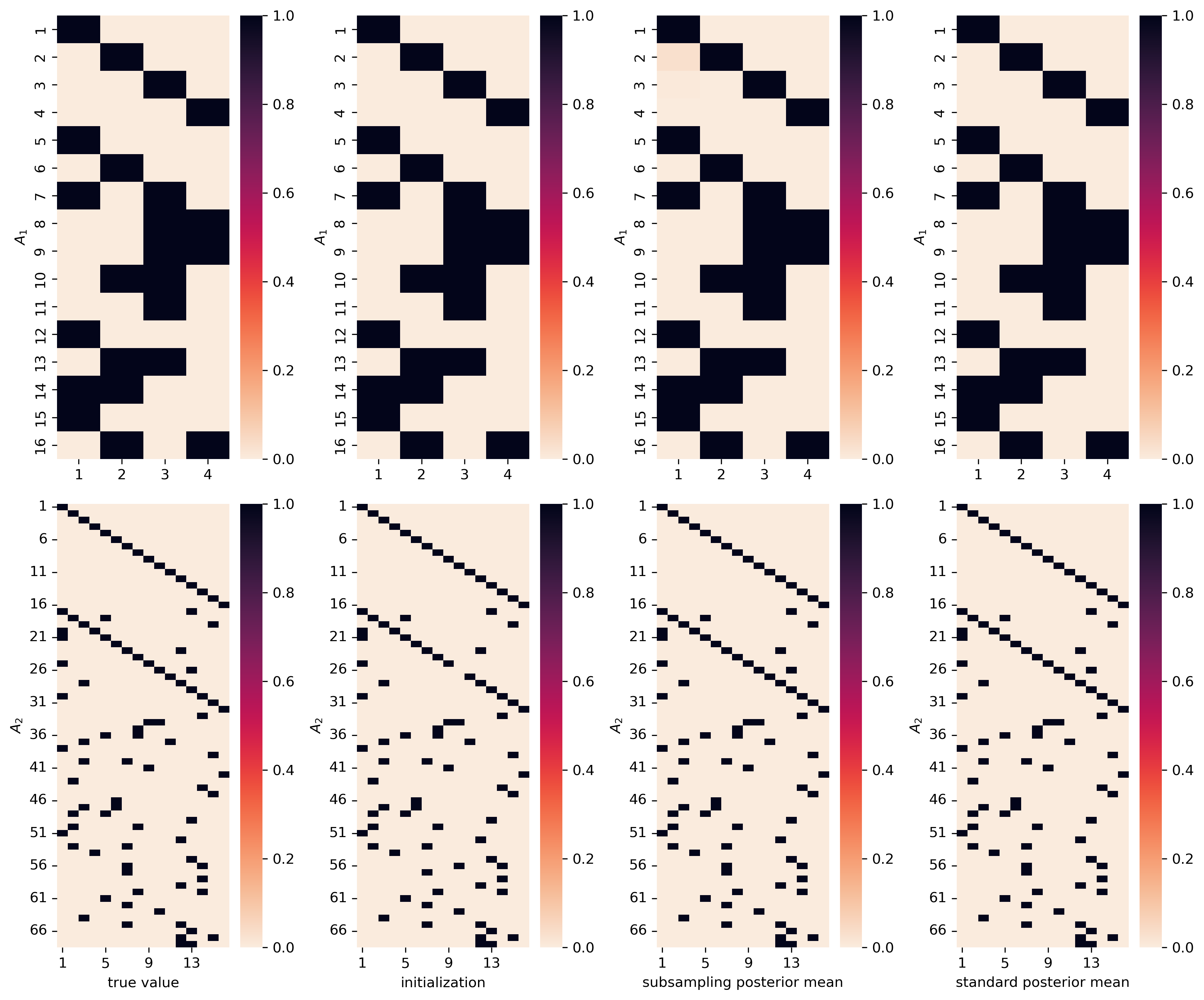}
\caption{True value, spectral initialization, and posterior means of the connection matrices $\Ab$.
}
\label{fig:sim_A}
\end{figure}

\begin{figure}[ht]
\centering
\includegraphics[width = \textwidth]{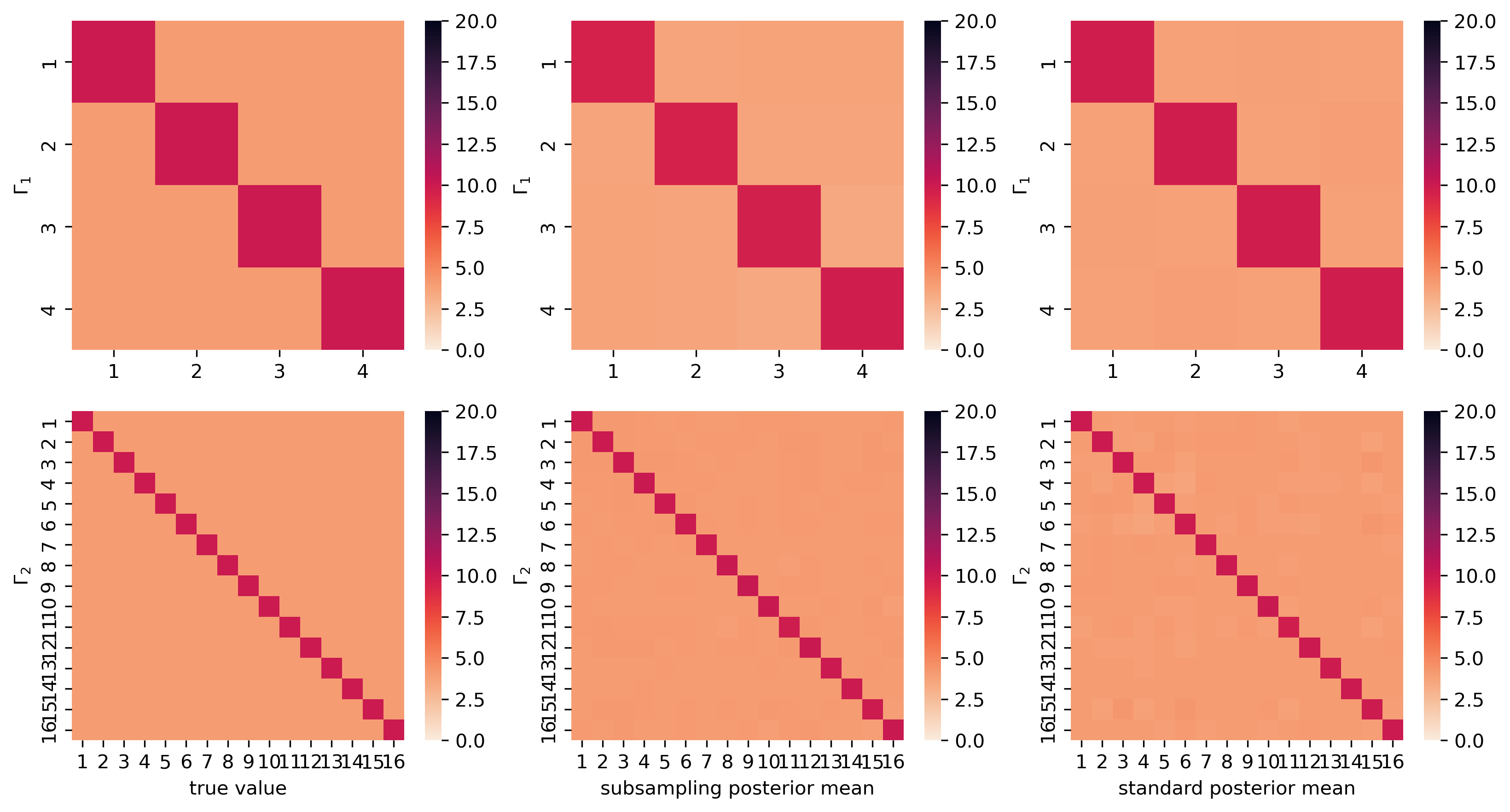}
\caption{True values and posterior means of the continuous parameters $\bGamma_1, \bGamma_2$.
}
\label{fig:sim_Gamma}
\end{figure}

\begin{figure}[!ht]
\centering
\includegraphics[width = \textwidth]{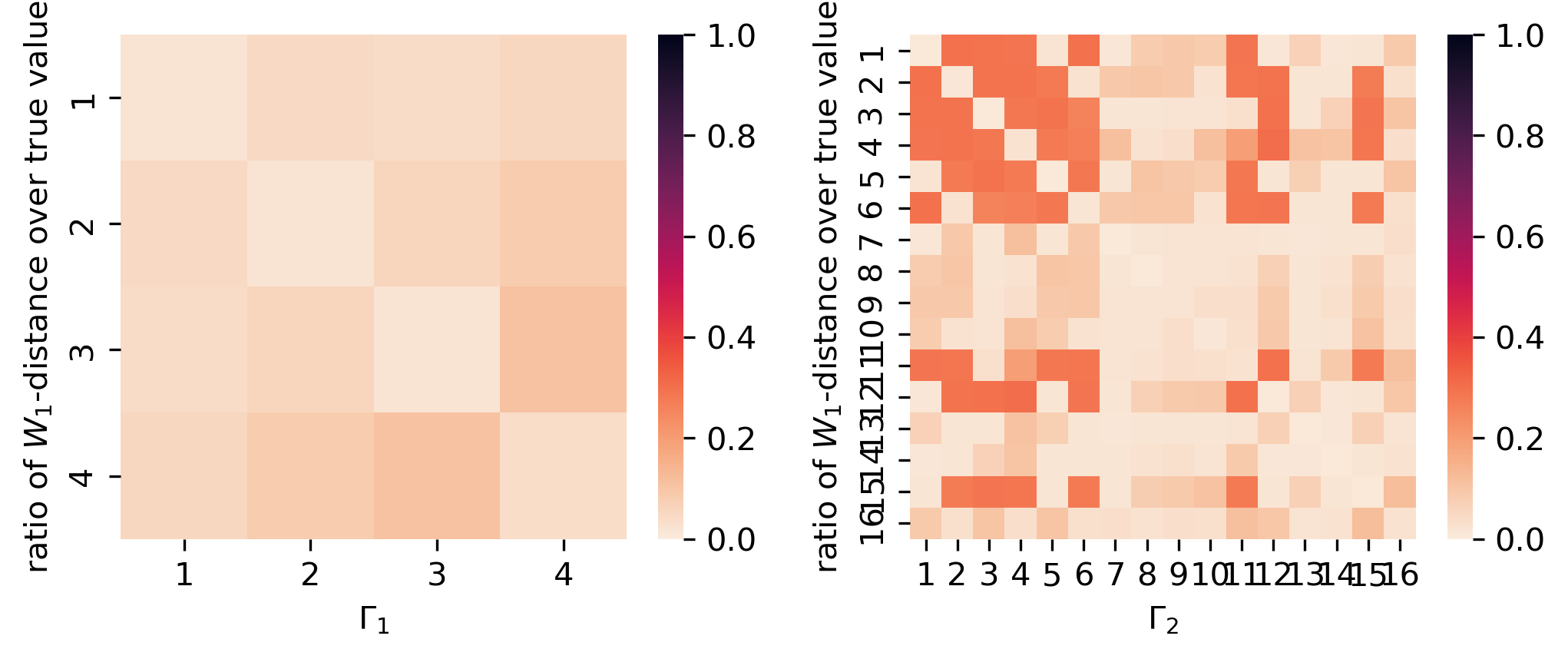}
\caption{Ratio of $W_1$ approximation error over true value for each entry of $\bGamma_1, \bGamma_2$.}
\label{fig:sim_Gamma_W1}
\end{figure}

For the continuous parameters $\bGamma_k$ at each layer $k$, we plot their true values and their posterior means under both Gibbs samplers in Figure \ref{fig:sim_Gamma}.
Both subsampling Gibbs sampler and standard Gibbs sampler produce samples that are close to the true values of $\bGamma_1, \bGamma_2$.
While the subsampling Gibbs sampler tends to underestimate some parameters in the first layer and overestimate some parameters in the second layer, the samples from the standard Gibbs sampler have demonstrated to estimate the true values of $\bGamma$ consistently.
The posterior estimates of the continuous parameters $C_1, C_2$ are also found to demonstrate similar behaviors.
This agrees with the trace plots of $C_k, \bGamma_k$ shown in Figure \ref{fig:sim_trace} of the main paper.

As discussed in Section \ref{ssec:sub_gibbs} of the main paper, the stationary distribution of standard Gibbs sampler is the true posterior distribution, whereas the stationary distribution of subsampling Gibbs sampler only approximates the true posterior distribution.
We now numerically compare the error of this approximation in $W_1$-distance and justify the use of our subsampling Gibbs sampler.
For any two probability measures $\mu, \nu$ on a Polish space $\cZ$, the Wasserstein $1$-distance between them is defined as $W_1(\mu, \nu) := \inf_{\pi \in \Gamma(\mu, \nu)} \int_{\cZ \times \cZ} \|x - y\|_1 \ud \pi(x, y)$, where $\Gamma(\mu, \nu)$ denotes the collection of all couplings of $\mu, \nu$.
For each entry of $\bGamma_1, \bGamma_2$, we compute the $W_1$-distance between the empirical distributions of posterior samples under subsampling and standard Gibbs samplers.
We present the ratio of $W_1$-distance over the true value for each entry in Figure \ref{fig:sim_Gamma_W1}.
We find for each entry of $\bGamma_1, \bGamma_2$ that the $W_1$ approximation error of the stationary distribution for subsampling Gibbs sampler is relatively small compared to the scale of the true value.
This is also consistent with our observations from Figure \ref{fig:sim_trace} of the main paper.

\subsection{Sensitivity Analysis}\label{ssec:sens}

\begin{table}[ht]
\centering
\begin{tabular}{c|ccc|ccc}
\toprule
& $|M - M'|$ & $|SD - SD'|$ & $W_1(P, P')$ & $\frac{|M - M'|}{SD}$ & $\frac{|SD - SD'|}{SD}$ & $\frac{W_1(P, P')}{SD}$ \\
\midrule
$\mu_C \leftarrow \mu_C + 1$ & 0.0029 & 0.0020 & 0.0034 & 0.3156 & 0.2244 & 0.3733 \\
$\mu_C \leftarrow \mu_C - 1$ & 0.0029 & 0.0021 & 0.0035 & 0.3234 & 0.2321 & 0.3843 \\
$\sigma_C^2 \leftarrow 2 \sigma_C^2$ & 0.0038 & 0.0021 & 0.0043 & 0.4125 & 0.2261 & 0.4721 \\
$\sigma_C^2 \leftarrow \frac12 \sigma_C^2$ & 0.0039 & 0.0020 & 0.0042 & 0.4259 & 0.2245 & 0.4645 \\
$\mu_\gamma \leftarrow \mu_\gamma + 1$ & 0.0035 & 0.0021 & 0.0041 & 0.3870 & 0.2311 & 0.4476 \\
$\mu_\gamma \leftarrow \mu_\gamma - 1$ & 0.0032 & 0.0021 & 0.0039 & 0.3534 & 0.2293 & 0.4243 \\
$\sigma_\gamma^2 \leftarrow 2 \sigma_\gamma^2$ & 0.0036 & 0.0023 & 0.0042 & 0.3946 & 0.2532 & 0.4598 \\
$\sigma_\gamma^2 \leftarrow \frac12 \sigma_\gamma^2$ & 0.0031 & 0.0020 & 0.0035 & 0.3368 & 0.2203 & 0.3881 \\
$\mu_\delta \leftarrow \mu_\delta + 1$ & 0.0034 & 0.0022 & 0.0040 & 0.3753 & 0.2435 & 0.4424 \\
$\mu_\delta \leftarrow \mu_\delta - 1$ & 0.0033 & 0.0023 & 0.0039 & 0.3659 & 0.2486 & 0.4261 \\
$\sigma_\delta^2 \leftarrow 2 \sigma_\delta^2$ & 0.0031 & 0.0021 & 0.0035 & 0.3382 & 0.2304 & 0.3862 \\
$\sigma_\delta^2 \leftarrow \frac12 \sigma_\delta^2$ & 0.0032 & 0.0021 & 0.0038 & 0.3486 & 0.2279 & 0.4189 \\
$\alpha \leftarrow 2 \alpha$ & 0.0034 & 0.0022 & 0.0040 & 0.3787 & 0.2386 & 0.4343 \\
$\alpha \leftarrow \frac12 \alpha$ & 0.0036 & 0.0023 & 0.0042 & 0.3906 & 0.2513 & 0.4580 \\
\bottomrule
\end{tabular}
\caption{Sensitivity analysis.}
\label{tab:sens}
\end{table}

To evaluate the robustness of our proposed method with respect to hyperparameter specification, we conduct a systematic sensitivity analysis by perturbing each prior hyperparameter around its baseline value.
The hyperparameters include the prior mean and variance $\mu_C, \sigma_C^2$ of the parameter $C$, which models the baseline logit adjacency probability between two nodes that share neither a community nor adjacent communities; the prior mean and variance $\mu_\gamma, \sigma_\gamma^2$ of the parameter $\gamma$, which governs the increase in logit adjacency probability when two nodes share a community; the prior mean and variance $\mu_\delta, \sigma_\delta^2$ of the parameter $\delta$, which captures the additional increase when two nodes share a pair of adjacent communities; and the Dirichlet prior concentration $\alpha$ for the parameter $\bnu$, which controls the prior distribution of the top-layer adjacency matrices $\Xb_0$.

For each hyperparameter $\mu_C$, $\sigma_C^2$, $\mu_\gamma$, $\sigma_\gamma^2$, $\mu_\delta$, $\sigma_\delta^2$, and $\alpha$, we consider multiplicative or additive perturbations of moderate size, i.e. $\pm 1$ shifts for location parameters and $\times 2$ or $\times \frac12$ scalings for variance and concentration parameters.
Under each perturbed configuration, we refit the model using the same data via the standard Gibbs sampler and assess the influence on posterior inference using several summary metrics: the posterior mean $M$, posterior standard deviation $SD$, and the posterior distribution $P$ itself. 
Deviations from the baseline posterior are quantified through absolute $L_1$ differences $|M - M'|$ and $|SD - SD'|$, as well as the $1$-Wasserstein distance $W_1(P,P')$ between the baseline and perturbed posterior distributions.

Table~\ref{tab:sens} reports both absolute deviations and scale-free deviations obtained by normalizing with respect to the baseline posterior standard deviation.
Across all perturbations examined, the absolute changes in posterior mean and standard deviation remain on the order of $10^{-3}$, with corresponding Wasserstein distances also below $5\times 10^{-3}$.
The normalized deviations exhibit the same pattern: shifts in the posterior mean are typically between $31\%$ and $43\%$ standard deviations, while changes in posterior spread fall between $22\%$ and $25\%$ standard deviations. 
The normalized Wasserstein distances range between $37\%$ and $47\%$, indicating modest but controlled movement of the entire posterior distribution.

Overall, the magnitude and stability of these deviations demonstrate that the proposed model is robust to reasonable misspecification of prior hyperparameters.
The posterior summaries remain stable under a range of plausible perturbations, supporting the practical reliability of the method in applied settings.

\subsection{\texorpdfstring{Large-$p_K$, small-$N$ regime}{Large-pK, small-N regime}}\label{ssec:large_p_small_N}

\begin{figure}[!ht]
\centering
\includegraphics[width = 0.6\textwidth]{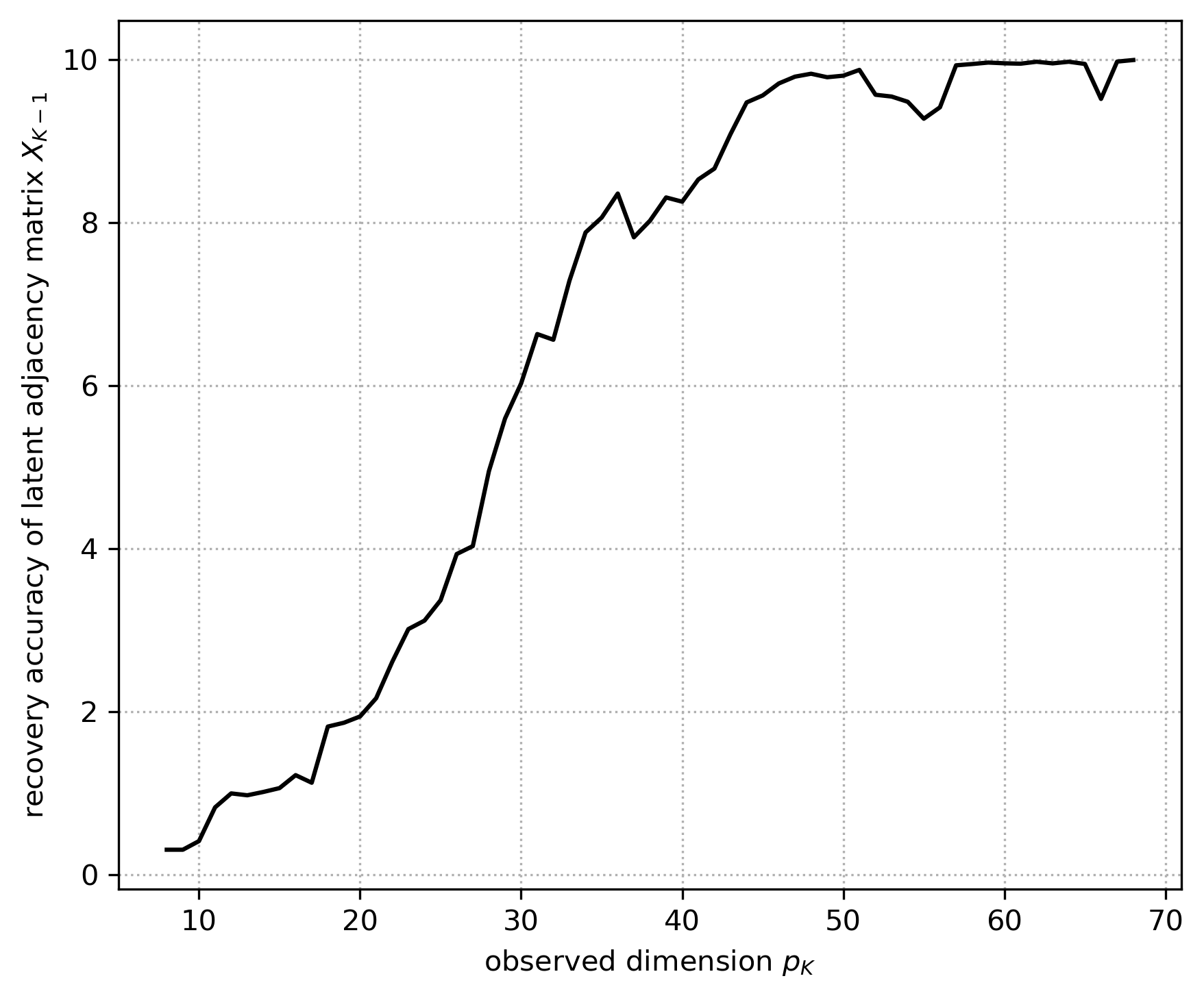}
\caption{
Recovery accuracy of latent adjacency matrices $Xb_{k - 1}^{(n)}$ as observed dimension $p_K$ increases between $[8, 68]$ for small sample size $N = 10$.
The $y$-axis measures the average number of exactly recovered latent adjacency matrices over all Gibbs posterior samples, with a maximum value of $10$ corresponding to exact recovery of all $\Xb_{k - 1}^{(n)}$ for $n \in [N]$ by all posterior samples.
}
\label{fig:large_p_small_N}
\end{figure}

Most of our theoretical results and simulation studies focus on the asymptotic regime where the sample size $N$ (i.e. number of observed adjacency matrices) grows while the observed number of nodes $p_K$ (i.e. dimension of observed adjacency matrices) remains fixed.
This is the regime under which posterior consistency of the model parameters $(\Ab, \bTheta)$ can be established, and it is the primary focus of our contribution.
In contrast, a different line of recent research work focuses on the regime in which $p_K$ is large while the sample size $N$ is small, commonly focusing on a single high-dimensional network with $N = 1$.
This high-dimensional setting is typical in community detection and multiplex network analysis, with examples including \citet{lei2023bias, jing2021community, ma2023community, lei2020consistent, paul2020spectral, chen2022global, macdonald2022latent, lei2024computational}.

In this alternative asymptotic framework with $p_K$ growing and $N$ fixed, theoretical guarantees for parameter consistency become difficult to obtain: the size of the connection matrix $\Ab_K$ grows proportionally with $p_K$, making the parameter space increasingly high-dimensional.
Accordingly, it is natural in this regime to focus not on parameter estimation but on the recovery of latent variables.
In our model, this corresponds to the exact recovery of the latent adjacency matrices $\Xb_{k - 1}^{(n)}$, which play a role analogous to community labels in classical network clustering.

To assess the behavior of our method in the large-$p_K$, small-$N$ setting, we consider a simple two-layer Bayesian network ($K = 1$) with $p_1 = 68$ observed nodes and $p_0 = 4$ latent nodes.
The model parameters match those used in earlier simulations:
$\bnu = 2^{- p_0 (p_0 - 1) / 2} \one$, $C_1 = -7$, and $\bGamma_1 = 4 \one \one^\top + 6 \Ib$.
We fix a small sample size $N = 10$ and vary the observed-layer dimension $p_1$ over the range $[8, 68]$ by extracting the upper-left $p_1 \times p_1$ submatrix of each observed adjacency matrix $\Xb_1^{(n)}$.

For each value of $p_1$, we run the Gibbs sampler for 2000 iterations starting from the spectral initialization and discard the first 1000 as burn-in.
Because of the simplicity of the model in this setting, 1000 iterations are sufficient for the sampler to reach stationarity.
We then post-process the latent adjacency matrices to resolve label-switching among the $p_0 = 4$ latent nodes.

Let $\Xb_0^{(n, *)}$ denote the true latent adjacency matrix and $\Xb_0^{(n; t)}$ denote the posterior sample at iteration $t$.
The recovery accuracy is defined as
$$
\sum_{n = 1}^N \frac{1}{1000} \sum_{t = 1001}^{2000} 1_{\Xb_0^{(n; t)} = \Xb_0^{(n, *)}}
,
$$
representing the average number of exactly recovered latent adjacency matrices over all Gibbs posterior samples.
This metric evaluates whether the entire latent adjacency matrix is correctly recovered, not just individual entries.
Viewing each distinct latent adjacency matrix value as a distinct cluster in community detection, this recovery accuracy focuses on the correct inference of the cluster of each observed network sample.

Figure~\ref{fig:large_p_small_N} shows how recovery accuracy varies with $p_1$.
As $p_1$ increases, recovery accuracy improves steadily and approaches close to exact recovery for $p_1 \ge 50$.
This indicates that even with a small number of network slices $N$, the latent adjacency matrices can be accurately inferred when the observed networks are sufficiently high-dimensional.

This simulation demonstrates that the method performs well not only in the classical large-$N$ regime analyzed theoretically, but also in the high-dimensional large-$p_K$ regime commonly encountered in empirical network applications.
The two regimes emphasize different inferential targets, parameter consistency versus latent-variable recovery, yet both represent meaningful and practically relevant behaviors of the model.

\subsection{Comparison with Hierarchical Community Detection}\label{ssec:compare_hcd}

\begin{table}[!htbp]
\centering
\begin{tabular}{c|cc|c|ccc}
\toprule
\shortstack{sample\\size $N$} & \shortstack{tree\\layer} & \shortstack{number of\\communities} & method & \shortstack{normalized mutual\\information} & \shortstack{label matching\\accuracy} & \shortstack{tree\\distance} \\
\midrule
\multirow{6}{*}{10} & \multirow{3}{*}{1} & \multirow{3}{*}{3} & BDGM      & 0.0000 & 0.3333 & 0.3333 \\
& & & HCD-Sign  & 0.7883 & 0.9111 & 0.1007 \\
& & & HCD-Spec  & 0.9248 & 0.9741 & 0.0747 \\
\cline{2-7}
& \multirow{3}{*}{2} & \multirow{3}{*}{9} & BDGM      & 0.8795 & 0.8640 & 0.2788 \\
& & & HCD-Sign  & 0.6686 & 0.5593 & 0.1397 \\
& & & HCD-Spec  & 0.7021 & 0.5704 & 0.1028 \\
\hline
\multirow{6}{*}{20} & \multirow{3}{*}{1} & \multirow{3}{*}{3} & BDGM      & 0.5916 & 0.7284 & 0.1078 \\
& & & HCD-Sign  & 0.8130 & 0.9185 & 0.0961 \\
& & & HCD-Spec  & 0.8950 & 0.9611 & 0.0806 \\
\cline{2-7}
& \multirow{3}{*}{2} & \multirow{3}{*}{9} & BDGM      & 0.9188 & 0.8930 & 0.1024 \\
& & & HCD-Sign  & 0.6795 & 0.5667 & 0.1337 \\
& & & HCD-Spec  & 0.6919 & 0.5556 & 0.1113 \\
\hline
\multirow{6}{*}{30} & \multirow{3}{*}{1} & \multirow{3}{*}{3} & BDGM      & 1.0000 & 1.0000 & 0.0000 \\
& & & HCD-Sign  & 0.8046 & 0.9173 & 0.0977 \\
& & & HCD-Spec  & 0.9075 & 0.9642 & 0.0785 \\
\cline{2-7}
& \multirow{3}{*}{2} & \multirow{3}{*}{9} & BDGM      & 0.9294 & 0.9259 & 0.0107 \\
& & & HCD-Sign  & 0.6798 & 0.5654 & 0.1368 \\
& & & HCD-Spec  & 0.6924 & 0.5605 & 0.1094 \\
\hline
\multirow{6}{*}{40} & \multirow{3}{*}{1} & \multirow{3}{*}{3} & BDGM      & 1.0000 & 1.0000 & 0.0000 \\
& & & HCD-Sign  & 0.8184 & 0.9241 & 0.0949 \\
& & & HCD-Spec  & 0.9091 & 0.9657 & 0.0779 \\
\cline{2-7}
& \multirow{3}{*}{2} & \multirow{3}{*}{9} & BDGM      & 0.9295 & 0.9259 & 0.0107 \\
& & & HCD-Sign  & 0.6841 & 0.5713 & 0.1333 \\
& & & HCD-Spec  & 0.6992 & 0.5713 & 0.1086 \\
\hline
\multirow{6}{*}{50} & \multirow{3}{*}{1} & \multirow{3}{*}{3} & BDGM      & 1.0000 & 1.0000 & 0.0000 \\
& & & HCD-Sign  & 0.8232 & 0.9259 & 0.0937 \\
& & & HCD-Spec  & 0.9101 & 0.9659 & 0.0779 \\
\cline{2-7}
& \multirow{3}{*}{2} & \multirow{3}{*}{9} & BDGM      & 1.0000 & 1.0000 & 0.0000 \\
& & & HCD-Sign  & 0.6916 & 0.5800 & 0.1306 \\
& & & HCD-Spec  & 0.7076 & 0.5852 & 0.1073 \\
\bottomrule
\end{tabular}
\caption{Comparison with Hierarchical Community Detection.}
\label{tab:compare_hcd}
\end{table}

Through the multilayer Bayesian network specified by the connection matrices $\{\Ab_k\}$, our model naturally induces a hierarchical community structure and can therefore be applied to hierarchical community detection tasks.
In this section, we investigate the clustering accuracy of our model in a controlled community-detection experiment and compare it to the state-of-the-art hierarchical community detection procedure proposed by \citet{li2022hierarchical}.
Their method performs hierarchical community detection through a top-down recursive spectral partitioning procedure that begins with all nodes in a single cluster and repeatedly splits each cluster into two communities, producing a full binary community tree.
Their paper introduces two variants based on their different splitting criteria: HCD-Sign (Algorithm 1 in \citet{li2022hierarchical}), which splits by eigenvalue sign, and HCD-Spec (Algorithm 2 in \citet{li2022hierarchical}), which applies regularized spectral clustering.
These are the two benchmark methods used in our comparison.

Before presenting the comparison, we highlight a fundamental modeling distinction.
The HCD methods are designed specifically for detecting hierarchical community structure in a single high-dimensional adjacency matrix.
In contrast, our model is a Bayesian generative model for multiplex networks, leveraging information across a population of networks.
Consequently, the HCD methods need to be applied separately to each individual adjacency matrix in the sample, while our model learns a shared population-level hierarchical structure by jointly analyzing all adjacency matrices.
This difference has practical consequences: as shown in our numerical study, HCD methods behave favorably when the sample size $N$ is very small, but as $N$ grows, shared structure can be learned more reliably and our model achieves substantially higher accuracy.

For our simulation experiment, we let the true data generating distribution be a hierarchical stochastic block model on 27 nodes with a fixed community tree of depth 3 and structure $3 \to 9 \to 27$.
This represents a fully-connected ternary tree with its root node removed: the 27 nodes at the finest level are partitioned into 9 leaf communities of size 3, and these leaf communities are further grouped into 3 coarser communities of size 9.
Let $\mathcal{C}^{(2)}_1, \dots, \mathcal{C}^{(2)}_9$ denote the 9 leaf communities of size $3$, and  let $\mathcal{C}^{(1)}_1, \dots, \mathcal{C}^{(1)}_3$ denote the coarser communities of size $9$, each formed by merging three leaf communities.
For two distinct nodes $i \neq j$, define their shared tree depth
$$
d(i,j) :=
\begin{cases}
2, & \text{if } i,j \text{ belong to the same leaf community } \mathcal{C}^{(2)}_r,\\[4pt]
1, & \text{if } i,j \text{ belong to the same coarse community } \mathcal{C}^{(1)}_\ell 
      \text{ but different leaf communities},\\[4pt]
0, & \text{if } i,j \text{ belong to different coarse communities}.
\end{cases}
$$
Conditional on the tree structure, we first generate a symmetric matrix of edge probabilities 
$\{p_{ij}\}_{1 \le i < j \le 27}$ according to
$$
p_{ij} \,\big|\, d(i,j) \sim
\begin{cases}
\mathrm{Uniform}(0.7,\,0.8), & d(i,j) = 2,\\
\mathrm{Uniform}(0.4,\,0.5), & d(i,j) = 1,\\
\mathrm{Uniform}(0,\,0.1),   & d(i,j) = 0,
\end{cases}
\qquad p_{ji} = p_{ij}, \quad p_{ii} \text{ fixed.}
$$
Given these edge probabilities, we sample $N$ independent adjacency matrices $\{\Xb^{(n)}\}_{n=1}^N$ with each $\Xb^{(n)} \in \{0,1\}^{27 \times 27}$, where for each $n$ and each unordered pair $i < j$, the entries follow
$$
X^{(n)}_{ij} \mid p_{ij} \sim \mathrm{Bernoulli}(p_{ij}), 
\qquad X^{(n)}_{ji} = X^{(n)}_{ij},
$$
independently across $n$ and all unordered pairs $(i,j)$.
Diagonal entries $X^{(n)}_{ii}$ are excluded from modeling.
This construction defines a hierarchical stochastic block model in which edges between nodes that share a deeper common ancestor in the community tree (same leaf community) are highly likely (probabilities in $[0.7, 0.8]$), edges between nodes that only share a coarser ancestor are moderately likely (probabilities in $[0.4, 0.5]$), and edges between nodes in different top-level branches are sparse (probabilities in $[0, 0.1]$).
This hierarchical stochastic block model provides a natural and challenging testbed for evaluating hierarchical community detection performance.

While different hierarchical community detection methods and our model each have their own procedures for model selection, for simplicity and to ensure a meaningful comparison, we assume that the true numbers of communities, specifically the 3 coarse communities and 9 leaf communities in the data generating hierarchical stochastic block model, are known to all methods so that comparisons can be made over detected communities of equal sizes.
We vary the sample size $N$, i.e. the number of $27 \times 27$ adjacency matrices generated from the hierarchical stochastic block model, over the range of $\{10, 20, 30, 40, 50\}$.
To measure the clustering accuracy of the detected community structure relative to the true structure, we follow \citet{li2022hierarchical} and use the three metrics: normalized mutual information, label matching accuracy, and tree distance (see \citet{li2022hierarchical} Section 4 for their definitions).
The hierarchical community detection methods HCD-Sign and HCD-Spec are fitted separately to each adjacency matrix, which produces a potentially different detected community structure for each sample.
We compare each detected structure with the true community structure and average the resulting metrics over all samples.
For our model, which is designed for multiplex networks and therefore provides a single population level hierarchical community structure, we infer the posterior distribution by running the standard Gibbs sampler as before and compute the clustering metrics for each posterior sample, averaging the results over all posterior samples.

The metrics measuring the clustering accuracy of our model (BDGM), HCD-Sign, and HCD-Spec are reported in Table \ref{tab:compare_hcd} for both the upper layer of the data generating hierarchical stochastic block model with 3 coarse communities and the lower layer with 9 leaf communities, with sample size $N$ varying from 10 to 50.
Several key observations can be made.
Across all methods, community layers, and sample sizes, the clustering accuracy measured by normalized mutual information, label matching accuracy, and tree distance is highly consistent, so that larger normalized mutual information is accompanied by larger label matching accuracy and smaller tree distance.
For both BDGM and the hierarchical community detection methods HCD-Sign and HCD-Spec, detecting the 3 coarse communities in the upper layer tends to be easier and more accurate than detecting the 9 leaf communities in the lower layer.
As expected, since hierarchical community detection methods are designed for single adjacency matrices and can only be applied to each sample separately, their clustering accuracy remains roughly the same as the sample size $N$ increases.
In contrast, our model benefits from pooling information across the population of networks, and its clustering accuracy improves as $N$ becomes larger.
Within the two hierarchical community detection methods, HCD Spec tends to perform slightly better than HCD Sign, while our model underperforms both HCD methods for $N \in \{10, 20\}$, surpasses them for $N \in \{30, 40\}$, and achieves exact recovery of the true hierarchical community structure for $N = 50$.

\section{More Data Analyses Results}\label{sec:more_app}

In this section, we provide more data analyses results complementary to Section \ref{sec:appl}.
We provide the details of our model structure selection in Section \ref{ssec:model_sel}, deferred figures of Section \ref{ssec:appl_part_clus} in Section \ref{ssec:appl_deferred_part_clus}, and deferred figures of Section \ref{ssec:appl_cluster} in Section \ref{ssec:appl_deferred_cluster}.
As promised in Section \ref{sec:appl}, we provide two additional data analyses of our model in the supplementary material, Sections \ref{ssec:appl_relate} and \ref{ssec:mask}, where we relate the inferred latent features to cognitive traits and evaluate performance on missing-edge prediction, with comparisons to baseline methods such as DCMM \citep{jin2023mixed}.

\subsection{Model Structure Selection}\label{ssec:model_sel}

As discussed in Section \ref{sec:post}, the sparsity hyperparameter $S$ controls the number of deeper-layer communities that each node can belong to, and we recommend $S \in \{1,2\}$ for interpretability and computational/statistical efficiency.
We carry out our analyses for both $S = 1$ and $S = 2$.
For $S = 2$, we obtain a multilayer overlapping clustering of brain regions, where each region belongs to up to two clusters, each of which further belongs to up to two deeper clusters.
For $S = 1$, we obtain a multilayer hierarchical partition of brain regions, where each region belongs to one cluster, each of which further belongs to one deeper cluster.

We focus on three-layer Bayesian network structures ($K = 2$).
For model selection, we vary the number of nodes in each layer within $p_0 \in [2, 4], p_1 \in [2p_0, 23], p_2 = 68$, excluding larger $p_0, p_1$ because they imply less interpretable model structures.
Prior specification and posterior computation are the same as our simulation study in Section \ref{sec:sim}.
For each choice of $p_k$, we initialize the connection matrices $\Ab$ using the Mixed-SCORE algorithm, run 10,000 iterations of the subsampling Gibbs sampler with 1\% subset ratio, then run the standard Gibbs sampler for 100 iterations.
Based on the last 100 iterations, we compute WAIC as defined in \eqref{eq:waic} and use it to choose the optimal $p_k$ in each layer.

\begin{figure}[ht]
\centering
\subfigure[
Sparsity hyperparameter $S = 1$.
]{
\includegraphics[width = \textwidth]{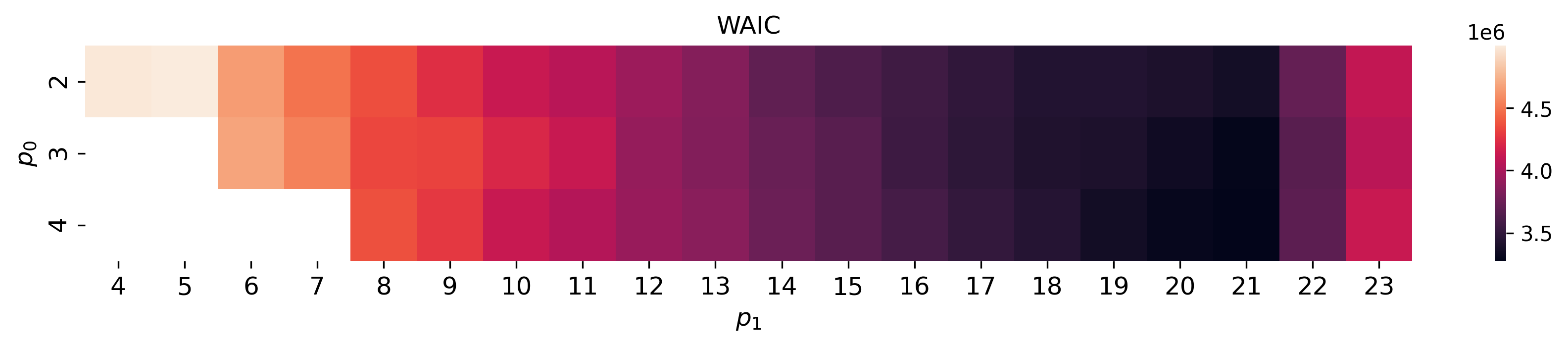}
\label{fig:app_waic_s1}
}
\hfill
\subfigure[
Sparsity hyperparameter $S = 2$.
]{
\includegraphics[width = \textwidth]{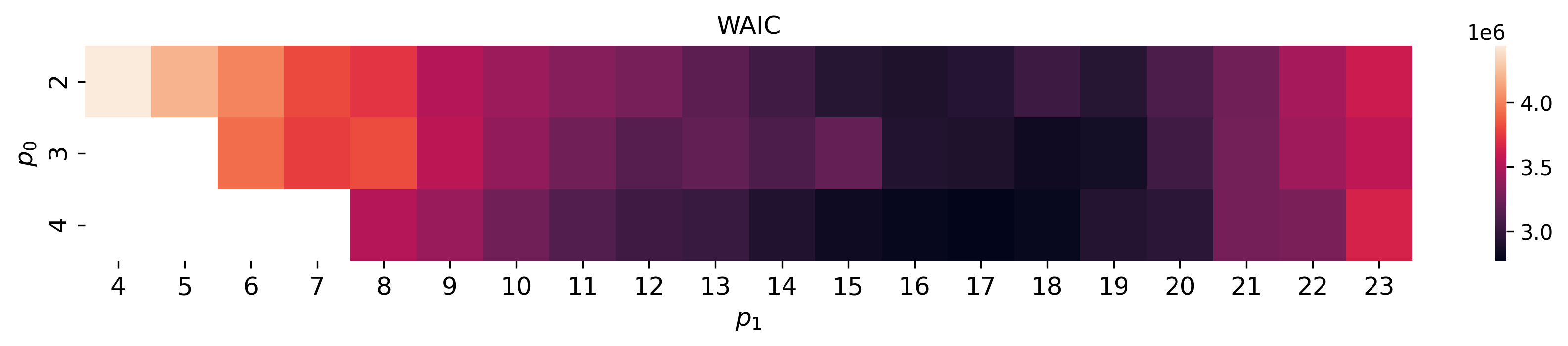}
\label{fig:app_waic_s2}
}
\caption{WAIC over the $p_0-p_1$ plane for sparsity hyperparameter $S \in \{1, 2\}$.}
\label{fig:app_waic}
\end{figure}

We visualize the computed WAIC for each choice of $p_0, p_1$ under sparsity hyperparameter $S = 1, 2$ in Figure \ref{fig:app_waic}.
We find that the landscape of WAIC over the $p_0-p_1$ plane is approximately unimodal and smooth under each sparsity hyperparameter $S$.
For $S = 1$, WAIC is minimized at $p_0 = 4, p_1 = 21$.
For $S = 2$, WAIC is minimized at $p_0 = 4, p_1 = 18$.
As expected, the optimal choice of $p_k$ is smaller under sparsity parameter $S = 2$ compared to $S = 1$, indicating that allowing overlapping memberships under $S = 2$ relieves the need of higher latent dimensions compared to $S = 1$. 
The WAIC for the optimal choice of $p_k$ is also much lower under $S = 2$ compared to $S = 1$.

\subsection{Multi-resolution Partitioning and Clustering of Brain Regions}\label{ssec:appl_deferred_part_clus}

\begin{figure}[ht]
\centering
\subfigure[
Visualization of $p_0 = 4$ higher-level clusters of ROIs on 
cerebral cortex.
]{
\includegraphics[width = 0.5\textwidth]{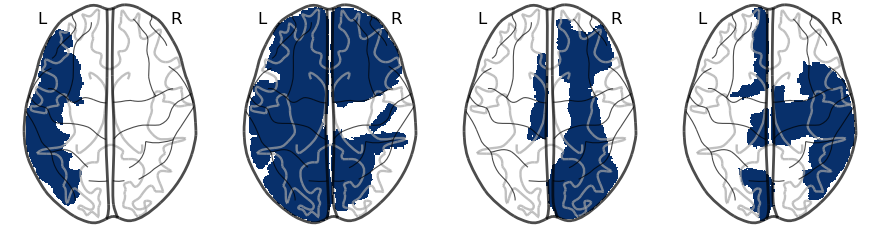}
\label{fig:app_brain_S1_A1}
}
\hfill
\subfigure[
Visualization of $p_1 = 21$ lower-level clusters of ROIs on cerebral cortex.
]{
\includegraphics[width = \textwidth]{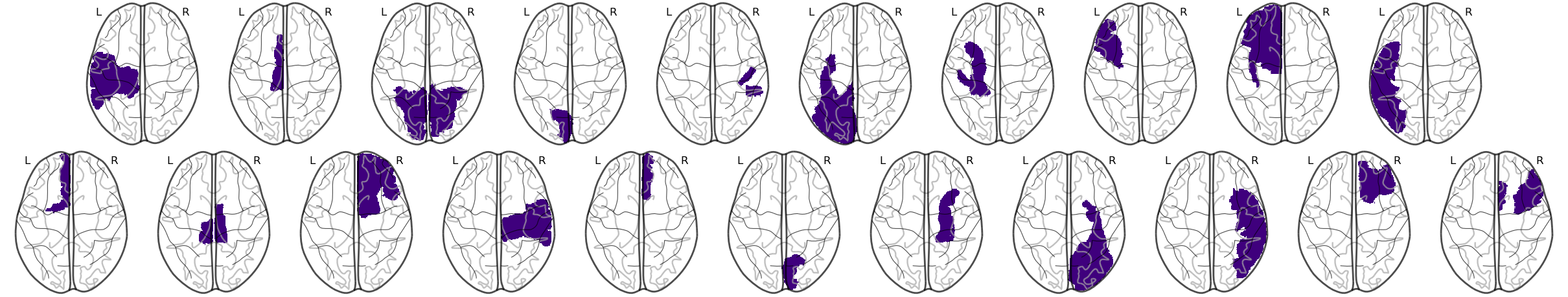}
\label{fig:app_brain_S1_A2}
}
\caption{Posterior distribution of the higher-level and lower-level clusters given by connection matrices $\Ab$, under sparsity hyperparameter $S = 1$ and number of nodes in each Bayesian network layer $p_0 = 4, p_1 = 21, p_2 = 68$.}
\label{fig:app_brain_S1}
\end{figure}

In Section \ref{ssec:appl_part_clus}, we have visualized the marginal posterior distributions of the $p_0$ higher-level clusters and the $p_1$ lower-level clusters for the optimal model choice $p_0 = 4, p_1 = 18$ under sparsity hyperparameter $S = 2$.
We now visualize for the optimal model choice $p_0 = 4, p_1 = 21$ under sparsity hyperparameter $S = 1$ in Figures \ref{fig:app_brain_S1_A1} and \ref{fig:app_brain_S1_A2}.
Observe that almost all lower-level and higher-level clusters are formed by brain regions that are geographically close together on the cerebral cortex even though we do not include spatial locations of nodes in our data analysis.

The learned clusters are highly interpretable.
The lower-level (shallower-layer) clusters in Figure \ref{fig:app_brain_S1_A2} form a partition of the cerebral cortex into 21 regions, which roughly correspond to the frontal lobes, parietal lobes, temporal lobes, occipital lobes of the left and right hemispheres of the brain, as well as some smaller regions.
At a coarser scale, the four higher-level (deeper-layer) clusters in Figure \ref{fig:app_brain_S1_A1} separate the brain into left- and right-hemisphere structures: two clusters correspond roughly to the left hemisphere and two to the right.
Within each hemisphere, one cluster spans most of the cortex, while the other is primarily concentrated in the temporal lobe.
Overall, these clusters align well with known anatomical organization of the human brain.

\subsection{Clustering Individuals}\label{ssec:appl_deferred_cluster}

\begin{figure}[!ht]
\centering
\subfigure[
$S = 1$ and $p_0 = 4, p_1 = 21, p_2 = 68$.
]{
\includegraphics[width = 0.45\textwidth]{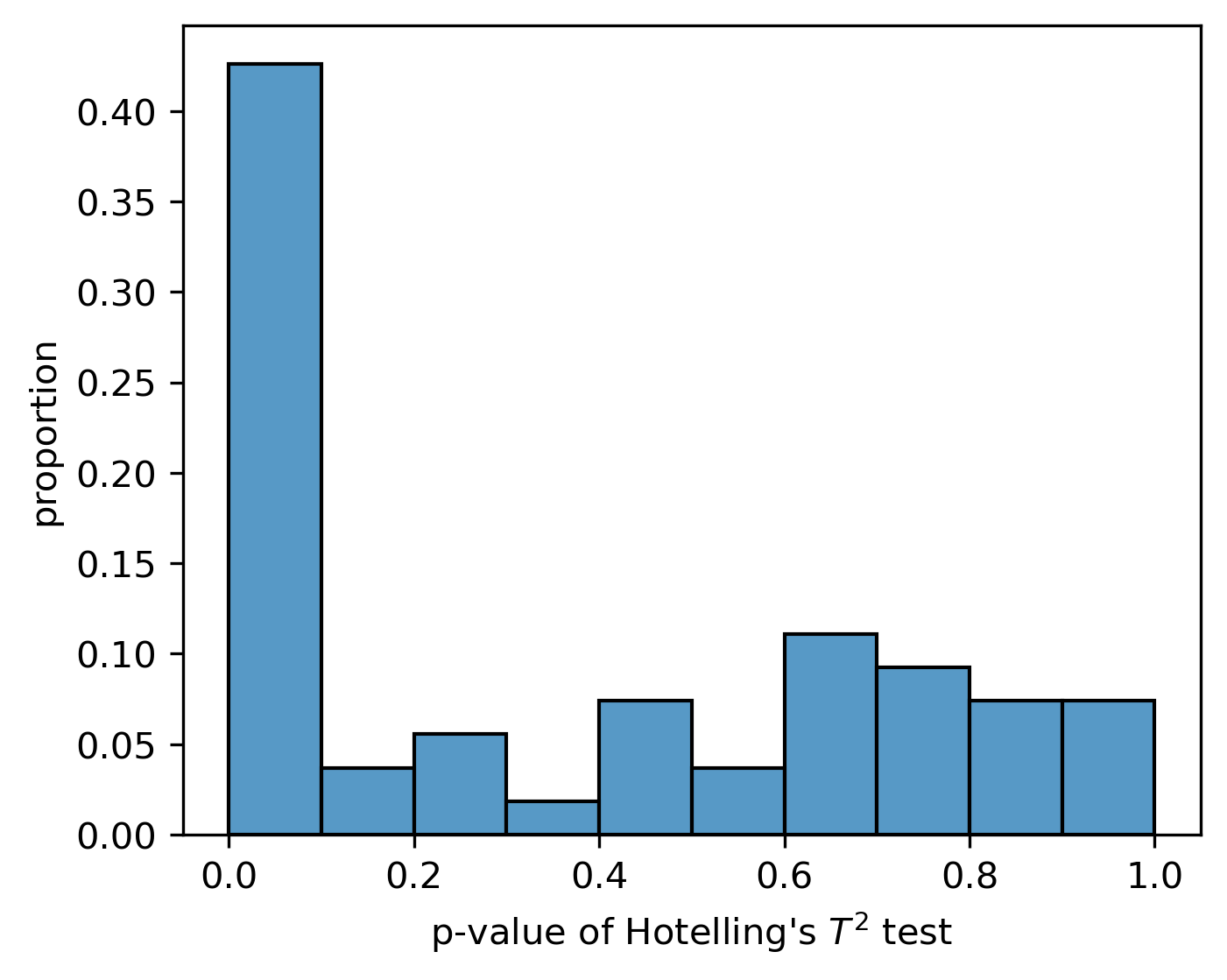}
\label{fig:hotelling_S1}
}
\hfill
\subfigure[
$S = 2$ and $p_0 = 4, p_1 = 18, p_2 = 68$.
]{
\includegraphics[width = 0.45\textwidth]{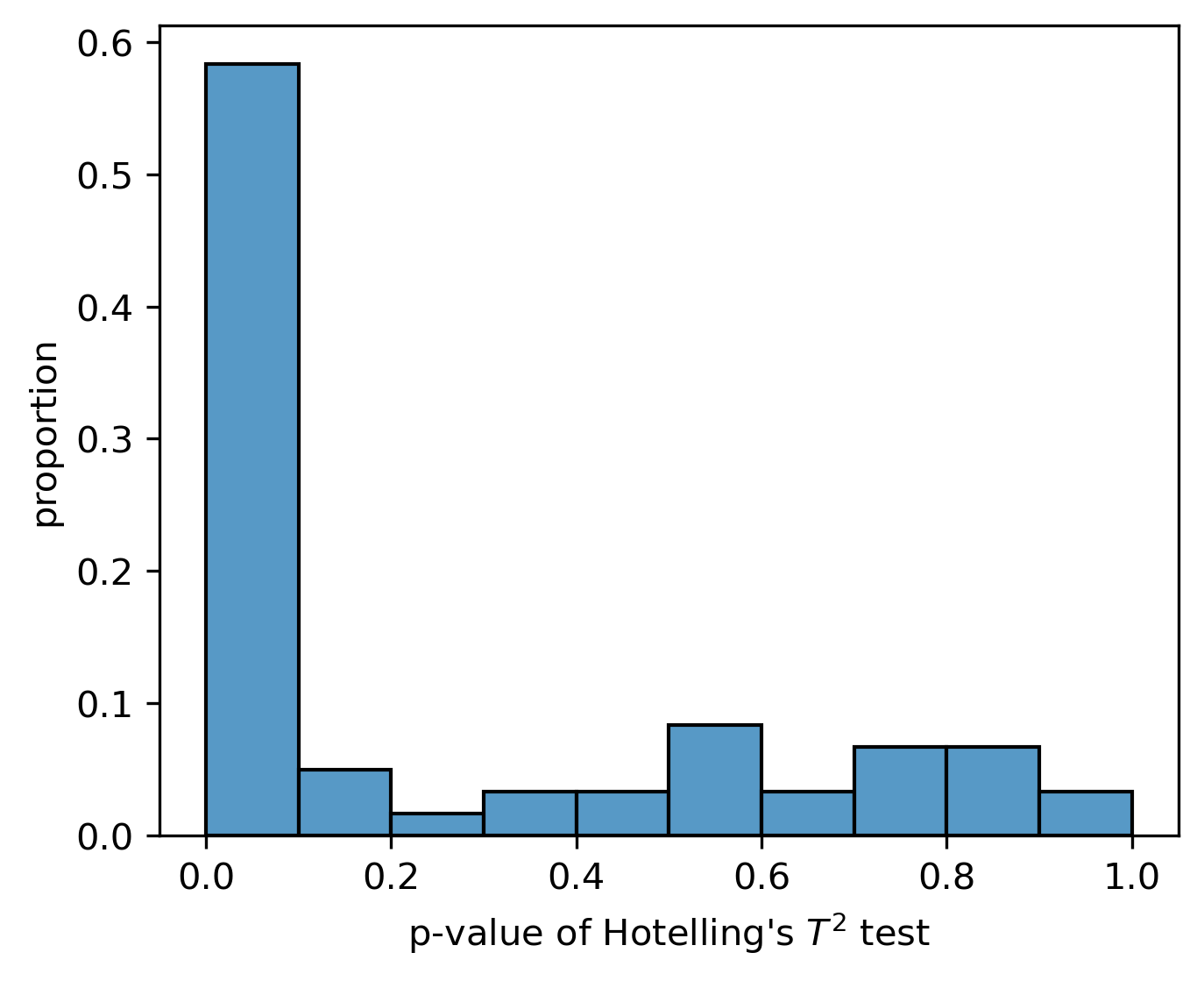}
\label{fig:hotelling_S2}
}
\caption{
Histograms of $p$-values in Hotelling's two-sample $T^2$ tests.
A significant $p$-value represents an active entry $(i, j)$ in $\Xb_1$ that generates an interpretable clustering of the population.
}
\label{fig:hotelling}
\end{figure}

In Section \ref{ssec:appl_cluster}, we are assessing how different clusters of individuals vary in their cognitive traits.
Hotelling's $T^2$ tests \citep{hotelling1931generalization} are performed for each active entry of $\Xb_1$ under the model with $S = 1, p_0 = 4, p_1 = 21$ and under the model with $S = 2, p_0 = 4, p_1 = 18$.
Histograms of the results $p$-values are plotted in Figure \ref{fig:hotelling}.
We observe that under both $S = 1$ and $S = 2$, around 50\% of the active entries in $\Xb_1$ have significant $p$-values ($< 0.1$), indicating that our unsupervised model manages to infer latent edges that are related to many cognitive traits.

\subsection{Relating Inferred Latent Features to Cognitive Traits}\label{ssec:appl_relate}

We next relate individual's cognitive traits to their brain connectivity.
We use 33 categories of cognitive traits, including language skills measured by oral reading recognition, dexterity measured by the game of 9-hole pegboard, strength measured by grip strength dynamometry, dependence on tobacco, etc.
Within each category, multiple ordinal or continuous measures are available, together forming 175 cognitive traits.

Instead of directly relating the very high-dimensional observed adjacency matrices $\Xb_2^{(1:N)}$ to the cognitive traits, we focus on the latent adjacency matrices in the shallower latent layer of our model, i.e. $\Xb_1^{(1:N)}$.
As explained in Section \ref{sec:model}, the nodes in layer 1 can be viewed as the communities for the observed nodes in layer 2 and the adjacency matrix $\Xb_1^{(n)}$ represents the adjacency between these communities for each individual $n$.
So intuitively, $\Xb_1^{(n)}$ should preserve the important adjacency information in $\Xb_2^{(n)}$.
We want to study if $\Xb_1^{(1:N)}$, which is learned in an unsupervised manner, is related to the cognitive traits.

MCMC samples of latent adjacency matrix $\Xb_1^{(n)}$ are not comparable when the corresponding samples of connection matrix $\Ab$ are different.
In the remaining part of this section, for interpretability, we therefore condition on $\Ab$ being equal to its posterior mode $\Ab^*$.
For brain connectome data, the posterior mean of $\Xb_1^{(n)}$ often has many entries $\approx$ 0 or 1 for all individuals.
We exclude these entries and focus on the {\em active} entries $(i, j)$ having $\frac{1}{N} \sum_{n = 1}^N \EE[X_{1, i, j}^{(n)} | \Xb_2^{(n)}, \Ab^*] \in (0.01, 0.99)$.
For our best models selected in Section \ref{ssec:model_sel} with $p_0 = 4, p_1 = 21, p_2 = 68$ for $S = 1$ and $p_0 = 4, p_1 = 18, p_2 = 68$ for $S = 2$, there remain around 100 active entries, forming the latent features of our model.
This is a considerable reduction from the 2278 entries in an observed adjacency matrix $\Xb_2^{(n)}$.

Since our spectral initialization is based on the Degree Corrected Mixed Membership model \citep[DCMM,][]{jin2023mixed} and its Mixed-SCORE algorithm, we compare our results with those obtained by applying DCMM to each individual network.
This comparison will allow us to see the improvement in our model over the approximation given by simpler models.
The DCMM is a state-of-the-art flexible method for single network community detection.
Descriptions of DCMM are provided in Section \ref{ssec:warm} and Section \ref{sec:more_app}.
From \eqref{eq:dcmm}, the entries in the latent probability matrix $\Zb^{(n)}$ for individual $n$ are viewed as the latent features of DCMM and can be estimated using the Mixed-SCORE algorithm.
The number of communities $q$ in DCMM roughly corresponds to the number of nodes $p_1$ in the shallower latent layer of our Bayesian network.
So for comparison with our model, we consider two DCMM models with $q = 18$ and $21$, corresponding to the selected $p_1$ values under $S = 2$ and $S = 1$, respectively.

\begin{table}[!htbp]
\centering
\small
\begin{tabular}{p{6.8cm}cc}
\toprule
Cognitive Group & BDGM PLS Corr & DCMM PLS Corr \\
\midrule
Delay Discounting & 0.1525 & 0.0896 \\
Education and Income & 0.2129 & 0.0861 \\
Emotion Recognition & 0.1731 & 0.0829 \\
Episodic Memory (Picture Sequence Memory) & 0.1759 & 0.1456 \\
Fluid Intelligence (Penn Progressive Matrices) & 0.2006 & 0.1057 \\
Health and Family History & 0.2493 & 0.1216 \\
Alcohol Use and Dependence & 0.1440 & 0.0859 \\
Audition (Words in Noise) & 0.2123 & 0.0860 \\
Contrast Sensitivity (Mars Contrast Sensitivity) & 0.1555 & 0.0638 \\
Dexterity (9-hole Pegboard) & 0.2056 & 0.1752 \\
Endurance (2 minute walk test) & 0.1659 & 0.1022 \\
Executive Function / Cognitive Flexibility & 0.1894 & 0.0713 \\
Executive Function / Inhibition (Flanker Task) & 0.1639 & 0.0666 \\
Five Factor Model (NEO-FFI) & 0.1845 & 0.0883 \\
Illicit Drug Use & 0.1788 & 0.0864 \\
Language / Reading Decoding (Oral Reading Recognition) & 0.1763 & 0.1023 \\
Locomotion (4-meter walk test) & 0.2376 & 0.0805 \\
Marijuana Use and Dependence & 0.2734 & 0.1148 \\
Negative Affect & 0.2402 & 0.0713 \\
Spatial Orientation (Line Orientation Test) & 0.2005 & 0.0951 \\
Strength (Grip Strength Dynamometry) & 0.2361 & 0.1057 \\
Stress and Self Efficacy & 0.1659 & 0.0858 \\
Sustained Attention (Continuous Performance Test) & 0.1431 & 0.1032 \\
Tobacco Use and Dependence & 0.3091 & 0.1693 \\
Verbal Episodic Memory (Penn Word Memory Test) & 0.1848 & 0.0902 \\
Language / Vocabulary Comprehension (Picture Vocabulary) & 0.1980 & 0.1132 \\
Processing Speed (Pattern Completion) & 0.1540 & 0.0740 \\
Psychological Well-being & 0.2389 & 0.0852 \\
Restricted Instrument: Psychiatric History & 0.2230 & 0.0837 \\
Social Relationships & 0.2166 & 0.2152 \\
Working Memory (List Sorting) & 0.2037 & 0.0804 \\
Pain Intensity (Odor) & 0.2297 & 0.0938 \\
Lift Function & 0.2079 & 0.0577 \\
\bottomrule
\end{tabular}
\caption{
Partial least squares canonical correlations (PLS Corr) between latent network features learned from BDGM or DCMM and cognitive trait groups.
}
\label{tab:pls_corr}
\end{table}

\begin{figure}[ht]
\centering
\includegraphics[width = 0.8\textwidth]{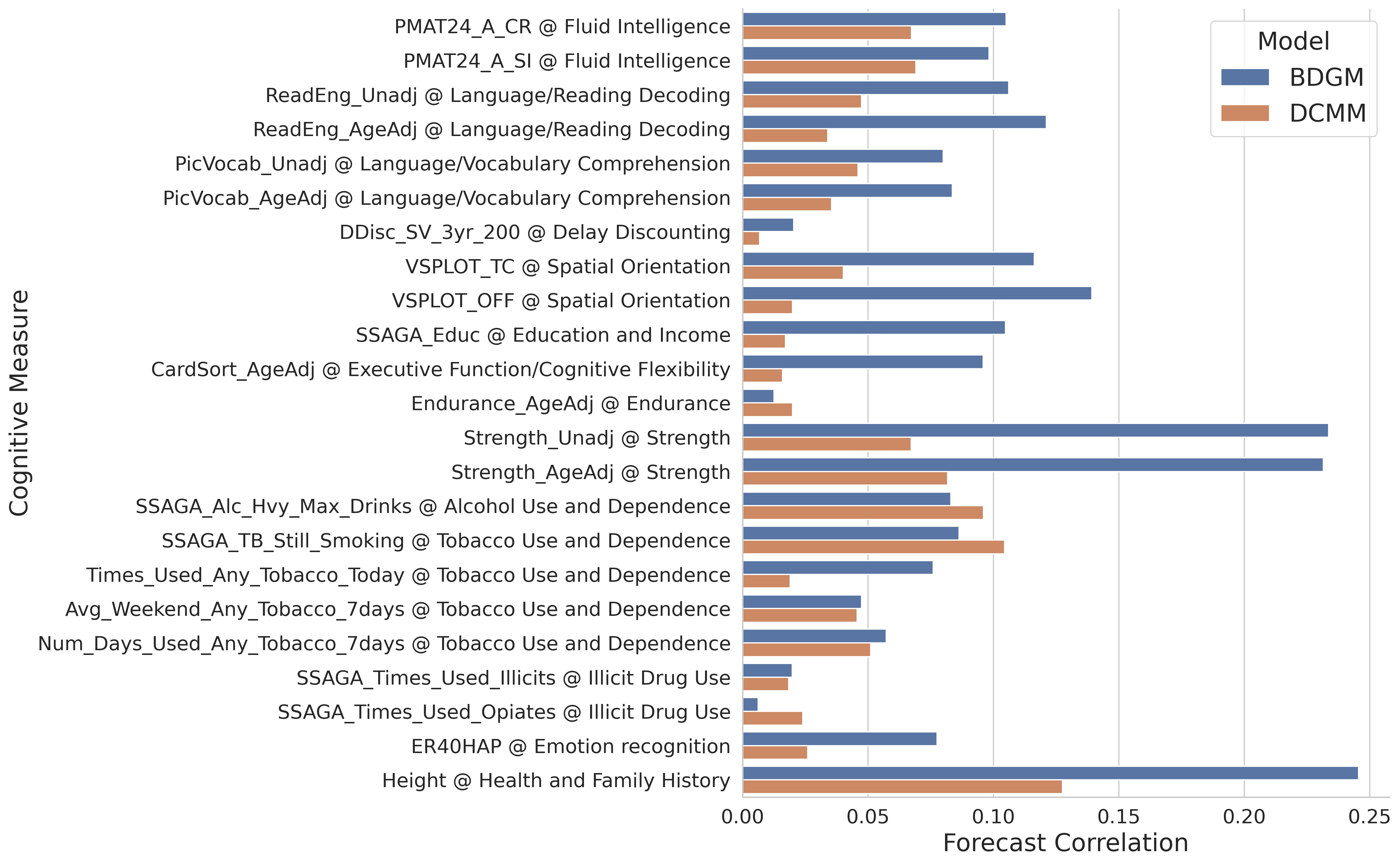}
\caption{
The out-of-sample $R^2$'s are presented for ridge regressions that predict different (Gaussian-transformed) dominant principal components of cognitive traits.
Four groups of predictors are considered, which are the latent features in $\Xb_1^{(n)}$ of our model with $S = 1$ and $S = 2$, the latent features in $\Zb^{(n)}$ of the DCMM with $q = 18$ and $q = 21$.
}
\label{fig:low_rep_cog_pred}
\end{figure}

We evaluate how the latent brain connectivity features learned from our model and from DCMM relate to the cognitive traits, noting that the latent features themselves are learned entirely from the brain networks without using any cognitive information.
We consider two complementary analyses.
First, to quantify global associations between the latent network representations and each category of cognitive traits, we compute partial least squares canonical correlations \citep{hastie2009elements} between the latent features and the grouped cognitive measures, with results summarized in Table~\ref{tab:pls_corr}.
Second, to assess the practical predictive value of the latent representations, we fit ridge regression models that use the latent features to forecast individual cognitive measures, tuning the regularization parameter by 10-fold cross validation and evaluating out-of-sample forecast correlations.
The results are shown in Figure~\ref{fig:low_rep_cog_pred}, where we focus on the subset of cognitive traits achieving at least 1\% predictive correlation under either BDGM or DCMM, a criterion that highlights traits containing recoverable signal without favoring either method.
Across both sets of analyses, the latent adjacency matrices $\Xb_1^{(n)}$ inferred by our model display markedly stronger associations with, and explanatory power for, the cognitive traits compared with the latent features obtained from DCMM.
Given that cognitive traits are well known to be subject to considerable measurement error and we are using an unsupervised approach, our model's predictive performance is very promising in terms of capturing important aspects of variation in human brain networks.

\subsection{Prediction of Missing Edges in Multiplex Networks}\label{ssec:mask}

\begin{figure}[ht]
\centering
\includegraphics[width = 0.6\textwidth]{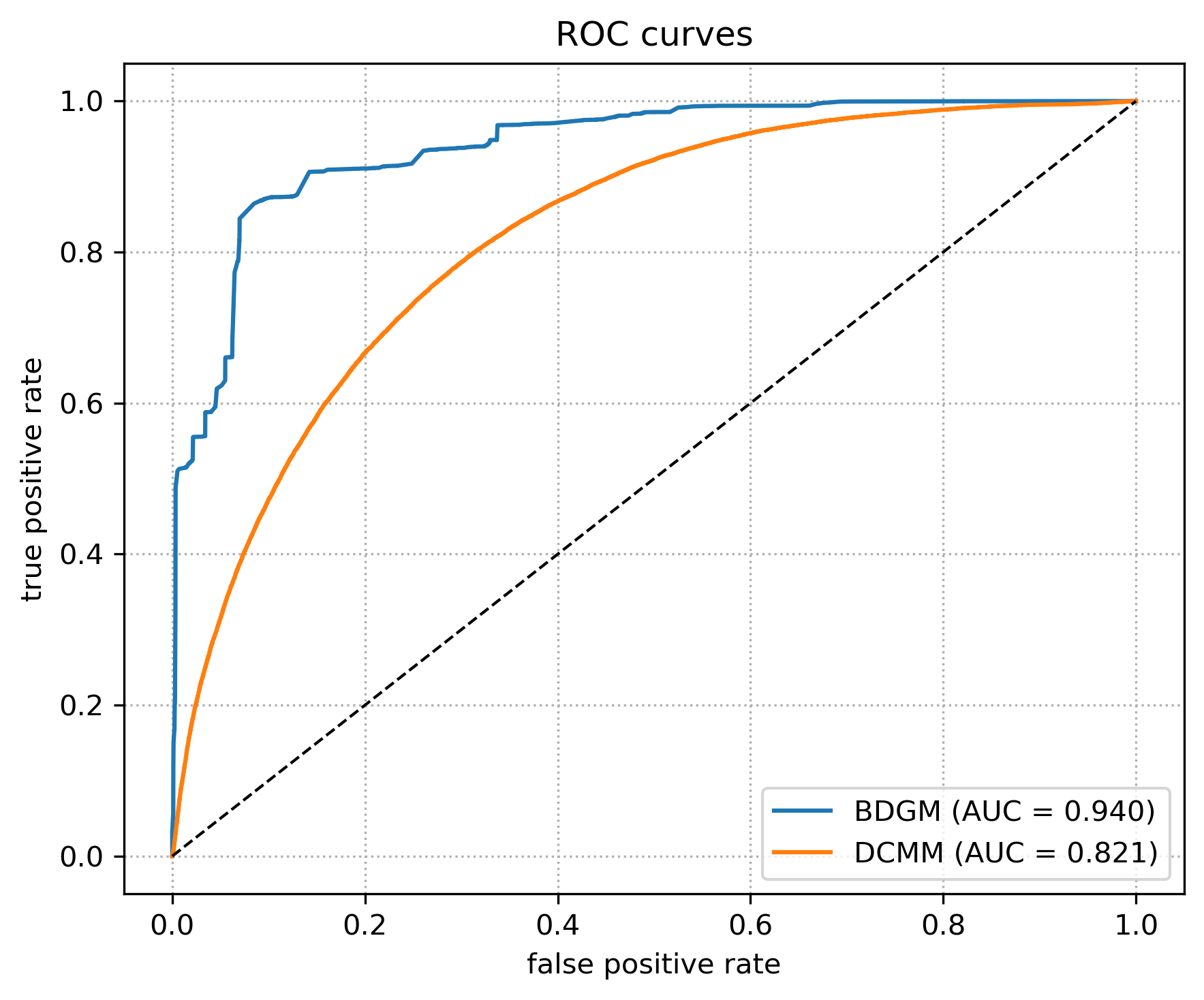}
\caption{
ROC curves for predicting masked edges in the brain connectivity networks of 1065 individuals.
Missing entries are created by masking a fixed $20 \times 20$ submatrix corresponding to 20 regions of the Desikan atlas.
}
\label{fig:mask}
\end{figure}

Although many existing methods are available for clustering individuals or detecting communities in observed adjacency matrices, most competing approaches are primarily algorithmic and either do not require model parameter estimation or rely on only partial parameter estimation.
In contrast, our model is a fully generative Bayesian model that performs joint inference on all latent variables and model parameters.
This feature enables applications that go beyond clustering, including prediction and uncertainty quantification.
To further illustrate the practical utility of full Bayesian inference, we include an additional application that evaluates the model’s ability to predict missing edges in multiplex networks.

To construct this example, we use the brain connectivity dataset consisting of 1065 individuals, each with a fully observed $68 \times 68$ adjacency matrix representing structural connectivity between regions defined by the Desikan atlas.
For demonstration purposes, we manually introduce missingness by randomly selecting 20 regions and masking the corresponding $20 \times 20$ submatrix in every individual’s adjacency matrix.
The same submatrix is masked across individuals, resulting in 8.3\% missing entries.
These artificially removed entries serve as ground truth for evaluating predictive performance.

We compare the predictive accuracy of the proposed model with that of the Degree Corrected Mixed Membership model (DCMM) \citep{jin2023mixed}.
Both models need to be adapted to accommodate missing data.
For our model, the standard Gibbs sampler is modified by excluding the missing entries from the likelihood when sampling parameters and latent variables.
For DCMM, its Mixed-SCORE algorithm does not directly handle missing values, so we embed it in an expectation maximization scheme that alternates between estimating DCMM parameters and inferring the conditional distributions of missing edges under the current parameter estimates.

To evaluate predictive performance, we vary the binary classification threshold for missing edges and compute the receiver operating characteristic (ROC) curves for both models.
As shown in Figure \ref{fig:mask}, the area under the curve for BDGM is 0.940, while that for DCMM is 0.821.
This suggests that BDGM achieves better recovery of missing edges under the structured missingness pattern used in this example.

This demonstrative result highlights two advantages of the fully generative Bayesian approach in our model.
First, by integrating information across individuals and inferring the latent hierarchical structure, the model gains predictive power and robustness in the presence of missingness or irregularities in the observed networks.
Second, although full model inference requires higher computational effort, it enables applications such as edge prediction, posterior uncertainty assessment, and principled handling of missing data, which are difficult to achieve using algorithmic clustering methods alone.
Together, these results complement the previous clustering example and provide additional evidence of the practical relevance of the model in real network applications.

\section{Implementation Details and Computational Scalability}

\subsection{Software Environment and Open Source Code}

All numerical experiments and simulations presented in this work were implemented in \textsf{Python} (version~3.11). 
To ensure reproducibility, numerical stability, and computational efficiency, we relied on widely used open-source scientific computing libraries. 
The core packages used in our implementation are listed below, with full details documented in the accompanying code repository.

\begin{itemize}
\item
\textsf{NumPy}~1.20.3 for vectorized numerical computation and array-based linear algebra;
\item
\textsf{SciPy}~1.6.2 for probability distributions, numerical routines, and random variate generation;
\item
\textsf{polyagamma}~1.3.2 for exact sampling of Pólya--Gamma latent variables used in Bayesian posterior inference;
\item
\textsf{Numba}~0.53.1 for just-in-time (\textsf{JIT}) compilation of performance-critical computational kernels;
\item
\textsf{scikit-learn}~1.3.1 for baseline models, comparative methods, and standardized preprocessing utilities;
\item
\textsf{nilearn}~0.10.2 and \textsf{nibabel}~4.0.2 for neuroimaging data manipulation and extraction when handling MRI data formats;
\item
\textsf{Pillow}~9.2.0, \textsf{matplotlib}~3.7.3, and \textsf{seaborn}~0.11.1 for visualization, diagnostic plots, and figure generation.
\end{itemize}

All experiments were executed on a standard multi-core CPU workstation, and no GPU acceleration was required. 
Reproducible code for our Gibbs sampler and all simulation and application results are available on Github at \url{https://github.com/BayesianModels01/Bayesian-Deep-Generative-Models-for-Multiplex-Networks-with-Multiscale-Overlapping-Clusters}.

\subsection{Computational Scalability for Large Multiplex Networks}

Analyzing large multiplex networks presents notable computational challenges due to the high dimensionality 
of latent variables, the combinatorial structure of adjacency patterns, and the repeated evaluation of likelihood contributions within each iteration of the Gibbs sampler. 
This subsection outlines both methodological strategies and implementation-level tools that substantially improve scalability.

\subsubsection{Methodological Tools}

\paragraph{Spectral Initialization.}
Random initial points often lead to slow burn-in for MCMC methods with high-dimensional latent structures. 
To address this, we initialize the latent connection matrices using a spectral initialization approach, as introduced in Section \ref{ssec:spec_init_alg}.
This provides high-quality starting values that guide the sampler toward high-probability regions of the posterior, thus reducing the total number of required iterations.

\paragraph{Subsampling Gibbs Sampler.}
In early stages of posterior computation, evaluating full conditional distributions over all network data can be prohibitively expensive for large networks. 
We therefore employ a subsampling Gibbs sampler during the initial iterations, in which only a fixed fraction (e.g., 1\%) of node pairs is randomly sampled and updated, as introduced in Section \ref{ssec:sub_gibbs}.
This produces rapid initial convergence at a fraction of the computational cost. 
After the sampler has stabilized, we switch to the full Gibbs sampler for accurate posterior inference. 
This hybrid approach provides substantial computational savings while preserving estimation accuracy.

\subsubsection{Implementation Tools}

\paragraph{Vectorization with NumPy.}
Whenever possible, numerical operations are expressed in vectorized form using \textsf{NumPy} arrays. 
This allows Python code to leverage highly optimized C and Fortran routines under the hood, dramatically accelerating dense matrix operations such as vector–matrix and matrix–matrix multiplications, log-likelihood evaluations, and probability computations.

\paragraph{Loop Acceleration with Numba.}
Certain components of the Gibbs sampler inherently require Python-level loops, such as iterating over nodes, communities, or conditional updates for non-vectorizable components. 
To accelerate these steps, we use \textsf{Numba}, a just-in-time compiler that translates annotated Python functions into optimized machine code using LLVM. 
Numba preserves readable Python syntax while enabling near-C execution speed for computational bottlenecks.

\paragraph{Parallelization via Multiprocessing.}
Although not used in the experiments reported here, the structure of our model admits parallelization through Python’s \textsf{multiprocessing} module. 
Node-level updates, independent Gibbs chains, or evaluations of WAIC across different model configurations can all be parallelized on multi-core machines. 
This provides a direct and practical path toward scaling the methodology to substantially larger networks or to extensive model-selection grids in future applications.

\vspace{0.1in}
\noindent
Together, spectral initialization, subsampling Gibbs sampling, vectorized numerical routines, Numba-accelerated loops, and potential multiprocessing support provide a cohesive framework 
for scaling the proposed Bayesian generative model to large multiplex networks while maintaining practical runtime performance.

%% file: tikz/tikz_figure_additive.tex
\begin{figure}[ht]
\centering
\begin{tikzpicture}[scale = 0.95]
\tikzmath{
for \x in {1, ..., 4}{
if {
(\x <= 2)
} then {
let \c = red!50;
} else {
let \c = black!10;
}; {
\fill[\c] (-0.5, 5 - \x) rectangle +(.95,.95);
}; };
for \y in {1, ..., 4}{
if {
(\y != 2)
} then {
let \c = red!50;
} else {
let \c = black!10;
}; {
\fill[\c] (\y, 5.5) rectangle + (0.95, 0.95);
}; };
for \x in {1, ..., 4}{ for \y in {1, ..., 4}{
if {
(\x <= 2) && (\y != 2)
} then {
let \c = red!50;
} else {
let \c = black!10;
}; {
\fill[\c] (\y, 5 - \x) rectangle + (0.95, 0.95);
}; }; };
for \x in {1, ..., 4}{ for \y in {1, ..., 4}{
if {
(\x == \y) || (\x + \y <= 3) || (\x + \y >= 7) || ((\x - \y) ^ 2 == 9)
} then {
let \c = blue!50;
} else {
let \c = black!10;
}; {
\fill[\c] (5 + \y, 5 - \x) rectangle + (0.95, 0.95);
}; }; };
for \x in {1, ..., 4}{ for \y in {1, ..., 4}{
if {
(\x + \y ^ 2 <= 3) || (\y - \x == 3)
} then {
let \c = green!50;
} else {
let \c = black!10;
}; {
\fill[\c] (10 + \y, 5 - \x) rectangle + (0.95, 0.95);
}; }; };
}
\node[anchor = west] at (-1.5, 3) {$\ab_{k, i}$};
\node[anchor = west] at (2.7, 7) {$\ab_{k, j}^\top$};
\node[anchor = west] at (2.6, 0.5) {$\ab_{k, i} \ab_{k, j}^\top$};
\node[anchor = west] at (5.3, 3) {$*$};
\node[anchor = west] at (5.3, 0.5) {$*$};
\node[anchor = west] at (7.7, 0.5) {$\Xb_{k - 1}^{(n)}$};
\node[anchor = west] at (10.3, 3) {$=$};
\node[anchor = west] at (10.3, 0.5) {$=$};
\node[anchor = west] at (11.5, 0.5) {$(\ab_{k, i} \ab_{k, j}^\top) * \Xb_{k - 1}^{(n)}$};
\node[anchor = west] at (-1.5, -0.5) {$\logit~ \PP(X_{k, i, j}^{(n)} = 1 | \Xb_{k - 1}^{(n)}, \Ab_k, \bTheta_k) = C_k + \left\langle (\ab_{k, i} \ab_{k, j}^\top) * \Xb_{k - 1}^{(n)}, \bGamma_k \right\rangle = C_k + \Gamma_{k, 1, 1} + \Gamma_{k, 2, 1} + \Gamma_{k, 1, 4}$.};
\end{tikzpicture}
\caption{
An illustration of the additive effects of $\Gamma_{k, i, j}$ in calculating the logit of the probability of nodes $\cN_{k, i}$ and $\cN_{k, j}$ being connected, where $\cN_{k, i}, \cN_{k, j}$ both belong to community $\cN_{k - 1, 1}$, and they separately belong to pairs of communities $(\cN_{k - 1, 2}, \cN_{k - 1, 1})$ and $(\cN_{k - 1, 1}, \cN_{k - 1, 4})$ adjacent for individual $n$.
}
\label{fig:additive_effects}
\end{figure}

%% file: arxiv.bbl
\begin{thebibliography}{}

\bibitem[Abbe, 2017]{abbe2017community}
Abbe, E. (2017).
\newblock Community detection and stochastic block models: Recent developments.
\newblock {\em Journal of Machine Learning Research}, 18(1):6446--6531.

\bibitem[Ahn et~al., 2012]{ahn2012bayesian}
Ahn, S., Korattikara, A., and Welling, M. (2012).
\newblock Bayesian posterior sampling via stochastic gradient {Fisher} scoring.
\newblock {\em Proceedings of the 29th International Conference on Machine Learning}.

\bibitem[Airoldi et~al., 2008]{airoldi2008mixed}
Airoldi, E.~M., Blei, D., Fienberg, S., and Xing, E. (2008).
\newblock Mixed membership stochastic blockmodels.
\newblock {\em Journal of Machine Learning Research}, 9(65):1981--2014.

\bibitem[Allman et~al., 2009]{allman2009identifiability}
Allman, E.~S., Matias, C., and Rhodes, J.~A. (2009).
\newblock Identifiability of parameters in latent structure models with many observed variables.
\newblock {\em Annals of Statistics}, 37(6A):3099--3132.

\bibitem[Ara{\'u}jo et~al., 2001]{araujo2001successive}
Ara{\'u}jo, M. C.~U., Saldanha, T. C.~B., Galvao, R. K.~H., Yoneyama, T., Chame, H.~C., and Visani, V. (2001).
\newblock The successive projections algorithm for variable selection in spectroscopic multicomponent analysis.
\newblock {\em Chemometrics and Intelligent Laboratory Systems}, 57(2):65--73.

\bibitem[Arroyo et~al., 2021]{arroyo2021inference}
Arroyo, J., Athreya, A., Cape, J., Chen, G., Priebe, C.~E., and Vogelstein, J.~T. (2021).
\newblock Inference for multiple heterogeneous networks with a common invariant subspace.
\newblock {\em Journal of Machine Learning Research}, 22(1):6303--6351.

\bibitem[Arroyo and Levina, 2020]{arroyo2020simultaneous}
Arroyo, J. and Levina, E. (2020).
\newblock Simultaneous prediction and community detection for networks with application to neuroimaging.
\newblock {\em arXiv preprint arXiv:2002.01645}.

\bibitem[Baker et~al., 2019]{baker2019control}
Baker, J., Fearnhead, P., Fox, E.~B., and Nemeth, C. (2019).
\newblock Control variates for stochastic gradient {MCMC}.
\newblock {\em Statistics and Computing}, 29(3):599--615.

\bibitem[Bardenet et~al., 2014]{bardenet2014towards}
Bardenet, R., Doucet, A., and Holmes, C. (2014).
\newblock Towards scaling up {Markov} chain {Monte} {Carlo}: An adaptive subsampling approach.
\newblock {\em Proceedings of the 31st International Conference on Machine Learning}.

\bibitem[Bardenet et~al., 2017]{bardenet2017markov}
Bardenet, R., Doucet, A., and Holmes, C.~C. (2017).
\newblock On {Markov} chain {Monte} {Carlo} methods for tall data.
\newblock {\em Journal of Machine Learning Research}, 18(47).

\bibitem[Bhattacharya and Dunson, 2011]{bhattacharya2011sparse}
Bhattacharya, A. and Dunson, D.~B. (2011).
\newblock Sparse {Bayesian} infinite factor models.
\newblock {\em Biometrika}, 98(2):291--306.

\bibitem[Binette, 2019]{binette2019note}
Binette, O. (2019).
\newblock A note on reverse pinsker inequalities.
\newblock {\em IEEE Transactions on Information Theory}, 65(7):4094--4096.

\bibitem[Chen et~al., 2022]{chen2022global}
Chen, S., Liu, S., and Ma, Z. (2022).
\newblock Global and individualized community detection in inhomogeneous multilayer networks.
\newblock {\em The Annals of Statistics}, 50(5):2664--2693.

\bibitem[Cheng et~al., 2019]{cheng2019decreased}
Cheng, W., Rolls, E.~T., Robbins, T.~W., Gong, W., Liu, Z., Lv, W., Du, J., Wen, H., Ma, L., Quinlan, E.~B., et~al. (2019).
\newblock Decreased brain connectivity in smoking contrasts with increased connectivity in drinking.
\newblock {\em Elife}, 8:e40765.

\bibitem[Cox et~al., 2013]{cox2013ideals}
Cox, D., Little, J., and OShea, D. (2013).
\newblock {\em Ideals, Varieties, and Algorithms: An Introduction to Computational Algebraic Geometry and Commutative Algebra}.
\newblock Springer Science \& Business Media.

\bibitem[Desikan et~al., 2006]{desikan2006automated}
Desikan, R.~S., S{\'e}gonne, F., Fischl, B., Quinn, B.~T., Dickerson, B.~C., Blacker, D., Buckner, R.~L., Dale, A.~M., Maguire, R.~P., Hyman, B.~T., et~al. (2006).
\newblock An automated labeling system for subdividing the human cerebral cortex on {MRI} scans into gyral based regions of interest.
\newblock {\em NeuroImage}, 31(3):968--980.

\bibitem[Diquigiovanni and Scarpa, 2019]{diquigiovanni2019analysis}
Diquigiovanni, J. and Scarpa, B. (2019).
\newblock Analysis of association football playing styles: An innovative method to cluster networks.
\newblock {\em Statistical modelling}, 19(1):28--54.

\bibitem[Doob, 1949]{doob1949application}
Doob, J.~L. (1949).
\newblock Application of the theory of martingales.
\newblock {\em Actes du Colloque International Le Calcul des Probabilites et ses Applications}.

\bibitem[Duan et~al., 2023]{duan2023bayesian}
Duan, L.~L., Michailidis, G., and Ding, M. (2023).
\newblock Bayesian spiked {Laplacian} graphs.
\newblock {\em Journal of Machine Learning Research}, 24(3):1--35.

\bibitem[Duan et~al., 2020]{duan2020bayesian}
Duan, L.~L., Young, A.~L., Nishimura, A., and Dunson, D.~B. (2020).
\newblock Bayesian constraint relaxation.
\newblock {\em Biometrika}, 107(1):191--204.

\bibitem[Durante and Dunson, 2016]{durante2016locally}
Durante, D. and Dunson, D.~B. (2016).
\newblock Locally adaptive dynamic networks.
\newblock {\em Annals of Applied Statistics}, 10(4):2203--2232.

\bibitem[Durante et~al., 2017]{durante2017nonparametric}
Durante, D., Dunson, D.~B., and Vogelstein, J.~T. (2017).
\newblock Nonparametric {Bayes} modeling of populations of networks.
\newblock {\em Journal of the American Statistical Association}, 112(520):1516--1530.

\bibitem[Durrett, 2019]{durrett2019probability}
Durrett, R. (2019).
\newblock {\em Probability: Theory and Examples}, volume~49.
\newblock Cambridge university press.

\bibitem[Elam et~al., 2021]{elam2021human}
Elam, J.~S., Glasser, M.~F., Harms, M.~P., Sotiropoulos, S.~N., Andersson, J.~L., Burgess, G.~C., Curtiss, S.~W., Oostenveld, R., Larson-Prior, L.~J., Schoffelen, J.-M., et~al. (2021).
\newblock The human connectome project: A retrospective.
\newblock {\em NeuroImage}, 244:118543.

\bibitem[Fan et~al., 2022]{fan2022alma}
Fan, X., Pensky, M., Yu, F., and Zhang, T. (2022).
\newblock Alma: Alternating minimization algorithm for clustering mixture multilayer network.
\newblock {\em Journal of machine learning research}, 23(330):1--46.

\bibitem[Federer, 2014]{federer2014geometric}
Federer, H. (2014).
\newblock {\em Geometric Measure Theory}.
\newblock Springer.

\bibitem[Frank and Strauss, 1986]{frank1986markov}
Frank, O. and Strauss, D. (1986).
\newblock Markov graphs.
\newblock {\em Journal of the American Statistical Association}, 81(395):832--842.

\bibitem[Gassiat et~al., 2020]{gassiat2020identifiability}
Gassiat, E., Le~Corff, S., and Leh{\'e}ricy, L. (2020).
\newblock Identifiability and consistent estimation of nonparametric translation hidden markov models with general state space.
\newblock {\em Journal of Machine Learning Research}, 21(115):1--40.

\bibitem[Gassiat et~al., 2022]{gassiat2022deconvolution}
Gassiat, E., Le~Corff, S., and Leh{\'e}ricy, L. (2022).
\newblock Deconvolution with unknown noise distribution is possible for multivariate signals.
\newblock {\em The Annals of Statistics}, 50(1):303--323.

\bibitem[Gelman, 2006]{gelman2006prior}
Gelman, A. (2006).
\newblock Prior distributions for variance parameters in hierarchical models.
\newblock {\em Bayesian Analysis}, 1(3):515--534.

\bibitem[Gelman et~al., 2014]{gelman2014understanding}
Gelman, A., Hwang, J., and Vehtari, A. (2014).
\newblock Understanding predictive information criteria for {Bayesian} models.
\newblock {\em Statistics and Computing}, 24(6):997--1016.

\bibitem[Gelman and Rubin, 1992]{gelman1992inference}
Gelman, A. and Rubin, D.~B. (1992).
\newblock Inference from iterative simulation using multiple sequences.
\newblock {\em Statistical Science}, 7(4):457--472.

\bibitem[Geweke, 1992]{geweke1992evaluating}
Geweke, J.~F. (1992).
\newblock Evaluating the accuracy of sampling-based approaches to the calculation of posterior moments.
\newblock {\em Bayesian Statistics}.

\bibitem[Ghosal and Van~der Vaart, 2017]{ghosal2017fundamentals}
Ghosal, S. and Van~der Vaart, A. (2017).
\newblock {\em Fundamentals of Nonparametric {Bayesian} Inference}.
\newblock Cambridge University Press.

\bibitem[Gu and Dunson, 2023]{gu2021bayesian}
Gu, Y. and Dunson, D.~B. (2023).
\newblock Bayesian pyramids: Identifiable multilayer discrete latent structure models for discrete data.
\newblock {\em Journal of the Royal Statistical Society: Series B}, 85(2):399--426.

\bibitem[Gu and Xu, 2020]{gu2020partial}
Gu, Y. and Xu, G. (2020).
\newblock Partial identifiability of restricted latent class models.
\newblock {\em Annals of Statistics}, 48(4):2082--2107.

\bibitem[Gyllenberg et~al., 1994]{gyllenberg1994non}
Gyllenberg, M., Koski, T., Reilink, E., and Verlaan, M. (1994).
\newblock Non-uniqueness in probabilistic numerical identification of bacteria.
\newblock {\em Journal of Applied Probability}, 31(2):542--548.

\bibitem[Han et~al., 2015]{han2015consistent}
Han, Q., Xu, K., and Airoldi, E. (2015).
\newblock Consistent estimation of dynamic and multi-layer block models.
\newblock {\em Proceedings of the 32nd International Conference on Machine Learning}.

\bibitem[Han and Dunson, 2018]{han2018multiresolution}
Han, S. and Dunson, D.~B. (2018).
\newblock Multiresolution tensor decomposition for multiple spatial passing networks.
\newblock {\em arXiv preprint arXiv:1803.01203}.

\bibitem[Handcock et~al., 2007]{handcock2007model}
Handcock, M.~S., Raftery, A.~E., and Tantrum, J.~M. (2007).
\newblock Model-based clustering for social networks.
\newblock {\em Journal of the Royal Statistical Society: Series A}, 170(2):301--354.

\bibitem[Hastie et~al., 2009]{hastie2009elements}
Hastie, T., Tibshirani, R., Friedman, J., et~al. (2009).
\newblock The elements of statistical learning.

\bibitem[Hinton et~al., 2006]{hinton2006fast}
Hinton, G.~E., Osindero, S., and Teh, Y.-W. (2006).
\newblock A fast learning algorithm for deep belief nets.
\newblock {\em Neural Computation}, 18(7):1527--1554.

\bibitem[Hoff et~al., 2002]{hoff2002latent}
Hoff, P.~D., Raftery, A.~E., and Handcock, M.~S. (2002).
\newblock Latent space approaches to social network analysis.
\newblock {\em Journal of the American Statistical Association}, 97(460):1090--1098.

\bibitem[Holland et~al., 1983]{holland1983stochastic}
Holland, P.~W., Laskey, K.~B., and Leinhardt, S. (1983).
\newblock Stochastic blockmodels: First steps.
\newblock {\em Social Networks}, 5(2):109--137.

\bibitem[Hotelling, 1931]{hotelling1931generalization}
Hotelling, H. (1931).
\newblock The generalization of student’s ratio.
\newblock {\em Annals of Mathematical Statistics}, 2(3):360--378.

\bibitem[Jiang et~al., 2022]{jiang2022associations}
Jiang, R., Westwater, M.~L., Noble, S., Rosenblatt, M., Dai, W., Qi, S., Sui, J., Calhoun, V.~D., and Scheinost, D. (2022).
\newblock Associations between grip strength, brain structure, and mental health in $>$ 40,000 participants from the {UK} biobank.
\newblock {\em BMC Medicine}, 20(1):1--14.

\bibitem[Jin, 2015]{jin2015fast}
Jin, J. (2015).
\newblock Fast community detection by score.
\newblock {\em Annals of Statistics}, 43(1):57--89.

\bibitem[Jin et~al., 2023]{jin2023mixed}
Jin, J., Ke, Z.~T., and Luo, S. (2023).
\newblock Mixed membership estimation for social networks.
\newblock {\em Journal of Econometrics}.

\bibitem[Jing et~al., 2021]{jing2021community}
Jing, B.-Y., Li, T., Lyu, Z., and Xia, D. (2021).
\newblock Community detection on mixture multilayer networks via regularized tensor decomposition.
\newblock {\em The Annals of Statistics}, 49(6):3181--3205.

\bibitem[Johndrow and Mattingly, 2017a]{johndrow2017coupling}
Johndrow, J.~E. and Mattingly, J.~C. (2017a).
\newblock Coupling and decoupling to bound an approximating {Markov} chain.
\newblock {\em arXiv preprint arXiv:1706.02040}.

\bibitem[Johndrow and Mattingly, 2017b]{johndrow2017error}
Johndrow, J.~E. and Mattingly, J.~C. (2017b).
\newblock Error bounds for approximations of {Markov} chains used in {Bayesian} sampling.
\newblock {\em arXiv preprint arXiv:1711.05382}.

\bibitem[Johndrow et~al., 2015]{johndrow2015optimal}
Johndrow, J.~E., Mattingly, J.~C., Mukherjee, S., and Dunson, D. (2015).
\newblock Optimal approximating {Markov} chains for {Bayesian} inference.
\newblock {\em arXiv preprint arXiv:1508.03387}.

\bibitem[Josephs et~al., 2023]{josephs2023nested}
Josephs, N., Amini, A.~A., Paez, M., and Lin, L. (2023).
\newblock Nested stochastic block model for simultaneously clustering networks and nodes.
\newblock {\em arXiv preprint arXiv:2307.09210}.

\bibitem[Karrer and Newman, 2011]{karrer2011stochastic}
Karrer, B. and Newman, M.~E. (2011).
\newblock Stochastic blockmodels and community structure in networks.
\newblock {\em Physical Review E}, 83(1):016107.

\bibitem[Kemp et~al., 2006]{kemp2006learning}
Kemp, C., Tenenbaum, J.~B., Griffiths, T.~L., Yamada, T., and Ueda, N. (2006).
\newblock Learning systems of concepts with an infinite relational model.
\newblock {\em Proceedings of the 21st national conference on Artificial intelligence}.

\bibitem[Kolda and Bader, 2009]{kolda2009tensor}
Kolda, T.~G. and Bader, B.~W. (2009).
\newblock Tensor decompositions and applications.
\newblock {\em SIAM review}, 51(3):455--500.

\bibitem[Korattikara et~al., 2014]{korattikara2014austerity}
Korattikara, A., Chen, Y., and Welling, M. (2014).
\newblock Austerity in {MCMC} land: Cutting the {Metropolis-Hastings} budget.
\newblock {\em Proceedings of the 31st International Conference on Machine Learning}.

\bibitem[Krantz and Parks, 2008]{krantz2008geometric}
Krantz, S.~G. and Parks, H.~R. (2008).
\newblock {\em Geometric Integration Theory}.
\newblock Springer Science \& Business Media.

\bibitem[Kruskal, 1977]{kruskal1977three}
Kruskal, J.~B. (1977).
\newblock Three-way arrays: Rank and uniqueness of trilinear decompositions, with application to arithmetic complexity and statistics.
\newblock {\em Linear Algebra and Its Applications}, 18(2):95--138.

\bibitem[Le et~al., 2018]{le2018estimating}
Le, C.~M., Levin, K., and Levina, E. (2018).
\newblock Estimating a network from multiple noisy realizations.
\newblock {\em Electronic Journal of Statistics}, 12(2):4697--4740.

\bibitem[Legramanti et~al., 2020]{legramanti2020bayesian}
Legramanti, S., Durante, D., and Dunson, D.~B. (2020).
\newblock Bayesian cumulative shrinkage for infinite factorizations.
\newblock {\em Biometrika}, 107(3):745--752.

\bibitem[Lei et~al., 2020]{lei2020consistent}
Lei, J., Chen, K., and Lynch, B. (2020).
\newblock Consistent community detection in multi-layer network data.
\newblock {\em Biometrika}, 107(1):61--73.

\bibitem[Lei and Lin, 2023]{lei2023bias}
Lei, J. and Lin, K.~Z. (2023).
\newblock Bias-adjusted spectral clustering in multi-layer stochastic block models.
\newblock {\em Journal of the American Statistical Association}, 118(544):2433--2445.

\bibitem[Lei et~al., 2024]{lei2024computational}
Lei, J., Zhang, A.~R., and Zhu, Z. (2024).
\newblock Computational and statistical thresholds in multi-layer stochastic block models.
\newblock {\em The Annals of Statistics}, 52(5):2431--2455.

\bibitem[Li et~al., 2022]{li2022hierarchical}
Li, T., Lei, L., Bhattacharyya, S., Van~den Berge, K., Sarkar, P., Bickel, P.~J., and Levina, E. (2022).
\newblock Hierarchical community detection by recursive partitioning.
\newblock {\em Journal of the American Statistical Association}, 117(538):951--968.

\bibitem[Ma et~al., 2015]{ma2015complete}
Ma, Y.-A., Chen, T., and Fox, E. (2015).
\newblock A complete recipe for stochastic gradient {MCMC}.
\newblock {\em Proceedings of the 28th International Conference on Neural Information Processing Systems}.

\bibitem[Ma and Nandy, 2023]{ma2023community}
Ma, Z. and Nandy, S. (2023).
\newblock Community detection with contextual multilayer networks.
\newblock {\em IEEE Transactions on Information Theory}, 69(5):3203--3239.

\bibitem[MacDonald et~al., 2022]{macdonald2022latent}
MacDonald, P.~W., Levina, E., and Zhu, J. (2022).
\newblock Latent space models for multiplex networks with shared structure.
\newblock {\em Biometrika}, 109(3):683--706.

\bibitem[Maire et~al., 2019]{maire2019informed}
Maire, F., Friel, N., and Alquier, P. (2019).
\newblock Informed sub-sampling {MCMC}: Approximate {Bayesian} inference for large datasets.
\newblock {\em Statistics and Computing}, 29(3):449--482.

\bibitem[Mantziou et~al., 2021]{mantziou2021bayesian}
Mantziou, A., Lunagomez, S., and Mitra, R. (2021).
\newblock Bayesian model-based clustering for multiple network data.
\newblock {\em Preprint at https://arxiv. org/abs/2107.03431}.

\bibitem[Mantziou et~al., 2024]{mantziou2024bayesian}
Mantziou, A., Lunag{\'o}mez, S., and Mitra, R. (2024).
\newblock Bayesian model-based clustering for populations of network data.
\newblock {\em The Annals of Applied Statistics}, 18(1):266--302.

\bibitem[Meil{\u{a}} and Shi, 2001]{meilua2001random}
Meil{\u{a}}, M. and Shi, J. (2001).
\newblock A random walks view of spectral segmentation.
\newblock {\em Proceedings of the 8th International Workshop on Artificial Intelligence and Statistics}.

\bibitem[Merkle et~al., 2019]{merkle2019bayesian}
Merkle, E.~C., Furr, D., and Rabe-Hesketh, S. (2019).
\newblock Bayesian comparison of latent variable models: Conditional versus marginal likelihoods.
\newblock {\em Psychometrika}, 84:802--829.

\bibitem[Mukherjee et~al., 2017]{mukherjee2017clustering}
Mukherjee, S.~S., Sarkar, P., and Lin, L. (2017).
\newblock On clustering network-valued data.
\newblock {\em Proceedings of the 31st Conference on Neural Information Processing Systems}.

\bibitem[Nascimento and Dias, 2005]{nascimento2005vertex}
Nascimento, J.~M. and Dias, J.~M. (2005).
\newblock Vertex component analysis: A fast algorithm to unmix hyperspectral data.
\newblock {\em IEEE Transactions on Geoscience and Remote Sensing}, 43(4):898--910.

\bibitem[Neal, 1992]{neal1992connectionist}
Neal, R.~M. (1992).
\newblock Connectionist learning of belief networks.
\newblock {\em Artificial Intelligence}, 56(1):71--113.

\bibitem[Negrea and Rosenthal, 2021]{negrea2021approximations}
Negrea, J. and Rosenthal, J.~S. (2021).
\newblock Approximations of geometrically ergodic reversible {Markov} chains.
\newblock {\em Advances in Applied Probability}, 53(4):981--1022.

\bibitem[Ng et~al., 2001]{ng2001spectral}
Ng, A., Jordan, M., and Weiss, Y. (2001).
\newblock On spectral clustering: Analysis and an algorithm.
\newblock {\em Proceedings of the 14th Conference on Neural Information Processing Systems}.

\bibitem[Nicolaescu, 2011]{nicolaescu2011coarea}
Nicolaescu, L.~I. (2011).
\newblock The coarea formula.
\newblock {\em https://www3.nd.edu/~lnicolae/Coarea.pdf}.

\bibitem[Nielsen and Witten, 2018]{nielsen2018multiple}
Nielsen, A.~M. and Witten, D. (2018).
\newblock The multiple random dot product graph model.
\newblock {\em arXiv preprint arXiv:1811.12172}.

\bibitem[Nowicki and Snijders, 2001]{nowicki2001estimation}
Nowicki, K. and Snijders, T. A.~B. (2001).
\newblock Estimation and prediction for stochastic blockstructures.
\newblock {\em Journal of the American Statistical Association}, 96(455):1077--1087.

\bibitem[Park and Friston, 2013]{park2013structural}
Park, H.-J. and Friston, K. (2013).
\newblock Structural and functional brain networks: From connections to cognition.
\newblock {\em Science}, 342(6158):1238411.

\bibitem[Paul and Chen, 2020a]{paul2020random}
Paul, S. and Chen, Y. (2020a).
\newblock A random effects stochastic block model for joint community detection in multiple networks with applications to neuroimaging.
\newblock {\em The Annals of Applied Statistics}, 14(2):993--1029.

\bibitem[Paul and Chen, 2020b]{paul2020spectral}
Paul, S. and Chen, Y. (2020b).
\newblock Spectral and matrix factorization methods for consistent community detection in multi-layer networks.
\newblock {\em The Annals of Statistics}, 48(1):230--250.

\bibitem[Peng and Carvalho, 2016]{peng2016bayesian}
Peng, L. and Carvalho, L. (2016).
\newblock Bayesian degree-corrected stochastic blockmodels for community detection.
\newblock {\em Electronic Journal of Statistics}, 10(2):2746--2779.

\bibitem[Pillai and Smith, 2014]{pillai2014ergodicity}
Pillai, N.~S. and Smith, A. (2014).
\newblock Ergodicity of approximate {MCMC} chains with applications to large data sets.
\newblock {\em arXiv preprint arXiv:1405.0182}.

\bibitem[Polson et~al., 2013]{polson2013bayesian}
Polson, N.~G., Scott, J.~G., and Windle, J. (2013).
\newblock Bayesian inference for logistic models using {P{\'o}lya--Gamma} latent variables.
\newblock {\em Journal of the American Statistical Association}, 108(504):1339--1349.

\bibitem[Quiroz et~al., 2018a]{quiroz2018speeding}
Quiroz, M., Kohn, R., Villani, M., and Tran, M.-N. (2018a).
\newblock Speeding up {MCMC} by efficient data subsampling.
\newblock {\em Journal of the American Statistical Association}, 114(526):831--843.

\bibitem[Quiroz et~al., 2018b]{quiroz2018speeding2}
Quiroz, M., Tran, M.-N., Villani, M., and Kohn, R. (2018b).
\newblock Speeding up {MCMC} by delayed acceptance and data subsampling.
\newblock {\em Journal of Computational and Graphical Statistics}, 27(1):12--22.

\bibitem[Quiroz et~al., 2018c]{quiroz2018subsampling}
Quiroz, M., Villani, M., Kohn, R., Tran, M.-N., and Dang, K.-D. (2018c).
\newblock Subsampling {MCMC} - an introduction for the survey statistician.
\newblock {\em Sankhya A}, 80:33--69.

\bibitem[Reyes and Rodriguez, 2016]{reyes2016stochastic}
Reyes, P. and Rodriguez, A. (2016).
\newblock Stochastic blockmodels for exchangeable collections of networks.
\newblock {\em arXiv preprint arXiv:1606.05277}.

\bibitem[Rhodes, 2010]{rhodes2010concise}
Rhodes, J.~A. (2010).
\newblock A concise proof of {Kruskal’s} theorem on tensor decomposition.
\newblock {\em Linear Algebra and Its Applications}, 432(7):1818--1824.

\bibitem[Rudolf and Schweizer, 2018]{rudolf2018perturbation}
Rudolf, D. and Schweizer, N. (2018).
\newblock Perturbation theory for {Markov} chains via {Wasserstein} distance.
\newblock {\em Bernoulli}, 24(4A):2610--2639.

\bibitem[Schwartz, 1965]{schwartz1965bayes}
Schwartz, L. (1965).
\newblock On {Bayes} procedures.
\newblock {\em Zeitschrift f{\"u}r Wahrscheinlichkeitstheorie und verwandte Gebiete}, 4(1):10--26.

\bibitem[Shaughnessy et~al., 2020]{shaughnessy2020narrative}
Shaughnessy, K.~A., Hackney, K.~J., Clark, B.~C., Kraemer, W.~J., Terbizan, D.~J., Bailey, R.~R., and McGrath, R. (2020).
\newblock A narrative review of handgrip strength and cognitive functioning: Bringing a new characteristic to muscle memory.
\newblock {\em Journal of Alzheimer's Disease}, 73(4):1265--1278.

\bibitem[Shi and Malik, 2000]{shi2000normalized}
Shi, J. and Malik, J. (2000).
\newblock Normalized cuts and image segmentation.
\newblock {\em IEEE Transactions on Pattern Analysis and Machine Intelligence}, 22(8):888--905.

\bibitem[Simon, 2014]{simon2014introduction}
Simon, L. (2014).
\newblock Introduction to geometric measure theory.
\newblock {\em Tsinghua Lectures}, 2(2):3--1.

\bibitem[Srivastava, 2008]{srivastava2008course}
Srivastava, S.~M. (2008).
\newblock {\em A Course on Borel Sets}, volume 180.
\newblock Springer Science \& Business Media.

\bibitem[Stanley et~al., 2016]{stanley2016clustering}
Stanley, N., Shai, S., Taylor, D., and Mucha, P.~J. (2016).
\newblock Clustering network layers with the strata multilayer stochastic block model.
\newblock {\em IEEE transactions on network science and engineering}, 3(2):95--105.

\bibitem[Stein and Shakarchi, 2010]{stein2010complex}
Stein, E.~M. and Shakarchi, R. (2010).
\newblock {\em Complex Analysis}.
\newblock Princeton University Press.

\bibitem[Stephens, 2000]{stephens2000dealing}
Stephens, M. (2000).
\newblock Dealing with label switching in mixture models.
\newblock {\em Journal of the Royal Statistical Society: Series B}, 62(4):795--809.

\bibitem[Tropp, 2015]{tropp2015introduction}
Tropp, J. (2015).
\newblock An introduction to matrix concentration inequalities.
\newblock {\em Foundations and Trends in Machine Learning}, 8(1-2):1--230.

\bibitem[van~der Pas and van~der Vaart, 2018]{van2018bayesian}
van~der Pas, S. and van~der Vaart, A. (2018).
\newblock Bayesian community detection.
\newblock {\em Bayesian Analysis}, 13(3):767--796.

\bibitem[Van~Essen et~al., 2012]{van2012human}
Van~Essen, D.~C., Ugurbil, K., Auerbach, E., Barch, D., Behrens, T.~E., Bucholz, R., Chang, A., Chen, L., Corbetta, M., Curtiss, S.~W., et~al. (2012).
\newblock The human connectome project: A data acquisition perspective.
\newblock {\em NeuroImage}, 62(4):2222--2231.

\bibitem[Wang et~al., 2019]{wang2019common}
Wang, L., Zhang, Z., and Dunson, D. (2019).
\newblock Common and individual structure of brain networks.
\newblock {\em Annals of Applied Statistics}, 13(1):85--112.

\bibitem[Watanabe, 2013]{watanabe2013widely}
Watanabe, S. (2013).
\newblock A widely applicable {Bayesian} information criterion.
\newblock {\em Journal of Machine Learning Research}, 14(1):867--897.

\bibitem[Watanabe and Opper, 2010]{watanabe2010asymptotic}
Watanabe, S. and Opper, M. (2010).
\newblock Asymptotic equivalence of {Bayes} cross validation and widely applicable information criterion in singular learning theory.
\newblock {\em Journal of Machine Learning Research}, 11(12):3571--3594.

\bibitem[Welling and Teh, 2011]{welling2011bayesian}
Welling, M. and Teh, Y.~W. (2011).
\newblock Bayesian learning via stochastic gradient {Langevin} dynamics.
\newblock {\em Proceedings of the 28th International Conference on Machine Learning}.

\bibitem[Yang and Leskovec, 2013]{yang2013overlapping}
Yang, J. and Leskovec, J. (2013).
\newblock Overlapping community detection at scale: A nonnegative matrix factorization approach.
\newblock In {\em Proceedings of the Sixth ACM International Conference on Web Search and Data Mining}, pages 587--596.

\bibitem[Young et~al., 2022]{young2022clustering}
Young, J.-G., Kirkley, A., and Newman, M.~E. (2022).
\newblock Clustering of heterogeneous populations of networks.
\newblock {\em Physical Review E}, 105(1):014312.

\bibitem[Young and Scheinerman, 2007]{young2007random}
Young, S.~J. and Scheinerman, E.~R. (2007).
\newblock Random dot product graph models for social networks.
\newblock {\em Algorithms and Models for the Web-Graph: 5th International Workshop}.

\bibitem[Zhou et~al., 2025]{zhou2025bayesian}
Zhou, Y., Gu, Y., and Dunson, D.~B. (2025).
\newblock Bayesian deep latent class regression.
\newblock {\em arXiv preprint arXiv:2503.17531}.

\end{thebibliography}
